\documentclass[a4paper,twocolumn,11pt,unpublished]{quantumarticle}
\pdfoutput=1
\usepackage[utf8]{inputenc}
\usepackage[english]{babel}
\usepackage[T1]{fontenc}
\usepackage{graphicx,amsmath,amssymb,amsthm,bm,booktabs,array,adjustbox}
\usepackage[square,numbers,sort&compress]{natbib}
\usepackage[colorlinks=true,linkcolor=blue,citecolor=blue,urlcolor=blue]{hyperref}
\makeatletter\let\switch@array\relax\makeatother
\graphicspath{{figures/}}
\newlength{\apscolumnwidth}
\newtheorem{theorem}{Theorem}
\newtheorem{proposition}{Proposition}
\newtheorem{corollary}{Corollary}
\newtheorem{lemma}{Lemma}
\newcommand{\im}{\operatorname{im}}
\newcommand{\rank}{\operatorname{rank}}

\newcommand{\Var}{\operatorname{Var}}
\newcommand{\ket}[1]{\lvert #1\rangle}
\newcommand{\bra}[1]{\langle #1\rvert}
\newcommand{\K}{\mathcal K}
\newcommand{\T}{\mathcal T}
\newcommand{\extra}{\mathcal E}

\begin{document}\title{Beyond Bond Gauge: Exact Tensor-Network Tangent Spaces at Weighted Graph States}
\author{Song Cheng}
\email{chengsong@bimsa.cn}
\affiliation{Beijing Institute of Mathematical Sciences and Applications, Beijing 101408, China}
\begin{abstract}
For normal PEPS, the fundamental theorem settles at the finite level
whether virtual-bond gauge exhausts a tensor network's representation
freedom. We ask the first-order question at a fixed representation,
from both sides of the differential: which tensor variations leave the
state unchanged without being bond gauges, and which physical
directions can no variation reach? We answer both exactly for weighted
graph state (WGS) tensor networks with equal physical and bond
dimension $q\ge2$, which include the graph-state resources of
measurement-based quantum computation. Dividing a single-tensor
variation by the nonvanishing amplitude turns it into an arbitrary
function on the closed neighborhood of its vertex. Overlapping
neighborhoods give the complete kernel. Its non-bond part is generated
on periodic regular-polygon tilings by triangles, chordless squares,
and edge-sharing diamonds. Under stated period conditions, generic
gauge completeness on all thirty-nine families studied identifies
this nullity as a WGS rank drop: for thirty-three families at every
$q\ge2$, and for six at $q=2$. Supports
outside every neighborhood give the missing directions: at qubit graph
states on the eleven Archimedean tilings every nearest-neighbor $XX$
and $YY$ direction is missing, including on honeycomb, where the
non-bond kernel is empty. On large honeycomb and square tori, small
nonzero coherent $XX$ errors leave the bond-two PEPS set on the same
graph. Removing
redundant parameters and supplying missing directions are therefore
distinct operations.
\end{abstract}
\maketitle

\section{Introduction}
Tensor networks describe many-body states by contracting local tensors
along a graph. Different tensors can represent the same state, because
invertible transformations on a contracted virtual leg cancel between
its two endpoints. Whether this bond gauge accounts for all
representation freedom is a basic structural problem for matrix product
states (MPS) and projected entangled-pair states (PEPS)
\cite{perezgarcia2007mps,cirac2021matrix}.

The fundamental theorem for normal PEPS extends the injective
case~\cite{perezgarcia2010symmetries} and characterizes finite equivalence:
under its injectivity and geometric hypotheses, two networks representing
the same state are related by bond gauges~\cite{molnar2018normal}.
Here we ask a first-order question at a specified representation: if
independent variations of the local tensors give zero wavefunction
derivative, must they be an infinitesimal bond gauge? For a graph with
$N$ vertices, physical and bond dimension $q$, site-tensor space
$\mathcal P_G$, bond action $\Phi_A$, and contraction differential
$J_A$ (Sec.~\ref{sec:bond}), the relevant maps are
\begin{equation}
\begin{gathered}
\bigoplus_e\mathfrak{gl}(q,\mathbb C)
\xrightarrow{\ \Phi_A\ }\mathcal P_G
\xrightarrow{\ J_A\ }(\mathbb C^q)^{\otimes N},\\
J_A\Phi_A=0,
\end{gathered}
\label{eq:complex}
\end{equation}
and $\ker J_A/\im\Phi_A$ measures the non-bond infinitesimal
redundancy (non-bond: not an infinitesimal bond gauge; see the Scope
paragraph below). Connecting a finite-equivalence theorem to this
differential quotient requires regularity of the parametrization as
well as the theorem's injectivity hypotheses.

The same differential answers a second question from the other side.
Its image, projected to rays, is the physical tangent space $\T_A$: the
first-order state changes that the representation can produce. The two
questions are logically independent. Removing redundant parameters
cannot restore a missing physical direction, and a redundant parameter
does not signal one. We therefore treat the kernel and the image of
$J_A$ together, as two halves of one exact computation.

Weighted graph states (WGS), generated by commuting two-site phase gates
\cite{verstraete2004valence,anders2006ground,anders2007weighted},
allow both halves to be computed exactly. At edge phase $\pi$ they
include graph-state resources for measurement-based quantum computation
(MBQC)~\cite{raussendorf2001oneway,raussendorf2003cluster,hein2006graph,vandennest2006universal}.
Nine Archimedean lattices support universal symmetry-protected cluster
phases~\cite{daniel2020archimedean}. The WGS PEPS
representations have bond dimension $D$ equal to the physical
dimension $d=2$
\cite{verstraete2004valence,gross2007measurement,gross2007novel}.
WGS points form a measure-zero subset of parameter space. On all
thirty-nine families studied here, generic tensors carry no non-bond
kernel under the period and dimension conditions of
Theorem~\ref{thm:generic}: every $q\ge2$ for thirty-three families
and $q=2$ for the other six.
For a tangent-space integrator initialized at a graph or cluster state,
the initial representation's exact zero modes and missing directions
enter the projected evolution. Nor is the effect confined to the point itself:
the zero modes open only quadratically along nearby paths, and not at
all along WGS phase evolution (Sec.~\ref{sec:projectiongap}).

We evaluate at nondegenerate WGS tensors with equal physical and bond
dimension $d=D=q$, while allowing arbitrary independent variations of
every tensor entry. Dividing a single-site variation by the nonzero WGS
amplitude identifies it with an arbitrary function on the closed
neighborhood $N[v]$. Consequently,
\begin{equation}
\im J_A=\psi_A\cdot\sum_{v\in V}\operatorname{Fun}(N[v]),
\label{eq:introimage}
\end{equation}
so the accessible first-order variations are exactly the diagonal
deformations $f\psi$ with $f(s)=\sum_v f_v(s_{N[v]})$.
In the language of correlator product states~\cite{changlani2009cps},
a WGS network is a
Jastrow state written with copy tensors and invertible edge
matrices~\cite{clark2018unifying}. Equation~\eqref{eq:introimage}
then says that an arbitrary PEPS tensor variation at such a point acts
as a correlator on one closed neighborhood $N[v]$: the PEPS tangent
space coincides with that of a correlator product state whose
correlators are these closed neighborhoods.
Figure~\ref{fig:concept} shows how both halves of the question follow.
A function on a set covered by two closed neighborhoods can be produced
by either tensor, which gives a kernel relation; a function on a set
covered by none cannot be produced at all, which gives a missing
direction.

\begin{figure*}[!tbp]
\centering
\includegraphics[width=\textwidth]{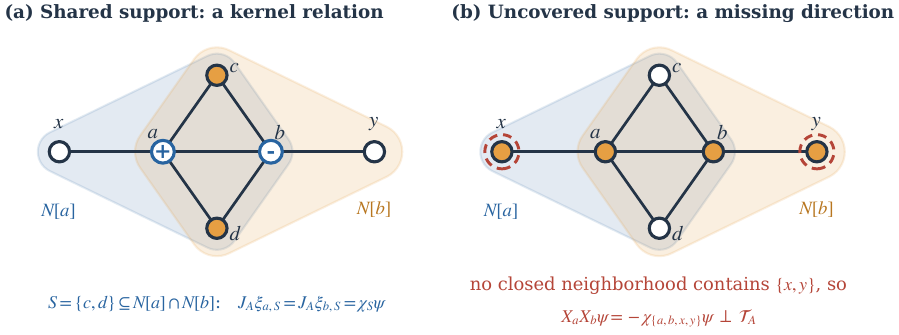}
\caption{Kernel and image from overlapping closed neighborhoods, on the
six-vertex graph $V=\{a,b,c,d,x,y\}$,
$E_G=\{ab,ac,ad,bc,bd,ax,by\}$ of
Corollary~\ref{cor:observableresponse} at the
qubit graph-state point (all edge phases $\pi$).
Shaded regions are $N[a]$ (blue) and $N[b]$ (orange). (a) The support
$S=\{c,d\}$ lies in both, so the site variations $\xi_{a,S}$ and
$\xi_{b,S}$ produce the same physical vector $\chi_S\psi$; their
difference is a non-bond kernel direction, the diamond motif of
Fig.~\ref{fig:motifs}. (b) The leaves $x,y$ lie in no common closed
neighborhood, so no tensor variation produces a character whose
support contains both. The edge error $X_aX_b\psi$ is such a character
and is missing from $\T_A$.}
\label{fig:concept}
\end{figure*}

Our main results are as follows.
\begin{enumerate}
\item[(i)] On any finite graph, the local isomorphism gives the exact
image, rank, and complete kernel of the differential, and identifies
ordinary bond gauge sector by sector, leaving an explicit basis of the
non-bond quotient (Theorem~\ref{thm:star},
Propositions~\ref{prop:bond} and~\ref{prop:quotient}).
\item[(ii)] A minimal-region polynomial classifies the local generators
on any graph (Theorem~\ref{thm:minimalgeneral}). On periodic
regular-polygon tilings exactly three motifs occur, triangles,
chordless squares, and shared-edge diamonds, with a count that depends
on the cyclic order of faces (Theorem~\ref{thm:regular}); a finite-cell
formula covers general periodic graphs, tied or untied
(Theorem~\ref{thm:cell}).
\item[(iii)] An exact criterion decides injective regions. The
non-bond kernel occurs at non-normal (square, triangular, kagome) and
normal (star) representations alike, where normality here concerns
infinite lattices and partitions into finite blocks
(Corollary~\ref{cor:normal}). Relations with nonadjacent centers
integrate to finite equivalences, while parallel pairs on adjacent
squares are obstructed at second order
(Propositions~\ref{prop:finite} and~\ref{prop:obstruction}).
\item[(iv)] On all thirty-nine periodic families
studied, generic tensors are gauge complete under the stated period
conditions: thirty-three families for every $q\ge2$, and six at
$q=2$. In this scope the WGS non-bond nullity is a rank drop of the
parametrization, not a dimension defect of the variety
(Theorem~\ref{thm:generic}).
\item[(v)] The exact image gives the tangent projector and TDVP
residual. Transverse-field Ising dynamics is tangent at every
nondegenerate WGS point, whereas at graph states on the eleven
Archimedean tilings every nearest-neighbor $XX$ and $YY$ direction is
missing, including on honeycomb, where the non-bond kernel is empty
(Corollaries~\ref{cor:tdvp} and~\ref{cor:support}). On honeycomb and
square tori, Schmidt bounds exclude all bond-two representations of
the perturbed state (Corollaries~\ref{cor:honeycombstate}
and~\ref{cor:squarestate}).
\item[(vi)] Kernel removal regularizes the reduced linear system
without changing any velocity, whereas a state-preserving bond
expansion supplies a missing direction. Off the WGS point the zero
modes open quadratically under an explicit condition on the path and
persist along WGS phase evolution (Proposition~\ref{prop:nearbygram}).
\end{enumerate}

These results complement the dimension theory of tensor-network
varieties
\cite{bernardi2023dimension,bernardi2026triangular,christandl2020geometry},
which our generic comparison extends to the subcritical regime $d=D$,
and work on gauges and tangent-space geometry
\cite{haegeman2014geometry,vanderstraeten2015peps,evenbly2018gauge}. Known singularities of
matrix-product and tensor-train parametrizations arise from rank drops
on a bond~\cite{holtz2012manifolds,haegeman2014geometry}, and Riemannian
fundamental theorems hold on open dense subsets for the specific
network families studied in Ref.~\cite{paezvelasco2026manifolds},
including MPS and isometric PEPS; WGS points instead have full bond
rank and are special points of general PEPS, whose state sets need not
be closed~\cite{barthel2022closedness}. Semi-injective, $G$-, and
MPO-injective PEPS treat finite equivalences and virtual
symmetries~\cite{molnar2018semiinjective,schuch2010pepsdegeneracy,sahinoglu2021mpo,bultinck2017anyons}.
The non-bond count quantifies the zero modes left by the gauge
projection of Wu and Nys~\cite{wu2025peps}, related metric zero modes
aid bond truncation in loopy networks~\cite{sokolov2025zeromode,dziarmaga2026loopy}, and
bond expansion links the exact image to subspace-expansion methods
\cite{hubig2015subspace,yang2020tdvp,li2024cbe,vanderstraeten2019tangent}.

\paragraph*{Scope.}
All statements concern the differential and the fiber of the
contraction map at a specified representation. ``Non-bond'' abbreviates
``not an infinitesimal bond gauge''; such directions are still
redundancies of the parametrization, and some integrate to finite
state-preserving transformations. Although ``gauge'' conventionally
means bond gauge~\cite{molnar2018normal}, we avoid the shorter
``non-gauge'' because Eq.~\eqref{eq:complex} excludes only the bond
action, not redundancy in general. Tangent and residual statements are pointwise:
a finite step of a tangent flow, such as an Ising step, can leave the
WGS family. Lattice statements concern finite-support algebraic
relations on locally faithful quotients; smaller quotients are covered
by the finite-graph theorem, and no infinite-volume Hilbert norm is
involved. For MBQC, the tangent and state-level bounds constrain the
tracking of coherent preparation errors at bond dimension two; they do
not describe closure under general measurements. Finite numerical scans
illustrate exact statements and their dependence on the chosen lift;
they do not measure solver performance.

Sections~\ref{sec:bond}--\ref{sec:relations} establish the gauge
baseline, differential theorem, and local generators.
Sections~\ref{sec:lattice}--\ref{sec:periodic} give the lattice and
periodic-cell counts; Sections~\ref{sec:fiber}--\ref{sec:generic} analyze
the fiber and generic rank. Section~\ref{sec:tdvp} develops the
physical projection, bond expansion, and reduced linear system.
Table~\ref{tab:notation} collects the notation. The appendices
contain the extended graph census,
generic-completeness proofs, finite-equivalence checks, further
physical examples, and reproducibility details.

\begin{figure*}[!tbp]
\centering
\includegraphics[width=\textwidth]{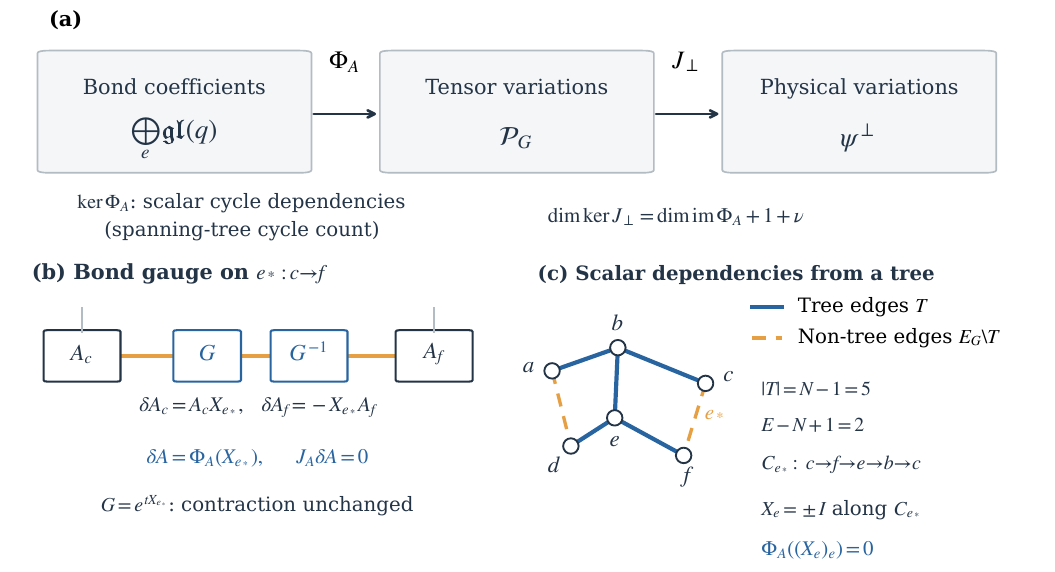}
\caption{Bond gauge and its generator dependencies at $q=2$.
(a) The bond action $\Phi_A$ maps generator coefficients to tensor
variations, and $J_\perp$ maps those to projective physical variations.
(b) Inserting $G=e^{tX_{e_*}}$ and $G^{-1}$ on the oriented edge
$e_*:c\to f$ of panel (c) changes the tensors,
$\delta A_c=A_cX_{e_*}$, $\delta A_f=-X_{e_*}A_f$, but not the state:
$J_A\Phi_A(X_{e_*})=0$.
(c) A spanning tree $T$ (blue) and its complement (orange, dashed).
Each non-tree edge closes a fundamental cycle; scalar generators
$X_e=\pm I$ along that cycle cancel at every vertex, giving the
$E-N+1=2$ scalar dependencies of Eq.~\eqref{eq:cycle}. Non-bond
kernel relations are illustrated in Fig.~\ref{fig:motifs}.}
\label{fig:maps}
\end{figure*}

\section{Tensor representations and the bond-gauge baseline}
\label{sec:bond}
\subsection{Parameter and physical maps}
Let $G=(V,E_G)$ be a finite simple connected graph, with $N=|V|$ vertices
and $E=|E_G|$ edges. Fix equal bond and physical dimensions
$D=d=q\ge2$.
At the WGS points introduced below, whose edge factors are invertible,
this equality makes two maps isomorphisms: the map from a single site
tensor to functions on its closed neighborhood (Sec.~\ref{sec:star})
and the map from a bond generator to functions on its edge
(Sec.~\ref{sec:relations}). The exact image and kernel classification
rests on both.

Every vertex carries one physical index
$s_v\in\{0,\ldots,q-1\}$ and one virtual index of dimension $q$
for each incident edge. The tensors
are independent, without translation tying or a symmetry mask. Their
parameter space is
\begin{equation}
\begin{gathered}
\mathcal P_G=\bigoplus_{v\in V}\mathbb C^{q^{\deg(v)+1}},\\
P=\dim\mathcal P_G=\sum_vq^{\deg(v)+1}.
\end{gathered}
\label{eq:parameters}
\end{equation}
Contraction defines a polynomial map $\Psi:\mathcal P_G\to(\mathbb C^q)^{\otimes N}$,
as in the algebraic description of tensor-network varieties
\cite{bernardi2023dimension}.
At a point with $\Psi(A)\ne0$, write $\psi=\Psi(A)/\|\Psi(A)\|$ and define
\begin{equation}
J_A=d\Psi_A,\qquad
J_{\perp,A}=(I-\ket\psi\bra\psi)J_A/\|\Psi(A)\|.
\label{eq:j}
\end{equation}
The latter is the horizontal representative of the differential to
projective Hilbert space~\cite{provost1980riemannian,hackl2020geometry}.
For a fixed base representation we write $A_0$ and
$\psi_0=\Psi(A_0)/\|\Psi(A_0)\|$.
All ranks and kernel dimensions below are
complex dimensions. A complex kernel direction represents two real
directions if the tensors are split into real and imaginary parts.

Orient the edges arbitrarily. A matrix $X_e\in\mathfrak{gl}(q,\mathbb C)$
acts on its tail tensor from the right and on its head tensor with the
negative left action, with the corresponding virtual leg held fixed.
This is the infinitesimal form of the standard virtual-bond gauge
\cite{perezgarcia2007mps,bernardi2023dimension}. Summing these actions defines
\begin{equation}
\Phi_A:\bigoplus_{e\in E_G}\mathfrak{gl}(q,\mathbb C)\longrightarrow\mathcal P_G.
\label{eq:phi}
\end{equation}
The two endpoint contributions cancel under contraction, as shown in
Fig.~\ref{fig:maps}(b), so $J_A\Phi_A=0$.
In particular $\im\Phi_A\subseteq\ker J_A\subseteq\ker J_{\perp,A}$.
The second inclusion has codimension one, supplied by the ray
(Sec.~\ref{sec:nullity}); whether the first is an equality is the
question of this paper.

In this language, $\ker\Phi_A$ collects coefficient dependencies among bond
generators, $\im\Phi_A$ the ordinary gauge variations, $\ker J_A$ and
$\ker J_{\perp,A}$ the variations invisible to the wavefunction and to
the projective state, and $\ker J_A/\im\Phi_A$ the redundancy beyond
ordinary gauge.

Let $\partial$ be the oriented vertex-edge incidence matrix, with $+1$ at a tail
and $-1$ at a head. Scalar generators $X_e=a_eI$ give
$\delta A_v=(\partial a)_v A_v$. Therefore
\begin{equation}
\begin{gathered}
\mathcal R_{\rm cyc}:=\{(a_eI)_e:\partial a=0\}\subseteq\ker\Phi_A,\\
\dim\mathcal R_{\rm cyc}=E-N+1.
\end{gathered}
\label{eq:cycle}
\end{equation}
The dimension and fundamental-cycle basis are standard incidence-matrix
facts~\cite{diestel2017graph}; their role as scalar PEPS gauge-generator
dependencies is explicit in Ref.~\cite{wu2025peps}.
This scalar cycle space exists at every tensor point. A spanning tree,
illustrated in Fig.~\ref{fig:maps}(c),
selects $N-1$ independent scalar directions; keeping these and the $q^2-1$
traceless directions on every edge preserves the full bond-gauge image.
The remaining issue is whether further coefficient dependencies survive.

\subsection{WGS representations and the complete bond-gauge rank}
Weighted graph states (WGS) are an exactly solvable family at which
the bond-gauge rank saturates the upper bound of Eq.~\eqref{eq:cycle};
genericity then transports this to almost every tensor assignment.
Following the controlled-phase construction of WGS
\cite{anders2007weighted}, the normalized amplitude is
\begin{equation}
\begin{gathered}
\psi(s)=q^{-N/2}\prod_{uv\in E_G}W_{uv}(s_u,s_v),\\
W_{uv}(s,t)=e^{i\phi_{uv}st},
\end{gathered}
\label{eq:wgs}
\end{equation}
where $\phi_{uv}\in\mathbb R$ and the Vandermonde nodes
$1,e^{i\phi_{uv}},\ldots,e^{i(q-1)\phi_{uv}}$ are distinct on every edge.
Equivalently, $e^{ik\phi_{uv}}\ne1$ for $1\le k<q$. Each $q\times q$
weight matrix is then invertible and has no zero entries; invertibility
follows from the Vandermonde determinant~\cite{golub2013matrix}. In particular,
$\phi_{uv}=2\pi/q$ is admissible for every integer $q\ge2$; no prime-dimension
assumption is needed. For $\phi_{uv}=2\pi w/q$, admissibility requires
$\gcd(w,q)=1$. The results also hold for arbitrary invertible $q\times q$
phase-weight matrices with unit-modulus entries.
Choose an invertible factorization
\begin{equation}
W_{uv}(s,t)=\sum_{a=0}^{q-1} F_{u,e}(s,a)F_{v,e}(t,a).
\label{eq:factor}
\end{equation}
The tensors
\begin{equation}
A_v^{s_v}(\{a_e\})=q^{-1/2}\prod_{e\ni v}F_{v,e}(s_v,a_e)
\label{eq:site}
\end{equation}
contract to Eq.~\eqref{eq:wgs}. Factorization choices are related by
invertible virtual gauges and do not affect the physical tangent image.

\begin{proposition}[Complete coefficient dependencies]
At the WGS representation~\eqref{eq:site},
\begin{equation}
\begin{aligned}
\ker\Phi_A&=\mathcal R_{\rm cyc},\\
\rank\Phi_A&=q^2E-(E-N+1)\\
&=E(q^2-1)+N-1.
\end{aligned}
\label{eq:phirank}
\end{equation}
The same conclusion holds on a Zariski-open dense subset of the full
parameter space $\mathcal P_G$ for each finite simple connected graph.
\label{prop:bond}
\end{proposition}
\begin{proof}
Fix a physical slice of a site tensor. It is a nonzero product of virtual
vectors. If a sum of single-leg actions on this product vanishes, then
each acted-on leg must remain parallel to its original vector. This
follows by projecting that leg modulo its original span and contracting
the other legs with nonvanishing linear functionals. For an oriented edge,
write its factors as rows $f_s$ and columns $g_t$. A coefficient null vector
must therefore satisfy
\begin{equation}
f_sX_e=\alpha_s f_s,\qquad X_eg_t=b_tg_t.
\end{equation}
Pairing these equations gives $(\alpha_s-b_t)W_e(s,t)=0$. All $q^2$ weights
are nonzero, so every $\alpha_s$ and $b_t$ equals one scalar $a_e$.
Invertibility of the factors then forces $X_e=a_eI$. The remaining site
equations are exactly $\partial a=0$.

The entries of $\Phi_A$ are linear polynomials in the tensor entries.
Equation~\eqref{eq:cycle} bounds its rank from above, and the WGS construction
attains that bound on every graph. An attaining minor is thus a nonzero
polynomial. Its nonvanishing defines the required open dense set, by the
standard maximal-minor criterion for generic rank
\cite{hartshorne1977algebraic,bernardi2023dimension}.
\end{proof}

Proposition~\ref{prop:bond} fixes the bond-gauge rank even at WGS points
with a non-bond physical kernel. It concerns compatible actions at
both endpoints of every edge; individual WGS tensors may still have
local multi-leg stabilizers. Appendix~\ref{SM-sec:stabilizerdetails}
gives their explicit construction and explains the role of nonzero
edge weights.

\subsection{The ray and the non-bond nullity}
\label{sec:nullity}
One further ingredient is needed before the two baselines --- affine and
projective --- can be compared on equal footing: an overall rescaling of
the wavefunction is itself a first-order direction, and it does not come
from any bond generator. Choose a ray representative $r_A$ with $J_A
r_A=\Psi(A)$, for example a variation equal to $A_v$ at one site and
zero elsewhere. Since $J_A\Phi_A=0$,
the ray is independent of the bond-gauge image. Define
\begin{equation}
\mathcal G_A=\im\Phi_A+\mathbb C r_A,\qquad
\extra_A=\ker J_{\perp,A}/\mathcal G_A.
\label{eq:extra}
\end{equation}
At the points of Proposition~\ref{prop:bond},
$\dim\mathcal G_A=E(q^2-1)+N$.
Changing the choice of $r_A$ among site rescalings leaves this sum
unchanged: their differences are scalar bond gauges on a connected graph.
Indeed, the difference of rescalings at sites $u,v$ corresponds to
the vertex vector $\boldsymbol e_u-\boldsymbol e_v\in\im\partial$.

We write $g_0=E(q^2-1)+N-1$ for this affine gauge rank and
$g=g_0+1$ for the projective baseline. The non-bond nullity is
$\nu_A=\dim\extra_A$; we write $\nu$ when the point is clear, and
$\nu_{\rm WGS}$, $\nu_{\rm generic}$ when comparing points. Its
elements are null directions of the physical differential that lie
outside $\im\Phi_A$; as stated in the Scope paragraph, they remain
redundancies of the parametrization, and
Proposition~\ref{prop:finite} integrates some of them to finite
state-preserving transformations. There is a natural isomorphism
\begin{equation}
\ker J_A/\im\Phi_A\simeq\extra_A.
\label{eq:affineprojective}
\end{equation}
Indeed, a projective null variation has $J_A\delta A=c\Psi(A)$;
subtracting $c r_A$ produces an affine null variation, and adding an
element of $\im\Phi_A$ does not change its class. The map is injective:
if $\delta A\in\ker J_A\cap\mathcal G_A$, then
$\delta A=\Phi_A(X)+cr_A$ and $0=J_A\delta A=c\Psi(A)$ force $c=0$.
Thus both quotients measure the same non-bond redundancy, which we
now compute.

Five objects recur throughout: the physical tangent space $\T_A$, the
support family $\K(G)$ with its center sets $C_S$, the gauge-and-ray
space $\mathcal G_A$, the non-bond quotient $\extra_A$ of dimension
$\nu$, and the minimal regions that generate it.
Table~\ref{tab:notation} lists these with the remaining notation and
where each is defined.

\begin{table*}[!tbp]
\centering
\caption{Notation. Ranks and dimensions are complex; $q=d=D$
throughout.}
\label{tab:notation}
\footnotesize
\setlength{\tabcolsep}{3pt}
\centering
\begin{adjustbox}{max width=\textwidth}
\begin{tabular}{>{$}l<{$}>{\raggedright\arraybackslash}p{8.07cm}l}\toprule
\text{Symbol} & Meaning & Def.\\
\midrule
N,\ E & numbers of vertices and edges of $G=(V,E_G)$ & Sec.~\ref{sec:bond}\\
\mathcal P_G,\ P & tensor parameter space and its dimension & \eqref{eq:parameters}\\
J_A,\ J_{\perp,A} & differential and its projective part & \eqref{eq:j}\\
\Phi_A,\ \partial & bond action; vertex-edge incidence matrix & \eqref{eq:phi}, \eqref{eq:cycle}\\
\mathcal R_{\rm cyc} & scalar cycle dependencies, $\dim=E-N+1$ & \eqref{eq:cycle}\\
\mathcal G_A,\ g_0,\ g & bond gauge plus ray; affine and projective gauge ranks & \eqref{eq:extra}\\
\extra_A,\ \nu & non-bond quotient and non-bond nullity & \eqref{eq:extra}\\
N(v),\ N[v] & open and closed neighborhoods & \eqref{eq:localmap}\\
\chi_{S,\boldsymbol k},\ \mathcal F_S & product character; nonzero frequencies on $S$ & \eqref{eq:walsh}\\
\xi_{v,S}^{\boldsymbol k} & one-site parameter vector producing $\chi_{S,\boldsymbol k}\psi$ & \eqref{eq:q}\\
\K(G),\ C_S,\ m_S & supported sets; their centers and multiplicity & \eqref{eq:family}\\
K_q(G),\ r & weighted support count; $r=q-1$ & \eqref{eq:weightedcount}\\
\T_A,\ e_{S,\boldsymbol k} & physical tangent space; basis vectors $\chi_{S,\boldsymbol k}\psi$ & \eqref{eq:tangent}\\
z_{S;u,v}^{\boldsymbol k} & center-difference kernel relation & \eqref{eq:z}\\
Q_R,\ u_R,\ a_R & minimal-region quotient; universal vertices; pinned pairs & Sec.~\ref{sec:minimalgeneral}\\
F_3,\ F_4,\ E_{33} & triangles, squares, edges shared by two triangles & \eqref{eq:archmotifs}\\
n_3,\ n_4,\ n_{33} & per-vertex triangle, square, consecutive-triangle counts & \eqref{eq:multiorbit}\\
b_k,\ \kappa(q),\ M & periodic-cell support counts; number of cells & \eqref{eq:cellbk}\\
\mathcal M,\ \Xi & projective QGT; explicit local lift & Sec.~\ref{sec:liftcoords}\\
P_{\T_A},\ R_{\rm TDVP} & tangent projector; TDVP residual & \eqref{eq:projector}, \eqref{eq:residual}\\
\widetilde J(\varepsilon),\ L_\Delta & reduced Jacobian near $A_0$; opening map & \eqref{eq:openingmap}\\
\bottomrule
\end{tabular}
\end{adjustbox}
\end{table*}

\section{Exact tangent spaces at weighted graph states}
\label{sec:star}
Identifying every single-site variation with a function on its closed
neighborhood turns the full differential into a sum over Fourier
sectors, from which both its image and its complete kernel can be read
off directly.

\subsection{A local tensor is a closed-neighborhood function}
The open neighborhood of $v$ is $N(v)$ and its closed neighborhood is
$N[v]=N(v)\cup\{v\}$. Let $\delta A_v$ be an arbitrary variation of the
single tensor at $v$. Dividing its contracted variation by the nonzero
amplitude~\eqref{eq:wgs} gives
\begin{equation}
\begin{split}
\frac{\delta\psi_v(s)}{\psi(s)}
={}&\frac{\sqrt q}{\prod_{e=vw}W_e(s_v,s_w)}\\
&\times\sum_{\{a_e\}}\delta A_v^{s_v}(\{a_e\})
       \prod_{e=vw}F_{w,e}(s_w,a_e).
\end{split}
\label{eq:localmap}
\end{equation}
Only physical variables in $N[v]$ occur. For each fixed $s_v$, the numerator is a tensor
product of invertible $q$-dimensional maps, and the denominator has no
zero entries. Hence Eq.~\eqref{eq:localmap} is an isomorphism between the
single-site parameter block and all functions of $N[v]$. Both spaces
have dimension $q^{\deg(v)+1}$.

This is stronger than the locality of a Hamiltonian term. It identifies
every possible single-tensor variation and establishes that there are
no missing functions within the closed neighborhood. Conversely, it
excludes functions depending on variables outside that neighborhood. It also
shows why local virtual degeneracies of a physical slice do not by
themselves imply that a single-site Jacobian has a kernel.

For $S\subseteq V$, let $\mathcal F_S=\{1,\ldots,q-1\}^{S}$ denote
the nonzero frequency assignments on $S$. The empty support has one
empty assignment. With $\omega=e^{2\pi i/q}$, define the character
\begin{equation}
\chi_{S,\boldsymbol k}(s)=\prod_{v\in S}\omega^{k_vs_v},
\quad \boldsymbol k\in\mathcal F_S,\quad \chi_\varnothing=1.
\label{eq:walsh}
\end{equation}
These are the product characters of the finite Abelian group
$\mathbb Z_q^N$~\cite{terras1999fourier}; for $q=2$ they are the
Fourier--Walsh basis~\cite{odonnell2014boolean}.
There is a unique one-site parameter vector $\xi_{v,S}^{\boldsymbol k}$
whenever $S\subseteq N[v]$, with
\begin{equation}
J_A \xi_{v,S}^{\boldsymbol k}=\chi_{S,\boldsymbol k}\psi.
\label{eq:q}
\end{equation}
Here and below the WGS representation is normalized. The representative
has explicit local entries. Extend $k_w$ by zero outside $S$, and set
\begin{align}
b_{v,e}^{(k_w)}(s_v,\cdot)
 &=F_{w,e}^{-1}[W_e(s_v,t)\omega^{k_wt}]_{t=0}^{q-1},\nonumber\\
(\xi_{v,S}^{\boldsymbol k})^{s_v}(\{a_e\})
 &=q^{-1/2}\omega^{k_vs_v}\prod_{e=vw}b_{v,e}^{(k_w)}(s_v,a_e).
\label{eq:qexplicit}
\end{align}
Here $W_e(s_v,t)$ is read with $v$ as its first argument, and
$b_{v,e}^{(k)}$ is the unique column with
$\sum_a F_{w,e}(t,a)\,b_{v,e}^{(k)}(s_v,a)=W_e(s_v,t)\omega^{kt}$;
all contractions are bilinear. Contracting each factor with $F_{w,e}$
proves Eq.~\eqref{eq:q}. No global Jacobian is required to construct
these representatives. At $q=2$ each support has one assignment,
recovering the Walsh character $\chi_S$ and representative $\xi_{v,S}$.

\subsection{Rank and all relations}
Different vertices produce the same function when their closed
neighborhoods overlap. Define the downward-closed family of allowed supports
\begin{equation}
\begin{gathered}
\K(G)=\bigcup_{v\in V}2^{N[v]},\\
C_S=\{v:S\subseteq N[v]\},\quad m_S=|C_S|.
\end{gathered}
\label{eq:family}
\end{equation}
A support is counted once in $\K$, irrespective of how many closed
neighborhoods contain it, but carries $(q-1)^{|S|}$ distinct characters.
Define the weighted count
\begin{equation}
K_q(G)=\sum_{S\in\K(G)}(q-1)^{|S|}.
\label{eq:weightedcount}
\end{equation}
In particular, $K_2(G)=|\K(G)|$.

\begin{theorem}[Exact physical tangent image]
For the WGS representation~\eqref{eq:site}, the projective tangent image is
\begin{equation}
\begin{gathered}
\T_A:=\im J_{\perp,A}
=\operatorname{span}\{\chi_{S,\boldsymbol k}\psi\},\\
S\in\K(G)\setminus\{\varnothing\},\quad
\boldsymbol k\in\mathcal F_S.
\end{gathered}
\label{eq:tangent}
\end{equation}
In particular,
\begin{align}
\rank J_{\perp,A}&=K_q(G)-1,\label{eq:rank}\\
\dim\extra_A&=P-K_q(G)+1-[E(q^2-1)+N].\label{eq:extrarank}
\end{align}
All unprojected relations are generated by
\begin{equation}
z_{S;u,v}^{\boldsymbol k}
=\xi_{u,S}^{\boldsymbol k}-\xi_{v,S}^{\boldsymbol k},
\quad u,v\in C_S,\quad \boldsymbol k\in\mathcal F_S.
\label{eq:z}
\end{equation}
Projective normalization adds one ray direction.
\label{thm:star}
\end{theorem}
\begin{proof}
In the invertible coordinates of Eq.~\eqref{eq:localmap}, write a variation
as $\{x_{v,S}^{\boldsymbol k}:S\subseteq N[v],
\boldsymbol k\in\mathcal F_S\}$. The full differential becomes
\begin{equation}
J_A\delta A=\sum_{S\in\K(G)}\sum_{\boldsymbol k\in\mathcal F_S}
       \left(\sum_{v\in C_S}x_{v,S}^{\boldsymbol k}\right)
       \chi_{S,\boldsymbol k}\psi.
\label{eq:sector}
\end{equation}
For fixed $(S,\boldsymbol k)$ we call the coordinates
$\{x_{v,S}^{\boldsymbol k}\}_{v\in C_S}$ the $(S,\boldsymbol k)$ sector.
The WGS Born probability is $q^{-N}$, so the vectors
$\chi_{S,\boldsymbol k}\psi$ are orthonormal. Each distinct character
contributes one image dimension.
Projection removes just $S=\varnothing$, proving Eq.~\eqref{eq:rank}.
For each frequency assignment in a nonempty support sector the kernel
is the hyperplane $\sum_{v\in C_S}x_{v,S}^{\boldsymbol k}=0$,
generated by differences between centers.
The unprojected empty sector has the same form; the projected empty sector
is all of $\mathbb C^N$. Finally subtract
$\dim\mathcal G_A=E(q^2-1)+N$
from $P-\rank J_\perp$ to obtain Eq.~\eqref{eq:extrarank}.
\end{proof}

The unprojected rank is $K_q(G)$. The extra $+1$ in
Eq.~\eqref{eq:extrarank} thus comes specifically from projection to rays.
This distinction enters every nullity formula below. Equivalently, a nonempty support with multiplicity $m_S$
has $(m_S-1)(q-1)^{|S|}$ relations, while the empty projective sector has $N$.

For each support and frequency assignment, connect its centers by any tree.
Its $m_S-1$ edges provide independent differences~\eqref{eq:z}. Applying this
construction to every character gives a complete kernel basis on any finite
graph. The choice of tree concerns centers that generate the same
function; it is conceptually separate from the spanning tree used to
resolve scalar bond cycles. Figure~\ref{fig:motifs}(b) illustrates
the shared-support relation for the two opposite centers of a square.

The physical basis in Eq.~\eqref{eq:tangent} is orthonormal, but the
parameter vectors $\xi_{v,S}^{\boldsymbol k}$ generally are not, since
their construction uses nonunitary inverse edge factors; the physical
projector does not require a choice of Euclidean parameter metric.

At the Fourier point, balanced edge factors also yield an exact Euclidean
QGT positive spectrum in terms of the same center multiplicities.
Appendix~\ref{SM-sec:positive_spectrum} derives this spectrum and a
system-size-independent condition bound at fixed nondegenerate weights
and bounded degree, separating positive-spectrum conditioning from
exact nullity.

\section{Relations modulo bond gauge and their local generators}
\label{sec:relations}
Theorem~\ref{thm:star} located every kernel relation inside a support
sector. This section separates the relations already accounted for by
ordinary bond gauge from the genuinely new ones, first algebraically
and then geometrically, by the smallest region on which each new
relation can be certified.

\subsection{Locating ordinary gauge within the function sectors}
For an edge $e=uv$, the map
\begin{equation}
X\longmapsto f_X(s,t)
=\frac{[F_{u,e}X F_{v,e}^{\mathsf T}]_{st}}{W_e(s,t)}
\label{eq:edgefunction}
\end{equation}
is an isomorphism from $q\times q$ matrices to functions of $(s,t)$.
The bond action generates $\xi_{u,f_X}-\xi_{v,f_X}$.
Thus its image can be expressed in the four support sectors
$\varnothing$, $\{u\}$, $\{v\}$, and $\{u,v\}$.

Across the graph, scalar gauges span the $N-1$ constant differences.
For a singleton $S=\{v\}$, its centers are $N[v]$; the edges incident to
$v$ connect them in a star and generate all $\deg(v)$ differences.
Each singleton has $q-1$ frequency assignments, so these supports
contribute $\sum_v\deg(v)(q-1)=2E(q-1)$ independent directions.
Each edge support $S=\{u,v\}$ has one ordinary relation per assignment
connecting its endpoints, giving another $E(q-1)^2$. With the ray this yields
\begin{equation}
N+2E(q-1)+E(q-1)^2=N+E(q^2-1).
\label{eq:ordinarysectors}
\end{equation}
This sector description identifies precisely what must be removed before
calling a relation non-bond. A shared neighborhood does not automatically
imply a new direction: constants, singletons, and the endpoint relation
of each edge have already been counted.

\begin{proposition}[Sector description of the non-bond kernel]
For each nonempty $S\in\K(G)$ and $\boldsymbol k\in\mathcal F_S$,
the non-bond quotient has dimension zero when $|S|=1$, dimension
$m_S-2$ when $S$ is an edge, and dimension $m_S-1$ otherwise.
The empty projective sector is exhausted by scalar gauges and the ray.
A basis is obtained by taking a tree on each center set $C_S$;
for edge supports first identify the two endpoint centers.
\label{prop:quotient}
\end{proposition}
\begin{proof}
The kernel in each affine character sector is the sum-zero subspace
of $\mathbb C^{C_S}$, of dimension $m_S-1$.
The preceding calculation exhausts the empty and singleton sectors.
On an edge support it removes exactly the difference between the two
endpoints. No bond action has support on a nonedge pair or on three
or more vertices. Quotienting that one endpoint difference identifies
its centers and reduces the dimension by one. Tree differences give
bases of the remaining sum-zero spaces.
\end{proof}
In particular, the quotient can be constructed sector by sector without
finding a global orthogonal complement of the gauge image. This is a
linear classification at $A$; which representatives integrate to
finite non-bond transformations is the subject of
Sec.~\ref{sec:integrability}.

\subsection{Locality with the external legs retained}
We now make precise how large a piece of the graph is needed to
witness a relation, so that later sections can speak of a
\emph{minimal region} without ambiguity.

For a connected induced region $R$, contract the tensors inside $R$ but
leave every physical leg and every external virtual leg open. Denote its
unprojected differential by $J_R$. Such boundary-to-bulk maps are standard
in PEPS injectivity arguments~\cite{molnar2018normal,cirac2021matrix}.
Here we define its intrinsically new kernel to be the
quotient of $\ker J_R$ by internal bond gauges and all kernels supported
on proper connected induced subregions, embedded by zero extension.
We call a region minimal if this quotient is nonzero.
Appendix~\ref{SM-sec:sminimal} proves the open-region sector decomposition
and makes the subregion embeddings explicit.

The open-leg definition prevents a cancellation that requires a special
external environment from being mistaken for a local identity. For a
nonempty support, a relation $z_{S;u,v}^{\boldsymbol k}$ has a certificate
in the connected induced region on $R=S\cup\{u,v\}$. For an empty
support, take any connected induced region containing both centers.
The untouched factors on the external legs agree exactly. Usually only
two tensors change; the region size specifies where their cancellation
can be certified, not how many tensor blocks are nonzero.

\subsection{The intrinsic minimal-region polynomial}
\label{sec:minimalgeneral}
Two structural features of a region control its minimal kernel
completely: how many of its vertices are adjacent to every other vertex
in the region (its \emph{universal} vertices), and which pairs of
vertices are pinned down as the unique common neighborhood of everyone
else. Neither periodicity nor a planar embedding is needed.
Set $r=q-1$. Let $Q_R$ be the open-region quotient just defined, with all proper
connected induced subregions removed. For a connected induced graph $R$ with $|R|\ge3$ vertices, let $u_R$
be its number of universal vertices. Let $a_R$ count unordered pairs
$\{a,b\}$ such that
\begin{equation}
\bigcap_{s\in R\setminus\{a,b\}}N_R[s]=\{a,b\}.
\label{eq:aregion}
\end{equation}
\begin{theorem}[Intrinsic minimal-region polynomial]
\label{thm:minimalgeneral}
For a nondegenerate WGS representation, with its external legs retained,
\begin{equation}
\begin{split}
\dim Q_R={}&a_Rr^{|R|-2}+u_R\mathbf1_{u_R\ge2}r^{|R|-1}\\
&+\max(u_R-1,0)r^{|R|}.
\end{split}
\label{eq:Qgeneral}
\end{equation}
For $|R|=1,2$, the quotient is zero.
\end{theorem}
\begin{proof}
Fix a nonempty supported set $S\subseteq R$, with center set
$C_S\ne\varnothing$ computed in $R$, and one frequency assignment on
$S$. The sector kernel is the sum-zero subspace of $\mathbb C^{C_S}$,
spanned by center differences. Define the \emph{reduction graph} on
the vertex set $C_S$ by joining two centers $a,b$ whenever their
difference already lies in the subspace quotiented out in $Q_R$: either
(i) $S\cup\{a,b\}\subsetneq R$, so the difference occurs on this proper
subregion, which is connected because $S$ is nonempty and both vertices
are centers; or (ii) $S\subseteq\{a,b\}$ and $ab$ is an edge, so the
difference is an ordinary bond gauge. Differences along the edges of the
reduction graph span a subspace of codimension $c-1$ in the sum-zero
space, where $c$ is its number of connected components, so the sector
contributes $c-1$ to $\dim Q_R$. If $|R\setminus S|\ge3$, condition (i)
joins every center pair.
If it equals two, only the pair of omitted vertices can fail to be joined.
Any additional center $z\in S$ is joined to both omitted vertices by
(i), making the reduction graph connected. Thus the sector contributes
precisely when $C_S$ is that pair, as in
Eq.~\eqref{eq:aregion}. The ordinary-edge exception cannot satisfy this
condition for $|R|\ge3$. If one vertex is omitted, it supplies a new
component exactly when it and at least one other vertex are universal;
there are $u_R$ such supports. If none is omitted, all centers are
universal and contribute $\max(u_R-1,0)$ differences (zero if there
is no center). Weighting these sectors by
$r^{|S|}$ gives the formula. Empty and singleton ordinary sectors introduce
no non-bond quotient, as in Appendix~\ref{SM-sec:sminimal}.
\end{proof}

\begin{corollary}[Minimal regions on at most four vertices]
\label{cor:smallregions}
Every tree has $Q_R=0$. Among connected graphs on $|R|\le4$ vertices,
exactly four have $Q_R\ne0$:
\begin{center}
\normalfont
\begin{adjustbox}{max width=\columnwidth}
\begin{tabular}{lcccl}
\hline
region & $|R|$ & $u_R$ & $a_R$ & $\dim Q_R$\\
\hline
triangle $K_3$ & 3 & 3 & 0 & $3r^2+2r^3$\\
chordless square $C_4$ & 4 & 0 & 2 & $2r^2$\\
diamond $K_4-e$ & 4 & 2 & 1 & $r^2+2r^3+r^4=q^2r^2$\\
$K_4$ & 4 & 4 & 0 & $4r^3+3r^4$\\
\hline
\end{tabular}
\end{adjustbox}
\end{center}
\end{corollary}
\begin{proof}
In a tree on $|R|\ge3$ vertices at most one vertex is universal, so the
last two terms of Eq.~\eqref{eq:Qgeneral} vanish, and a pair $\{a,b\}$
counted by $a_R$ would need every other vertex adjacent to both $a$ and
$b$, which for $|R|\ge4$ closes a four-cycle and for $|R|=3$ gives
$N_R[s]=R\ne\{a,b\}$. The remaining connected graphs on three or four
vertices are the four listed ones and the paw (a triangle with a
pendant edge), which has $u_R=1$, $a_R=0$; evaluating
Eq.~\eqref{eq:Qgeneral} gives the table.
\end{proof}

The three regions of Corollary~\ref{cor:smallregions} other than $K_4$
are the minimal generating types that occur on periodic regular-polygon
tilings, where $K_4$ is excluded by unit-edge geometry
(Lemma~\ref{lem:shortcycles}); Fig.~\ref{fig:motifs} shows them for
$q=2$. Their relations have the following explicit sector templates,
which the lattice counts of Sec.~\ref{sec:lattice} use. The superscript
$\boldsymbol k$ is suppressed: each template is included for every
$\boldsymbol k\in\mathcal F_S$ on its displayed support $S$.

\paragraph{Triangle.}
For vertices $a,b,c$ joined pairwise, the five support-center templates are
\begin{align}
&z_{\{a,b\};a,c},\quad z_{\{a,c\};a,b},\quad z_{\{b,c\};b,a},\nonumber\\
&z_{\{a,b,c\};a,b},\quad z_{\{a,b,c\};a,c}.
\label{eq:triangle}
\end{align}
The first three express the ability of the third vertex to produce the
function on an edge ($3r^2$ directions, the $u_R\mathbf 1_{u_R\ge2}r^{|R|-1}$
term); the last two express the three centers' ability to produce each
full-triangle character ($2r^3$ directions, the $\max(u_R-1,0)r^{|R|}$
term).

\paragraph{Chordless quadrilateral.}
For the cycle $a-b-c-d-a$, with no diagonal edges, the two templates are
\begin{equation}
z_{\{a,c\};b,d},\qquad z_{\{b,d\};a,c}.
\label{eq:square}
\end{equation}
The two centers see the same pair of opposite vertices, and neither
center lies in the support; these are the $a_R=2$ pairs of
Eq.~\eqref{eq:aregion}.

\paragraph{Diamond.}
Let $u,v$ share an edge and have two common neighbors $w,x$ that are not
adjacent. Beyond the $2(3r^2+2r^3)$ directions of its two triangles, the
four templates
\begin{equation}
z_{S;u,v},\qquad \{w,x\}\subseteq S\subseteq\{u,v,w,x\},
\label{eq:diamond}
\end{equation}
whose supports contain both $w$ and $x$ and are therefore absent from
either triangle, supply $r^2(1+2r+r^2)=q^2r^2$ new directions: the
unique pair counted by $a_R=1$ is $\{u,v\}$, with complementary
support $R\setminus\{u,v\}=\{w,x\}$; $u,v$ are also the two universal
vertices. At $q=2$ the diamond carries fourteen non-bond directions,
ten of which already live on its triangles.

Regions with more universal vertices or larger joins occur beyond the
regular-polygon class (Appendix~\ref{SM-sec:periodiccatalogue}).
Unlike the independent sector basis,
\emph{summing all overlapping minimal-region dimensions need not give
$\nu$}; the total in Table~\ref{SM-tab:periodic} is computed by
Theorem~\ref{thm:cell}, which removes all dependencies.

\begin{figure*}[!tbp]
\centering
\includegraphics[width=\textwidth]{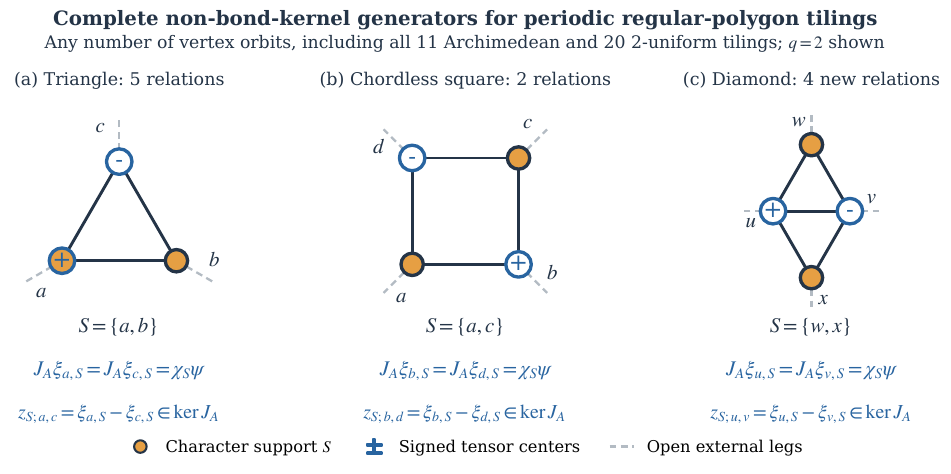}
\caption{The three minimal generating types of non-bond kernel relations
on periodic regular-polygon tilings, at $q=2$
(Theorem~\ref{thm:regular}). In each panel two site variations
(blue $\pm$ centers) generate the same physical vector $\chi_S\psi$
on the orange support $S$, so their difference lies in $\ker J_A$.
(a) One of the five triangle relations, Eq.~\eqref{eq:triangle}.
(b) One of the two chordless-square relations, Eq.~\eqref{eq:square}:
the opposite centers $b,d$ generate the character on $\{a,c\}$.
(c) One of the four diamond relations additional to those of its two
triangles, Eq.~\eqref{eq:diamond}. Dashed stubs are retained external
virtual legs; physical legs are suppressed.}
\label{fig:motifs}
\end{figure*}

\section{WGS classification on regular-polygon tilings}
\label{sec:lattice}
The unit-edge geometry of regular-polygon tilings restricts the
closed-neighborhood intersections of Sec.~\ref{sec:relations} to three
minimal motif types, with a face-order formula valid for any number of
vertex orbits. Tensor variations at distinct vertices remain
independent throughout, even when the graph has translation or point
symmetries.

\subsection{General classification and face-order formula}
The count depends on the \emph{cyclic order} of faces at a vertex,
not only on face densities. At each vertex let $n_3,n_4$ count incident triangular and square faces,
and let $n_{33}$ count consecutive pairs of triangular faces in the
cyclic vertex configuration, including the last--first pair. Brackets
$\langle\cdot\rangle$ denote the average over vertices of a translation
cell, and $r=q-1$.

Here and below, a finite quotient of a periodic graph is called
\emph{locally faithful} if the quotient map is injective on every
graph-distance-two vertex ball. This is exactly the condition under
which every support, its center set, and the multiplicities $m_S$ agree
with those of the infinite graph (Theorem~\ref{thm:cell}); in
particular it excludes wrapping three- and four-cycles.

\begin{theorem}[Periodic regular-polygon classification]
\label{thm:regular}
For nondegenerate WGS representations with $D=d=q\ge2$ on an
edge-to-edge, noncrossing, unit-edge Euclidean tiling by regular polygons
with a finite translation cell, the only minimal regions are
triangles, chordless squares, and shared-edge diamonds, whenever present.
These types generate every finite-support relation beyond ordinary bond
gauge, with no restriction on the number of vertex orbits. On locally
faithful tori, let $F_3,F_4,E_{33}$ count triangular faces, square
faces, and shared triangle edges. Their independent total non-bond nullity is
\begin{align}
\nu&=F_3(3r^2+2r^3)+2F_4r^2+E_{33}q^2r^2,
\label{eq:archmotifs}\\
\frac\nu N&=r^2\left[
\frac{\langle n_3\rangle}{3}(2q+1)
+\frac{\langle n_4\rangle}{2}
+\frac{\langle n_{33}\rangle}{2}q^2\right].
\label{eq:multiorbit}
\end{align}
The regions are minimal with every external virtual leg retained.
\end{theorem}

\begin{lemma}[Short cycles in unit-edge regular-polygon tilings]
\label{lem:shortcycles}
In an edge-to-edge, noncrossing, unit-edge Euclidean tiling by regular
polygons, (i) two distinct vertices have at most two common neighbors
and $K_4$ does not occur; (ii) every three-cycle bounds a triangular
face; (iii) every four-cycle is either the boundary of a square face or
the union of two triangular faces sharing an edge; in particular a
chordless four-cycle is a square face.
\end{lemma}
Part (i) holds because two unit circles meet in at most two points.
For (ii) and (iii), the region bounded by a unit three- or four-cycle
has area at most one, too small to contain an interior vertex together
with its at least three incident faces, so it is tiled by faces with
corners on the cycle (Appendix~\ref{SM-sec:latticeproofs}).

\begin{proof}[Proof of Theorem~\ref{thm:regular}]
By Theorem~\ref{thm:star} and Proposition~\ref{prop:quotient}, it suffices
to classify common closed neighborhoods after removing the ordinary
sectors in Eq.~\eqref{eq:ordinarysectors}.
By Lemma~\ref{lem:shortcycles}(i), adjacent centers have at most four
vertices in their closed-neighborhood intersection and nonadjacent
centers at most two. The remaining supports are precisely triangle,
chordless-square, and diamond supports, and by
Lemma~\ref{lem:shortcycles}(ii)--(iii) chordless four-cycles are
exactly square faces while diamonds correspond to shared triangle edges.
If nonadjacent centers have two adjacent common neighbors, their support
is the shared edge of two triangles. Each center then joins the ordinary
endpoint component through its own triangle, so this sector is already
counted by the triangle contributions below.

For each frequency assignment, each triangle edge support connects a new common
neighbor to the ordinary endpoint component. A triangle's full support
has three centers and two independent differences. Each remaining
square or diamond support has exactly two centers. These sectors are
independent by Eq.~\eqref{eq:sector}. Including the $(q-1)^{|S|}$
assignments on each support gives the stated contributions without
further dependencies. The support and centers of each new
relation require its stated region; no proper connected subregion
contains the required relation in that sector. The open-leg argument
in Appendix~\ref{SM-sec:sminimal} makes this last step independent of external tensors.

The motif dimensions in Eqs.~\eqref{eq:triangle}--\eqref{eq:diamond}
therefore give Eq.~\eqref{eq:archmotifs}. Counting face incidences and
consecutive triangular pairs over one cell yields
$F_3/N=\langle n_3\rangle/3$, $F_4/N=\langle n_4\rangle/4$, and
$E_{33}/N=\langle n_{33}\rangle/2$, proving
Eq.~\eqref{eq:multiorbit}. No step uses vertex transitivity.
\end{proof}
The square coefficient combines four vertex incidences per face with
$2r^2$ directions per square. Likewise, each shared triangle edge has
two endpoint incidences and adds $q^2r^2$ directions beyond its two
triangles. The last term therefore records cyclic face order, not just
the separate triangle and square densities.

\subsection{Archimedean tilings}
The eleven Archimedean tilings are vertex transitive, so the cell averages
reduce to a single cyclic configuration
\cite{daniel2020archimedean,lukin2024archimedean}.

\begin{corollary}[Archimedean-lattice counts]
\label{cor:lattices}
For nondegenerate WGS representations with $D=d=q\ge2$ on any of
the eleven Archimedean tilings, the three motif types give the complete
independent counts in Table~\ref{tab:counts} on locally faithful tori,
with
\begin{equation}
\frac\nu N=r^2\left[\frac{n_3}{3}(2q+1)
+\frac{n_4}{2}+\frac{n_{33}}{2}q^2\right].
\label{eq:archdensity}
\end{equation}
Square requires only chordless squares; honeycomb requires none;
kagome requires only triangles; triangular requires triangles and diamonds.
\end{corollary}
\begin{proof}
Specializing Eq.~\eqref{eq:multiorbit} to one vertex orbit gives
Eq.~\eqref{eq:archdensity} and the table. On honeycomb, adjacent
centers share only their endpoints and nonadjacent centers at most one
neighbor, so no non-bond support exists. On triangular, adjacent
centers share two nonadjacent neighbors $w,x$; supports containing
both are exactly the diamond templates~\eqref{eq:diamond}, the rest
lie on the two triangles. Square and kagome are analogous (Appendix~\ref{SM-sec:latticeproofs}).
\end{proof}
For square, honeycomb, triangular, and kagome in the coordinate
conventions of Appendix~\ref{SM-sec:coordinates}, $L_x,L_y\ge5$ suffices
for the bulk formulas. The other seven geometries likewise require
quotients preserving the relevant short neighborhoods. These conditions
are separate from the generic-completeness bounds of Sec.~\ref{sec:generic}.

\begin{table*}[!tbp]
\centering
\caption{Exact WGS bulk counts for the eleven Archimedean lattices on
locally faithful tori, $D=d=q\ge2$, $r=q-1$. Densities are per
vertex ($N$ vertices). $P$: tensor parameters; $g=E(q^2-1)+N$:
bond gauge plus ray; $\nu=\dim\extra_A$: non-bond projective null
directions. The last three columns evaluate $\nu/N$ at the indicated
$q$.}
\label{tab:counts}
\small
\setlength{\tabcolsep}{3.5pt}
\centering
\begin{adjustbox}{max width=\textwidth}
\begin{tabular}{lccccrrr}\toprule
Lattice & Configuration & $P/N$ & $g/N$ & $\nu/N$ & $q=2$ & $q=3$ & $q=4$\\
\midrule
Square &$4^4$&$q^5$&$2q^2-1$&$2r^2$&2&8&18\\
Honeycomb &$6^3$&$q^4$&$(3q^2-1)/2$&0&0&0&0\\
Triangular &$3^6$&$q^7$&$3q^2-2$&$r^2(3q^2+4q+2)$&22&164&594\\
Kagome &$3.6.3.6$&$q^5$&$2q^2-1$&$2r^2(2q+1)/3$&$10/3$&$56/3$&54\\
Star &$3.12^2$&$q^4$&$(3q^2-1)/2$&$r^2(2q+1)/3$&$5/3$&$28/3$&27\\
Square-octagon &$4.8^2$&$q^4$&$(3q^2-1)/2$&$r^2/2$&$1/2$&2&$9/2$\\
Cross &$4.6.12$&$q^4$&$(3q^2-1)/2$&$r^2/2$&$1/2$&2&$9/2$\\
Ruby &$3.4.6.4$&$q^5$&$2q^2-1$&$2r^2(q+2)/3$&$8/3$&$40/3$&36\\
Elongated triangular &$3^3.4^2$&$q^6$&$(5q^2-3)/2$&$r^2(q^2+2q+2)$&10&68&234\\
Snub-square &$3^2.4.3.4$&$q^6$&$(5q^2-3)/2$&$r^2(q^2/2+2q+2)$&8&50&162\\
Maple-leaf &$3^4.6$&$q^6$&$(5q^2-3)/2$&$r^2(9q^2+16q+8)/6$&$38/3$&$274/3$&324\\
\bottomrule
\end{tabular}
\end{adjustbox}
\end{table*}

The table follows from the total motif counts ($N$ squares on the
square lattice, $2N$ triangles and $3N$ shared edges on triangular,
$2N/3$ triangles on kagome) and agrees with the support-union
formula~\eqref{eq:extrarank}. Figure~\ref{fig:lattices} shows the
motifs on each lattice and compares these densities with independent
brute-force counts on periodic tori.

The comparison separates three graph characteristics. The scalar cycle
count grows extensively on all eleven lattices, including honeycomb, yet
honeycomb has no non-bond kernel. Coordination fixes the number of
local tensor parameters, but square and kagome both have degree four
and have different non-bond nullities. The relevant information is the
intersection pattern of closed neighborhoods, which specifies which
functions several tensors can produce.

Elongated triangular and snub-square have the same coordination and
the same triangle and square densities, but $n_{33}=2$ and $1$,
respectively. Their non-bond nullity densities differ by
$q^2(q-1)^2/2$. Thus even face densities do not determine the kernel:
the cyclic arrangement of faces matters. Honeycomb is the only one
of the eleven tilings with zero bulk WGS non-bond nullity.

Small periodic systems must be treated as their actual quotient graphs.
For example, a $4\times4$ square torus has wrapping four-cycles in addition
to its unit squares. At $q=2$ its WGS non-bond nullity is $48$, rather than
the bulk value $2N=32$. Equation~\eqref{eq:extrarank} remains exact.
The finite-graph theorem covers such quotients directly.

\begin{figure*}[!t]
\centering
\includegraphics[width=\textwidth]{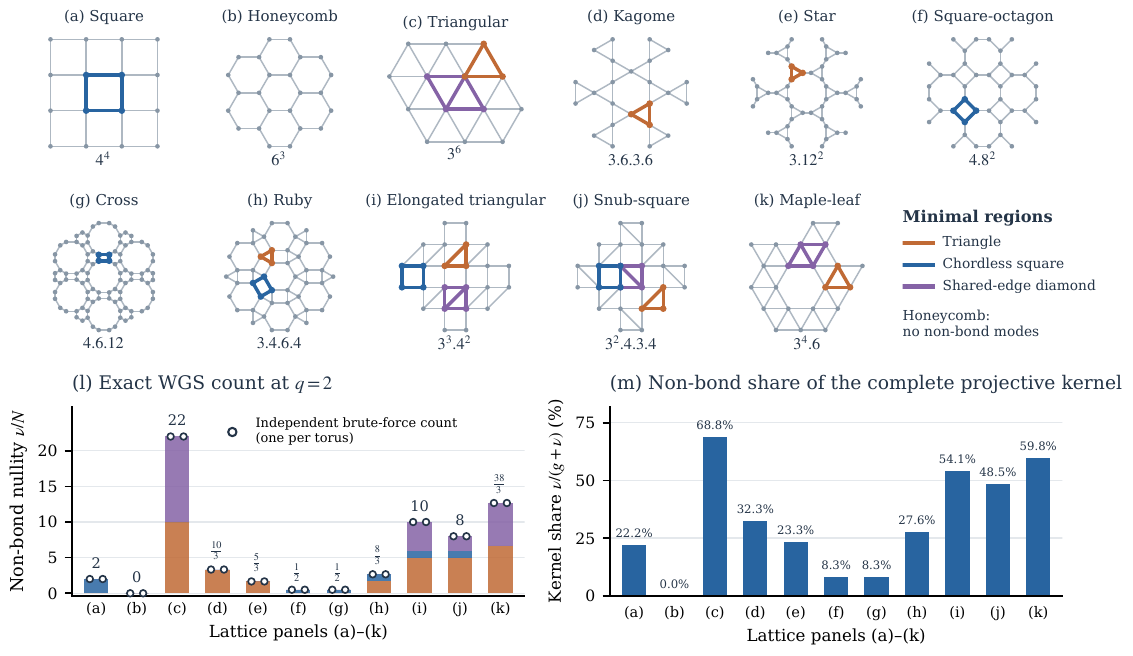}
\caption{All eleven Archimedean lattices and their complete WGS non-bond
nullity at $q=2$. Panels (a)--(k) show planar schematic patches, with
vertex configurations below; drawn edges need not have equal length.
Colored subgraphs mark representative triangles, chordless squares, and
shared-edge diamonds. Panel (l) gives the exact $q=2$ density as stacked
contributions $5F_3/N$, $2F_4/N$, and $4E_{33}/N$ in the same colors;
the diamond contribution is additional to its triangles. Open circles
are independent brute-force counts of the support union, one per
periodic torus: sizes $L=5,8$ for (a)--(d) and construction parameters
$s=7,8$ for (e)--(k), with cell conventions in Appendix~\ref{SM-sec:coordinates}; small horizontal offsets separate the
two tori. Panel (m) gives the exact non-bond share
$\nu/(g+\nu)$, where $g$ includes ordinary bond gauge and the ray,
in the same lattice order as panel (l). Honeycomb alone has zero non-bond nullity.
Completeness follows from Corollary~\ref{cor:lattices}.}
\label{fig:lattices}
\end{figure*}

\subsection{2-uniform tilings}
\label{sec:twouniform}
The twenty 2-uniform tilings have two vertex orbits under their full
Euclidean symmetry groups~\cite{maity2021quotient}, and belong to the
same regular-polygon class, so Theorem~\ref{thm:regular} gives their
complete WGS census from the cell averages
$\langle n_j\rangle=(n_An_j(A)+n_Bn_j(B))/(n_A+n_B)$ of the two vertex
configurations $A,B$ with cell multiplicities $n_A,n_B$. Appendix~\ref{SM-sec:twouniformSM} lists all twenty, using the identifiers
\texttt{t2001}--\texttt{t2020} of the Galebach coordinate data
\cite{soto2019acquiring,sotoGalebachData} together with their
Gr\"unbaum--Shephard configuration pairs $[A;B]$~\cite{grunbaum1987tilings}, and draws them from
the same data. All twenty have positive WGS non-bond nullity; at $q=2$ the densities
$\nu/N$ range from $11/9$ (\texttt{t2001}, $[3.4.6.4;4.6.12]$) to
$52/3$ (\texttt{t2020}, $[3^6;3^4.6]$). Two features of the census are
worth recording. Tilings with the same configuration pair but a
different cell ratio have different non-bond nullities (\texttt{t2003} and
\texttt{t2004}, $[3^3.4^2;4^4]$ with ratios $2{:}2$ and $2{:}1$, give
$6$ and $22/3$), whereas distinct tilings with the same pair and ratio
(\texttt{t2006} and \texttt{t2007}) necessarily coincide, since
Eq.~\eqref{eq:multiorbit} sees only cell-averaged vertex data.

\section{General graphs: periodic counts and tied tensors}
\label{sec:periodic}
We now drop the regular-polygon geometry of Sec.~\ref{sec:lattice}.
Theorem~\ref{thm:minimalgeneral} already classifies the minimal regions
of an arbitrary graph; this section computes the total non-bond
nullity at a WGS representation, which is not a sum of overlapping
minimal-region dimensions, from a periodic cell and extends it to translation-tied tensors.
Eight connectivity examples and the effects of subdivision and line
graphs are collected in Appendix~\ref{SM-sec:periodiccatalogue}.

\subsection{Exact rank and nullity from a periodic cell}
Let an infinite connected locally finite simple graph admit a free $\mathbb Z^2$
translation action with a finite cell $V_0$, containing $n$ vertices
and $e$ undirected edge orbits. Let $d_\alpha$ be the degree of cell
vertex $\alpha$. Closed neighborhoods and multiplicities $m_S=|C_S|$
are computed in the infinite graph. Define
\begin{equation}
\begin{gathered}
b_k=\sum_{\alpha\in V_0}\ 
\sum_{\substack{\varnothing\ne S\subseteq N[\alpha]\\ |S|=k}}
\frac1{m_S},\\
\kappa(q)=\sum_{k\ge1}b_k(q-1)^k.
\end{gathered}
\label{eq:cellbk}
\end{equation}
Only $k\le1+\max_\alpha d_\alpha$ contribute.
\begin{theorem}[Periodic-cell rank and nullity]
\label{thm:cell}
The numbers $b_k$ are integers counting translation orbits of supported
nonempty $k$-vertex sets. On a torus with $M$ cells whose quotient map is
injective on every graph-distance-two vertex ball, a nondegenerate WGS
representation has
\begin{align}
P&=M\sum_{\alpha\in V_0}q^{d_\alpha+1},\nonumber\\
\rank J&=1+M\kappa(q),\nonumber\\
\rank J_\perp&=M\kappa(q),\nonumber\\
g&=M\bigl[e(q^2-1)+n\bigr],\nonumber\\
\nu&=M\Bigl[\sum_{\alpha\in V_0}q^{d_\alpha+1}\nonumber\\
&\qquad-\kappa(q)-e(q^2-1)-n\Bigr].
\label{eq:cellrank}
\end{align}
Here $g$ includes the ray, so $\dim\im\Phi_A=g-1$.
These formulas do not tie tensor variations between translated vertices.
\end{theorem}
\begin{proof}
A finite nonempty support has no nonzero translation stabilizer. Counting
incidences $(\alpha,S)$ over one cell counts every supported-set orbit
$m_S$ times; division by $m_S$ gives Eq.~\eqref{eq:cellbk}.
A supported set lies in one closed neighborhood. Its other centers lie
within distance two of that center. Injectivity on these balls preserves
support identities and center multiplicities, so each support orbit has
$M$ distinct translates on the torus. The constant sector contributes
one affine direction, while each nonempty support contributes
$(q-1)^{|S|}$ directions. Theorem~\ref{thm:star} and the bond-gauge rank
then give Eq.~\eqref{eq:cellrank}.
\end{proof}
For edge labels $(\alpha,\beta,\delta_x,\delta_y)$ with
$|\delta_x|,|\delta_y|\le \ell_0$, rectangular periods
$L_x,L_y>4\ell_0$ suffice: the cell coordinates of a distance-two ball
span at most $4\ell_0$ in either direction. Smaller tori can instead be evaluated by the original finite-graph formula.

\subsection{Translation-tied tensors}
\label{sec:tied}
Practical PEPS calculations on periodic lattices usually tie the tensors
at translated vertices. The Fourier coordinates make this case a direct
corollary. Let $\Gamma\simeq\mathbb Z_{L_x}\times\mathbb Z_{L_y}$ act
freely on a locally faithful torus with $M=|\Gamma|$ cells, and restrict
to the tied parameter space
$\mathcal P_G^\Gamma=\bigoplus_{\alpha\in V_0}\mathbb C^{q^{d_\alpha+1}}$
with $\delta A_{\gamma v}=\delta A_v$, and to tied bond generators
$X_{\gamma e}=X_e$. We evaluate these maps at a $\Gamma$-invariant WGS
representation $A\in\mathcal P_G^\Gamma$: edge phases and the chosen
half-edge factors agree along translation orbits, with a fixed
translation-compatible ordering of virtual legs.
\begin{corollary}[Tied tensors on a locally faithful torus]
\label{cor:tied}
Under the hypotheses of Theorem~\ref{thm:cell}, at such a
$\Gamma$-invariant WGS representation, the tied differential
$J_A^\Gamma$ and tied bond map $\Phi_A^\Gamma$ satisfy
\begin{equation}
\begin{gathered}
\im J_{\perp,A}^\Gamma=\T_A^{\Gamma}\\
=\operatorname{span}\Bigl\{\sum_{\gamma\in\Gamma}\chi_{\gamma S,\boldsymbol k}\psi\Bigr\},\\
\rank J^\Gamma_{\perp,A}=\kappa(q),\\
\rank\Phi_A^\Gamma=e(q^2-1)+n-1,\\
\nu^\Gamma=\nu/M,
\end{gathered}
\label{eq:tied}
\end{equation}
with $\nu$ from Eq.~\eqref{eq:cellrank}: the tied non-bond nullity is
the per-cell value, and the tied physical tangent is the
$\Gamma$-invariant part of the untied one.
\end{corollary}
The proof, in the coordinates of Eq.~\eqref{eq:sector}, is given in
Appendix~\ref{SM-sec:latticeproofs}.
For translation-invariant Hamiltonians, the projector
of Sec.~\ref{sec:tdvp} restricted to $\T_A^\Gamma$ is the tied
projector, and the support criterion of
Corollary~\ref{cor:support} applies unchanged.

Further graph operations give closed counts: complete edge subdivision
eliminates the WGS non-bond kernel, so the Lieb, heavy-hex, and
heavy-square-octagon connectivities have $\nu=0$, while line graphs admit a formula
in terms of vertex degrees and four-cycles of the original graph
(Proposition~\ref{SM-prop:operations}).

\section{The fiber at a WGS point: injectivity, normality, and finite equivalences}
\label{sec:fiber}
This section asks what the fiber of the contraction map looks like at
a WGS point: which regions are injective, so that the normal-PEPS
theorem applies, and which non-bond null directions are initial
velocities of finite constant-state paths.

\subsection{Injective regions and normality}
\label{sec:injectivity}
A single WGS site with $\deg(v)\ge2$ is never injective. For fixed
$s_v$, Eq.~\eqref{eq:site} has virtual tensor rank one; an invertible
change on each virtual leg turns the rows of the virtual-to-physical
site map into the independent vectors $\langle s_v|^{\otimes\deg(v)}$,
so the map has rank $q$ on a domain of dimension $q^{\deg(v)}>q$. The
same dimension count applies to generic site tensors. It neither
decides injectivity after blocking nor implies a non-bond kernel:
honeycomb has $\nu=0$ despite it. For regions, the open-region
amplitude of Appendix~\ref{SM-sec:sminimal},
Eq.~\eqref{SM-eq:openamplitude}, decides injectivity completely.
The following criterion applies the geometric argument of
Ref.~\cite{perezgarcia2008parent}, Sec.~6.1, to complex-phase WGS:
the same proof uses invertible boundary factors and a nonvanishing
interior amplitude.
\begin{lemma}[Injective regions]
\label{lem:injective}
At a nondegenerate WGS representation, the boundary-to-bulk map $B_R$
of a connected induced region $R$ is injective if and only if every
vertex of $R$ has at most one edge leaving $R$.
\end{lemma}
\begin{proof}
After the invertible boundary change of basis in
Eq.~\eqref{SM-eq:openamplitude}, the dependence of $\psi_R$ on the
boundary variables $t_e$ of the edges leaving a vertex $v\in R$ is
$\prod_{e}W_e(s_v,t_e)$, a function of the single $q$-valued variable
$s_v$. If $v$ has $m_v\ge2$ leaving edges, the $q^{m_v}$ boundary
configurations at $v$ are mapped into a space of dimension $q$, so
$B_R$ has a kernel. If every $m_v\le1$, then
$\psi_R(s_R,t_\partial)=C(s_R)\prod_{v}W_{e(v)}(s_v,t_{e(v)})$ with
$C$ nowhere zero and one invertible factor per boundary vertex, and
$B_R$ is injective.
\end{proof}
\begin{corollary}[Normality of WGS representations]
\label{cor:normal}
On the infinite square, triangular, and kagome lattices, WGS
representations admit no partition into finite injective blocks and
are not normal PEPS in this sense. Honeycomb and star-lattice WGS
representations are normal,
although the star lattice has $\nu=N(3r^2+2r^3)/3>0$.
\end{corollary}
On the first three infinite lattices every finite region has a vertex with at
least two external edges. Honeycomb and star instead admit partitions
into injective hexagons and triangles, respectively
(Appendix~\ref{SM-sec:normalitydetails}). The infinite-lattice
qualification is essential: on a square torus with $L_x\ge3$ and even
$L_y\ge6$, successive pairs of complete rows form injective blocks,
since each vertex has only one edge leaving its block. Such winding
blocks are not finite blocks of the infinite lattice.

The star lattice does not conflict with the fundamental theorem,
because the two statements concern different maps. The theorem applies
to the blocked network: two representations of the same state are
related by gauges on the coarse bonds between triangle blocks. On
star, Theorem~\ref{thm:regular} leaves only triangle motifs, and each
triangle is one block. The non-bond quotient therefore has a basis of
relations certified inside single blocks with their external legs
retained (Sec.~\ref{sec:relations}); each leaves every blocked tensor
unchanged to first order and is invisible at the coarse level. What
normality does not constrain is how a blocked tensor is parametrized
by its constituent site tensors, and on star this internal
parametrization carries the entire non-bond kernel. Normality and the
non-bond nullity are independent in this precise sense: star is normal
with $\nu>0$, honeycomb is normal with $\nu=0$, and the infinite square
lattice is not normal under finite blocking despite $\nu>0$. We next
construct finite equivalences in the original tensors.

\subsection{Exact finite equivalences}
\label{sec:integrability}
Which non-bond null directions are initial velocities of finite
constant-state paths? The local isomorphism answers this exactly
whenever the centers involved are pairwise nonadjacent. For a vertex
$v$ and a function $f$ on the variables of $N[v]$, let $L_v^{-1}(f)$
denote the unique site tensor whose contraction with the unchanged WGS
environment equals $f\psi$, i.e., the preimage of $f$ under
Eq.~\eqref{eq:localmap}; thus $L_v^{-1}(1)=A_v$ and
$L_v^{-1}(\chi_{S,\boldsymbol k})=\xi_{v,S}^{\boldsymbol k}$.

\begin{proposition}[Independent-set equivalences]
\label{prop:finite}
Let $U\subseteq V$ be an independent set of vertices, and for each
$u\in U$ let $f_u$ be a nowhere-vanishing function on the variables of
$N[u]$ such that $\prod_{u\in U}f_u\equiv1$. Then the tensors
\begin{equation}
A'_u=L_u^{-1}(f_u)\ (u\in U),\qquad A'_w=A_w\ (w\notin U),
\label{eq:finitefamily}
\end{equation}
satisfy $\Psi(A')=\Psi(A)$ exactly. In particular, for local functions
$h_u$ with $\sum_uh_u\equiv0$, choosing $f_u(t)=e^{t h_u}$ gives a
constant-state curve with initial velocity
$\delta=\sum_{u\in U}L_u^{-1}(h_u)$. In the coordinates of
Eq.~\eqref{eq:sector}, these are exactly the kernel elements whose
nonzero coefficients have all their centers in $U$.
\end{proposition}
\begin{proof}
Every unchanged tensor is a product over its legs,
Eq.~\eqref{eq:site}. Since no two vertices of $U$ are adjacent, each
leg of a vertex $u\in U$ ends at an unchanged tensor, and the factor
$F_{w,e}$ of that tensor on the shared leg is exactly the factor
entering Eq.~\eqref{eq:localmap}. Contracting the legs of each $u\in U$
separately therefore multiplies the amplitude by $f_u(s)$, and the
remaining factors of the unchanged tensors reassemble $\psi$. Hence
$\Psi(A')=\bigl(\prod_{u\in U}f_u\bigr)\psi=\psi$. The exponential
choice obeys the product constraint for every $t$. Differentiating at
$t=0$ gives the stated velocity, and expanding each $h_u$ in characters
identifies its coefficients with the sector coordinates.
\end{proof}

\begin{corollary}[Single-character lines in the ray fiber]
\label{cor:lines}
Let $S$ be a support with two nonadjacent centers $u,v\in C_S$, and
write $\chi=\chi_{S,\boldsymbol k}$. Then
$f_u=e^{\theta\chi}$, $f_v=e^{-\theta\chi}$ give an exact constant-state
curve with initial velocity $z_{S;u,v}^{\boldsymbol k}$, which is a
non-bond direction whenever $|S|\ge2$. If $\chi^2=1$, the associated
straight line satisfies
\begin{equation}
\Psi(A_0+t\,z_{S;u,v}^{\boldsymbol k})=(1-t^2)\,\Psi(A_0),
\qquad |t|<1.
\label{eq:squareline}
\end{equation}
Thus every individual qubit chordless-square relation gives a straight
line in the fiber of the physical ray. Relations whose centers together
form an independent set compose through the multiplicative construction
of Proposition~\ref{prop:finite}; that family need not be affine linear
in the tensor entries.
\end{corollary}
\begin{proof}
Nonadjacent $u,v\in C_S$ force $u,v\notin S$, so $S$ consists of
common neighbors and $|S|\ge2$ places $z_{S;u,v}^{\boldsymbol k}$ outside
the ordinary gauge image by Proposition~\ref{prop:quotient}.
The straight-line factors at the two centers are $1+t\chi$ and
$1-t\chi$, whose product is $1-t^2\chi^2$; this proves
Eq.~\eqref{eq:squareline} when $\chi^2=1$. In that case
$e^{\pm\theta\chi}=\cosh\theta\pm\sinh\theta\,\chi$, so the
constant-state exponential curve is obtained from the line at
$t=\tanh\theta$ by multiplying each changed tensor by $\cosh\theta$.
These rescalings distinguish the exact-state curve from the ray-preserving
line. Multiplying local factors at any shared center proves the
composition statement.
\end{proof}
The non-bond character of these equivalences follows locally from
orbit geometry, without assuming trivial local tensor stabilizers.
The ordinary bond-gauge orbit through $A_0$ is a smooth locally closed
orbit with tangent space $\im\Phi_{A_0}$. A constant-state curve whose
initial velocity has nonzero class in
$\ker J_{A_0}/\im\Phi_{A_0}$ has velocity outside the orbit tangent and hence lies
outside the orbit for sufficiently small nonzero parameter values. This
applies to the pair curves above when $|S|\ge2$, and to any
independent-center curve with nonzero quotient class.

Appendix~\ref{SM-sec:finiteequivalencedetails} gives a finite slice-rank
invariant for the square curves and an explicit example in which
composition requires nonlinear corrections.

\subsection{Second-order obstructions and the tangent cone}
\label{sec:obstruction}
Relations whose centers are adjacent are not covered by
Proposition~\ref{prop:finite}; for them the second-order expansion of
the contraction map is the first test. Since $\Psi$ is multilinear in
the site tensors, a curve
$A(t)=A_0+t\delta+t^2\delta'+O(t^3)$ has
\begin{equation}
\begin{aligned}
\Psi(A(t))={}&\Psi(A_0)+tJ_A\delta\\
&+t^2\bigl[J_A\delta'+\Psi_2(\delta,\delta)\bigr]+O(t^3),\\
\Psi_2(\delta,\delta)={}&\sum_{u<v}\Psi\bigl(A_0;\,A_u\to\delta_u,\\
&\hphantom{\sum_{u<v}\Psi\bigl(A_0;\,}A_v\to\delta_v\bigr),
\end{aligned}
\label{eq:secondorder}
\end{equation}
For $\delta\in\ker J_A$, a constant-ray curve requires
$\Psi_2(\delta,\delta)\in\im J_A$, since the ray itself lies in
$\im J_A$. Equivalently, the projected Hessian obstruction must vanish:
\begin{equation}
\mathfrak B_A(\delta,\delta)
=2(I-P_{\im J_A})\Psi_2(\delta,\delta)=0.
\end{equation}
Here $P_{\im J_A}$ is the orthogonal projector onto the affine image.
This condition is necessary for integrability, but need not suffice at
higher orders.
\begin{proposition}[Second-order obstructions]
\label{prop:obstruction}
At the $q=2$ graph state on a $3\times3$ open square patch, every
single chordless-square relation $z$ of Eq.~\eqref{eq:square} satisfies
$\Psi_2(z,z)\in\im J_A$, but for two parallel relations $z_1,z_2$ in
edge-adjacent squares, $\Psi_2(z_1+z_2,z_1+z_2)\notin\im J_A$. Hence
the fiber of the physical ray through $A_0$ is not smooth at $A_0$:
it contains curves with initial velocities $z_1$ and $z_2$, but no
analytic curve with initial velocity $z_1+z_2$.
\end{proposition}
\begin{proof}
Exact rational contractions certify that the single-square
second-order terms lie in $\im J_A$, while the mixed term of the
parallel pair does not (Appendix~\ref{SM-sec:obstructionchecks}).
A smooth ray fiber would have a tangent space containing both $z_i$
and their sum, and hence a curve with the latter velocity, contradicting
the second-order condition.
\end{proof}
Jointly independent centers admit the nonlinear corrections of
Proposition~\ref{prop:finite}; the obstructed parallel pair has
adjacent centers. For triangle-based motifs, numerical continuation
finds constant-ray paths for every tested direction, but a general
integrability proof remains open
(Appendix~\ref{SM-sec:obstructionchecks}).

\section{Generic gauge completeness and the WGS rank drop}
\label{sec:generic}
The preceding sections are pointwise in the WGS representation
family. We now compare with independent generic tensors in the full
space $\mathcal P_G$: is the WGS non-bond nullity a rank drop of the
parametrization at a special point, or a dimension defect of the
tensor-network variety itself? The lattice-by-lattice constructions
are in the appendices; this section states the general mechanism and the
result. Generic means on a nonempty Zariski-open subset
for a fixed finite graph and fixed $q$; tensors are not tied by
translation symmetry.

The WGS bond-gauge rank is already maximal. Consequently,
$\nu_{\rm generic}\le\nu_{\rm WGS}$, and their difference equals the
increase in physical differential rank at generic tensors
(Proposition~\ref{SM-prop:genericbounds}). The generic comparison
therefore reduces to establishing gauge completeness on the families
below. The normal-PEPS theorem enters through a statement about fibers,
not by differentiating a set-theoretic equality at a special point.
\begin{lemma}[Gauge fibers imply generic tangent completeness]
Let $X$ be an irreducible affine tensor parameter space with polynomial
contraction $\Psi$ and an algebraic bond-gauge action. Suppose a nonempty
Zariski-open set $U\subset X$ has the following property: any two points
of $U$ with equal contraction are related by bond gauge. If the generic
gauge-orbit dimension is $g_0$, then
\begin{equation}
\begin{gathered}
\rank d\Psi_{\rm generic}=\dim X-g_0,\\
\ker d\Psi_A=\im\Phi_A
\end{gathered}
\label{eq:genericfiber}
\end{equation}
on a possibly smaller nonempty open subset.
\label{lem:fiber}
\end{lemma}
\begin{proof}
Intersect $U$ with the open set of maximal orbit dimension.
Every nonempty fiber of the restricted map lies in one gauge orbit;
its intersection with that orbit is open and has dimension $g_0$.
The fiber-dimension theorem therefore gives
$\dim\overline{\Psi(X)}=\dim X-g_0$, and in characteristic zero the
generic differential rank equals this image dimension
\cite{hartshorne1977algebraic,borel1991groups}. Since gauge tangents
always lie in the kernel, equality of dimensions gives equality of
spaces.
\end{proof}

\begin{theorem}[Generic completeness]
For independent tensors with $D=d=q\ge2$, the thirty-three families
in the first four rows of Table~\ref{tab:genericsummary}, and at
$q=2$ the six families in its fifth row, satisfy
\begin{equation}
\begin{gathered}
\ker J_A=\im\Phi_A,\\
\rank J_A=P-[E(q^2-1)+N-1]
\end{gathered}
\label{eq:genericrank}
\end{equation}
on a nonempty Zariski-open subset of its parameter space, under the
period bounds listed there. Table~\ref{SM-tab:generic} gives the
family-by-family conditions.
\label{thm:generic}
\end{theorem}

\begin{table*}[!tbp]
\centering
\caption{Generic completeness by proof mechanism. Each row lists the
families for which $\nu_{\rm generic}=0$ is established, with the
sufficient period bounds in parentheses ($L$ denotes $L_x,L_y$; cell
conventions as in Appendix~\ref{SM-sec:coordinates}). The complete
family-by-family conditions for all thirty-nine families and the
proof-dependency map are in Tables~\ref{SM-tab:generic}
and~\ref{SM-tab:proofdependencies}.}
\label{tab:genericsummary}
\footnotesize
\setlength{\tabcolsep}{4pt}
\centering
\begin{adjustbox}{max width=\textwidth}
\begin{tabular}{>{\raggedright\arraybackslash}p{4.59cm}>{\raggedright\arraybackslash}p{7.65cm}c}\toprule
Mechanism & Families & $q$\\
\midrule
WGS rank witness ($\nu_{\rm WGS}=0$) & honeycomb ($L\ge5$); Lieb, heavy-hex, heavy-square-octagon (any complete subdivision) & all\\
Injective partitions on the original graph & square ($L\ge20$), kagome ($L\ge64$), snub-square (even $L\ge48$), dice, \texttt{t2016}--\texttt{t2018} ($L\ge30$) & all\\
Injective blocks, single-edge coarse bonds & star, square-octagon, cross (intact blocks); \texttt{t2001} ($L\ge32$) & all\\
Blocks with composite-bond splitting (triangle, dimer, bowtie, diamond, wheel, $K_4$, row) & ruby, maple-leaf, elongated triangular; \texttt{t2002}--\texttt{t2007}, \texttt{t2009}--\texttt{t2013}, \texttt{t2020}; checkerboard, shuriken ($L\ge3$ to $64$); king ($L_x\ge4,L_y\ge18$, or $x\leftrightarrow y$) & all\\
Normal or injective blocks with exact integer minors & \texttt{t2008} ($L\ge3$), \texttt{t2014}, \texttt{t2015} ($L\ge32$), \texttt{t2019} ($3\mid L\ge24$); triangular ($6\mid L\ge54$ or $19\mid L\ge57$); centered-square ($2\mid L_x\ge40$, $6\mid L_y\ge42$, or $x\leftrightarrow y$) & $2$\\
Open & the six qubit rows at $q>2$ & ---\\
\bottomrule
\end{tabular}
\end{adjustbox}
\end{table*}

\begin{proof}[Proof outline]
When $\nu_{\rm WGS}=0$, the WGS point attains the gauge upper bound on
$\rank J$, so a nonzero maximal minor proves generic completeness.
The other cases use injective partitions and blocking: the normal-PEPS
theorem confines constant-state curves to coarse gauge orbits, and
splitting arguments reduce those coarse actions to original-edge
gauges. For the six qubit-only families, exact integer minors certify
the internal block kernels and simultaneous splitting conditions.
Appendix~\ref{SM-sec:genericmechanism} states the required
hypotheses and proves the reduction; the family-specific constructions
and sufficient period bounds follow there.
\end{proof}

The triangular lattice admits no direct three-block edge partition,
but blocking into twelve- and nineteen-site hexagonal clusters gives
$\nu_{\rm generic}=0$ at $q=2$ on the period classes of
Table~\ref{tab:genericsummary}; exact modular certificates also settle
the $4\times3$ and $4\times4$ tori (Appendices~\ref{SM-sec:triangulargeneric} and~\ref{SM-sec:open5qubits}).
For king, blocking entire rows and an analytic simultaneous splitting
of the two row contacts prove the all-$q$ result (Appendix~\ref{SM-sec:kingrows}); no exponentially large contact-rank
calculation is required. The all-$q$ statements for the six qubit rows
remain open. The sufficient period bounds
in Table~\ref{tab:genericsummary} belong to the proof constructions;
the finite-graph WGS theorem applies independently of those bounds.

\begin{corollary}[Exact rank drop at WGS representations]
On every family and dimension covered by Theorem~\ref{thm:generic},
\begin{equation}
\rank J_{\rm generic}-\rank J_{\rm WGS}
=\nu_{\rm WGS},
\label{eq:rankdrop}
\end{equation}
with $\nu_{\rm WGS}$ from Theorem~\ref{thm:regular} or
Eq.~\eqref{eq:cellrank} on locally faithful quotients and from
Eq.~\eqref{eq:extrarank} otherwise.
\end{corollary}
\begin{proof}
Both gauge ranks equal $g_0$, so Eq.~\eqref{SM-eq:generalrankdrop}
reduces to Eq.~\eqref{eq:rankdrop} once $\nu_{\rm generic}=0$.
\end{proof}
The non-bond nullity is thus a rank drop of the parametrization at
the WGS point rather than a dimension defect of the variety. By
Corollary~\ref{cor:normal}, in the infinite-lattice sense with finite
blocks, it occurs both at non-normal points (square, kagome) and at a
normal one (star), so it is not a consequence of non-injectivity.
Generic completeness does not imply that the WGS non-bond
nullity disappears along every path leaving the WGS family.
Section~\ref{sec:projectiongap} gives the condition for quadratic
lifting and paths with persistent nullity.

\section{Physical tangent projection, blocked directions, and the reduced linear system}
\label{sec:tdvp}
Theorem~\ref{thm:star} determines both which physical directions a
tensor representation can reach and which of its parameters are
redundant. These are independent questions, and this section separates
them into three operations on a representation of a fixed physical
state: kernel removal selects one parameter representative per
physical direction, image projection returns the realizable velocity
and its residual, and bond expansion supplies a direction the image
lacks. A final subsection analyzes how the reduced linear system
responds to an inexact force, at and near the WGS point.

At a normalized WGS representation, write
$e_{S,\boldsymbol k}=\chi_{S,\boldsymbol k}\psi$. The orthogonal
projector onto $\T_A$,
\begin{equation}
P_{\T_A}=\sum_{\substack{S\in\K(G)\\S\ne\varnothing}}
\sum_{\boldsymbol k\in\mathcal F_S}
\ket{e_{S,\boldsymbol k}}\bra{e_{S,\boldsymbol k}},
\label{eq:projector}
\end{equation}
does not depend on the choice of parameter-space complement.
For a Hermitian Hamiltonian $H$, define
\begin{equation}
\begin{split}
h(s)&=(H\psi)(s)/\psi(s),\\
h_{S,\boldsymbol k}&=q^{-N}\sum_s
\overline{\chi_{S,\boldsymbol k}(s)}h(s).
\end{split}
\label{eq:energywalsh}
\end{equation}
The conjugate is required at general $q$; for $q=2$ the characters are
real. We use the exact, unregularized orthogonal-projection form of the
time-dependent variational principle (TDVP), with $\hbar=1$
\cite{mclachlan1964variational,haegeman2011tdvp,hackl2020geometry};
because $\T_A$ is a complex-linear subspace, the McLachlan and
Dirac--Frenkel formulations coincide here, which would not be automatic
for a real parametrization.

\begin{corollary}[Instantaneous representability]
The optimal physical variational velocity is
\begin{equation}
\dot\psi_{\rm TDVP}=-i\sum_{\substack{S\in\K(G)\\S\ne\varnothing}}
\sum_{\boldsymbol k\in\mathcal F_S}
h_{S,\boldsymbol k}e_{S,\boldsymbol k}.
\end{equation}
Its squared residual is
\begin{equation}
\begin{split}
R_{\rm TDVP}^2
&=\|(I-P_{\T_A})(H-\langle H\rangle)\psi\|^2\\
&=\sum_{S\notin\K(G)}\sum_{\boldsymbol k\in\mathcal F_S}
|h_{S,\boldsymbol k}|^2.
\end{split}
\label{eq:residual}
\end{equation}
The exact instantaneous Schr\"odinger direction is representable if and
only if every unsupported Fourier coefficient vanishes.
\label{cor:tdvp}
\end{corollary}
\begin{proof}
The vectors $e_{S,\boldsymbol k}$ over all supports form an orthonormal
basis. Subtraction of $\langle H\rangle$ removes the empty support.
Orthogonal projection keeps precisely the remaining allowed supports,
and Parseval's identity gives the residual.
\end{proof}
This is a WGS expression for the local-in-time variational error
\cite{martinazzo2020error}; terms producing the same character are
summed before squaring. For positive energy variance we also use
$\epsilon_{\rm TDVP}=R_{\rm TDVP}/\sqrt{\Var_\psi(H)}$.

Since $\T_A$ contains the centered deformation
$(f-\langle f\rangle_\psi)\psi$ for every function $f$ on a closed
neighborhood, the residual criterion has a purely local form.
\begin{corollary}[Support criterion]
\label{cor:support}
Let $H=\sum_iH_i$ be such that, for each term, the amplitude ratio
$(H_i\psi)(s)/\psi(s)$ depends only on the variables in some set
$U_i\in\K(G)$, i.e., $U_i\subseteq N[v_i]$ for some vertex $v_i$. Then
$R_{\rm TDVP}(H)=0$ at every nondegenerate WGS representation on $G$.
In particular this holds for every operator diagonal in the
computational basis whose support lies in one closed neighborhood,
in particular every diagonal edge operator, and for every single-site
operator, whose amplitude
ratio is a function on the closed neighborhood of its site.
\end{corollary}
\begin{proof}
A function of the variables in $U_i$ has Fourier support inside
$2^{U_i}\subseteq\K(G)$, and $\K(G)$ is downward closed, so every
coefficient of $h$ outside $\K(G)$ vanishes in Eq.~\eqref{eq:residual}.
For a one-site operator $O_v$, $(O_v\psi)(s)/\psi(s)$ is a combination
of $\psi(s')/\psi(s)$ over configurations $s'$ differing from $s$ only
at $v$, and this ratio involves only the edges at $v$.
\end{proof}

For a diagonal operator the exact condition is that each nonzero
Fourier component has support in $\K(G)$: $\sum_vZ_v$ is tangent on
any qubit WGS graph, whereas $Z_aZ_b$ is not if
$\operatorname{dist}(a,b)>2$.

The criterion of Corollary~\ref{cor:support} is sufficient but not
necessary: a term whose ratio lies outside $\K(G)$ can still have
vanishing unsupported coefficients, or be fully or partially blocked;
the $XX$ example below exhibits both regimes as the edge phase varies.

\subsection{Coordinates, kernel removal, and the image}
\label{sec:liftcoords}
Choose $\alpha_{vS}$ with $\sum_{v\in C_S}\alpha_{vS}=1$.
The map
\begin{equation}
\Xi c=\sum_{\substack{S\in\K(G)\\S\ne\varnothing}}
\sum_{\boldsymbol k\in\mathcal F_S}c_{S,\boldsymbol k}
\sum_{v\in C_S}\alpha_{vS}\xi_{v,S}^{\boldsymbol k}
\label{eq:lift}
\end{equation}
is an explicit local lift with
$J_\perp\Xi c=\sum_{S,\boldsymbol k}c_{S,\boldsymbol k}e_{S,\boldsymbol k}$,
so the physical metric in these tangent coordinates is the identity
at the evaluation point. At the qubit graph-state point the same lift gives the projected
first-order ground-state response to a perturbation of the stabilizer
parent Hamiltonian; its infidelity is governed by the unsupported
coefficients of Eq.~\eqref{eq:residual}, weighted by inverse excitation
energies (Appendix~\ref{SM-sec:groundresponse}).

The projective QGT at $A$ is $\mathcal M_A=J_{\perp,A}^\dagger J_{\perp,A}$,
written $\mathcal M$ when the base point is fixed, so its kernel is exactly
$\ker J_\perp$~\cite{provost1980riemannian,haegeman2014geometry}.
Restricting to a complement of the complete kernel selects a unique
parameter representative of each accessible physical direction; it
leaves Eq.~\eqref{eq:projector} unchanged and cannot reduce
Eq.~\eqref{eq:residual}. A finite invertible bond gauge likewise
preserves the physical image by an invertible parameter change,
although it can change Euclidean parameter-space condition numbers.
Appendix~\ref{SM-sec:triangleidentifiability} gives the complementary
example of a full physical tangent with redundant tensor coordinates.

\subsection{Representable and blocked directions}
\label{sec:blockade}
\paragraph{Transverse-field Ising is instantaneously tangent.}
As a first application of Corollary~\ref{cor:support}, specialize to
$q=2$ and consider
\begin{equation}
H_{\rm I}=-\sum_{uv\in E_G}J_{uv}Z_uZ_v-\sum_v(h^x_vX_v+h^z_vZ_v).
\label{eq:ising}
\end{equation}
A diagonal edge term has amplitude ratio $\chi_{\{u,v\}}$, whose support
lies in both endpoint closed neighborhoods, and flipping the bit at $v$
gives
\begin{equation}
\frac{(X_v\psi)(s)}{\psi(s)}
=\exp\!\left[i(1-2s_v)\sum_{w\in N(v)}\phi_{vw}s_w\right],
\label{eq:xenergy}
\end{equation}
a function of $N[v]$. Hence $R_{\rm TDVP}(H_{\rm I})=0$ on every finite
graph, for arbitrary real couplings and fields, at every nondegenerate
WGS representation.

\paragraph{Complete $XX$ blockade at the graph state.}
The $XX$ term is the simplest operator outside the scope of
Corollary~\ref{cor:support}: for adjacent $u,v$ its amplitude ratio is
a function on $N[u]\cup N[v]$, which is not contained in any single
closed neighborhood on the lattices below, so the outcome depends on
which Fourier coefficients actually appear. Set every edge phase to
$\pi$. The resulting graph state obeys the stabilizer equations
\begin{equation}
K_v\psi=\psi,\qquad K_v=X_v\prod_{w\in N(v)}Z_w,
\label{eq:stabilizer}
\end{equation}
which also follow directly from Eq.~\eqref{eq:xenergy}
\cite{hein2004multiparty,hein2006graph}. For adjacent $u,v$, moving
$X_u$ through the $Z_u$ in the second stabilizer introduces a minus
sign:
\begin{equation}
X_uX_v\psi=-\chi_{B_{uv}}\psi,\qquad
B_{uv}=N(u)\mathbin{\triangle}N(v),
\label{eq:xxsupport}
\end{equation}
with open neighborhoods, so both endpoints belong to $B_{uv}$.
On all eleven bulk Archimedean lattices $|B_{uv}|\ge z+1$ for
coordination number $z$, and $B_{uv}$ has no center, so it lies outside
$\K$ (Lemma~\ref{SM-lem:xxcenter}). For
\begin{equation}
H_{XX}=\sum_{uv\in E_G}X_uX_v
\label{eq:xx}
\end{equation}
all Fourier components lie outside the tangent space and, even if two
edges produce the same support, their equal negative coefficients add.
It follows that
\begin{equation}
\begin{gathered}
P_{\T_A}H_{XX}\psi=0,\\
R_{\rm TDVP}^2=\Var_\psi(H_{XX})>0,\quad
\epsilon_{\rm TDVP}=1.
\end{gathered}
\label{eq:freeze}
\end{equation}
For a deterministic minimum-norm parameter update, zero forcing selects
zero tensor velocity, with or without a ridge, so the graph state is an
exact stationary solution of the projected equation for \emph{all}
times, while for a single edge with unit coupling the exact survival
probability is $\cos^2t$: the projected evolution misses an infidelity
$\sin^2t$ that can be of order one at times of order one.

The obstruction is not a kernel effect: honeycomb has $\nu=0$, so no
coordinate null direction remains to be removed, yet the blockade is
complete. Gauge reduction selects representatives of tangent vectors
\cite{haegeman2014geometry} but cannot add an absent direction; this is
a support-level instance of the TDVP projection error
\cite{haegeman2011tdvp,haegeman2016unifying,hackl2020geometry}.

\paragraph{State-level obstructions.}
A missing velocity excludes a differentiable lift through the specified
tensor, but maximal differential rank there alone does not establish
smoothness of the state variety. The following Schmidt bound excludes
all bond-two representations, independently of normality.
\begin{corollary}[State-level honeycomb obstruction]
\label{cor:honeycombstate}
Let $G$ be a honeycomb torus with $L_x,L_y\ge5$ in the conventions of
Appendix~\ref{SM-sec:coordinates}, and let $uv$ be any nearest-neighbor edge.
For $0<|\varepsilon|<\pi/2$, the state
$\ket{\psi_\varepsilon}=e^{-i\varepsilon X_uX_v}\ket G$ has no PEPS
representation with every bond dimension at most two on the same graph.
Moreover, if $\mathcal P_2(G)$ denotes the normalized states admitting
such a representation, then for all real $\varepsilon$
\begin{equation}
\max_{\phi\in\mathcal P_2(G)}
|\langle\phi|\psi_\varepsilon\rangle|^2
=\max\{\cos^2\varepsilon,\sin^2\varepsilon\}.
\label{eq:honeycombfidelity}
\end{equation}
\end{corollary}
\begin{proof}
For every edge $uv$, Appendix~\ref{SM-sec:schmidtobstruction} gives
an eight-edge matching cut through $uv$, with a same-side neighbor
incident to no cut edge at each endpoint. Removing internal controlled-$Z$
gates leaves eight maximally entangled pairs; the conjugated $X_uX_v$ flips a free
$\ket+$ site to $\ket-$ on each side. The two components therefore
have orthogonal Schmidt supports, with coefficients
$|\cos\varepsilon|/16$ and $|\sin\varepsilon|/16$, each repeated
$256$ times. Their combined rank is $512$ when both are nonzero,
exceeding the bond-two cut bound $256$. The largest $256$
coefficients give Eq.~\eqref{eq:honeycombfidelity}, attained by
$\ket G$ or $X_uX_v\ket G$, whichever has the larger overlap.
\end{proof}
A cylinder cut proves the corresponding state-level obstruction on
the square torus.
\begin{corollary}[State-level square-lattice obstruction]
\label{cor:squarestate}
Let $G$ be an $L_x\times L_y$ square torus with $L_x,L_y\ge6$ and let
$uv$ be any nearest-neighbor edge. Then the conclusions of
Corollary~\ref{cor:honeycombstate}, including
Eq.~\eqref{eq:honeycombfidelity}, hold for $G$.
\end{corollary}
\begin{proof}
For horizontal $uv$, a cylinder of $\lfloor L_x/2\rfloor$ columns
ending at $u$ has $2L_y$ crossing edges forming a matching. Each
endpoint has a same-side neighbor away from the cut. Lemma~\ref{SM-lem:matchingcut} gives rank $2^{2L_y+1}$ against the
bond-two bound $2^{2L_y}$ and the same fidelity. Exchanging axes
covers vertical edges; Appendix~\ref{SM-sec:squarecut} gives the coordinates.
\end{proof}
The argument applies wherever these cut conditions hold. Kagome and
triangular lattices have no nontrivial matching cut, because every
edge lies in a triangle. Their state-level exclusion remains open;
tangent exclusion alone does not decide it.

\paragraph{$YY$ terms and the Heisenberg velocity.}
On the same eleven bulk tilings, or locally faithful tori, the
stabilizer algebra gives, for adjacent $u,v$,
\begin{equation}
\begin{gathered}
Y_uY_v\psi=-Z_uZ_vX_uX_v\psi=\chi_{B'_{uv}}\psi,\\
B'_{uv}=B_{uv}\setminus\{u,v\}.
\end{gathered}
\label{eq:yysupport}
\end{equation}
where $B'_{uv}$ is the set of neighbors exclusive to one endpoint. The same
center-exclusion argument, using Lemma~\ref{lem:shortcycles}(i) and
an inspection of every edge orbit of the eleven tilings, shows that
$B'_{uv}$ has no center either (Appendix~\ref{SM-sec:yysupport}). Hence
nearest-neighbor $YY$ terms are blocked exactly as $XX$ terms are, and
at a graph state on any of the eleven lattices the projected velocity
of a Heisenberg Hamiltonian
$\sum_{uv}(J_xX_uX_v+J_yY_uY_v+J_zZ_uZ_v)$ is that of its $ZZ$ part
alone: $\dot\psi_{\rm TDVP}=-iJ_z\sum_{uv}(Z_uZ_v-\langle Z_uZ_v\rangle)\psi$.
The $ZZ$ part is diagonal and moves the state within the WGS family,
so the projected Heisenberg flow starting at a graph state initially
follows a pure Ising phase evolution; the $XX$ and $YY$ parts enter
only once the edge phases have left $\pi$, through the partial
residuals below.

\paragraph{A twelve-qubit trajectory.}
Figure~\ref{fig:dynamics} compares exact and projected evolution of a
twelve-qubit periodic kagome graph state under $H_{XX}$. Integrating
the projected equation for all independent $D=2$ tensor entries
leaves the state stationary, while the exact state decays rapidly.
Enlarging the family to
\begin{equation}
\ket{\psi(\boldsymbol\alpha)}=\prod_{uv\in E_G}e^{-i\alpha_{uv}X_uX_v}\ket{\psi_G},
\label{eq:gatefamily}
\end{equation}
supplies every missing edge direction $-iX_uX_v\psi_G$.
Because the $XX$ terms commute, this family contains the exact
trajectory, which the projected equation reproduces. Each gate can be
absorbed into the bond tensors at bond dimension at most four;
numerical checks are in Appendix~\ref{SM-sec:trajectorydetails}.

\begin{figure}[!tbp]
\centering
\includegraphics[width=\columnwidth]{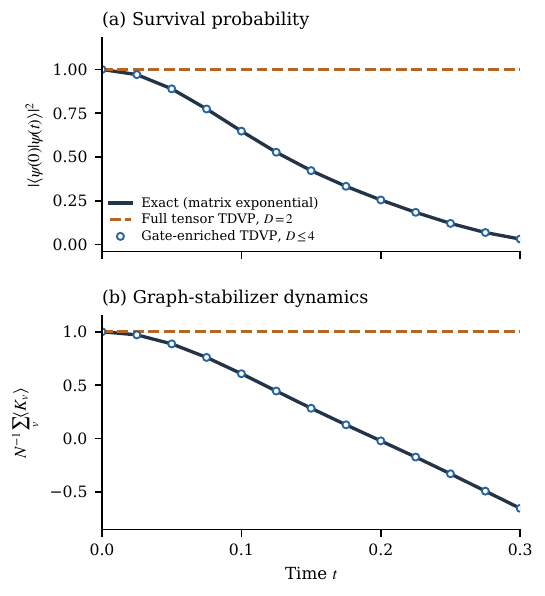}
\caption{Time evolution of the twelve-qubit periodic kagome graph
state under $H_{XX}$ with unit couplings and $\hbar=1$. (a) Survival
probability. (b) Mean graph stabilizer $N^{-1}\sum_v\langle K_v\rangle$.
The exact evolution (dark curves) departs rapidly, while projection
onto the tangent of all independent $D=2$ tensors (orange dashed)
leaves the state stationary. The augmented family of
Eq.~\eqref{eq:gatefamily} (blue circles) supplies the missing
$-iX_uX_v\psi$ directions and reproduces the exact trajectory.
All expectation values use exact Hilbert-space statistics;
numerical checks are in Appendix~\ref{SM-sec:trajectorydetails}.}
\label{fig:dynamics}
\end{figure}

\paragraph{General edge phase.}
The complete blockade is a property of the graph-state point.
At a uniform nondegenerate phase $\phi$ the single-edge ratio
$(X_uX_v\psi)/\psi$ has Fourier weight both inside and outside $\K$,
and Eq.~\eqref{eq:residual} gives its residual in closed form.
Writing $c=\cos\phi$, on the bulk honeycomb lattice
\begin{equation}
\begin{split}
R_{\rm TDVP}(X_uX_v)&=\frac{(1-c)(3+c)}{4}\\
&=\sin^2\tfrac\phi2\,\bigl(1+\cos^2\tfrac\phi2\bigr),
\end{split}
\label{eq:xxphi}
\end{equation}
so complete blockade at $\phi=\pi$ gives way to a partial residual
away from the graph-state point. The square and triangular closed
forms, limiting behavior, and scans on all eleven tilings are given in
Appendix~\ref{SM-sec:xxscan}.

\paragraph{A missing direction witnessed by an observable.}
\label{sec:mbqccase}
A blocked direction is not only a residual: it is a first-order change
of a local observable that no tensor variation can produce. Consider
the six-qubit graph state $\ket{G_\diamond}$ on a diamond with two leaves,
$V=\{a,b,c,d,x,y\}$, $E_G=\{ab,ac,ad,bc,bd,ax,by\}$, with its
nondegenerate $q=2$ WGS representation $A_0$. The leaves $x,y$ lie in
no common closed neighborhood, so
$\K(G)=\{S\subseteq V:\{x,y\}\not\subseteq S\}$. Let
$\ket{G_{\diamond,\epsilon}}=e^{-i\epsilon X_aX_b}\ket{G_\diamond}$ and
$M=-Y_bZ_cZ_dZ_x$.
\begin{corollary}[A missing tangent witnessed by an observable]
\label{cor:observableresponse}
At $A_0$, the projective velocity
$\ket t=\partial_\epsilon\ket{G_{\diamond,\epsilon}}|_0=iZ_aZ_bZ_xZ_y\ket{G_\diamond}$ is a
unit vector orthogonal to $\T_{A_0}$. Every projective tensor
variation gives
$\delta\langle M\rangle=2\operatorname{Re}\langle t|\delta\psi\rangle=0$,
whereas along the curve
$\langle G_{\diamond,\epsilon}|M|G_{\diamond,\epsilon}\rangle=\sin(2\epsilon)$,
with slope $2$ at $\epsilon=0$.
\end{corollary}
\begin{proof}
Eq.~\eqref{eq:xxsupport} gives $X_aX_b\ket{G_\diamond}=-Z_aZ_bZ_xZ_y\ket{G_\diamond}$,
whose support contains both leaves, so Theorem~\ref{thm:star} gives
$t\perp\T_{A_0}$. The stabilizers give $M\ket{G_\diamond}=\ket t$ and
$\langle G_\diamond|M|G_\diamond\rangle=0$; since $(X_aX_b)^2=M^2=I$, the exact curve is
$\cos\epsilon\ket{G_\diamond}+\sin\epsilon\ket t$.
\end{proof}
Removing coordinate null directions leaves this obstruction unchanged.
The curve $\ket{G_{\diamond,\epsilon}}$ is the state produced by a coherent preparation
error on the shared edge of a measurement-based Hadamard primitive
\cite{raussendorf2001oneway,gross2007novel}; the protocol is given in
Appendix~\ref{SM-sec:mbqcprotocol}.

\subsection{State-preserving repair by bond expansion}
\label{sec:asymmetricrepair}
Removing parameters cannot restore a missing velocity; the
representation must be enlarged. This can be done for a specified edge
direction while keeping the physical state fixed. On $u\to v$,
replace its virtual label $a$ by $(a,b)$ with $b=0,1$, and set
\begin{equation}
\begin{aligned}
\widetilde A_u^{(0)}&=A_u,&\widetilde A_u^{(1)}&=0,\\
\widetilde A_v^{(0)}&=A_v,&\widetilde A_v^{(1)}&=X_vA_v.
\end{aligned}
\label{eq:asymmetricrepair}
\end{equation}
The superscript denotes the new virtual channel; $X_v$ acts on the
physical leg. The new-channel contraction vanishes at the base point,
so the unnormalized state is unchanged. Old-channel variations retain
every old derivative, while
$\delta\widetilde A_u^{(1)}=-iX_uA_u$ generates
$-iX_uX_v\Psi(A)$ and hence its projective velocity after normalization.
Thus the enlarged full tensor tangent contains both $\T_A$ and this
edge direction; its optimal residual cannot increase and is zero for
the single-edge Hamiltonian $X_uX_v$. Zero-padding both endpoints would
leave every new-channel first derivative zero and would not provide this
repair. This is an exactly solvable instance of the subspace- and
bond-expansion steps of single-site DMRG and TDVP
\cite{hubig2015subspace,yang2020tdvp,gleis2023cbe,li2024cbe,vanderstraeten2019tangent},
where the new channel is likewise seeded by the projected Hamiltonian
action rather than by zeros; here the exact projector identifies which
direction is missing. Only the selected bond dimension doubles.

\subsection{The reduced linear system at and near a WGS point}
\label{sec:projectiongap}
\begin{figure*}[!tbp]
\centering
\includegraphics[width=\textwidth]{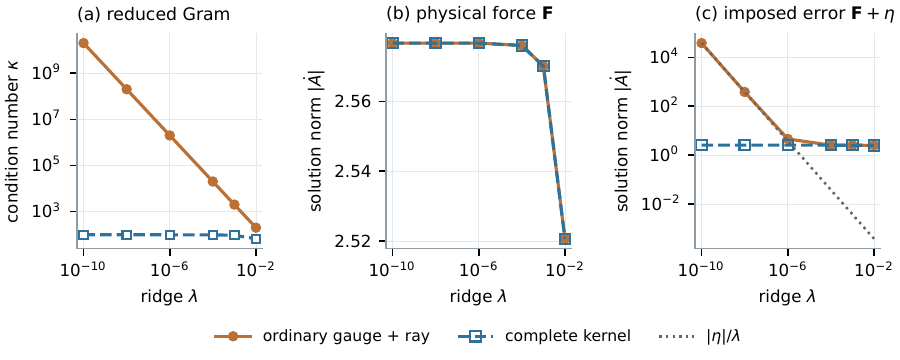}
\caption{Residual singularity and force compatibility on the open
$3\times3$ WGS square grid. (a) Condition number of the reduced Gram
matrix with ridge $\lambda$ after ordinary gauge and ray removal
(orange) or complete kernel removal (blue). In (b,c), $\dot A$ denotes
the tensor-coordinate solution lifted back to the raw parameter space,
with its Euclidean norm. (b) Solution norms for the
compatible force $\boldsymbol F=J_\perp^\dagger(H-\langle H\rangle)\psi$ with
$H=\sum_vZ_v$: both reductions remain bounded and agree to roundoff.
(c) The same force with an imposed error $\eta$ in the residual
non-bond kernel, $\|\eta\|=10^{-6}\|\boldsymbol F\|$: the gauge-only reduction
amplifies its null component as $\|\eta\|/\lambda$ (dotted), while
complete kernel removal eliminates it. The factor $-i$ of real-time
TDVP does not affect the norms. Numerical procedures and independent
checks are described in Appendix~\ref{SM-sec:reproducibility}.}
\label{fig:projectiongap}
\end{figure*}
Removing ordinary bond gauge and the ray leaves exactly $\nu$ zero
modes in the corresponding compression of the Gram matrix
$\mathcal M$ at a WGS point. The projection of Wu and Nys
\cite{wu2025peps} removes independent bond generators and the global
rescaling, which is exact for those directions; at points with
$\nu>0$ the remaining system still contains the $\nu$ non-bond zero
modes, so a Cholesky factorization needs either the complete kernel
removal below or regularization. Their additional
fixed-particle-number rescaling has no counterpart for the spin
systems considered here.

We verify this on the open $3\times3$ qubit square grid with four
chordless squares and independent nondegenerate edge phases, by dense
contraction over all $2^9$ configurations. Of $P=128$ complex
parameters, ordinary gauge and the ray remove $45$; the finite-graph
sector count of Proposition~\ref{prop:quotient} gives $\nu=8$, two
relations per square. The $83$-dimensional reduced Gram has eight
eigenvalues at roundoff scale and its unregularized Cholesky
factorization fails, whereas removing the complete kernel leaves a
$75$-dimensional positive-definite matrix with smallest eigenvalue
$0.020$.

At the exact WGS point the residual zero modes do not amplify a
compatible force. For a physical velocity source $b$ the exact
force is $\boldsymbol F=J_\perp^\dagger b\in\im\mathcal M=(\ker\mathcal M)^\perp$, so
\begin{equation}
\begin{gathered}
P_{\ker\mathcal M}(\mathcal M+\lambda I)^{-1}\boldsymbol F=0,\\
\lim_{\lambda\to0^+}(\mathcal M+\lambda I)^{-1}\boldsymbol F=\mathcal M^+\boldsymbol F,
\end{gathered}
\label{eq:compatibleforce}
\end{equation}
and the same holds in an orthonormal complement of ordinary gauge and
the ray: exact zero modes do not cause a $1/\lambda$ divergence, and
complete kernel removal and a consistent pseudoinverse give the same
physical projection. If an error $\eta$ with a component in the
residual kernel is added to the force, however,
\begin{equation}
P_{\ker\mathcal M}(\mathcal M+\lambda I)^{-1}(\boldsymbol F+\eta)
=\frac{P_{\ker\mathcal M}\eta}{\lambda},
\label{eq:incompatibleforce}
\end{equation}
a component that is invisible to the instantaneous physical map, since
$J_\perp P_{\ker\mathcal M}\eta=0$, and whose effect on a finite parameter step
is not determined by the parameter norm alone.

Figure~\ref{fig:projectiongap} illustrates both cases for
$H=\sum_vZ_v$, which is tangent by Corollary~\ref{cor:support}.
Finite sampling alone does not produce an incompatible force when
metric and force are formed from the same derivative samples (Appendix~\ref{SM-sec:sampledforce}). The same distinction underlies the
treatment of singular quantum geometric tensors in natural-gradient and
minimum-step stochastic-reconfiguration methods
\cite{stokes2020qng,chen2024minsr,rende2024minsr}: a singular metric
can be regularized or inverted on the range, but no such operation
enlarges the image.

\begin{figure}[!tbp]
\centering
\includegraphics[width=\columnwidth]{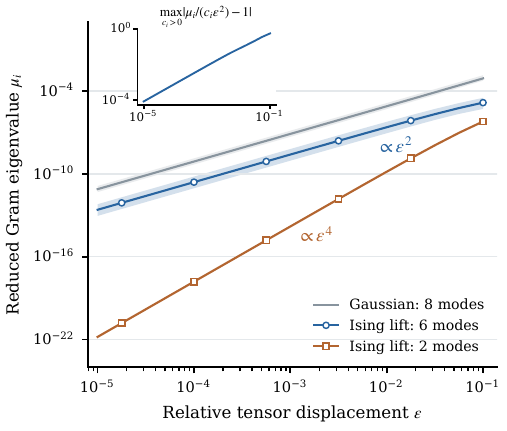}
\caption{Path-dependent opening on the open $3\times3$ WGS square
graph, with the current gauge/ray complement recomputed throughout.
Bands span six quadratic and two numerically quartic modes for the
normalized minimum-norm Ising lift, and eight quadratic modes for one
Gaussian direction (lifts defined in
Appendix~\ref{SM-sec:nearbygramproof}). Lines are geometric means; the inset tests the
six positive Hessian coefficients. The exact $ZZ$ path retains eight
zeros (below $10^{-28}$ numerically). Here $\varepsilon$ is relative
tensor displacement along each
straight line, not physical time.}
\label{fig:nearbygram}
\end{figure}

\paragraph{Lifting the zero modes off the WGS point.}
At $A_0$, the reduced Gram has $\nu$ exact zero modes. Fix the
Euclidean metric on tensor entries and a $C^2$ path
$A(\varepsilon)=A_0+\varepsilon\Delta+O(\varepsilon^2)$ through a
nondegenerate WGS representation, with real $\varepsilon$.
This parameter specifies tensor displacement; the phase-aligned state
distance is $|\varepsilon|\|J_{\perp,A_0}\Delta\|+O(\varepsilon^2)$,
so a first-order conversion requires $J_{\perp,A_0}\Delta\ne0$.
The gauge and ray space $\mathcal G_{A(\varepsilon)}$ has locally
constant dimension, since its WGS rank is maximal. Choose a smooth
orthonormal complement $Q(\varepsilon)$ and set
\begin{equation}
\begin{gathered}
\widetilde J(\varepsilon)=J_{\perp,A(\varepsilon)}Q(\varepsilon),\\
L_\Delta=(I-P_{\mathcal T_{A_0}})\widetilde J'(0)Z_0,
\end{gathered}
\label{eq:openingmap}
\end{equation}
where $Z_0$ spans $\ker\widetilde J(0)$ orthonormally. The projection removes
changes already accessible at $A_0$.
\begin{proposition}[Quadratic opening of the reduced Gram]
\label{prop:nearbygram}
Let $\nu=\dim\ker\widetilde J(0)$ and let $c_1\le\cdots\le c_\nu$ be the
eigenvalues of $L_\Delta^\dagger L_\Delta$. The $\nu$ eigenvalues of
$\widetilde J(\varepsilon)^\dagger\widetilde J(\varepsilon)$ tending to zero satisfy
\begin{equation}
\mu_i(\varepsilon)=c_i\varepsilon^2+O(|\varepsilon|^3).
\label{eq:nearbygram}
\end{equation}
If $\nu>0$ and $L_\Delta$ has full column rank, the reduced condition
number scales as $\varepsilon^{-2}$. If $c_i=0$, the corresponding
eigenvalue is $O(\varepsilon^4)$ and may remain exactly zero.
\end{proposition}
\begin{proof}
Split the domain and codomain by the kernel and image of $\widetilde J(0)$.
Left and right unitary transformations remove the off-diagonal blocks
to first order, leaving small singular values
$|\varepsilon|\sigma_i(L_\Delta)+O(\varepsilon^2)$.
Squaring proves the claim. The block form needs $\widetilde J(0)\ne0$,
which holds because the WGS projective image is nonzero.
Appendix~\ref{SM-sec:nearbygramproof} gives the block calculation.
\end{proof}

\paragraph{Physical paths and their lifts.}
Evolution under edge $ZZ$ and on-site $Z$ terms only changes WGS
edge and local phases, so along this WGS lift $\nu$ stays constant
while the weights are nondegenerate: gauge-only Gram singularity
persists, although compatible forces still admit bounded physical
projections. This does not extend to general many-body diagonal
Hamiltonians. A physical velocity does not fix its lift, and lifts
differing by $z\in\ker J_{\perp,A_0}$ can have different opening
coefficients (Eq.~\eqref{SM-eq:openinghessian}). On the open
$3\times3$ square graph (Fig.~\ref{fig:nearbygram}), the Euclidean
minimum-norm lift of a transverse-field Ising velocity opens six modes
quadratically and two numerically as $\varepsilon^4$, the exact WGS
$ZZ$ path retains all eight zeros, and random Gaussian directions open
all eight quadratically. For a lifted mode, ridge regularization
amplifies an independent force error by $1/(\lambda+\mu_i)$, with
crossover $|\varepsilon|\sim\sqrt{\lambda/c_i}$. Removing current exact
null directions preserves the image; discarding modes that have become
positive changes the projection. Constructions, coefficients, and
protocols are in Appendix~\ref{SM-sec:nearbygramproof}.

\section{Discussion}
The results (i)--(vi) listed in the Introduction separate two questions
that are often merged in PEPS practice. The kernel side concerns coordinates. Its
non-bond part arises from overlapping closed neighborhoods, is local and
explicitly enumerable, and is a rank drop of the parametrization on
all thirty-nine families under the dimension and period conditions of
Theorem~\ref{thm:generic}. It can be removed sector by sector without a
global orthogonal complement; doing so regularizes the reduced linear
system without changing any physical velocity. The image side concerns
expressivity. Honeycomb has no non-bond kernel, yet it misses every
nearest-neighbor $XX$ direction. On the honeycomb and square tori of
Corollaries~\ref{cor:honeycombstate} and~\ref{cor:squarestate}, no
bond-two representation on the same graph reaches the state perturbed
by $e^{-i\varepsilon X_uX_v}$ for $0<|\varepsilon|<\pi/2$. No gauge
fixing, kernel removal, or regularization addresses such an
obstruction. Only an enlarged representation does, and the exact
projector identifies which direction to add.

For the finite-equivalence theory, the pair equivalences with
nonadjacent centers---a diagonal multiplication and its inverse at two
sites---are explicit non-bond equivalences that the semi-injective
theorem~\cite{molnar2018semiinjective} must accommodate, and the
second-order obstruction shows that the ray fiber at a graph state is
not smooth. For computation, graph states supply exactly solvable
benchmarks: zero modes whose opening depends on the tensor path, and
missing directions with a known state-preserving repair. Both are
controlled tests for regularization and bond-expansion strategies in
PEPS evolution.

Natural extensions include unequal physical and bond dimensions,
edge weights with zero entries such as the copy-tensor point $W=I$,
and imposed point-group symmetries, as well as an integrability
theorem for the triangle-based relations and generic completeness of
the six qubit-only families at $q>2$. The face-order formula also holds
for centered-square, beyond the unit-edge hypotheses of
Theorem~\ref{thm:regular} (Appendix~\ref{SM-sec:periodiccatalogue});
characterizing the larger graph class is a combinatorial question.

\section*{Acknowledgments}
The author thanks Yantao Wu and Yi-Yu Lin for valuable discussions, and acknowledges
the use of the language models GPT-6 Astra and GPT-5.6 Sol for
assistance with the details of mathematical proofs, and Claude Opus~5
for assistance with manuscript preparation. All proofs were verified
independently by the author. This work was
supported by the National Natural Science Foundation of China (Grant
No.~12004205 and No.~12574253).

\section*{Code and data availability}
The proof certificates, their independent checkers, and the scripts and
data behind every figure are archived on Zenodo at
\url{https://doi.org/10.5281/zenodo.22934461} (code under the MIT license,
certificates and data under CC BY 4.0), as described in
Appendix~\ref{SM-sec:reproducibility}.

\bibliographystyle{quantum}
\bibliography{gauge,gauge_the_extension}

\clearpage
\appendix
\counterwithin{equation}{section}

\twocolumn[{\begin{center}\sffamily\LARGE Appendices\end{center}
\vspace{-0.6em}\hrule\vspace{1.2em}}]
\phantomsection
\addcontentsline{toc}{section}{Appendices}

\paragraph*{Guide to the appendices.}
The appendices are grouped by purpose. Proofs and supplements for
Secs.~\ref{sec:bond}--\ref{sec:fiber} are in
Appendices~\ref{SM-sec:stabilizerdetails} (bond stabilizers and local
unitaries), \ref{SM-sec:latticeproofs} (regular-polygon and
periodic-cell proofs), \ref{SM-sec:sminimal} (open-region kernels and
minimality), and \ref{SM-sec:fiberdetails} (normality and finite
equivalences). The graph census and its conventions are in
Appendices~\ref{SM-sec:twouniformSM} (2-uniform tilings),
\ref{SM-sec:periodiccatalogue} (eight further periodic graphs), and
\ref{SM-sec:coordinates} (periodic conventions). The Euclidean QGT
spectrum and the physical applications of Sec.~\ref{sec:tdvp} are in
Appendices~\ref{SM-sec:positive_spectrum}
and~\ref{SM-sec:groundresponse}--\ref{SM-sec:triangleexperiment}. The
generic-completeness proofs of Sec.~\ref{sec:generic} are in
Appendix~\ref{SM-sec:genericappendix}, with their integer and routing
certificates in Appendix~\ref{SM-sec:certificates}.
Appendix~\ref{SM-sec:reproducibility} states what is proved
analytically and what is machine-verified, and how the checks are
replayed.

\section{Bond stabilizers and covariance under local physical unitaries}
\label{SM-sec:stabilizerdetails}
\subsection{Nonzero edge weights and compatible bond stabilizers}
Generic stabilizer results provide a related algebraic baseline
\cite{bernardi2023dimension}. Proposition~\ref{prop:bond}
also fixes the gauge rank at every nondegenerate WGS point, where the
physical differential can have a non-bond kernel.
The hypotheses matter here. Invertible weights with
zero entries, such as $W=I$, need not obey the same coefficient
conclusion: with $W=I$ the site tensor is a copy tensor, on which every
diagonal virtual matrix acts as the same diagonal on the physical leg,
so diagonal generators obeying a level-by-level scalar cycle condition
act trivially on all tensors and $\ker\Phi_A$ has dimension
$q(E-N+1)$. This is the mechanism
associated with virtual symmetries of copy tensors. A local tensor
stabilizer must, however, be distinguished from a compatible bond-gauge
stabilizer of the entire network. At a site of degree at least two,
write the factors on two legs as matrices $F_1,F_2$ whose rows are
$f_{s,1},f_{s,2}$. For any invertible diagonal matrix
$D=\operatorname{diag}(\lambda_0,\ldots,\lambda_{q-1})$, choose
\begin{equation}
U_1=F_1^{-1}DF_1,\qquad U_2=F_2^{-1}D^{-1}F_2.
\end{equation}
Then $f_{s,1}U_1=\lambda_s f_{s,1}$ and
$f_{s,2}U_2=\lambda_s^{-1}f_{s,2}$. With identity on the remaining
legs, these actions fix every physical slice even when $D$ is
nonscalar. Thus a nondegenerate WGS tensor can have nontrivial local
multi-leg stabilizers. Proposition~\ref{prop:bond}
uses the additional compatibility of the same bond generator at both
endpoints: the nonzero pairings $W_e(s,t)$ force its eigenvalues to
coincide. It establishes that the network's infinitesimal bond-gauge
stabilizer consists only of scalar cycles, not that each local tensor
has a trivial stabilizer. The non-bond quotient is defined relative
to the image of this compatible bond action.

\subsection{Local physical unitaries}
\label{SM-sec:localunitaries}
The geometric conclusions transport under fixed local physical unitaries.
Let $U=\bigotimes_vU_v$ and let $\mathcal U$ act on each tensor's physical
leg, so $A'=\mathcal U A$. Differentiating
$\Psi(\mathcal U A)=U\Psi(A)$ gives
$J_{\perp,A'}\mathcal U=UJ_{\perp,A}$ and hence
$\T_{A'}=U\T_A$; ranks and non-bond nullities are unchanged.
The Fourier vectors transport as $U(\chi_{S,\boldsymbol k}\psi)$, rather
than by reusing the diagonal-function formula in the rotated basis.
The projection residual is unchanged when the Hamiltonian is also
transported to $UHU^\dagger$, not for an arbitrary fixed Hamiltonian.

\section{Complete 2-uniform census}
\label{SM-sec:twouniformSM}
The twenty 2-uniform tilings have two vertex orbits under their full
Euclidean symmetry groups, rather than necessarily two vertices per
translation cell~\cite{maity2021quotient}. They belong to the same
regular-polygon class as the Archimedean tilings, so
Theorem~\ref{thm:regular} gives their complete WGS census.

For configurations $A,B$ occurring $n_A,n_B$ times in a cell, this means
$\langle n_j\rangle=(n_A n_j(A)+n_B n_j(B))/(n_A+n_B)$,
including $j=33$. The parameter and baseline densities are
\begin{align}
\frac PN&=\frac{n_Aq^{d_A+1}+n_Bq^{d_B+1}}{n_A+n_B},\nonumber\\
\frac gN&=1+\frac{n_A d_A+n_B d_B}{2(n_A+n_B)}(q^2-1).
\label{SM-eq:twoPg}
\end{align}
Together with Table~\ref{SM-tab:twouniform}, these determine the exact
physical rank density $P/N-g/N-\nu/N$ for every $q\ge2$.

We retain the identifiers \texttt{t2001}--\texttt{t2020} of the
Soto S\'anchez--Medeiros e S\'a--de Figueiredo representation of the
Galebach collection~\cite{soto2019acquiring,sotoGalebachData}.
Their integer seeds and two translations in $\mathbb Z[e^{i\pi/6}]$
fix each arrangement, not just its configuration pair. The
configuration columns $A$, $B$ of Table~\ref{SM-tab:twouniform} together
form the Gr\"unbaum--Shephard symbol $[A;B]$ of the tiling, so each
row can be identified independently of the dataset.
Figure~\ref{SM-fig:two_uniform} draws all twenty from these data.
For example, \texttt{t2006} and \texttt{t2007} have the same configurations,
cell ratio, and $\nu/N$, while being different tilings. In contrast,
\texttt{t2003} and \texttt{t2004} have the same pair of configurations
but different ratios, giving qubit densities $6$ and $22/3$.
All twenty have positive WGS non-bond nullity. This statement alone
makes no claim about their generic tensor ranks. The generic comparison is given in Sec.~\ref{sec:generic}.

\begin{table*}[!tbp]
\centering
\caption{Complete twenty-tiling WGS census. Identifiers are those of the
Galebach coordinate dataset. $A,B$ are cyclic configurations (up to
rotation and reflection); $n_A:n_B$ gives actual counts in the supplied
translation cell. The polynomial column is $\nu/[N(q-1)^2]$, valid for
all $q\ge2$. Equal configuration pairs are not merged.}
\label{SM-tab:twouniform}
\small
\centering
\begin{adjustbox}{max width=\textwidth}
\begin{tabular}{lccccc}\toprule
ID & Configuration $A$ & Configuration $B$ & $n_A:n_B$ & $\nu/(Nr^2)$ & $\nu/N$ at $q=2$\\\midrule
\texttt{t2001} & $3.4.6.4$ & $4.6.12$ & $6:12$ & $\frac{2}{9}q+\frac{7}{9}$ & $\frac{11}{9}$\\
\texttt{t2002} & $3.4.3.12$ & $3.12^{2}$ & $4:4$ & $q+\frac{3}{4}$ & $\frac{11}{4}$\\
\texttt{t2003} & $3^{3}.4^{2}$ & $4^{4}$ & $2:2$ & $\frac{1}{2}q^2+q+2$ & $6$\\
\texttt{t2004} & $3^{3}.4^{2}$ & $4^{4}$ & $2:1$ & $\frac{2}{3}q^2+\frac{4}{3}q+2$ & $\frac{22}{3}$\\
\texttt{t2005} & $3.4^{2}.6$ & $3.4.6.4$ & $12:6$ & $\frac{2}{3}q+\frac{4}{3}$ & $\frac{8}{3}$\\
\texttt{t2006} & $3.4^{2}.6$ & $3.6.3.6$ & $4:1$ & $\frac{4}{5}q+\frac{6}{5}$ & $\frac{14}{5}$\\
\texttt{t2007} & $3.4^{2}.6$ & $3.6.3.6$ & $4:1$ & $\frac{4}{5}q+\frac{6}{5}$ & $\frac{14}{5}$\\
\texttt{t2008} & $3^{3}.4^{2}$ & $3.4.6.4$ & $6:6$ & $\frac{1}{2}q^2+\frac{4}{3}q+\frac{5}{3}$ & $\frac{19}{3}$\\
\texttt{t2009} & $3^{2}.4.3.4$ & $3.4.6.4$ & $6:6$ & $\frac{1}{4}q^2+\frac{4}{3}q+\frac{5}{3}$ & $\frac{16}{3}$\\
\texttt{t2010} & $3^{6}$ & $3^{2}.6^{2}$ & $1:6$ & $\frac{6}{7}q^2+\frac{12}{7}q+\frac{6}{7}$ & $\frac{54}{7}$\\
\texttt{t2011} & $3^{2}.6^{2}$ & $3.6.3.6$ & $2:1$ & $\frac{1}{3}q^2+\frac{4}{3}q+\frac{2}{3}$ & $\frac{14}{3}$\\
\texttt{t2012} & $3^{4}.6$ & $3^{2}.6^{2}$ & $2:2$ & $q^2+2q+1$ & $9$\\
\texttt{t2013} & $3^{6}$ & $3^{2}.4.12$ & $2:12$ & $\frac{6}{7}q^2+\frac{12}{7}q+\frac{9}{7}$ & $\frac{57}{7}$\\
\texttt{t2014} & $3^{6}$ & $3^{3}.4^{2}$ & $1:2$ & $\frac{5}{3}q^2+\frac{8}{3}q+2$ & $14$\\
\texttt{t2015} & $3^{6}$ & $3^{3}.4^{2}$ & $2:2$ & $2q^2+3q+2$ & $16$\\
\texttt{t2016} & $3^{3}.4^{2}$ & $3^{2}.4.3.4$ & $4:8$ & $\frac{2}{3}q^2+2q+2$ & $\frac{26}{3}$\\
\texttt{t2017} & $3^{3}.4^{2}$ & $3^{2}.4.3.4$ & $4:4$ & $\frac{3}{4}q^2+2q+2$ & $9$\\
\texttt{t2018} & $3^{6}$ & $3^{2}.4.3.4$ & $1:6$ & $\frac{6}{7}q^2+\frac{16}{7}q+2$ & $10$\\
\texttt{t2019} & $3^{6}$ & $3^{4}.6$ & $2:6$ & $\frac{15}{8}q^2+3q+\frac{3}{2}$ & $15$\\
\texttt{t2020} & $3^{6}$ & $3^{4}.6$ & $6:6$ & $\frac{9}{4}q^2+\frac{10}{3}q+\frac{5}{3}$ & $\frac{52}{3}$\\
\bottomrule
\end{tabular}
\end{adjustbox}\end{table*}

\begin{figure*}[!tbp]
\centering
\includegraphics[width=0.98\textwidth]{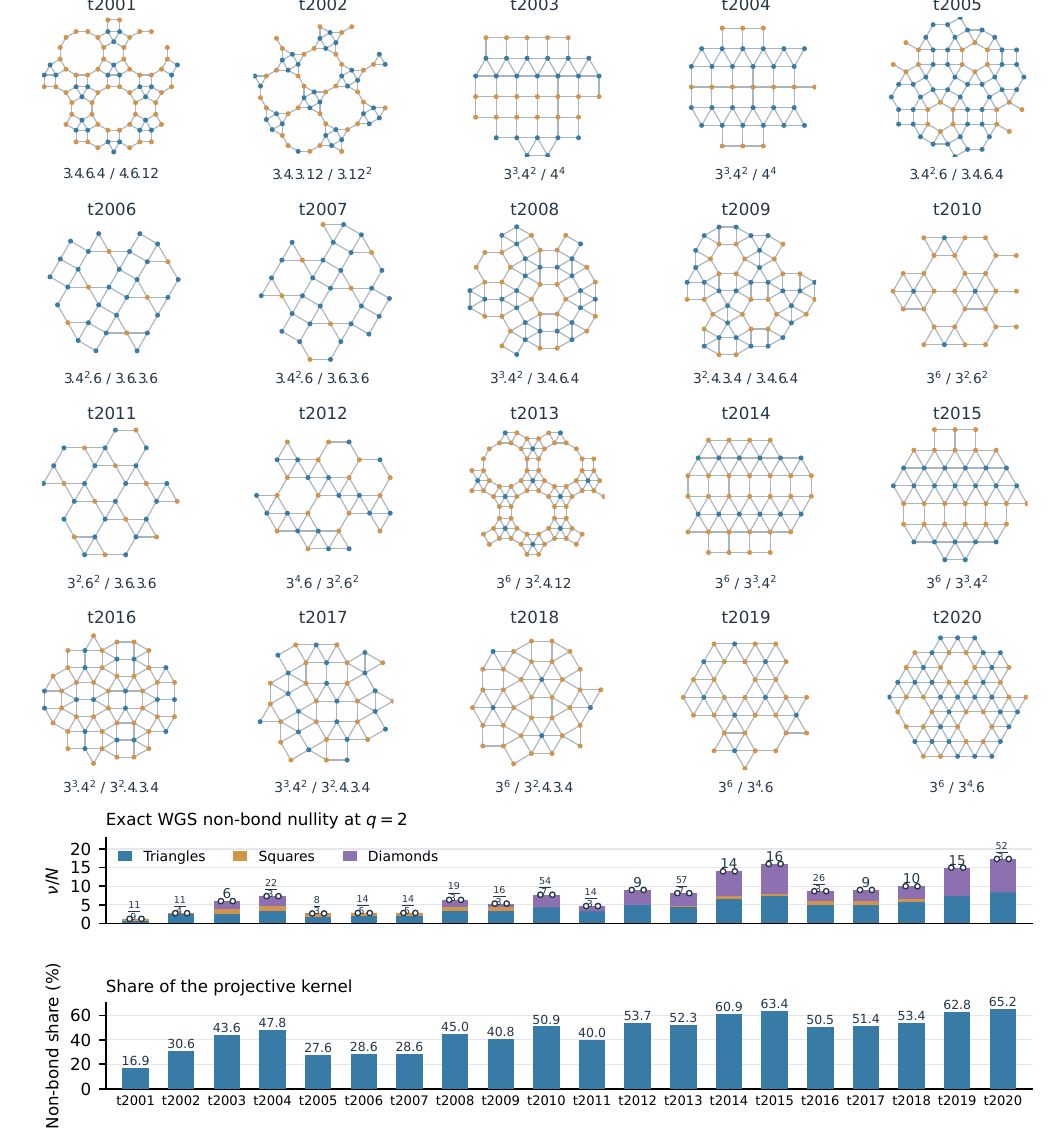}
\caption{Complete census of all twenty 2-uniform tilings,
\texttt{t2001}--\texttt{t2020}, reconstructed from the exact Galebach cell
data and drawn in Galebach order. The $4\times5$ panels show finite
patches; the two vertex colors distinguish the printed configurations.
Cell multiplicities $n_A:n_B$ are listed in
Table~\ref{SM-tab:twouniform}. The two $1\times20$ charts below use one bar
per tiling, aligned with the panel columns: the upper chart gives the
exact WGS non-bond nullity at $q=2$, with stacked independent triangle,
square, and diamond contributions $5F_3/N$, $2F_4/N$, and
$4E_{33}/N$, and open circles giving independent brute-force counts of
the support union on two tori with periods $4\ell_0+1$ and $4\ell_0+3$; the lower chart
gives $100\nu/(g+\nu)$, with $g$ the projective baseline of
Eq.~\eqref{SM-eq:twoPg}. Both charts are exact evaluations of
Theorem~\ref{thm:regular}; drawings are finite patches. Distinct
arrangements sharing a configuration pair are retained.}
\label{SM-fig:two_uniform}
\end{figure*}

\section{Proofs for the regular-polygon and periodic-cell sections}
\label{SM-sec:latticeproofs}
\subsection{Short cycles in unit-edge tilings}
\begin{proof}[Proof of Lemma~\ref{lem:shortcycles}]
(i) The common neighbors of $u\ne v$ lie on both unit circles about $u$
and $v$, which meet in at most two points; a unit-distance $K_4$ would
need a fourth point at unit distance from all three vertices of an
equilateral triangle, which does not exist in the plane.
(ii)--(iii) A three-cycle of unit edges is an equilateral triangle of
area $\sqrt3/4$, and a four-cycle of unit edges is a rhombus of area
$\sin\theta\le1$. Since the tiling is edge-to-edge and noncrossing, the
closed region bounded by either cycle is a union of faces, each a
regular polygon of unit side and hence of area at least $\sqrt3/4$. If
a vertex lay strictly inside the region, the at least three faces
meeting at it would all lie inside the region, with total area at
least $3\sqrt3/4>1$, a contradiction. So the region contains no
interior vertex and is tiled by faces whose corners are cycle vertices
only. For the triangle this leaves a single triangular face. For the
rhombus, a regular $m$-gon of unit side with $m\ge5$ has area greater
than one and a square has area exactly one, so the region is either one
square face or a union of triangular faces; four corners admit exactly
two triangles, which share a diagonal of unit length.
\end{proof}

\subsection{Motif identification on the square and kagome lattices}
The following completes the proof of Corollary~\ref{cor:lattices} for the two remaining Archimedean
examples treated there.
On the square lattice, two nonadjacent centers can share two neighbors
only as the opposite vertices of one unit square. Their two-neighbor
support gives Eq.~\eqref{eq:square}; all other intersections are already
covered by ordinary gauge.

On kagome, adjacent centers share precisely one third vertex. Their
intersection is a triangle. Nonadjacent centers share at most one
neighbor, so Eq.~\eqref{eq:triangle} is complete. 
\subsection{Translation-tied tensors}
\begin{proof}[Proof of Corollary~\ref{cor:tied}]
The translation-compatible factors at the $\Gamma$-invariant base point
make the local Fourier isomorphisms equivariant. Thus, in the coordinates
of Eq.~\eqref{eq:sector}, tying sets
$x_{\gamma v,\gamma S}^{\boldsymbol k}=x_{v,S}^{\boldsymbol k}$, so every
translate of a support receives the same coefficient
$\sum_{v\in C_S}x_{v,S}^{\boldsymbol k}$; the image is spanned by the
orbit sums, one for each of the $\kappa(q)$ orbits of nonempty
$(S,\boldsymbol k)$ plus the constant. A finite nonempty support has
trivial stabilizer on a locally faithful torus, so the $m_S$ tied
coordinates $x_{v,S}^{\boldsymbol k}$, $v\in C_S$, are distinct
variables and each orbit kernel is again the sum-zero hyperplane of
dimension $m_S-1$. The Vandermonde argument of
Proposition~\ref{prop:bond} is local and applies verbatim to tied
generators, giving $X_e=a_eI$ on each edge orbit with $a$ in the cycle
space of the quotient multigraph on $V_0$, of dimension $e-n+1$.
Subtracting $\rank\Phi^\Gamma_A+1$ and $\rank J^\Gamma_{\perp,A}$ from
$\dim\mathcal P_G^\Gamma=\sum_\alpha q^{d_\alpha+1}$ gives $\nu/M$.
\end{proof}

\section{Eight further periodic connectivity graphs}
\label{SM-sec:periodiccatalogue}
\subsection{Further periodic connectivity examples}
\label{SM-sec:hardwaregraphs}
Table~\ref{SM-tab:periodic} and Fig.~\ref{SM-fig:periodic} extend the calculation
to mixed-degree, non-regular-face, and crossing-edge graphs. Edges refer
to interaction connectivity; a crossing without a vertex does not
introduce a tensor. Here centered-square means that all square edges
are retained and a center is connected to all four corners (Union Jack).
Dice has degree-six hubs and degree-three rim vertices. The checkerboard
graph is $L(\text{square})$, where $L$ denotes the line graph; shuriken
is $L(\text{square-octagon})$. The king graph has all nearest square
edges and both diagonals of every square, without diagonal-crossing
vertices. These definitions remove ambiguities associated with names
such as ``face-centered square.''
The examples separate three mechanisms: complete subdivision removes
all triangle and four-cycle relations; line graphs create incident-edge
cliques; and dice, centered-square, and king test mixed coordination,
non-regular faces, or crossing edges. Their selection is structural,
without claiming to classify hardware layouts under a size or degree cut.

\begin{table*}[!tbp]
\centering
\caption{Eight further finite-cell periodic graphs. Degree notation $d^m$
means $m$ cell vertices of degree $d$, so $P/N=\sum_d m_dq^{d+1}/n$ and
$g/N=1+(e/n)(q^2-1)$. The nullities hold for all $q\ge2$, with $r=q-1$.
These rows specify eight examples outside the two regular-polygon
catalogues; they do not exhaust periodic graphs. Expanding in $r$ gives
coefficients by character-support size $|S|$, not by minimal-region
size; overlapping minimal-region dimensions must not be added.}
\label{SM-tab:periodic}
\small
\centering
\begin{adjustbox}{max width=\textwidth}
\begin{tabular}{lccclc}\toprule
Graph & $n$ & $e$ & Degrees/cell & $\nu/N$ & $q=2$\\\midrule
Lieb / heavy-square &3&4&$4^1,2^2$&$0$&0\\
Heavy-hex &5&6&$3^2,2^3$&$0$&0\\
Heavy-square-octagon &10&12&$3^4,2^6$&$0$&0\\
Dice &3&6&$6^1,3^2$&$2r^2$&2\\
Centered-square / Union Jack &2&6&$8^1,4^1$&$r^2(3q^2+4q+2)$&22\\
Checkerboard &2&6&$6^2$&$r^2[3(q+1)^2+2]/2$&$29/2$\\
Shuriken / square-kagome &6&12&$4^6$&$r^2(4q+3)/3$&$11/3$\\
King / full-diagonal square &1&4&$8^1$&$20r^2+36r^3+29r^4+12r^5+2r^6$&99\\
\bottomrule
\end{tabular}
\end{adjustbox}
\end{table*}

Centered-square and triangular have the same nullity polynomial, although
Theorem~\ref{thm:regular} does not apply to centered-square's non-regular
triangles. Its two-vertex cell has extra-sector coefficients
$18,20,6$ at support sizes $2,3,4$, respectively, and zero at larger
sizes. Thus Theorem~\ref{thm:cell} gives
$\nu/N=9r^2+10r^3+3r^4=r^2(3q^2+4q+2)$.
There is a useful distinction from the regular-tiling proof: the four
corners $a,b,c,d$ of an original square form a chordless cycle even
though its interior contains a center $h$. On support $\{a,c\}$ the
centers are $b,d,h$, and
$z_{\{a,c\};b,d}=z_{\{a,c\};b,h}-z_{\{a,c\};d,h}$.
The square relation is therefore already generated by the two diamond
relations involving $h$. Agreement with the face-density polynomial in
this example does not remove the geometric hypotheses of
Theorem~\ref{thm:regular}.

\begin{figure*}[!tbp]
\centering
\includegraphics[width=\textwidth]{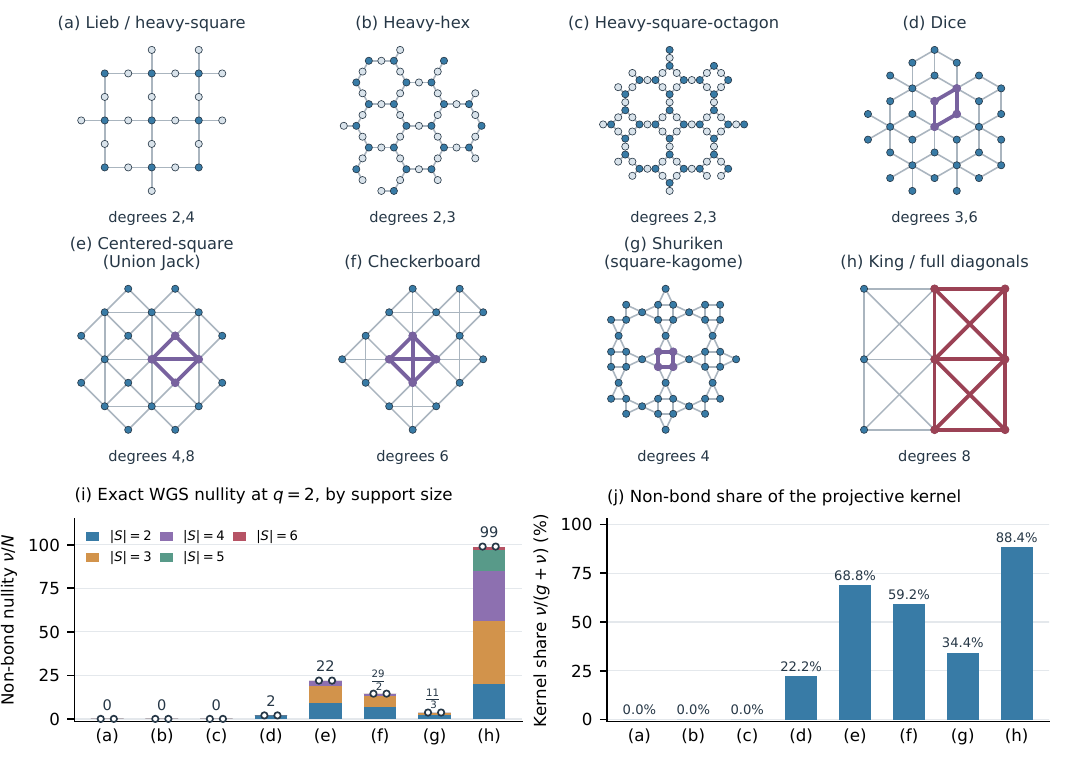}
\caption{Eight additional periodic connectivity graphs and their exact
WGS counts. Colored regions illustrate a largest minimal region present;
the first three graphs have none. Edge crossings without dots are not
vertices; in the king panel the six-vertex join includes such crossings. Stacks separate powers
$r^{|S|}$ in the non-bond-kernel polynomial at $q=2$, rather than summing
potentially dependent minimal regions. Open circles are independent
brute-force counts of the support union on $5\times5$ and $7\times7$ tori; all edge shifts in these
cell conventions have absolute value at most one. Panel (j) shows
$\nu/(g+\nu)$ with $g$ including the ray.}
\label{SM-fig:periodic}
\end{figure*}

Beyond the regular-polygon class, examples in Appendix~\ref{SM-sec:periodiccatalogue} exhibit $K_4$ and the joins
$K_2\vee(K_2\sqcup K_1)$ and $K_2\vee(K_2\sqcup K_2)$; the join connects
each of the first two vertices to every vertex of the second factor.
Their intrinsic polynomials are, respectively,
\begin{equation}
4r^3+3r^4,\qquad q^2r^3,\qquad q^2r^4.
\label{SM-eq:newminimal}
\end{equation}
Checkerboard and king contain $K_4$; king also contains the five- and
six-vertex joins. These regions have genuinely new open-leg relations
beyond all their proper subregions. 
\subsection{Subdivision and line-graph constructions}
The preceding catalogue can be extended by graph operations. Complete
edge subdivision gives infinitely many families with zero non-bond
nullity; line graphs have an explicit count in terms of the original
graph.

\begin{proposition}[Subdivisions and line graphs]
\label{SM-prop:operations}
Subdividing \emph{every} edge of a connected simple graph at least once
produces a graph with no WGS non-bond kernel. Ordinary gauge is also
complete at generic tensors on each such finite graph.
For a connected triangle-free simple graph $H$ with at least one edge,
\begin{align}
\nu_{L(H)}={}&\sum_{v\in V(H)}\Bigl[(d_v-1)(q^{d_v}-1)\nonumber\\
&\qquad-\binom{d_v}{2}(q^2-1)\Bigr]\nonumber\\
&+2C_4(H)(q-1)^2,
\label{SM-eq:linegraph}
\end{align}
where $C_4(H)$ counts unoriented four-cycles once. The periodic version
uses cell densities on locally faithful quotients.
\end{proposition}
\begin{proof}
Complete edge subdivision leaves no triangles or four-cycles. Adjacent
closed neighborhoods then intersect only in their endpoints, and
nonadjacent ones in at most one vertex. Proposition~\ref{prop:quotient}
gives $\nu=0$. The WGS point realizes the maximal possible differential
rank $P-(g-1)$; a nonzero minor remains nonzero on a Zariski-open set,
where the bond-gauge rank has the same value. Thus generic completeness
also follows, without a blocking construction.

In $L(H)$, the edges at a base vertex of degree $d$ form a clique $K_d$.
Triangle-freeness makes the common closed neighborhood of adjacent line
vertices precisely this incident-edge clique. Its independent additional
kernel is $(d-1)(q^d-1)-\binom d2(q^2-1)$, obtained from the full-support
union on $K_d$. Distinct such cliques meet in at most one vertex, so these
non-bond sectors do not overlap. Two nonadjacent line vertices are disjoint
base edges. They have at most two common neighbors; two occur exactly
when they are opposite in a base four-cycle. Each cycle supplies two
nonedge-pair sectors of size $(q-1)^2$, proving the formula.
\end{proof}
Thus Lieb/heavy-square, heavy-hex, and heavy-square-octagon have zero
additional WGS and generic kernel. Subdivision and line-graph pairs also
arise in routing between hardware connectivity and simulated interactions
\cite{kattemolle2023routing}; their kernel counts differ because the
underlying graphs differ.

\subsection{Generic completeness of the eight connectivity graphs}
The generic comparison for these graphs is listed in the last block of
Table~\ref{SM-tab:generic}. Seven families have zero generic non-bond
nullity for every $q$ under the listed conditions, and centered-square
has it at $q=2$; Equation~\eqref{SM-eq:genericbounds} bounds its
unresolved value at $q>2$, with the WGS upper bound in
Table~\ref{SM-tab:periodic} on locally faithful quotients.
Centered-square uses the cell convention of
Appendix~\ref{SM-sec:open5qubits}, and dice the integer cells of
Appendix~\ref{SM-sec:generic39proof}. King is established by the row
construction in Appendix~\ref{SM-sec:kingrows}. The proof dependencies
are in Table~\ref{SM-tab:proofdependencies}.

\subsection{Which graph represents a toric-code layout?}
In the square toric code, data qubits occupy edges, while star and
plaquette checks occupy vertices and faces~\cite{kitaev2003anyons}.
There are several associated two-body graphs. Connecting a data qubit
to its incident star ancillas alone gives the subdivided square (Lieb)
graph. Including both star and plaquette ancillas gives an ordinary
square graph with half the original lattice spacing. Connecting pairs
of data qubits sharing a star gives the checkerboard line graph. Connecting
pairs sharing either a star or a plaquette gives the king graph after a
rotation and rescaling. Thus the corresponding qubit WGS densities are
$0,2,29/2,99$, respectively, each normalized by the vertices of its
\emph{own} interaction graph. On a square code torus the last identification
has skew periods $(L,L)$ and $(L,-L)$ in king coordinates; it is not an
identification with a rectangular $L\times L$ king torus.

These are statements about nondegenerate WGS tensor representations on
specified connectivity graphs. They do not identify the toric-code state
with a WGS state on its local check graph: graph-state representations of
toric code can require nonlocal connectivity~\cite{liao2021toric}.
Likewise, the centered-square row concerns pairwise WGS edges, not the
three-qubit controlled-phase resource often associated with the Union
Jack lattice~\cite{miller2016hierarchy}.

\section{Periodic graph conventions}
\label{SM-sec:coordinates}
Square uses nearest-neighbor edges of $\mathbb Z^2$. Triangular uses
steps $\pm(1,0),\pm(0,1),\pm(1,1)$. Honeycomb has two vertices
$A_{xy},B_{xy}$ per cell and edges from $A_{xy}$ to
$B_{xy},B_{x-1,y},B_{x,y-1}$. Kagome is its line graph; an equivalent
three-site convention is specified in Appendix~\ref{SM-sec:m28generic}.
Periods identify cell coordinates without identifying tensor parameters.

The star, square-octagon, and cross constructors truncate honeycomb,
square, and kagome, respectively: replace each source vertex by the
cyclic polygon of its incident half-edges and retain each original
edge between its two half-edges. Ruby uses the vertex--face flags of
honeycomb: the three flags at a vertex form a triangle, and the six
flags around a face form a hexagon. This yields intervening squares.

For elongated triangular, use horizontal and vertical edges on
$\mathbb Z^2$ and add $(x,y)(x+1,y+1)$ whenever $y$ is even; one family of
bulk quotients has periods $(s,2s)$. For snub-square use the same
grid with periods $(2s,2s)$ and add one diagonal in each square with
even $x+y$: $(x,y)(x+1,y+1)$ for even $x$, and
$(x+1,y)(x,y+1)$ for odd $x$. These integer-grid drawings are
combinatorial realizations, not regular unit-edge embeddings.
For maple-leaf use triangular-lattice steps $(1,0),(0,1),(1,1)$,
and remove vertices with $x+2y=0\pmod7$.
For the generic theorem use the vacancy-cell period vectors
$L_1(2,-1),L_2(1,3)$; square periods divisible by seven also define
valid quotient graphs for the finite-graph WGS formula. The regular embedding has coordinates
$(x-y/2,\sqrt3y/2)$. In the maple-leaf TDVP scan, the construction-size
label $s$ instead selects square coordinate periods
$L_x=L_y=7(s-5)$, giving $N=42(s-5)^2$. Thus the displayed sizes
$s=7,9,11$ have $N=168,672,1512$, respectively; $s$ is not the
vacancy-cell period $L_1$ or $L_2$ in this scan.

In the lattice classification a quotient must preserve the relevant
closed-neighborhood intersections, including nonfacial short cycles.
If it does not, the finite-graph theorem still applies directly.
The stronger generic period bounds are stated separately in
Tables~\ref{tab:genericsummary} and~\ref{SM-tab:generic}, and are not needed for the arbitrary-graph
WGS theorem.

\section{Generic completeness: injective blocks and routing}
\label{SM-sec:genericappendix}
Table~\ref{SM-tab:generic} collects the generic-completeness results for
all thirty-nine families: thirty-three for every $q\ge2$, and six
at $q=2$, under the listed period conditions. Theorem~\ref{thm:generic} covers every family;
Table~\ref{SM-tab:proofdependencies} maps each to its proof ingredients.

\begin{table*}[!tbp]
\centering
\caption{Generic non-bond nullity $\nu_{\rm generic}$ for the eleven
Archimedean tilings, the twenty 2-uniform tilings, and eight further
connectivity graphs (Appendix~\ref{SM-sec:periodiccatalogue}), with
independent $D=d=q\ge2$ tensors. Zero entries hold for every
integer $q\ge2$ under the listed sufficient conditions, except for the
six families restricted to $q=2$. Thus all thirty-nine families are
settled at $q=2$; ``Open'' refers only to the six unresolved $q>2$
extensions. Equation~\eqref{SM-eq:genericbounds} applies to every row.
Coordinate conventions are those of Appendix~\ref{SM-sec:coordinates};
kagome and ruby periods count honeycomb source cells, maple-leaf
periods count vacancy cells, shuriken and checkerboard periods count
original cells of their $2\times2$ supercells, dice uses the integer
cells of Appendix~\ref{SM-sec:generic39proof}, and centered-square the
cell convention of Appendix~\ref{SM-sec:open5qubits}. The proof mechanism
names the instance of Theorem~\ref{SM-thm:template} used; all
constructions and certificates are in Appendix~\ref{SM-sec:genericappendix},
with the proof-dependency map in Table~\ref{SM-tab:proofdependencies}.}
\label{SM-tab:generic}
\footnotesize
\setlength{\tabcolsep}{4pt}
\centering
\begin{adjustbox}{max width=\textwidth}
\begin{tabular}{lp{2.2cm}p{5.0cm}p{4.2cm}}\toprule
Graph & $\nu_{\rm generic}$ & Sufficient conditions & Proof mechanism\\
\midrule
\multicolumn{4}{l}{\emph{Eleven Archimedean tilings}}\\
Square & $0$ & $L_x,L_y\ge20$ & Injective partitions\\
Honeycomb & $0$ & $L_x,L_y\ge5$ & WGS rank witness\\
Triangular & $0$ ($q=2$); Open ($q>2$) & $6\mid L_x,L_y\ge54$ or $19\mid L_x,L_y\ge57$; also $4{\times}3$, $4{\times}4$ & Hexagonal blocks; exact minors\\
Kagome & $0$ & $L_x,L_y\ge64$ & Injective partitions\\
Star & $0$ & Intact disjoint triangles; simple coarse graph & Injective blocking\\
Square-octagon & $0$ & Intact disjoint squares; simple coarse graph & Injective blocking\\
Cross & $0$ & Intact disjoint squares; simple coarse graph & Injective blocking\\
Ruby & $0$ & $L_x,L_y\ge64$ & Triangle blocking and splitting\\
Elongated triangular & $0$ & $L_x\ge20$, even $L_y\ge40$ & Dimer blocking and splitting\\
Snub-square & $0$ & Even $L_x,L_y\ge48$ & Injective partitions\\
Maple-leaf & $0$ & Period vectors $L_1(2,-1),L_2(1,3)$; $L_1,L_2\ge40$ & Triangle blocking and splitting\\
\midrule
\multicolumn{4}{l}{\emph{Twenty 2-uniform tilings}}\\
\texttt{t2001} & $0$ & $L_x,L_y\ge32$ & Triangle and singleton blocking\\
\texttt{t2002}--\texttt{t2005}, \texttt{t2011} & $0$ & $L_x,L_y\ge32$ & Triangle/dimer blocking and splitting\\
\texttt{t2006}, \texttt{t2007} & $0$ & $L_x,L_y\ge32$ & Bowtie blocking and splitting\\
\texttt{t2008} & $0$ ($q=2$); Open ($q>2$) & $L_x,L_y\ge3$ & Injective block; exact minors\\
\texttt{t2009}, \texttt{t2012}, \texttt{t2020} & $0$ & $L_x,L_y\ge32$ & Diamond blocking and splitting\\
\texttt{t2010}, \texttt{t2013} & $0$ & $L_x,L_y\ge3$ & Wheel blocking\\
\texttt{t2014} & $0$ ($q=2$); Open ($q>2$) & $L_x,L_y\ge32$ & Normal partitions; exact minors\\
\texttt{t2015} & $0$ ($q=2$); Open ($q>2$) & $L_x,L_y\ge32$ & Normal partitions; exact minors\\
\texttt{t2016}--\texttt{t2018} & $0$ & $L_x,L_y\ge30$ & Open-region partitions\\
\texttt{t2019} & $0$ ($q=2$); Open ($q>2$) & $3\mid L_x,L_y\ge24$ & Normal partitions; exact minors\\
\midrule
\multicolumn{4}{l}{\emph{Eight further connectivity graphs}}\\
Lieb, heavy-hex, heavy-square-octagon & $0$ & Complete subdivision of a finite simple graph & WGS rank witness\\
Dice & $0$ & $L_x,L_y\ge30$ & Open-region partitions\\
Checkerboard & $0$ & Even $L_x,L_y\ge64$ & $K_4$ blocking and splitting\\
Shuriken & $0$ & Even $L_x,L_y\ge64$ & Triangle blocking and splitting\\
Centered-square & $0$ ($q=2$); Open ($q>2$) & $2\mid L_x\ge40$, $6\mid L_y\ge42$, or $x\leftrightarrow y$ & Twelve-site blocks; exact minors\\
King & $0$ & $L_x\ge4$, $L_y\ge18$, or $x\leftrightarrow y$ & Row blocking; analytic splitting\\
\bottomrule
\end{tabular}
\end{adjustbox}
\end{table*}

\subsection{Generic rank bounds}
\begin{proposition}[Generic rank bounds]
\label{SM-prop:genericbounds}
For every finite simple connected graph and fixed integer $q\ge2$, put
$g_0=E(q^2-1)+N-1$. On a nonempty Zariski-open subset,
\begin{gather}
\rank\Phi_A=g_0,\qquad
\rank J_A=P-g_0-\nu_{\rm generic},\nonumber\\
\max\{0,P-g_0-q^N\}\le\nu_{\rm generic}\le\nu_{\rm WGS}.
\label{SM-eq:genericbounds}
\end{gather}
Moreover, the affine non-bond nullity at any tensor assignment is
at least $\nu_{\rm generic}$, and
\begin{equation}
\rank J_{\rm generic}-\rank J_{\rm WGS}
=\nu_{\rm WGS}-\nu_{\rm generic}.
\label{SM-eq:generalrankdrop}
\end{equation}
\end{proposition}
\begin{proof}
Proposition~\ref{prop:bond} gives the maximal gauge rank $g_0$ on a
nonempty open set. Let $R$ be the maximal rank of $J$; an attaining
minor defines another nonempty open set, and irreducibility of the
affine parameter space makes their intersection nonempty. There
$\nu_{\rm generic}=P-R-g_0$. Gauge inclusion and the physical
dimension bound give $R\le\min(P-g_0,q^N)$, while evaluation at a
nondegenerate WGS point gives $R\ge P-g_0-\nu_{\rm WGS}$,
proving both displayed statements.
At any other point both ranks are bounded above by their respective
maxima, so $P-\rank J_A-\rank\Phi_A\ge\nu_{\rm generic}$.
\end{proof}
For a locally faithful torus with $M$ cells, Eq.~\eqref{eq:cellrank}
makes the upper bound explicit:
\begin{equation}
\begin{split}
\nu_{\rm generic}\le{}&M\Bigl[
\sum_{\alpha\in V_0}q^{d_\alpha+1}\\
&-\kappa(q)-e(q^2-1)-n\Bigr].
\end{split}
\label{SM-eq:genericcellbound}
\end{equation}
Table~\ref{tab:counts} and Tables~\ref{SM-tab:twouniform} and
\ref{SM-tab:periodic} therefore bound $\nu_{\rm generic}$ from above
for every family studied here. For short periodic quotients
use the finite-graph WGS value in Eq.~\eqref{SM-eq:genericbounds} instead.

\subsection{From finite equivalence to the generic differential}
\label{SM-sec:genericmechanism}
The normal-PEPS theorem enters through a statement about fibers.
A finite-equivalence theorem alone need not fix a differential at a
special point: $x\mapsto x^3$ has a single-point fiber at zero but
vanishing derivative there. The following generic argument supplies
the required regularity.

We apply Lemma~\ref{lem:fiber} to the original tensor parameter space.
Its conclusion, Eq.~\eqref{eq:genericfiber}, supplies the required
generic differential rank once the gauge-fiber hypothesis is established.

For direct applications of the normal-PEPS theorem, both of its
partition requirements must be checked~\cite{molnar2018normal}. For
every edge $uv$, three nonempty injective blocks must separate $u$ and
$v$, with $uv$ the only edge between their two blocks. For every vertex
$v$, there must be a region $R$ excluding $v$ such that $R$,
$R\cup\{v\}$, and both complements are injective. When blocking creates
composite virtual legs, coarse gauge completeness must also be lifted to
the original edges. The following theorem isolates what any such
argument must supply.

\begin{theorem}[Generic completeness by blocking]
\label{SM-thm:template}
Let $G$ be a finite simple connected graph with independent
$D=d=q$ tensors, and suppose the following data exist:
\begin{enumerate}
\item[(i)] a partition of $V$ into connected blocks, inducing a coarse
graph whose bonds group one or more original edges;
\item[(ii)] for every coarse edge and every coarse vertex, the
partitions required by the normal-PEPS theorem on the coarse graph,
each of whose regions is certified injective by a routing witness on
the original $q$-dimensional legs;
\item[(iii)] generic internal tangent completeness of each block: the
kernel of the block's own open differential is its internal bond gauge
on a nonempty open set;
\item[(iv)] for every composite coarse bond, a detector: quotient maps
under which all single-site variations and all other coarse-bond
actions vanish, together with an all-$q$ witness showing that the
target generator lies in a sum of operator spaces supported on single
original sites; and
\item[(v)] a common-refinement argument showing that the two endpoint
site partitions of each coarse bond intersect exactly in the original
single-edge operator algebras.
\end{enumerate}
Then $\ker J_A=\im\Phi_A$ and
$\rank J_A=P-[E(q^2-1)+N-1]$ on a nonempty Zariski-open subset of
$\mathcal P_G$.
\end{theorem}
\begin{proof}
Each routing witness in (ii) is a nonzero injectivity minor valid for
every integer $q\ge2$ [Eq.~\eqref{SM-eq:routeisometry}], so all
conditions in (ii)--(iv) hold on a common nonempty open set, which we
intersect with the maximal original gauge-rank set of
Proposition~\ref{prop:bond} and with the open set on which $d\Psi$ has
its maximal rank; irreducibility keeps the intersection nonempty. On
the maximal-rank set the rank is locally constant, so by the
constant-rank theorem the fibers of $\Psi$ are smooth submanifolds
whose tangent space is $\ker d\Psi$, and every kernel tangent at such
a point integrates to a local constant-state curve. Its blocked
curve lies in a single coarse gauge orbit by (ii) and the normal-PEPS
theorem; the coarse generators split into original-edge gauges by (iv)
and (v); the residual block variations are internal gauges by (iii).
Hence the tangent lies in $\im\Phi_A$, and Lemma~\ref{lem:fiber} gives
the rank. No genericity of blocked tensors in an unrestricted coarse
parameter space is assumed.
\end{proof}

\subsection{Proof of the generic-completeness theorem}
\begin{proof}
Two mechanisms occur. A WGS point with $\nu=0$ attains the bond-gauge
upper bound on $\rank J$, so a nonzero maximal minor proves generic
completeness directly; this covers honeycomb and the complete
subdivisions by Corollary~\ref{cor:lattices} and
Proposition~\ref{SM-prop:operations}. Every other zero entry is an
instance of Theorem~\ref{SM-thm:template}. Square, kagome, snub-square,
dice, and \texttt{t2016}--\texttt{t2018} use the original graph as its
own blocking, so (iii)--(v) are vacuous and only the partitions (ii)
are needed; the covering argument that extends finite routing
templates to all stated periods is given for snub-square in
Appendix~\ref{SM-sec:archgenericproof}. The remaining families block
triangles, dimers, bowties, diamonds, wheels, $K_4$'s, or entire rows; the corresponding
sections supply for each the injective grouped site maps of (iii) and the
detectors and refinements of (iv)--(v); ruby (Appendix~\ref{SM-sec:m28generic}) is the simplest complete example, with a
two-site Vandermonde detector, Eq.~\eqref{SM-eq:rubysplit}, that forces
each coarse $q^2$-dimensional bond generator to be a sum of two
single-edge actions. All splitting witnesses are analytic and valid
for every $q$. King's row construction (Appendix~\ref{SM-sec:kingrows})
uses a simultaneous joint-rank witness proved by cyclic displacement
sectors, rather than a dense rank calculation. The six $q=2$ results, which block paths, a twelve-site
cluster, or lattice-periodic hexagonal clusters, replace (iii)--(iv) by
nonzero integer minors and are not extrapolated to higher $q$
(Appendices~\ref{SM-sec:open4qubits} and~\ref{SM-sec:open5qubits}).
\end{proof}

The period bounds describe the chosen routing templates and their
complement covers, not physical thresholds; for example the square
bound $20=16+4$ combines a width-sixteen template box with a width-four
covering block. Below these bounds the finite-graph WGS theorem remains
exact, and Proposition~\ref{SM-prop:genericbounds} bounds
$\nu_{\rm generic}$.

\subsection{The triangular lattice}
\label{SM-sec:triangulargeneric}
The direct three-block edge partition is impossible on the triangular
lattice for any tensor assignment. Injectivity of a region requires, at
each degree-six vertex in that region, at least three neighbors inside
it: the singleton case of the local cut condition [Eq.~\eqref{SM-eq:routecut}] reads $6-d_{\rm in}\le1+d_{\rm in}$. If $uv$
were the only edge between blocks $A$ and $B$, its two common neighbors
would have to lie in $C$; the other three neighbors of $u$ must then
belong to $A$ and those of $v$ to $B$, and each common neighbor has
four neighbors forced outside $C$, violating the same condition. A
blocking route needs coarse partitions and splitting arguments adapted
to the triangular graph; Appendix~\ref{SM-sec:open5qubits} supplies one at
$q=2$ with hexagonal blocks. The obstruction concerns the proof
route, not generic completeness itself: small tori are settled
directly, since a nonzero minor of $J$ at any point is a
lower bound on the generic rank, and the rank of $J$ at a point with
integer entries can be computed exactly modulo a prime. At $q=2$ on
the $4\times3$ triangular torus ($P=1536$, $g_0=119$), a random
integer point has $\rank J\ge1417=P-g_0$ modulo $p=8191$, and on the
$4\times4$ torus ($P=2048$, $g_0=159$) $\rank J\ge1889=P-g_0$, so
$\nu_{\rm generic}=0$ on both graphs
(\nolinkurl{experiments/triangular_generic_rank.py}; the rank of a
random subset of $1700$ and $2100$ rows already attains the bound).
These are finite-graph certificates. The blocked qubit proof covers
the period classes of Appendix~\ref{SM-sec:open5qubits}; the all-$q$
statement is open. Tensor-network varieties of the single triangle
graph are analyzed in Ref.~\cite{bernardi2026triangular}.

\subsection{The common structure of the blocking arguments}
\label{SM-sec:blockingtemplate}
Theorem~\ref{SM-thm:template} supplies the common reduction for the
proofs in Appendices~\ref{SM-sec:strivalent}--\ref{SM-sec:t2020extension},
\ref{SM-sec:open4qubits}, and \ref{SM-sec:open5qubits}. Each construction specifies the blocks,
injective partitions, internal kernels, composite-bond detectors, and
endpoint refinements. Its conclusion is applied in the original tensor
parameter space.

The sections below specify (i)--(v) for each graph family and prove
the corresponding detectors; their new ingredients are the trivalent
block inversion of Appendix~\ref{SM-sec:strivalent}, the Vandermonde two-site
splitting of Appendix~\ref{SM-sec:m28generic}, the joint dimer detector of
Appendix~\ref{SM-sec:m29generic}, the open-region routing extension of
Appendix~\ref{SM-sec:generic39proof}, the clique extension of the detector in
Appendix~\ref{SM-sec:fiveextension}, the noninjective-center bowtie block of
Appendix~\ref{SM-sec:bowtieextension}, the region detectors of
Appendix~\ref{SM-sec:diamondextension}, the wheel core decomposition of
Appendix~\ref{SM-sec:wheelextension}, the unequal-bond joint detector of
Appendix~\ref{SM-sec:t2020extension}, the exact qubit minors of
Appendix~\ref{SM-sec:open4qubits}, and the lattice blockings with disc covers
of Appendix~\ref{SM-sec:open5qubits}.

\subsection{Generic rank witnesses and trivalent blocking}
\label{SM-sec:strivalent}
The triangles of star and the squares of square-octagon and cross
are disjoint and cover all vertices. Each site has exactly one
external edge. Assume these blocks remain intact under the periodic
identification and their coarse graph is simple. The coarse graphs
are honeycomb, square, and kagome, respectively.

For a block $C$ of length $k=3$ or $4$, the contracted map is
\begin{equation}
B_C:(\mathbb C^q)^{\otimes k}_{\partial C}
\longrightarrow (\mathbb C^q)^{\otimes k}_{\mathrm{phys}}.
\end{equation}
It is generically invertible: set each site's physical index equal
to its external index and both internal indices to zero. This gives
$B_C=I$, proving that its determinant is a nonzero polynomial.
Independently, group each site's physical and external legs into one
$q^2$-dimensional output. Its two internal legs together also have dimension
$q^2$, so this site flattening is generically invertible. Both conditions
are nonempty Zariski-open conditions and hence hold simultaneously
on a nonempty open set. The sparse witness for the first condition
need not satisfy the second.

An injective tensor network on a simple graph has only the ordinary
bond-gauge fiber. The dimension count is the injective, or
supercritical, regime of Ref.~\cite{bernardi2023dimension}; the
identification of the fiber with a single gauge orbit is the injective
case of Ref.~\cite{molnar2018normal}. MPO-injective generalizations of
such fiber identifications are reviewed in Ref.~\cite{cirac2021matrix}.
Take two original tensor points
in the common open set with the same global contraction. Their injective
coarse networks are related by finite coarse bond gauges. Each coarse
edge here is one original external edge, so apply those gauges directly
to align the block contractions. Within each aligned block, the two
networks have identical tensors as outputs on their physical and external
legs. The grouped site maps remain injective under the external gauges.
The injective MPS theorem on the internal triangle or square therefore
relates them by internal bond gauges. Thus the original fiber within
this open set is contained in a single original gauge orbit.
Lemma~\ref{lem:fiber}, applied to the original parameter space, gives
\begin{equation}
\begin{gathered}
\nu_{\rm generic}=0,\\
\rank J_{\rm generic}=Nq^4-E(q^2-1)-N+1
\end{gathered}
\label{SM-eq:trivalentgeneric}
\end{equation}
for all three tilings and every $q\ge2$ under the stated quotient
conditions. The original sites have $d=q<q^3$; the supercritical
statement is used only after the specified grouping.

\subsection{Injective partitions and the generic snub-square theorem}
\label{SM-sec:archgenericproof}
We use Theorem~4 of Ref.~\cite{molnar2018normal} with its two
partition hypotheses checked explicitly. For every edge $uv$, there
must be a partition into three nonempty injective blocks $A,B,C$
with $u\in A$, $v\in B$, and $uv$ the sole edge between $A$ and $B$.
For every vertex $v$, there must be a region $R$ not containing $v$
such that $R$, $R\cup\{v\}$, and both of their complements are
injective. Injectivity always refers to the map from all boundary
virtual legs to the physical legs of the region. Multiple edges
between a pair of blocks can be combined into one larger virtual leg.

\paragraph{A certificate valid for every \texorpdfstring{$q$}{q}.}
Route each boundary input of a region to a distinct physical endpoint
along internal graph edges, with no internal edge used twice. Paths
may share vertices if they use distinct pairs of virtual legs there.
Kronecker deltas transmit each input label along its path and into
the physical leg at its endpoint; all unused indices are fixed to zero.
The resulting boundary map obeys
\begin{equation}
B_R^\dagger B_R=I_{q^{|\partial R|}}.
\label{SM-eq:routeisometry}
\end{equation}
This proves a nonzero injectivity minor for every integer $q\ge2$.
Each region may use a different witness: a finite product of nonzero
minor polynomials is nonzero, so all requested regions are injective
on a common nonempty Zariski-open set.

The maximum-flow/minimum-cut theorem and integral-capacity flow
construction~\cite{ford1956maximal,diestel2017graph} find such routes
using capacity one per
internal edge and per physical endpoint. Opposite flows are canceled
before path decomposition. An independent checker verifies each
boundary input, path edge, edge capacity, and endpoint. In particular,
the scalar dimension test $|\partial R|\le |R|$ alone is insufficient.
Every $S\subseteq R$ must satisfy the necessary local cut condition
\begin{equation}
|E(S,V\setminus R)|\le |S|+|E(S,R\setminus S)|.
\label{SM-eq:routecut}
\end{equation}
Failure of this condition obstructs that region's injectivity; it is
not evidence for an additional global tangent kernel.

\paragraph{Snub-square templates and extension to large periods.}
Use the grid-with-diagonals convention above and even periods $L_x,L_y$.
Translations by $(2,0)$ and $(0,2)$ have four vertex orbits and ten
undirected edge orbits. The saved open-region certificates contain
four $6\times6$ rectangles, one for each origin parity; ten edge
partitions $A,B,C_0$ inside $24\times24$ boxes; and four vertex
partitions $R,R\cup\{v\},C_0,C'_0$ inside $40\times40$ boxes.
Every listed region has a routing certificate. In particular the
complement regions retain their boundary legs on the outside of the
box, so their injectivity does not depend on closing a finite torus.

For even $L_x,L_y\ge48$, the complement of either box interval in
each coordinate has at least eight sites. Every site outside the box
is covered by a $6\times6$ rectangle disjoint from that box: choose
a length-six interval in a complementary coordinate interval containing
the site, and any length-six interval containing its other coordinate.
All such rectangles are translates of the four certified parity types.
A complementary interval of length at least eight is covered by
consecutive length-six windows, with successive windows overlapping.
All four origin parities are already certified, so no parity restriction
on the window containing a given site is needed. The union of injective regions
is injective (Lemma~4 of Ref.~\cite{molnar2018normal}), applied step by
step along the overlapping chain; joining these rectangles to $C_0$
or $C'_0$ therefore proves injectivity of the entire torus complement.
This verifies every edge and vertex partition required by the theorem
for all the stated rectangular periods.

\paragraph{Generic rank.}
The finite set of injectivity conditions defines a common nonempty
Zariski-open subset. The normal theorem identifies its fibers up to
bond gauge, and Lemma~\ref{lem:fiber} gives
\begin{equation}
\rank J_{\rm generic}=P-g_0,\qquad \nu_{\rm generic}=0.
\label{SM-eq:snubgeneric}
\end{equation}
Since snub-square has $E=5N/2$ and $P=Nq^6$, its generic affine rank is
$Nq^6-\frac{5N}{2}(q^2-1)-N+1$.

The bound $48$ is sufficient, not optimal. Separate full-torus
certificates also cover $24\times24$ and $26\times26$, checking
all $1440$ and $1690$ edges and all $576$ and $676$ vertices,
respectively. Independent template verification checks the mappings
to $48\times48$, $50\times50$, and $64\times64$ tori, the orbit
coverage, all routes, and the complement covering. The extension
to every larger even rectangular period is the preceding covering
proof, not a numerical size extrapolation.

\subsection{Square, kagome, and ruby: three further global results}
\label{SM-sec:m28generic}
For every integer $q\ge2$, independent generic tensors have
$\nu=0$ on square tori with coordinate periods $L_x,L_y\ge20$,
and on kagome and ruby tori with honeycomb source-cell periods
$L_x,L_y\ge64$. These are sufficient bounds. In all three cases the
coordination number is $z=4$, hence
\begin{equation}
\rank J_{\rm generic}=Nq^5-2N(q^2-1)-N+1.
\label{SM-eq:m28rank}
\end{equation}
The periods impose no equality between tensors at different sites.

\paragraph{Open partitions and their extension.}
Square has two edge types and one vertex type under cell translation.
For the edge $((0,0),(1,0))$, use
$A=[-3,0]\times[0,3]$, $B=[1,4]\times[-3,0]$ and the remaining
sites in the box $[-8,7]^2$. Exchanging coordinates gives the other
edge type. For the vertex $v=(0,0)$, use $R\cup\{v\}=[0,4]^2$
and the two complements in the same box. All these regions have
integer routing certificates, as does a $4\times4$ square.

For kagome, label the three sites per source cell by $t=0,1,2$.
One triangle is inside each cell; the other consists of
$(x,y,0),(x+1,y,1),(x,y+1,2)$. This graph agrees with the
honeycomb line-graph constructor. Six edge templates and three
vertex templates fit inside open $48\times48$ cell boxes.
Each region is certified on the original graph, retaining all
external legs of the box. A full $4\times4$ cell block is also
injective for generic tensors, by the same routing construction.

If a coordinate complement of a box has length at least four,
every site outside the box belongs to a $4\times4$ cell block
disjoint from it, and consecutive blocks can be chosen to overlap.
The union lemma then extends the certified open complement to the
whole torus. Thus $20=16+4$ suffices for square;
the chosen kagome bound $64$ exceeds $48+4$. Complete edge and
vertex types, cell-translation automorphisms, all paths, and the
covering across periodic seams are independently checked. The
generic-fiber argument leading to Eq.~\eqref{SM-eq:snubgeneric}
therefore proves Eq.~\eqref{SM-eq:m28rank} for these two lattices.

\paragraph{The ruby coarse graph.}
Ruby's disjoint triangles cover all vertices. Contracting each
triangle gives a honeycomb graph, with two original edges forming
each coarse bond of dimension $q^2$ and a physical dimension $q^3$
per coarse site. The three coarse edges incident on a triangle
touch its three different pairs of original sites; each coarse
edge acts on one external leg at each of its two sites.

Three coarse-edge templates and two coarse-vertex templates fit
inside open $48\times48$ source-cell boxes. Their regions are
unions of complete triangles, and every injectivity certificate
is checked on the original ruby graph with its $q$-dimensional
legs. Full $4\times4$ source-cell blocks are also certified.
The preceding union argument verifies the fundamental theorem's
conditions on the coarse graph for $L_x,L_y\ge64$.
No single coarse tensor is assumed injective. The resulting
$\mathrm{GL}(q^2)$ gauges must still be resolved into original
$\mathrm{GL}(q)$ edge gauges.

\paragraph{A two-site splitting lemma for every \texorpdfstring{$q$}{q}.}
Group the physical and two external legs of an original site as
outputs of its two internal legs, giving a generically full-column-rank
map $A_i:\mathbb C^{q^2}\to E_i=\mathbb C^{q^3}$.
Let $Q_i:E_i\to E_i/\im A_i$ be its quotient map, represented
by any complete left annihilator. For two adjacent sites $i,j$,
contract their shared internal leg and leave their other internal
legs open. Act with $T\in\operatorname{End}(\mathbb C^q\otimes
\mathbb C^q)$ on a selected external leg at each site, and then
with $Q_i\otimes Q_j$. Generically the kernel of this operator map is
\begin{equation}
\operatorname{End}(\mathbb C^q)\otimes I
+I\otimes\operatorname{End}(\mathbb C^q).
\label{SM-eq:rubysplit}
\end{equation}
The inclusion of this space follows from $Q_iA_i=Q_jA_j=0$.

To prove the reverse inclusion for every $q$, choose $q^2$ distinct
numbers $t_{ab}$ and write
$v_{ab}=(1,t_{ab},\ldots,t_{ab}^{q-1})^T$. Let $W$ be an
invertible, entrywise nonzero $q\times q$ matrix. With $p$ the
physical output, $y$ the spectator external output, $x$ the selected
external output, and $a,b$ the internal inputs, take
\begin{align}
A[p,y,x;a,b]&=\delta_{p,a}\delta_{y,b}(v_{ab})_x,\nonumber\\
B[p,y,x;a,b]&=W_{ap}\delta_{y,b}(v_{pb})_x.
\label{SM-eq:rubysplitwitness}
\end{align}
Both site maps have full column rank. Their images are direct sums
of the lines spanned by $v_{py}$ in the output fibers $(p,y)$.
After the shared internal leg is contracted, the coefficient of
output fibers $(a,b)$ and $(c,d)$ is $W_{ac}\ne0$. Thus the
double projection contains all detection functionals
\begin{equation}
T\longmapsto(\lambda\otimes\mu)T(v_{ab}\otimes v_{cd}),
\qquad \lambda v_{ab}=\mu v_{cd}=0.
\label{SM-eq:rubyfunctionals}
\end{equation}
The single-site functionals $K\mapsto\lambda Kv_{ab}$ have
common kernel $\mathbb C I$. A matrix preserving all these lines
is diagonal in any $q$ of them, while a further Vandermonde line
has all nonzero coordinates in that basis and forces its eigenvalues
to coincide. Such $q+1$ lines exist since $q^2\ge q+1$.
These functionals span the dual of
$\operatorname{End}(\mathbb C^q)/\mathbb C I$; their products
span the dual of its tensor square. This proves
Eq.~\eqref{SM-eq:rubysplit}, with detected rank $(q^2-1)^2$.
The property is Zariski open on full-column-rank site maps, so
all required site and pair conditions hold simultaneously on a
nonempty open set. Their separate witnesses need not coincide.

\paragraph{Lifting the coarse fiber to original edge gauges.}
Work on the common nonempty open set of all the preceding
conditions and the maximal-rank open set of the original
contraction map. Every original kernel tangent is then tangent
to a constant-state curve (constant-rank theorem). The coarse fundamental theorem places
its blocked tensors in a single coarse gauge orbit. Since this
orbit is smooth, the induced block tangent is a sum of coarse-edge
generator actions in $\operatorname{End}(\mathbb C^{q^2})$.
This does not assume tangent completeness at an arbitrary
special point of the coarse tensor space.

Within a triangle, apply $Q_i,Q_j$ to a pair of sites.
Every single-site variation vanishes, as do the actions of the
other two coarse edges, each of which leaves one of these sites
untouched. Applying the third site's left inverse opens the two
remaining internal legs. The surviving coarse generator satisfies
Eq.~\eqref{SM-eq:rubysplit}, so it is $X\otimes I+I\otimes Y$.
Each coarse generator splits into two original edge gauges.
Subtract them from the original variation. Each triangle block then
varies only by a block-internal null variation, and its grouped sites
are injective; the injective-network result on that simple triangle
leaves only internal edge gauges. This proves generic absence of additional
kernel for ruby and completes Eq.~\eqref{SM-eq:m28rank}.

\subsection{Elongated triangular and maple-leaf: all-\texorpdfstring{$q$}{q} global results}
\label{SM-sec:m29generic}
Two further blockings establish $\nu_{\rm generic}=0$ for every
$q\ge2$. Elongated triangular uses the original grid with horizontal
and vertical edges and diagonals $(x,y)(x+1,y+1)$ for even $y$.
Sufficient periods are $L_x\ge20$ and even $L_y\ge40$.
For maple-leaf, use triangular integer coordinates with edge directions
$\pm(1,0),\pm(0,1),\pm(1,1)$ and vacancies $x+2y=0\pmod7$.
Write the vacancy translations as $a=(2,-1)$ and $b=(1,3)$.
Sufficient period vectors are $L_1a,L_2b$ with $L_1,L_2\ge40$.
These lengths count vacancy cells, not original coordinate steps.
Both original graphs have degree five, so the affine rank is
\begin{equation}
\rank J=Nq^6-\left[\frac{5N}{2}(q^2-1)+N-1\right].
\label{SM-eq:m29rank}
\end{equation}

\paragraph{Block graphs and original-leg injectivity.}
For elongated triangular, pair $(x,2j)$ with $(x,2j+1)$.
The coarse graph is square, with physical dimension $q^2$,
horizontal bond dimension $q^3$, and vertical bond dimension $q$.
A horizontal bond contains the two horizontal original edges and
one diagonal. At a dimer, the left and right horizontal bonds
partition their three legs between the original sites as $1+2$
and $2+1$, respectively. Each site also has one spectator external
leg belonging to a vertical coarse bond.

For maple-leaf, label the six neighbors of a vacancy by
\begin{equation}
\begin{aligned}
(h_0,h_1,h_2)&=((1,0),(1,1),(0,1)),\\
(h_3,h_4,h_5)&=((-1,0),(-1,-1),(0,-1)).
\end{aligned}
\end{equation}
A site $(r,s,t)$ is at $ra+sb+h_t$. The triangles
\begin{align}
T^0_{rs}&=\{(r,s,0),(r+1,s,2),(r+1,s+1,4)\},\nonumber\\
T^1_{rs}&=\{(r,s,1),(r,s+1,5),(r+1,s+1,3)\}
\label{SM-eq:mapletriangles}
\end{align}
are disjoint and cover the original vertices. Their coarse graph
is honeycomb, with physical and bond dimensions both $q^3$.
Each coarse bond contains three original edges and touches two
sites at each endpoint triangle, in a $2+1$ distribution.
The three coarse bonds incident on a triangle touch its three
different pairs of sites.

All required injective regions are certified on the original graph,
with one $q$-dimensional output per original site. Elongated
triangular has two coarse-edge types and one coarse-vertex type;
maple-leaf has three and two. Their open templates fit inside
$16\times16$ dimer-cell boxes and $32\times32$ vacancy-cell boxes,
respectively. Full $4\times4$ cell blocks are also injective.
The exterior-cover argument of Appendix~\ref{SM-sec:m28generic} proves
the normal fundamental theorem's conditions on the coarse graphs.
In dimer cells, periods of at least $20$ suffice, giving the original
coordinate bounds above. Vacancy-cell periods of at least $40$
suffice for maple-leaf. The coarse tensors are constrained by
their internal networks; no genericity in the full coarse tensor
space is assumed.

\paragraph{Simultaneous splitting at a dimer.}
A site grouped as physical and four external outputs is a map
$\mathbb C^q\to\mathbb C^{q^5}$. Partition its external legs
into a small group $s$ of dimension $q$, a large group $(x,y)$
of dimension $q^2$, and a spectator $t$ of dimension $q$.
With $\Omega=\sum_j|j,j\rangle$, define
\begin{equation}
\begin{gathered}
v_0=\Omega_{s,x}\Omega_{y,t},\qquad
v_1=\Omega_{s,y}\Omega_{x,t},\\
v_p=v_0\quad(p\ge2).
\end{gathered}
\end{equation}
The two site maps are
\begin{equation}
A[p,\mathrm{ext};c]=\delta_{pc}v_p,\qquad
B[r,\mathrm{ext};c]=W_{cr}v_r,
\label{SM-eq:dimersplitwitness}
\end{equation}
where $W$ is invertible and entrywise nonzero. Both maps have
rank $q$, with images given by one line in each physical slice.
Let $Q_A,Q_B$ be their complete left annihilators.

Consider the joint local operator space
\begin{equation}
\mathcal D=(\operatorname{End}(\mathbb C^q)/\mathbb CI)
\oplus(\operatorname{End}(\mathbb C^{q^2})/\mathbb CI).
\end{equation}
If $(K_s+K_{xy})v_p$ is proportional to $v_p$ for every $p$,
the first two slices imply
\begin{equation}
K_{xy}=\lambda_0 I-K_s^{\mathsf T}\otimes I
=\lambda_1 I-I\otimes K_s^{\mathsf T}.
\end{equation}
Taking traces gives $\lambda_0=\lambda_1$, and then $K_s$
and $K_{xy}$ are scalar. Thus the local detection functionals
span $\mathcal D^*$. Contracting the internal dimer edge gives
the nonzero coefficient $W_{pr}$ for every pair of physical
slices. The double projection therefore contains every product
of these local functionals and is injective on
$\mathcal D\otimes\mathcal D$.

The non-single-site parts of the left and right horizontal
generators occupy different direct summands,
small$\otimes$large and large$\otimes$small. Their combined
detected rank is $2(q^2-1)(q^4-1)$, so cancellation between the
two coarse edges is excluded. Each generator splits into
single-site actions at this dimer endpoint. Actions on a vertical
coarse edge already act on a single original edge and vanish under
the double projection.

\paragraph{Asymmetric splitting at a maple-leaf triangle.}
Each grouped site is a map $\mathbb C^{q^2}\to\mathbb C^{q^4}$.
For a chosen pair, one site has two selected external legs and
the other has one. Write $Q=q^2$. At the first site, regard the
selected outputs and the spectator outputs (physical plus one
external leg) as two $Q$-dimensional factors. Choose its $Q$
input columns, viewed as $Q\times Q$ matrices, to span
\begin{equation}
\mathcal L=\operatorname{span}\{I,D,D^2,\ldots,D^{Q-2},S\},
\label{SM-eq:mapleleftspace}
\end{equation}
where $D$ is diagonal with $Q$ distinct entries and $S$ is a
cyclic shift. This space has dimension $Q$. Its left stabilizer
is scalar: $K\mathcal L\subseteq\mathcal L$ first gives
$K=f(D)+cS$ from $I\in\mathcal L$; the off-diagonal part of
$KD\in\mathcal L$ forces $c=0$; and $KS\in\mathcal L$
forces $f(D)$ to be scalar. Hence $K\mapsto Q_AKA$ detects
all of $\operatorname{End}(\mathbb C^Q)/\mathbb CI$.

At the second site use selected output $x$, spectator external
outputs $y,z$, physical output $p$, and internal inputs $c,d$:
\begin{equation}
B[p,z,y,x;c,d]=\delta_{pc}\delta_{zd}\delta_{xy}.
\label{SM-eq:mapleBell}
\end{equation}
Its image is the full $(p,z)$ space times the Bell line on $(y,x)$.
The quotient detects every non-scalar operator on $x$, while
$p$ exposes the shared internal input $c$. After contracting that
input with $A$, the pair detector is the tensor product of the
large-site detector $K\mapsto Q_AKA$ and the Bell detector.
It has rank $(q^4-1)(q^2-1)$; its kernel is exactly
\begin{equation}
\operatorname{End}(\mathbb C^{q^2})\otimes I
+I\otimes\operatorname{End}(\mathbb C^q).
\label{SM-eq:maplesplit}
\end{equation}
Within a triangle, double projection at the selected pair kills
every original single-site variation and the other two coarse-edge
actions. The third site's left inverse exposes the remaining
internal legs. Thus Eq.~\eqref{SM-eq:maplesplit} applies to each
coarse-edge generator separately. The column-rank and splitting
conditions are nonempty Zariski-open conditions, so they hold
simultaneously with all required region injectivities.

\paragraph{Intersecting the two endpoint splittings.}
For either lattice, label the three original legs of a coarse
$q^3$ edge. At both endpoints the generator lies in the sum of
operator algebras supported on the two original sites. Expand
operators using $\operatorname{End}(\mathbb C^q)=
\mathbb CI\oplus\mathfrak{sl}(q)$ on each original leg.
An allowed nontrivial support must lie within one part of each
endpoint's site partition. The common refinement consists of
single legs: two legs in the same part at both endpoints would be
parallel original edges, excluded by simplicity of the original
graph. Therefore the intersection is exactly the sum of the
three original single-edge operator algebras.
The transpose and minus sign at the opposite end of a bond preserve
these support spaces and give the usual matched gauge actions.

Apply the generic-fiber argument of Appendix~\ref{SM-sec:m28generic}
and subtract these original external-edge gauges. Each dimer or
triangle block then varies only by a block-internal null variation.
All grouped site maps are injective, so its remaining kernel consists
of internal bond gauges. This proves Eq.~\eqref{SM-eq:m29rank} for the stated
finite tori and every $q\ge2$.

\subsection{Dice and three 2-uniform tilings: open partitions}
\label{SM-sec:generic39proof}
We prove the four period-30 statements of Table~\ref{SM-tab:generic}
on the original graphs, for every integer $q\ge2$.
Write a vertex as $(x,y,\alpha)$, where $(x,y)\in\mathbb Z^2$
and $\alpha\in V_0$ is its cell label. For dice, $V_0=\{0,1,2\}$:
the rim vertex $(x,y,1)$ is adjacent to hubs
$(x,y,0),(x+1,y,0),(x,y+1,0)$, and $(x,y,2)$ to
$(x+1,y,0),(x,y+1,0),(x+1,y+1,0)$.
For \texttt{t2016}--\texttt{t2018}, use the exact integer cells of
Appendix~\ref{SM-sec:twouniformSM}. Their vertex and edge counts per cell are
$(n,e)=(12,30),(8,20),(7,18)$, respectively; dice has $(3,6)$.
Every edge changes either cell coordinate by at most one.

\paragraph{Open templates and simultaneous injectivity.}
The routing construction of Eq.~\eqref{SM-eq:routeisometry} applies
without change: each boundary input is routed to a different physical
endpoint, no internal edge is used twice, and independent paths may
pass through the same site on different legs. Delta tensors produce
an isometric boundary map for every $q$. Each certificate therefore
gives a nonzero injectivity minor in the original parameter space.
A finite product of these minors is nonzero, even though their sparse
witnesses need not be the same tensor assignment.

Let $I=[-12,11]\cap\mathbb Z$, $I_0=[-8,7]\cap\mathbb Z$, and
$I_+=[-11,10]\cap\mathbb Z$, and define
\begin{equation}
Q=I^2\times V_0,\qquad
K=I_0^2\times V_0,\qquad K^+=I_+^2\times V_0.
\label{SM-eq:extensionboxes}
\end{equation}
For each graph the certificates specify an outer region $O$ and a
routable covering region $D$ satisfying
\begin{gather}
K^+\subseteq O\subseteq Q,\nonumber\\
D\subseteq\{0,\ldots,5\}^2\times V_0,\nonumber\\
\{(3,3,\alpha):\alpha\in V_0\}\subseteq D.
\label{SM-eq:extensioncover}
\end{gather}
The regions have trimmed boundaries; an untrimmed rectangular corner
need not satisfy the necessary cut condition~\eqref{SM-eq:routecut}.

For every edge translation type, the data give disjoint
$A,B\subseteq K$, with $uv$ their sole connecting edge, and routing
certificates for $A,B,O\setminus(A\cup B)$. For every vertex type,
they give $R^+=R\cup\{v\}\subseteq K$, $v\notin R$, and
certificates for $R,R^+,O\setminus R,O\setminus R^+$.
All external legs are retained, including those leaving $Q$.
The independent checker verifies the set identities, the unique
connecting edge, and every routing path. The counts of edge and
vertex types are $(6,3),(30,12),(20,8),(18,7)$ for dice,
\texttt{t2016}, \texttt{t2017}, and \texttt{t2018}, respectively.
These exhaust the original edge and vertex translation orbits.

\paragraph{Extension to all rectangular periods at least thirty.}
For $L_x,L_y\ge30$, the templates and their external neighbors embed
without identifications: the coordinate span of $Q$ with its one-cell
neighbor shell is less than either period. To cover a cell
$p=(x,y)$ outside $K^+$, translate the center cell $(3,3)$ of $D$
to $p$. In at least one coordinate, the resulting six-cell interval
is disjoint from $I_0$. This remains true across a periodic seam:
the cyclic complement of $I_0$ has length at least fourteen, and
$p$ lies at least three cells beyond its endpoints whenever its
coordinate is outside $I_+$. Hence this translate of $D$ avoids $K$
and contains every vertex in cell $p$.

For an edge template the full complement is thus
\begin{equation}
V\setminus(A\cup B)
=[O\setminus(A\cup B)]\ \cup\!\bigcup_t(D+t),
\label{SM-eq:extensioncomplement}
\end{equation}
where the selected translates avoid $K$ and cover $V\setminus K^+$.
The part within $K^+$ is already contained in the first term.
Every region in this finite union is injective on a common nonempty
open set; the union lemma of Ref.~\cite{molnar2018normal} proves
injectivity of the full complement. The identical construction works
for both vertex-template complements. Translating the templates
checks every edge and vertex, including those across periodic seams.
This is an analytic extension to all stated periods, not an
extrapolation from a finite sequence of tori.

\paragraph{Fibers, differential rank, and the WGS rank drop.}
Both hypotheses of Theorem~4 of Ref.~\cite{molnar2018normal} now
hold on a common nonempty Zariski-open subset of the original tensor
space. It identifies equal-state tensors there by original-edge
gauges. Any site scalars whose product is one can be absorbed into
scalar bond gauges along a spanning tree, since the graph is connected.
Intersecting also with the maximal gauge-rank open set and applying
Lemma~\ref{lem:fiber} proves Eq.~\eqref{eq:genericrank}.
There is no coarse-parameter genericity assumption and no composite
bond to split.

With $M=L_xL_y$, write $\rank J_{\rm generic}=1+M\rho_G(q)$.
The four explicit rank polynomials are
\begin{align}
\rho_{\rm dice}(q)&=q^7+2q^4-6(q^2-1)-3,\nonumber\\
\rho_{\texttt{t2016}}(q)&=12q^6-30(q^2-1)-12,\nonumber\\
\rho_{\texttt{t2017}}(q)&=8q^6-20(q^2-1)-8,\nonumber\\
\rho_{\texttt{t2018}}(q)&=q^7+6q^6-18(q^2-1)-7.
\label{SM-eq:newgenericranks}
\end{align}
The first graph has degrees $6,3,3$ per cell, the middle two have
only degree-five vertices, and the last has one degree-six and six
degree-five vertices. Since these periods also preserve the WGS
short neighborhoods, the exact rank increases from WGS to generic
tensors, divided by $N$, are respectively
\begin{gather}
2r^2,\qquad r^2\bigl(\tfrac23q^2+2q+2\bigr),\nonumber\\
r^2\bigl(\tfrac34q^2+2q+2\bigr),\qquad
r^2\bigl(\tfrac67q^2+\tfrac{16}{7}q+2\bigr),
\label{SM-eq:newrankdrops}
\end{gather}
where $r=q-1$. The other families use the separate constructions
below; the six $q>2$ extensions marked Open in
Table~\ref{SM-tab:generic} remain unresolved.

\subsection{Triangle blocking for two further tilings and shuriken}
\label{SM-sec:triangleextension}
We prove Eq.~\eqref{eq:genericrank} for \texttt{t2001} and
\texttt{t2005} when $L_x,L_y\ge32$, and for shuriken when both
original periods are even and at least 64. All tensors at original
sites remain independent, and the proof applies to every integer
$q\ge2$.

\paragraph{The three blockings.}
The original cell of \texttt{t2001} has six degree-four and twelve
degree-three vertices. Its two disjoint triangles contain all six
degree-four vertices. Contract these triangles and retain the other
twelve vertices as singleton blocks. The coarse cell has fourteen
vertices and twenty-four edges, each containing one original edge.
For \texttt{t2005}, six disjoint triangles cover all eighteen
degree-four vertices per cell. The coarse cell has six vertices and
twelve edges: six single and six double bonds.

For shuriken, use its definition as the line graph of the square-octagon
graph $H$. Label the four vertices in each cell of $H$ cyclically by
$t=0,1,2,3$. Besides the intracell cycle, its edges are
$(x,y,0)\sim(x+1,y,2)$ and $(x,y,1)\sim(x,y+1,3)$.
The coloring
\begin{equation}
c(x,y,t)=x+y+(t\bmod2)\pmod2
\label{SM-eq:shurikencoloring}
\end{equation}
is bipartite. The three edges incident on each black vertex form a
triangle in the line graph. Every edge of $H$ has one black endpoint,
so these triangles partition the original shuriken sites. The coloring
is periodic in $2\times2$ original cells. Each such blocked cell
contains twenty-four original sites, eight triangle blocks, and twenty
coarse edges: sixteen single and four double bonds.

In each blocking the regrouped original edges shift cell coordinates
by at most one. Every triangle site has degree four, with two internal
and two external legs. Every double bond has a triangle endpoint
where its two original edges touch distinct sites $i,j$, and no other
incident coarse bond touches both $i,j$. The frozen blockings verify
this isolation property for all six double-bond types of
\texttt{t2005} and all four of shuriken.

\paragraph{Normal partitions in the original parameter space.}
Let $\mathcal B$ denote the block labels in one blocked cell. In
blocked-cell coordinates define
\begin{equation}
\begin{gathered}
Q=[-12,11]^2\times\mathcal B,\\
K=[-8,7]^2\times\mathcal B,\qquad
K^+=[-11,10]^2\times\mathcal B,
\end{gathered}
\label{SM-eq:triangleboxes}
\end{equation}
where intervals contain integer coordinates. The certificates provide
an open region $O$ and a covering region $D$ satisfying
\begin{equation}
\begin{gathered}
K^+\subseteq O\subseteq Q,\qquad
D\subseteq[0,5]^2\times\mathcal B,\\
\{(3,3,b):b\in\mathcal B\}\subseteq D.
\end{gathered}
\label{SM-eq:trianglecover}
\end{equation}
For each coarse-edge translation type there are disjoint $A,B\subseteq K$
with that edge as their only coarse connection, and routes certifying
$A$, $B$, and $O\setminus(A\cup B)$. For each coarse-vertex type $v$
there are $R$ and $R^+=R\cup\{v\}\subseteq K$, with routes for
$R$, $R^+$, $O\setminus R$, and $O\setminus R^+$. The region $D$
also has a routing certificate.

Every region is expanded to all original sites in its blocks before
checking its boundary inputs and paths. Thus all routed legs and
physical endpoints have dimension $q$, even when a coarse bond has
dimension $q^2$ and a triangle has physical dimension $q^3$.
The edge-disjoint paths to distinct physical endpoints give the delta
tensor witnesses of Eq.~\eqref{SM-eq:routeisometry}. Their injectivity
minors are nonzero polynomials in the original tensor entries for
every $q\ge2$. The finite product of these minors, including their
required translations, is nonzero.

Take both blocked periods at least 32. The templates and their exterior
neighbor shells embed without identifying original boundary legs.
For any cell $p$ outside $K^+$, translate the center $(3,3)$ of $D$
to $p$. In at least one coordinate its interval of length six avoids
$[-8,7]$, including across the periodic seam. This translate contains
every block at $p$ and is disjoint from $K$. Consequently the full
complement of $A\cup B$ is the union of $O\setminus(A\cup B)$ and
these translates of $D$. The same construction covers the complements
of $R$ and $R^+$. Applying the injective-region union lemma after
expanding all blocks proves injectivity of each full complement.
The complete lists of coarse-edge and coarse-vertex types therefore
verify both partition hypotheses of Theorem~4 of
Ref.~\cite{molnar2018normal}.
For \texttt{t2001} and \texttt{t2005} the blocked periods equal the
original periods; for shuriken they are half the original periods.
This proves the sufficient geometric bounds in Table~\ref{SM-tab:generic}.

\paragraph{Splitting double bonds and removing internal kernels.}
Within every triangle, group the physical and two external legs of
site $i$ into the output space $E_i=\mathbb C^{q^3}$, giving an
independent map $A_i:\mathbb C^{q^2}\to E_i$.
Impose full column rank for these maps and the generic two-site
detector condition~\eqref{SM-eq:rubysplit} for every selected double-bond
endpoint. The Vandermonde witnesses~\eqref{SM-eq:rubysplitwitness}
prove that each condition is nonempty and open for every $q\ge2$.

We also impose generic tangent completeness for each internal triangle
contraction, with the three spaces $E_i$ as physical outputs.
This is an independently justified local condition: the three maps
$\mathbb C^{q^2}\to\mathbb C^{q^3}$ range over their full parameter
spaces, single-site injectivity is nonempty and open, and the
injective fundamental theorem on the simple triangle identifies
its finite fibers with internal-edge gauge orbits. The generic-fiber
argument of Lemma~\ref{lem:fiber} then gives a nonempty open set
where the local differential kernel consists of internal gauges.
Reshaping the original tensors into the maps $A_i$ is an invertible
linear change of coordinates, so this local open set pulls back to
a nonempty open set of the original parameters.

The lifting step of Appendix~\ref{SM-sec:blockingtemplate} now applies:
intersect these conditions with all routed-region injectivity
conditions, the maximal original gauge-rank open set, and the
maximal-rank open set of the original global contraction. The
maximal-rank condition does not assume a value for the global rank.

For a double bond, choose the triangle endpoint and sites $i,j$
specified above. Apply $Q_i\otimes Q_j$, where
$Q_i:E_i\to E_i/\im A_i$, to the block tangent equation.
Every original single-site variation vanishes, since at least one
projector acts on an unchanged site. Every other coarse-bond action
also vanishes by the isolation property. Applying the third site's
left inverse opens the remaining two internal legs. Thus the target
generator is in the kernel of the two-site detector and has the form
$X\otimes I+I\otimes Y$ by Eq.~\eqref{SM-eq:rubysplit}.
The paired negative transpose at the other endpoint preserves this
decomposition, so one detecting endpoint suffices to split the double
bond into two original-edge gauges. No splitting is needed for
\texttt{t2001}.

Subtract these original external-edge gauges from $\delta A$.
Every block contraction now has zero variation. Singleton variations
are zero, and triangle variations are internal-edge gauges by the
local open condition established above. Hence
$\ker J_A\subseteq\im\Phi_A$; the reverse inclusion follows from
$J_A\Phi_A=0$. The maximal original gauge rank gives
Eq.~\eqref{eq:genericrank}.

\paragraph{Explicit ranks and WGS rank increases.}
With $M=L_xL_y$ counting \emph{original} cells, including for shuriken,
write $\rank J_{\rm generic}=1+M\rho_G(q)$. Then
\begin{align}
\rho_{\texttt{t2001}}(q)&=6q^5+12q^4-30(q^2-1)-18,\nonumber\\
\rho_{\texttt{t2005}}(q)&=18q^5-36(q^2-1)-18,\nonumber\\
\rho_{\rm shuriken}(q)&=6q^5-12(q^2-1)-6.
\label{SM-eq:trianglegenericranks}
\end{align}
The respective vertex and edge counts are $(N,E)=(18M,30M)$,
$(18M,36M)$, and $(6M,12M)$.
Since the stated periods preserve the WGS neighborhoods, the exact
rank increases from WGS to generic tensors, divided by $N$, are
\begin{equation}
\frac{r^2(2q+7)}9,\qquad
\frac{r^2(2q+4)}3,\qquad
\frac{r^2(4q+3)}3,
\label{SM-eq:trianglerankdrops}
\end{equation}
in the same order, with $r=q-1$.

\subsection{Four further 2-uniform tilings and checkerboard}
\label{SM-sec:fiveextension}
We establish Eq.~\eqref{eq:genericrank} for \texttt{t2002},
\texttt{t2003}, \texttt{t2004}, and \texttt{t2011} at original
periods $L_x,L_y\ge32$, and for checkerboard at even original
periods $L_x,L_y\ge64$. The statements hold for independent
original tensors at every integer $q\ge2$.

\paragraph{Exact blockings.}
In the integer coordinates of the 2-uniform dataset, select the
following triangles for \texttt{t2002}:
\begin{align}
T^0_{xy}&=\{(x,y,2),(x,y,3),(x,y,4)\},\nonumber\\
T^1_{xy}&=\{(x,y,5),(x,y,6),(x,y,7)\}.
\label{SM-eq:t2002blocks}
\end{align}
Labels 0 and 1 remain singleton blocks. Each triangle contains two
degree-four sites and one degree-three site. The coarse cell has
four vertices and seven edges, six single and one double.
The double bond has an endpoint where it touches the two degree-four
sites, and no other incident coarse bond touches both.

For \texttt{t2003} and \texttt{t2004}, respectively, use the dimers
\begin{align}
D^{(3)}_{xy}&=\{(x,y,0),(x,y,1)\},\nonumber\\
D^{(4)}_{xy}&=\{(x,y,0),(x,y-1,1)\}.
\label{SM-eq:twouniformdimers}
\end{align}
Both dimer sites have degree five. The remaining degree-four sites
are singletons: labels 2 and 3 for \texttt{t2003}, and label 2 for
\texttt{t2004}. Their coarse cells have three vertices and six edges,
and two vertices and four edges, respectively. Each has one triple-bond
translation type; all other bonds are single. At each dimer, the two
incident triple bonds partition their original legs as $1+2$ and
$2+1$ between its sites. Each site has one further spectator leg.
At both endpoints of a triple bond the common refinement of these
site partitions consists of three individual original edges.

For \texttt{t2011}, the triangles
\begin{equation}
T_{xy}=\{(x,y,2),(x,y+1,0),(x,y+1,1)\}
\label{SM-eq:t2011blocks}
\end{equation}
partition all original sites. There is one coarse vertex type and
two coarse-edge types, one single and one double. The double bond
touches a distinct pair of sites at a detecting endpoint, and other
incident bonds do not touch both sites. An edge between identical
block labels in different cells is not a self-loop of the finite graph.

For checkerboard, color the vertices $(x,y)$ of the square source graph
by $x+y\pmod2$. The four source edges incident on each black vertex
form a $K_4$ in the line graph. Every source edge has exactly one
black endpoint, so these cliques partition the original checkerboard
sites. Use a $2\times2$ source supercell to preserve the coloring.
Each blocked cell contains eight original sites, two $K_4$ blocks,
and eight coarse edges, four single and four double. Each original
site has three internal and three external legs. Every double bond
has a detecting endpoint at two distinct clique sites, with no other
incident coarse bond touching both. All these structural properties
are verified from the frozen original adjacency.

\paragraph{Open partitions and period bounds.}
Use the boxes and cover inclusions of
Eqs.~\eqref{SM-eq:triangleboxes} and \eqref{SM-eq:trianglecover} for each
new coarse graph. The frozen certificates provide all edge partitions
$A,B,O\setminus(A\cup B)$, all vertex partitions
$R,R\cup\{v\},O\setminus R,O\setminus(R\cup\{v\})$, and the
cover $D$, with the small regions contained in $K$.
They cover all 27 coarse-edge and 12 coarse-vertex translation types
across the five graphs. Every region expands to complete original
blocks before its paths are checked. Thus all routed virtual legs
and physical endpoints have dimension $q$; the physical dimensions
$q^2,q^3,q^4$ of dimers, triangles, and cliques are accounted for by
their original sites.

The regrouped original edges have coordinate shifts at most one.
The same translated-cover proof as in
Appendix~\ref{SM-sec:triangleextension} therefore applies at blocked
periods at least 32: the templates and their exterior shells embed
without boundary identifications, and translates of $D$ centered
outside $K^+$ avoid $K$ and cover every exterior block, including
across seams. The union lemma gives injectivity of the full complements.
For fixed $q$, the product of the nonzero original-parameter minors
over all required regions and translations is nonzero. Hence the
normal theorem's hypotheses hold on one nonempty open set of the
original parameters. Checkerboard's blocked periods are half its
original periods; the four tilings use their original cells directly.

\paragraph{A clique extension of the two-site detector.}
Let $k=2$ or $3$. Consider two adjacent sites with $k$ internal legs
each, one shared and the other $k-1$ open. Group their outputs to
give maps $\mathbb C^{q^k}\to\mathbb C^{q^{k+1}}$, and select one
$q$-dimensional external output at each site. The corresponding
double-quotient detector has generic kernel
\begin{equation}
\operatorname{End}(\mathbb C^q)\otimes I
+I\otimes\operatorname{End}(\mathbb C^q).
\label{SM-eq:cliquesplit}
\end{equation}
To prove this for every $q$, choose $q^k$ distinct values $t_{p,y}$,
where $p$ has $q$ values and $y$ denotes $k-1$ independent $q$-indices,
and let $v_{p,y}=(1,t_{p,y},\ldots,t_{p,y}^{q-1})^T$.
With an invertible, entrywise nonzero matrix $W$, take
\begin{align}
A[p,y,x;a,b]&=\delta_{pa}\delta_{yb}(v_{p,y})_x,\nonumber\\
B[p,y,x;a,b]&=W_{ap}\delta_{yb}(v_{p,y})_x.
\label{SM-eq:cliquewitness}
\end{align}
Here $p$ is physical, $y$ denotes spectator outputs, $x$ is selected,
$a$ is the shared input, and $b$ denotes the remaining inputs.
Both maps have rank $q^k$, and their images are one line in each
$(p,y)$ output fiber. After contraction, the detector contains
\begin{equation}
W_{pr}(\lambda\otimes\mu)T(v_{p,y}\otimes v_{r,z}),
\quad \lambda v_{p,y}=\mu v_{r,z}=0.
\label{SM-eq:cliquefunctionals}
\end{equation}
The single-site functionals have common kernel $\mathbb CI$:
$q$ Vandermonde lines diagonalize a line-preserving operator and
one further line forces equal eigenvalues. Since $q^k\ge q+1$,
their products detect the full tensor square of
$\operatorname{End}(\mathbb C^q)/\mathbb CI$.
This proves Eq.~\eqref{SM-eq:cliquesplit}, with rank $(q^2-1)^2$,
and supplies a nonzero minor defining a nonempty open condition.

For a checkerboard double bond, apply the two quotients at its
detecting clique sites $i,j$. Original single-site variations vanish,
as do all other coarse-bond actions by the isolation condition.
Apply left inverses at the remaining two clique sites. This opens
the other internal inputs of $i,j$ and leaves a Bell factor on the
edge between the remaining sites. Selecting a nonzero coordinate
of that factor yields exactly the $k=3$ detector above.
Thus the target generator splits into two original-edge generators.
The paired negative transpose at the other endpoint preserves the
decomposition, so one detecting endpoint suffices.
For \texttt{t2002}, the two detecting sites have degree four and
use $k=2$; the third site's map $\mathbb C^{q^2}\to\mathbb C^{q^2}$
is generically invertible and opens the remaining legs. Its zero
quotient is never used as a detector. The case \texttt{t2011}
uses the original $k=2$ argument directly.

\paragraph{Triple bonds and the original differential kernel.}
The dimers in \texttt{t2003} and \texttt{t2004} have precisely the
two incident $1+2$ and $2+1$ triple bonds and the spectator legs
required by Eq.~\eqref{SM-eq:dimersplitwitness}. The joint detector in
Appendix~\ref{SM-sec:m29generic} has rank $2(q^2-1)(q^4-1)$ on their
non-single-site parts; the small$\otimes$large and
large$\otimes$small direct summands cannot cancel.
Each generator therefore splits according to the site partition
at each endpoint. Expanding each original-leg operator space as
$\mathbb CI\oplus\mathfrak{sl}(q)$ shows that a nontrivial support
must lie in one part of both endpoint partitions. Their common
refinement is into individual legs, so each triple-bond generator
is a sum of three original-edge generators.

To complete the argument, impose all routed-region conditions and
the local detector conditions simultaneously with generic tangent
completeness of each internal dimer, triangle, or $K_4$ contraction.
The latter is an independent nonempty open condition: every grouped
site map ranges over its full matrix space and has output dimension
at least its input dimension. Single-site injectivity is nonempty;
the injective fundamental theorem on the internal simple graph and
Lemma~\ref{lem:fiber} then imply its generic tangent completeness.
Reshaping original site tensors is an invertible linear coordinate
change, so these conditions pull back to nonempty open sets in the
original parameter space.

The lifting step of Appendix~\ref{SM-sec:blockingtemplate} then applies:
intersect also with the maximal original gauge-rank open set and the
maximal-rank open set of the original global contraction. The coarse
generators split by the preceding arguments; subtracting the
corresponding original external gauges leaves zero variation of every
block contraction. Singleton variations vanish, and all other residual
variations are internal gauges by the local open conditions. This
proves $\ker J_A=\im\Phi_A$ and Eq.~\eqref{eq:genericrank}.

\paragraph{Ranks and WGS comparison.}
With $M=L_xL_y$ counting original cells, including for checkerboard,
write $\rank J_{\rm generic}=1+M\rho_G(q)$. The polynomials are
\begin{align}
\rho_{\texttt{t2002}}(q)&=4q^5+4q^4-14(q^2-1)-8,\nonumber\\
\rho_{\texttt{t2003}}(q)&=2q^6+2q^5-9(q^2-1)-4,\nonumber\\
\rho_{\texttt{t2004}}(q)&=2q^6+q^5-7(q^2-1)-3,\nonumber\\
\rho_{\texttt{t2011}}(q)&=3q^5-6(q^2-1)-3,\nonumber\\
\rho_{\rm checkerboard}(q)&=2q^7-6(q^2-1)-2.
\label{SM-eq:fivegenericranks}
\end{align}
In this order, the exact WGS-to-generic rank increases divided by
$N$, with $r=q-1$, are
\begin{gather}
r^2\bigl(q+\tfrac34\bigr),\qquad
r^2\bigl(\tfrac12q^2+q+2\bigr),\nonumber\\
r^2\bigl(\tfrac23q^2+\tfrac43q+2\bigr),\qquad
r^2\bigl(\tfrac13q^2+\tfrac43q+\tfrac23\bigr),\nonumber\\
\tfrac12r^2[3(q+1)^2+2].
\label{SM-eq:fiverankdrops}
\end{gather}
These follow from the independently proved generic completeness and
the WGS counts at locally faithful periods.

\subsection{Bowtie blocks for two overlapping-triangle tilings}
\label{SM-sec:bowtieextension}
For \texttt{t2006} and \texttt{t2007}, Eq.~\eqref{eq:genericrank}
holds for independent original tensors at every integer $q\ge2$
and original cell periods $L_x,L_y\ge32$. Each original cell has
five degree-four vertices and ten edges. The new ingredient is a
five-site block consisting of two triangles sharing one vertex,
whose central site is not itself injective.

\paragraph{Block geometry and global partitions.}
Write a block as $\{c\}\cup L\cup R$, where $c$ is its center
and $L,R$ are the two leaf pairs. For \texttt{t2006}, choose
\begin{equation}
\begin{aligned}
c_{xy}&=(x,y,0),\\
L_{xy}&=\{(x-1,y,3),(x,y,1)\},\\
R_{xy}&=\{(x-1,y-1,4),(x,y-1,2)\}.
\end{aligned}
\label{SM-eq:t2006bowtie}
\end{equation}
For \texttt{t2007}, choose
\begin{equation}
\begin{aligned}
c_{xy}&=(x,y,4),\\
L_{xy}&=\{(x,y,3),(x+1,y,2)\},\\
R_{xy}&=\{(x,y+1,1),(x+1,y+1,0)\}.
\end{aligned}
\label{SM-eq:t2007bowtie}
\end{equation}
In each case these translated blocks partition all original sites.
The internal graph has edges $ca,cb,ab,cd,ce,de$, writing
$L=\{a,b\}$ and $R=\{d,e\}$. The center has no external legs;
each leaf has two. The coarse graph of \texttt{t2006} has three
edge types with shifts $(-1,-1),(-1,0),(0,-1)$: the horizontal
bond is double and the others are single. For \texttt{t2007},
the horizontal and vertical bonds are both double.

At a detecting endpoint, a horizontal double bond touches one leaf
from each wing. The vertical double bond of \texttt{t2007} touches
the two leaves of a single wing. In every case no other incident
coarse bond touches both selected leaves; opposite incidences of
the same translation type are included in this check.

The open-region certificates use exactly the boxes and cover of
Eqs.~\eqref{SM-eq:triangleboxes} and \eqref{SM-eq:trianglecover}.
There are fourteen regions for \texttt{t2006} and eleven for
\texttt{t2007}, covering every coarse-edge and coarse-vertex type.
All routes are on original $q$-dimensional legs and physical sites,
not on hypothetical unrestricted coarse tensors. The regrouped edge
shifts are at most one. Thus the translated-cover argument of
Appendix~\ref{SM-sec:triangleextension} proves both normal-partition
hypotheses on every torus with original periods at least 32.

\paragraph{Internal completeness despite a noninjective center.}
Regard the physical and two external legs at each leaf $i$ as the
output space $E_i=\mathbb C^{q^3}$. Its map
$A_i:\mathbb C^{q^2}\to E_i$ is generically injective.
The center tensor $C$ has one physical and four internal indices.
For every choice of two internal indices, impose full column rank
on the flattening
\begin{equation}
C_{ij}:\mathbb C^{q^2}\longrightarrow\mathbb C^{q^3},
\label{SM-eq:bowtiecenter}
\end{equation}
whose outputs are the physical index and the other two internal
indices. Each condition has an elementary identity-matrix witness,
so all six conditions hold on a nonempty common open set. This
does not require injectivity of $C$ as a map
$\mathbb C^{q^4}\to\mathbb C^q$.

The wing map from its two center legs to $E_a\otimes E_b$ is
injective: left inverses of $A_a,A_b$ recover the two input legs
together with a fixed nonzero Bell vector on $ab$.
The region $\{c,a,b\}$ is also injective, by composing the
full-rank flattening of $C$ with this injective wing map.
The same statements hold for $\{d,e\}$ and $\{c,d,e\}$.

These regions verify the normal theorem for the internal bowtie
itself. For edge $ca$, take the three blocks
$\{c,d,e\},\{a\},\{b\}$; for edge $ab$, take
$\{a\},\{b\},\{c,d,e\}$. Relabeling leaves and exchanging
wings covers every edge. For the center vertex use $R=\{a,b\}$;
then $R$, $R\cup\{c\}$ and their complements are the two wings
and the two center-plus-wing regions. For leaf $a$, use $R=\{b\}$.
Its complements are $\{a,c,d,e\}$ and $\{c,d,e\}$; the first is
injective as the union of the injective singleton $\{a\}$ and
the injective region $\{c,d,e\}$. The other leaf cases are symmetric.
All separating connections are single original edges.
The normal theorem and Lemma~\ref{lem:fiber} therefore give
generic internal tangent completeness in the full original block
parameter space, with the external legs treated as physical outputs.

\paragraph{Adjacent and nonadjacent leaf detectors.}
Let $Q_i:E_i\to E_i/\im A_i$ be the leaf quotients. Apply
$Q_i\otimes Q_j$ at the selected endpoint of an external double
bond. Every original single-site variation vanishes, as do all
other coarse-bond actions by the isolation property.
Left inverses at the two unselected leaves expose the corresponding
center indices and any remaining leaf-edge indices. The left inverse
of Eq.~\eqref{SM-eq:bowtiecenter} then exposes the center legs of
the two selected leaves. Any Bell factor between unselected leaves
can be evaluated at a nonzero coordinate.

If the selected leaves are adjacent, one shared internal edge remains
and the other input of each leaf is open. The two-site lemma
\eqref{SM-eq:rubysplit} therefore splits the coarse generator.
If the leaves belong to different wings, both internal inputs at
each leaf are now independently open and there is no shared edge.
The single-leaf map
\begin{equation}
K\longmapsto Q_i(I\otimes K)A_i
\label{SM-eq:bowtieleafdetector}
\end{equation}
on its selected external leg has generic kernel $\mathbb CI$.
Indeed, the $q^2$ Vandermonde lines of
Eq.~\eqref{SM-eq:rubysplitwitness}, with both internal inputs open,
give the functionals $K\mapsto\lambda Kv_{ab}$; any operator
preserving these lines is scalar, as already proved there.
The two nonadjacent leaf detectors form a tensor product and hence
have kernel
\begin{equation}
\operatorname{End}(\mathbb C^q)\otimes I
+I\otimes\operatorname{End}(\mathbb C^q).
\label{SM-eq:bowtiesplit}
\end{equation}
Thus both kinds of double bond split into two original-edge gauges
for every $q\ge2$. The matched negative transpose at the other
endpoint preserves the splitting.

\paragraph{Global lifting and rank.}
The lifting step of Appendix~\ref{SM-sec:blockingtemplate} applies to
the intersection of the routed-region, leaf-injectivity,
center-flattening, detector, and internal generic-completeness open
sets. The preceding detectors split the coarse generators; after
subtracting the corresponding original external gauges, each block
has zero contraction variation, hence only internal gauges by the
internal bowtie result. This proves Eq.~\eqref{eq:genericrank}
without assuming an injective center.

With $M=L_xL_y$ counting original cells, both graphs have
\begin{equation}
\rank J_{\rm generic}
=1+M[5q^5-10(q^2-1)-5].
\label{SM-eq:bowtierank}
\end{equation}
Their exact WGS-to-generic rank increase per original site is
\begin{equation}
\frac{\rank J_{\rm generic}-\rank J_{\rm WGS}}{N}
=\frac{(q-1)^2(4q+6)}5.
\label{SM-eq:bowtierankdrop}
\end{equation}
The two arrangements have different coarse graphs despite identical
degree counts and rank polynomials.

\subsection{Diamond blocks and region detectors}
\label{SM-sec:diamondextension}
We establish Eq.~\eqref{eq:genericrank} for \texttt{t2009} and
\texttt{t2012}, for independent original tensors at every integer
$q\ge2$ and original cell periods $L_x,L_y\ge32$. A diamond has
vertices $a,b,c,d$ and edges $ab,ac,bc,bd,cd$, with tips $a,d$
and centers $b,c$. Each tip has original degree four and each center
degree five. Every site has two external legs. Grouping these with
its physical leg gives $E_i=\mathbb C^{q^3}$; the tip maps have
input dimension $q^2$, and the center maps have input dimension $q^3$.
They are all generically injective. In particular, the center maps
are square and their single-site quotients are zero. The following
region detectors avoid using those quotients.

\paragraph{Exact geometry and normal partitions.}
For \texttt{t2012}, use the translated diamonds
\begin{equation}
\begin{split}
D_{xy}=\{&(x,y,1),(x,y,2),\\
&(x,y,3),(x+1,y+1,0)\}.
\end{split}
\label{SM-eq:t2012diamond}
\end{equation}
The displayed order is $a,b,c,d$. There are two coarse-edge types:
a single bond with shift $(-1,-1)$ and a triple bond with shift
$(-1,0)$. At its first endpoint the triple bond has one leg at $a$
and two at $b$; at its other endpoint two are at $c$ and one at $d$.
The three original edges have distinct endpoint-site pairs, so the
common refinement of these two site partitions consists of single edges.
At either endpoint, every other coarse-bond incidence acts wholly
within the selected tip or wholly within the complementary triangle.

For \texttt{t2009}, the following three translated diamonds partition
all twelve sites of an original cell:
\begin{align}
D^0_{xy}&=\{(x,y,2),(x,y,3),(x,y,6),(x,y,8)\},\nonumber\\
D^1_{xy}&=\{(x,y,5),(x,y,9),\nonumber\\
        &\hspace{1.5em}(x,y+1,1),(x,y+1,4)\},\nonumber\\
D^2_{xy}&=\{(x,y,10),(x,y,11),(x,y+1,7),\nonumber\\
        &\hspace{1.5em}(x+1,y+1,0)\}.
\label{SM-eq:t2009diamonds}
\end{align}
The six coarse-edge types are all double. At every diamond, each
incident double bond touches one tip and one center. Each tip has
two such incidences, one through each center; each center has one
external leg in the pair of incidences selected by either tip.
These incidence statements include opposite ends of translated
bonds and are checked directly from the original adjacency.

The open certificates contain 31 regions for \texttt{t2009} and 11
for \texttt{t2012}, covering all eight coarse-edge and four
coarse-vertex types. They use the boxes and cover of
Eqs.~\eqref{SM-eq:triangleboxes} and \eqref{SM-eq:trianglecover}.
All 9934 paths use original virtual legs and distinct original
physical endpoints. The regrouped original edge shifts have size
at most one. Consequently the translated-cover proof in
Appendix~\ref{SM-sec:triangleextension} verifies the normal theorem
for every original torus with periods at least 32.

\paragraph{The region quotient.}
Choose a tip $a$, and let $B$ be the contraction of the remaining
triangle $\{b,c,d\}$, with its two legs to $a$ left open. Both
\begin{equation}
A:\mathbb C^{q^2}\longrightarrow E_a,\qquad
B:\mathbb C^{q^2}\longrightarrow E_b\otimes E_c\otimes E_d
\label{SM-eq:diamondregionmaps}
\end{equation}
are generically injective. This follows for $B$ by taking left
inverses of the three injective site maps and recovering its two
inputs and the three internal Bell vectors. Let $Q_A,Q_B$ be the
quotients by their images. Applying $Q_A\otimes Q_B$ to the block
tangent equation kills every original site variation: a variation
at $a$ is killed by $Q_B$, and a variation in the triangle is killed
by $Q_A$. It also kills every coarse action confined to one side
of this partition. Quotient matrices can be chosen rationally on
full-rank minor charts, so a nonzero detector minor defines a
nonempty Zariski-open condition in the original tensors.

\paragraph{An isolated $1+2$ triple bond.}
For \texttt{t2012}, let $x$ be the selected external leg at $a$,
and $x_1,x_2$ the selected external legs at $b$. Write the other
tip outputs as $p,s$, and its inputs as $i,j$. Choose distinct
$t_{ij}$ and set
\begin{equation}
\begin{gathered}
A[p,s,x;i,j]=\delta_{pi}\delta_{sj}(v_{ij})_x,\\
v_{ij}=(1,t_{ij},\ldots,t_{ij}^{q-1})^T.
\end{gathered}
\label{SM-eq:diamondtipwitness}
\end{equation}
The $q^2$ lines spanned by $v_{ij}$ have only scalar common
stabilizers, by the Vandermonde argument used above. Hence
$K\mapsto Q_A K_x A$ has kernel $\mathbb CI$.

A simultaneous witness for the triangular region can be given with
Kronecker deltas. Label its internal edges $bc,bd,cd$ by $u,v,w$,
and write the physical and two external outputs of $c,d$ as
$(p_c,y_1,y_2)$ and $(p_d,z_1,z_2)$. Take
\begin{align}
A_b[p_b,x_1,x_2;i,u,v]
 &=\delta_{p_b i}\delta_{x_1u}\delta_{x_2v},\nonumber\\
A_c[p_c,y_1,y_2;j,u,w]
 &=\delta_{p_c j}\delta_{y_1u}\delta_{y_2w},\nonumber\\
A_d[p_d,z_1,z_2;v,w]
 &=\delta_{p_d0}\delta_{z_1v}\delta_{z_2w}.
\label{SM-eq:diamondtriplewitness}
\end{align}
Thus $B|i,j\rangle$ consists of the visible label $|i,j\rangle$
at $p_b,p_c$, two Bell vectors pairing $(x_1,x_2)$ with $(y_1,z_1)$,
and a spectator Bell vector on $y_2,z_2$.
For $Y\in\operatorname{End}(\mathbb C^{q^2})$, the map
$Y\mapsto(Y\otimes I)|\Omega_{q^2}\rangle$ modulo the Bell line
has kernel $\mathbb CI$ and rank $q^4-1$. The visible $i,j$ labels
retain every column of the tip detector. Evaluating the spectator
Bell vector at a nonzero coordinate, the joint detector is therefore
the tensor product of maps of ranks $q^2-1$ and $q^4-1$.
Its kernel on the triple-bond generator is exactly
\begin{equation}
\operatorname{End}(\mathbb C^q)\otimes I
+I\otimes\operatorname{End}(\mathbb C^{q^2}).
\label{SM-eq:diamondtriplesplit}
\end{equation}
The geometry isolates this action from all other coarse incidences.
Repeating at the opposite endpoint gives the other $2+1$ partition.
Expanding original-leg operator spaces as
$\mathbb CI\oplus\mathfrak{sl}(q)$ shows that only supports contained
in a part of both partitions survive. Their common refinement is
single edges, so the triple generator splits into three original
bond generators.

\paragraph{Two double bonds detected jointly.}
For \texttt{t2009}, fix a tip $a$ and call its external legs $x,y$.
The two selected double bonds pair these with external legs $s_b,s_c$
at centers $b,c$, respectively. Let their other external legs be
$z_b,z_c$. All other coarse actions are confined to the triangular
region. The map $A$ can be chosen generically so that both
\begin{equation}
F_x(K)=Q_A K_x A,\qquad F_y(K)=Q_A K_y A
\label{SM-eq:diamondtipdetectors}
\end{equation}
have kernel $\mathbb CI$. Each condition has the witness
\eqref{SM-eq:diamondtipwitness}, with the two external outputs exchanged
when necessary; irreducibility makes their intersection nonempty.

To prove joint detection, choose the center tensors to route their
inputs by
\begin{align}
A_b[p_b,s_b,z_b;i,u,v]
 &=\delta_{p_bu}\delta_{s_bv}\delta_{z_bi},\nonumber\\
A_c[p_c,s_c,z_c;j,u,w]
 &=\delta_{p_cu}\delta_{s_cw}\delta_{z_cj}.
\label{SM-eq:diamonddoublewitness}
\end{align}
Choose any injective $A_d:\mathbb C^{q^2}\to E_d$.
Then $B|i,j\rangle$ factors as the visible labels at $z_b,z_c$,
a Bell vector at $p_b,p_c$, and a fixed vector
\begin{equation}
|T\rangle=\sum_{v,w}|v,w\rangle_{s_b,s_c}
                  \otimes A_d|v,w\rangle.
\label{SM-eq:diamondschmidt}
\end{equation}
Its Schmidt rank across $(s_b,s_c)|E_d$ is $q^2$.
Modulo $\mathbb C|T\rangle$, the maps
$G_b(K)=(K\otimes I)|T\rangle$ and
$G_c(K)=(I\otimes K)|T\rangle$, each defined on
$\operatorname{End}(\mathbb C^q)/\mathbb CI$, are injective
and have disjoint images. Indeed, applying a left inverse of $A_d$
reduces a relation between these images to
$K\otimes I+I\otimes H=\lambda I$, which forces both $K,H$
to be scalar.

Retaining the visible $i,j$ labels, the two non-single-leg parts of
the coarse generators are detected by
$F_x\otimes G_b$ and $F_y\otimes G_c$. Their images are disjoint,
so the joint detector has rank $2(q^2-1)^2$: the two coarse actions
cannot cancel. Both double bonds split into original-edge generators.
Repeating for the other tip covers every incident double bond.
All witnesses used here are assignments of original site tensors;
no arbitrary choice of a coarse tensor is assumed.

\paragraph{Lifting to the original network and explicit ranks.}
The internal diamond is an injective network on a simple graph when
its external legs are treated as physical outputs. The injective
fundamental theorem and Lemma~\ref{lem:fiber} therefore give generic
internal tangent completeness. The lifting step of
Appendix~\ref{SM-sec:blockingtemplate} then applies to the intersection
of this open set with the routed-region conditions and the region
detector minors. The preceding detectors split the coarse generators.
Subtracting these original external gauges leaves zero contraction
variation within every diamond, and internal completeness removes
the residual variations. This proves Eq.~\eqref{eq:genericrank}.

For $M=L_xL_y$ original cells, the ranks are $1+M\rho_G(q)$ with
\begin{align}
\rho_{\texttt{t2009}}(q)&=6q^6+6q^5-27(q^2-1)-12,\nonumber\\
\rho_{\texttt{t2012}}(q)&=2q^6+2q^5-9(q^2-1)-4.
\label{SM-eq:diamondgenericranks}
\end{align}
Their exact WGS-to-generic rank increases per original site are,
respectively,
\begin{equation}
\frac{(q-1)^2(3q^2+16q+20)}{12},\qquad
(q-1)^2(q+1)^2.
\label{SM-eq:diamondrankdrops}
\end{equation}

\subsection{Injective wheel blocks with noninjective rim sites}
\label{SM-sec:wheelextension}
For \texttt{t2010} and \texttt{t2013}, Eq.~\eqref{eq:genericrank}
holds for every integer $q\ge2$ and original periods $L_x,L_y\ge3$.
The proof uses seven-site wheels. Their boundary maps are injective,
but grouping the external legs with the physical outputs leaves
individual rim sites noninjective. We therefore prove internal
completeness separately.

\paragraph{Wheel geometry and block injectivity.}
In the integer coordinates of the dataset, take the closed neighborhoods
of $(x,y,5)$ for \texttt{t2010}, and of $(x,y,6)$ and $(x,y,12)$
for \texttt{t2013}. These neighborhoods partition the original sites.
Each consists of a degree-six center $c$ and a six-cycle of degree-four
rim sites $r_1,\ldots,r_6$, with spokes $cr_i$. Each rim site has
exactly one external leg. The wheel boundary map is generically
injective: copy each external index to its rim physical index and
fix every internal index and the center physical index to zero.
This gives a $q^6$-dimensional identity map into the $q^7$-dimensional
physical output.

For \texttt{t2010}, there is one coarse vertex per original cell;
its three single-edge types have shifts $(-1,0),(0,-1),(-1,1)$.
For \texttt{t2013}, there are two coarse vertices and three double-edge
types with shifts $(0,-1),(0,0),(1,-1)$ from the first type to the second.
At either endpoint these double bonds touch three disjoint adjacent
rim pairs that cover the six-cycle. The frozen adjacency verifies
these properties directly. For periods at least three, neighbor
offsets differing by at most two cannot alias. Thus the coarse graph,
after parallel original legs are combined, is simple on every stated
torus. The injective fundamental theorem applies to this coarse network
without further regional routing.

\paragraph{Three invertible rim dimers.}
Pair consecutive rim sites into three dimers, using the external
bond pairs just described for \texttt{t2013}. A contracted dimer
$D_i$ has output dimension $q^4$ from its two physical and two external
legs, and input dimension $q^4$ from its two spokes and the two
uncontracted rim edges. It is generically invertible: at either site,
copy the spoke and uncontracted rim inputs to the physical and external
outputs, respectively, and fix the shared dimer edge to zero.
These conditions and wheel injectivity need not use the same witness.

Write $n=q^2$ and group the two spokes of dimer $i$ into
$S_i=\mathbb C^n$. After applying the three $D_i^{-1}$, the internal
wheel tensor is
\begin{equation}
\begin{gathered}
C\otimes|\Omega_q\rangle_{12}
 \otimes|\Omega_q\rangle_{23}
 \otimes|\Omega_q\rangle_{31},\\
C\in\mathbb C^q\otimes S_1\otimes S_2\otimes S_3.
\end{gathered}
\label{SM-eq:wheelcore}
\end{equation}
The three Bell vectors are the remaining rim edges. Impose full column
rank on each one-$S_i$ flattening of $C$. Also impose, for each pair
$i\ne j$, the condition
\begin{equation}
\begin{gathered}
(X_i\otimes I+I\otimes Y_j)C=0\\
\Longrightarrow\quad
X_i=\lambda I,\quad Y_j=-\lambda I.
\end{gathered}
\label{SM-eq:wheelpairrigidity}
\end{equation}
Each is a nonempty open condition. For the latter, regard the physical
index and the third spoke pair as slice labels, and choose three slices
on $S_i\otimes S_j$ to be $I,D,S$, where $D$ is diagonal with distinct
entries and $S$ is a cyclic shift. The first slice gives $Y_j^T=-X_i$;
the other two make $X_i$ commute with both $D$ and $S$, hence scalar.
There are $q^3\ge3$ slice labels. The separate flattening and pair
conditions therefore hold on a common nonempty open set.

Now consider a zero tangent of this center-plus-three-dimer network.
Pulling back the dimer variations gives operators
$H_i=D_i^{-1}\delta D_i$ on $S_i$ and its two rim half-edges.
Expand each rim operator space as $\mathbb CI\oplus\mathfrak{sl}(q)$.
On a Bell pair, these two summands give independent vector subspaces,
and transposition at the other endpoint preserves the decomposition.
The tangent equation consequently separates by the set of rim edges
carrying a nonidentity operator.
A two-edge support occurs in only one $H_i$; its coefficients vanish
because the $S_i$ flattening of $C$ has full column rank.
A one-edge support has contributions from its two endpoints; by
Eq.~\eqref{SM-eq:wheelpairrigidity}, their coefficients are opposite
scalars on the spoke spaces. They are ordinary rim-edge gauges.
The empty support consists of arbitrary operators on the $S_i$,
absorbed by the center variation. Hence the only remaining coarse
internal directions are the three spoke-pair gauges and the three
uncontracted rim-edge gauges.

\paragraph{Separating spoke and external pair actions.}
A single original dimer can also be regarded as a rank-$q$ bipartite
matrix, contracting its shared edge between site maps
$N_i:\mathbb C^q\to\mathbb C^{q^4}$. The four outputs at each site
are its physical index $p$, its uncontracted rim index $y$, its
external index $x$, and its spoke index $z$. Let $Q_i$ be quotients
by the images of $N_i$. The product $Q_1\otimes Q_2$ kills every
original single-site variation.

We need to distinguish a joint action on the two external indices
from a joint action on the two spoke indices. Choose distinct $t_{py}$,
put $v_{py}=(1,t_{py},\ldots,t_{py}^{q-1})^T$, and take an invertible
entrywise nonzero $W$. The original site tensors admit the witnesses
\begin{align}
N_1[p,y,x,z;a]&=\delta_{pa}(v_{py})_x(v_{py})_z,\nonumber\\
N_2[p,y,x,z;a]&=W_{ap}(v_{py})_x(v_{py})_z.
\label{SM-eq:wheeljointwitness}
\end{align}
Both maps have full column rank. An annihilator of the $x$ line in
each $(p,y)$ output fiber annihilates $N_i$ even when an arbitrary
operator acts on $z$. Applying these annihilators at both sites and
evaluating $z=0$ therefore detects only the external-pair action.
Its rows are the two-site Vandermonde product functionals, multiplied
by nonzero entries of $W$, so its rank is $(q^2-1)^2$ on the
non-single-leg operator space. Exchanging $x,z$ gives a second channel
of the same rank, which detects only the spoke-pair action.
Thus the joint detector has rank $2(q^2-1)^2$ and forces both
operators separately to be sums of their two single-leg actions.
These annihilators factor through the quotients $Q_i$, so the
nonzero minors prove generic conditions in the original site parameters.
Input/output transposes arising from reshaping preserve each
single-leg operator subspace.

For internal wheel completeness, remove the uncontracted rim gauges
found above. Each remaining dimer variation has the form $D_iK_i$
for a spoke-pair operator $K_i$. The spoke channel splits $K_i$
into two original spoke gauges. After subtracting these, the dimer
variation is zero, and the full-column rank-$q$ matrix factorization
has only its shared-edge gauge kernel. This proves internal tangent
completeness for the original wheel.

For an external coarse tangent of \texttt{t2013}, let $T_i$ act on
the two external legs of dimer $i$. Subtracting $T_iD_i$ from the
original dimer variation gives an internal zero tangent. The preceding
support decomposition, after removing the uncontracted rim gauges,
therefore yields
\begin{equation}
\delta D_i=T_iD_i+D_iK_i.
\label{SM-eq:wheeljointtangent}
\end{equation}
The two-channel detector splits both $T_i$ and $K_i$. In particular,
every external double bond splits into original-edge gauges at its
wheel endpoint. For \texttt{t2010}, all external bonds are already
single and only internal completeness is needed.

\paragraph{Global conclusion and ranks.}
The lifting step of Appendix~\ref{SM-sec:blockingtemplate} applies to
the intersection of the wheel-injectivity, dimer-invertibility,
center, two-channel detector, and full-column $N_i$ conditions.
The external splitting above lifts the coarse generators to original
edges; subtracting these gauges leaves only the internal wheel
kernels, already proved to be gauges. This gives
Eq.~\eqref{eq:genericrank} without requiring the individual
rim maps $\mathbb C^{q^3}\to\mathbb C^{q^2}$ to be injective.

With $M=L_xL_y$ original cells, write the ranks as $1+M\rho_G(q)$:
\begin{align}
\rho_{\texttt{t2010}}(q)&=q^7+6q^5-15(q^2-1)-7,\nonumber\\
\rho_{\texttt{t2013}}(q)&=2q^7+12q^5-30(q^2-1)-14.
\label{SM-eq:wheelgenericranks}
\end{align}
On locally faithful quotients, in particular $L_x,L_y\ge5$, their
WGS-to-generic rank increases per site are, respectively,
\begin{equation}
\frac{6(q-1)^2(q+1)^2}{7},\qquad
\frac{3(q-1)^2(2q^2+4q+3)}7.
\label{SM-eq:wheelrankdrops}
\end{equation}
For smaller stated periods the generic ranks still hold; the finite-graph
WGS formula supplies the appropriate comparison.

\subsection{Joint unequal triple bonds in \texorpdfstring{\texttt{t2020}}{t2020}}
\label{SM-sec:t2020extension}
For \texttt{t2020}, Eq.~\eqref{eq:genericrank} holds for every integer
$q\ge2$ and original periods $L_x,L_y\ge32$. The appropriate diamonds
have degree-five tips and degree-six centers. Their physical-plus-external
output spaces are all $\mathbb C^{q^4}$, while the internal input spaces
have dimensions $q^2$ at tips and $q^3$ at centers. All internal site maps
are therefore generically injective.

\paragraph{A diamond packing with certified partitions.}
Use the three translated blocks
\begin{align}
D^0_{xy}&=\{(x,y,6),(x,y,7),(x,y,8),(x,y,9)\},\nonumber\\
D^1_{xy}&=\{(x,y,10),(x,y+1,1),\nonumber\\
 &\hspace{1.5em}(x,y+1,2),(x,y+1,4)\},\nonumber\\
D^2_{xy}&=\{(x,y,11),(x,y+1,3),\nonumber\\
 &\hspace{1.5em}(x,y+1,5),(x+1,y,0)\}.
\label{SM-eq:t2020diamonds}
\end{align}
They partition all twelve original sites. The coarse cell has three
vertices and six triple-edge types. Every incidence meets one tip and
one center, with leg counts $1+2$ or $2+1$. At each tip its two
incidences meet distinct centers, with counts $1+2$ at one and $2+1$
at the other. Each center has an external spectator leg in either
chosen tip partition. At the two ends of every triple bond, the common
refinement of the site partitions is into the three original edges.
All these statements are checked from the frozen original adjacency.

The 31 open-region certificates cover all six coarse-edge and three
coarse-vertex types. They satisfy the same boxes and cover inclusions
as Eqs.~\eqref{SM-eq:triangleboxes} and \eqref{SM-eq:trianglecover}; their
original edge shifts are at most one. The same translated-cover
argument therefore proves the normal theorem's hypotheses on every
torus with original periods at least 32.

\paragraph{A tip code with scalar stabilizers on both sides.}
Let $n=q^2$, let $D$ be an $n\times n$ diagonal matrix with distinct
entries, and let $S$ be a cyclic shift. The $n$-dimensional matrix space
\begin{equation}
\mathcal U=\operatorname{span}\{I,D,\ldots,D^{n-2},S\}
\label{SM-eq:t2020code}
\end{equation}
has only scalar left and right multiplicative stabilizers. Indeed,
$Y\mathcal U\subseteq\mathcal U$ and $I\in\mathcal U$ imply
$Y=P(D)+aS$. Requiring $YS\in\mathcal U$ forces $a=0$, since
$S^2$ has support disjoint from both the diagonal and $S$ for $n\ge4$.
The remaining $P(D)S$ is a multiple of $S$ only when $P(D)$ is scalar.
The right-sided argument uses $SY$ and is identical.

Vectorize a basis of $\mathcal U$ to define an injective tip map
$A:\mathbb C^n\to\mathbb C^n\otimes\mathbb C^n$.
Assign the two external legs of one triple bond to the first output
factor; assign the physical and remaining external leg to the second.
With $Q_A$ the image quotient, the large-leg detector on
$\operatorname{End}(\mathbb C^{q^2})$ has kernel $\mathbb CI$ by
left rigidity. Right rigidity gives the same conclusion for the
small-leg detector on the remaining $q$-dimensional external leg.
Thus the two tip detectors have ranks $q^4-1$ and $q^2-1$
simultaneously, at a witness consisting of one original site tensor.

\paragraph{A region witness for the two unequal actions.}
Partition a diamond into a tip $a$ and its complementary triangle
$\{b,c,d\}$, as in Eq.~\eqref{SM-eq:diamondregionmaps}. Name the centers
so that the small external leg of $a$ is paired with two selected
external legs at $b$, and its large external pair is paired with one
selected leg at $c$. Write the two boundary inputs of the triangle
as $i,j$, and its internal edges $bc,bd,cd$ as $u,v,w$.
A witness for its map $B$ is
\begin{align}
A_b[p_b,s_1,s_2,z_b;i,u,v]
 &=\delta_{p_b s_2}\delta_{s_1v}\delta_{z_bi}\delta_{u0},\nonumber\\
A_c[p_c,t,z_c,e_c;j,u,w]
 &=\delta_{p_c0}\delta_{tw}\delta_{z_cj}\delta_{e_c0}\delta_{u0},\nonumber\\
A_d[p_d,d_1,d_2,d_3;v,w]
 &=\delta_{p_dv}\delta_{d_1w}\delta_{d_20}\delta_{d_30}.
\label{SM-eq:t2020regionwitness}
\end{align}
The visible output labels $z_b,z_c$ retain $i,j$. The selected outputs
$(s_1,s_2,t)$ are paired with $(p_d,p_b,d_1)$ by three Bell vectors;
all other outputs are fixed. Thus $B$ is injective and the fixed vector
on the selected outputs and their spectators has Schmidt rank $q^3$.
The center maps need not be injective at this particular detector
witness; generic internal site injectivity is an independent nonempty
open condition.

Modulo the line of this fixed vector, actions on the selected
$q^2$-dimensional factor at $b$ and the selected $q$-dimensional
factor at $c$ have independent images, of dimensions $q^4-1$ and
$q^2-1$. This follows by applying a Schmidt left inverse: a relation
reduces to $Y\otimes I+I\otimes Z=\lambda I$, forcing scalar
$Y,Z$. Applying $Q_A\otimes Q_B$ kills every original site variation
and every coarse action confined to one side of the tip partition.
Keeping the visible $i,j$ labels, the two remaining triple-bond
contributions are tensor products of the small tip detector with
the large region detector, and the large tip detector with the
small region detector. Their images are disjoint. The joint rank is
therefore
\begin{equation}
2(q^2-1)(q^4-1).
\label{SM-eq:t2020jointrank}
\end{equation}
Each triple generator consequently splits according to its endpoint
site partition. Repeating for both tips of every diamond and using
the common refinement at the opposite endpoint splits all triple
bonds into original-edge generators.

\paragraph{Global lifting and rank.}
The region witnesses, tip code, routed partitions, and internal
site-injectivity conditions are nonempty open conditions in the
original parameters, and internal completeness follows from the
injective fundamental theorem on the simple diamond graph and
Lemma~\ref{lem:fiber}. The lifting step of
Appendix~\ref{SM-sec:blockingtemplate} applies to their intersection;
the joint detector splits the coarse generators, and after
subtracting original external gauges, internal completeness removes
the residual block kernels. Thus Eq.~\eqref{eq:genericrank}
holds at all stated $q$ and periods.

For $M=L_xL_y$ original cells, the rank is
\begin{equation}
\rank J_{\rm generic}
=1+M[6q^7+6q^6-33(q^2-1)-12].
\label{SM-eq:t2020genericrank}
\end{equation}
The WGS-to-generic rank increase per original site is
\begin{equation}
\frac{(q-1)^2(27q^2+40q+20)}{12}.
\label{SM-eq:t2020rankdrop}
\end{equation}

\subsection{Exact qubit certificates and reductions for \texttt{t2008}, \texttt{t2014}, \texttt{t2015}, and \texttt{t2019}}
\label{SM-sec:open4qubits}
The results in this subsection distinguish fixed bond dimension from
an all-$q$ statement. For independent original tensors with $q=2$,
we establish Eq.~\eqref{eq:genericrank} for \texttt{t2008} at
$L_x,L_y\ge3$ and for \texttt{t2014} at $L_x,L_y\ge32$.
The local algebraic step is computer assisted: explicit integer
matrices have nonzero square minors modulo two primes. This proves
nonempty Zariski-open conditions over $\mathbb C$ at $q=2$; it does
not extrapolate to $q>2$. We also give certified normal partitions
for a path blocking of \texttt{t2015} for every $q\ge2$, and an exact
$q=2$ local certificate for \texttt{t2019}, identifying the additional
steps supplied in Appendix~\ref{SM-sec:open5qubits}, which completes both
rows, the triangular lattice, and centered-square at $q=2$.

\paragraph{A block criterion using exact minors.}
Let $B$ be a connected block of $n_B$ original vertices and $e_B$
internal edges. Regard every original physical leg and every external
leg as an output of its local contraction $F_B$. Write $J_B$ for its
Jacobian, $P_B$ for its parameter count, and $\Phi_B$ for the internal
gauge map. At $q=2$ the maximal possible internal gauge rank is
\begin{equation}
g_B=3e_B+n_B-1.
\label{SM-eq:q2blockgauge}
\end{equation}
For each coarse incidence $h$, its original legs are partitioned
according to their original endpoint sites. If the nonempty parts have
sizes $k_{h,1},\ldots,k_{h,s_h}$ and $k_h=\sum_i k_{h,i}$, the
operators supported at a single original site span a linear
sum of dimension $\sum_i4^{k_{h,i}}-(s_h-1)$.
Use the basis $I,E_{01},E_{10},E_{11}$ on each qubit leg, and let
$\Theta_h$ contain the actions on $F_B$ of basis products whose
nonidentity support meets at least two original sites. Their number is
\begin{equation}
c_h=4^{k_h}-\sum_i4^{k_{h,i}}+s_h-1.
\label{SM-eq:q2nonlocalcount}
\end{equation}
A single original leg contributes no such columns. All same-site
operators act by a variation of that original tensor and hence already
belong to $\im J_B$.

It suffices to verify, at one integer tensor tuple,
\begin{equation}
\begin{split}
\rank\Phi_B&=g_B,\\
\rank J_B&=P_B-g_B,\\
\rank[J_B\mid\Theta_h:\ h]&=P_B-g_B+\sum_h c_h.
\end{split}
\label{SM-eq:q2jointcriterion}
\end{equation}
A row submatrix attaining the last rank is sufficient. Indeed, gauge
inclusion bounds the generic rank of $J_B$ from above by $P_B-g_B$.
The displayed minor attains the matching upper bound after the stated
number of columns is added. On the resulting nonempty open set,
$\ker J_B=\im\Phi_B$, and the nonlocal incidence actions are jointly
independent modulo $\im J_B$. Consequently, any sum of coarse actions
that is an original block tangent has only same-site terms at each
incidence. The joint test is essential: testing each incidence alone
would not exclude cancellations between different coarse edges.

At opposite endpoints of a coarse edge, the generator and its negative
transpose therefore belong to the respective same-site operator sums.
The common refinement of the two endpoint partitions consists of
individual original edges whenever no two such edges have the same
original endpoint pair. Expanding operator spaces into scalar and
traceless factors then shows that their intersection is precisely the
sum of original single-edge operator spaces. Thus the generator splits
into original-edge gauges.

\paragraph{The injective block for \texttt{t2008}.}
In the fixed integer coordinates of the dataset, translate the block
\begin{equation}
\begin{split}
B_8={}&\{(0,0,t):t\in\{0,3,8\}\}\\
 &\cup\{(-1,0,t):t\in\{1,2,4,5,6,7\}\}\\
 &\cup\{(-1,-1,t):t\in\{9,10,11\}\}.
\end{split}
\label{SM-eq:t2008qubitblock}
\end{equation}
Each source label occurs once, so these blocks partition every periodic
quotient. The block has twelve vertices and twenty-one internal edges.
Its four coarse incidences have shifts $(\pm1,0),(0,\pm1)$ and carry
three original legs each. At each incidence these legs meet three
distinct original sites. A frozen routing certificate sends the twelve
boundary inputs to the twelve distinct physical outputs along
edge-disjoint internal paths. The associated delta tensors prove block
injectivity for every integer $q\ge2$. Since all regrouped edge shifts
have magnitude at most one, the coarse graph is simple whenever both
original periods are at least three.

For the local algebra at $q=2$, $P_B=576$ and $g_B=74$.
Each incidence contributes $c_h=4^3-3\cdot4+2=54$ columns.
The exact ranks are
\begin{equation}
\begin{gathered}
\rank J_B=502,\\
\rank[J_B\mid\Theta_1\mid\cdots\mid\Theta_4]=718.
\end{gathered}
\label{SM-eq:t2008localranks}
\end{equation}
The full output has dimension $2^{24}$; only 4096 fixed output rows
are needed for these rank lower bounds. The explicit $718\times718$
minor has determinants 6153 and 58914 modulo 65521 and 65519,
respectively. The gauge map has rank 74 in both fields. Thus all
conditions of Eq.~\eqref{SM-eq:q2jointcriterion} are nonempty at $q=2$.

\paragraph{The path block for \texttt{t2014}.}
Use the translated three-site paths
\begin{equation}
B_{14}=\{(0,0,0),(0,0,1),(0,-1,2)\}.
\label{SM-eq:t2014qubitblock}
\end{equation}
Their internal edges join source site 0 to the other two sites.
The coarse graph is a square graph with five original legs on each
horizontal edge and one on each vertical edge. In the ordered source
labels $0,1,2$, a horizontal edge contains the pairs
\begin{equation}
(0,0),\ (0,1),\ (0,2),\ (1,1),\ (2,2).
\label{SM-eq:t2014fivepairs}
\end{equation}
Its endpoint partitions have sizes $3+1+1$ and $1+2+2$ and refine
to the five individual edges.

The frozen open templates satisfy the box and cover conditions
\eqref{SM-eq:triangleboxes}--\eqref{SM-eq:trianglecover} on the original
$q$-dimensional legs. They cover both coarse-edge types, the single
coarse-vertex type, and all required complementary regions: eleven
regions and 2072 paths in total. The same cover argument as in
Appendix~\ref{SM-sec:triangleextension} proves both normal-theorem
partition hypotheses for all $L_x,L_y\ge32$ and all $q\ge2$.
This is a statement in the independent original tensor parameters.

At $q=2$, $P_B=256$, $g_B=8$, and the two horizontal incidences
contribute 954 and 990 nonlocal columns. The certified ranks are
\begin{equation}
\rank J_B=248,\qquad
\rank[J_B\mid\Theta_-\mid\Theta_+]=2192.
\label{SM-eq:t2014localranks}
\end{equation}
The output dimension is $2^{15}$; the audit again uses 4096 fixed
rows. Its explicit $2192\times2192$ minor has determinants 35394
and 54211 modulo 65521 and 65519. Both fields also give gauge rank
8, establishing Eq.~\eqref{SM-eq:q2jointcriterion}.

\paragraph{Global qubit conclusion.}
For either tiling, the lifting step of
Appendix~\ref{SM-sec:blockingtemplate} applies to the intersection of
these local open conditions with all required block or regional
injectivity minors, maximal original gauge rank, and the smooth
generic fibers of the original contraction. For \texttt{t2008}, the
injective coarse fundamental theorem applies; for \texttt{t2014},
the certified partitions give the normal theorem.
Equation~\eqref{SM-eq:q2jointcriterion} and the endpoint refinement
split the coarse generators into original-edge generators.
Subtracting them leaves zero block variations, whose kernels are the
certified internal original gauges. This proves $\ker J_A=\im\Phi_A$
at $q=2$.
With $M=L_xL_y$ counting original cells,
\begin{align}
\rank J_{\texttt{t2008},q=2}&=1+483M,\nonumber\\
\rank J_{\texttt{t2014},q=2}&=1+229M.
\label{SM-eq:open4qubitranks}
\end{align}
On locally faithful quotients, the corresponding WGS-to-generic rank
increases are $76M$ and $42M$. For \texttt{t2008}, local faithfulness
is guaranteed by $L_x,L_y\ge5$; the generic qubit result itself
already holds at periods at least three.

\paragraph{Higher-dimensional conditions and the other qubit constructions.}
For \texttt{t2014} the path's grouped site maps are generically
injective for every $q$, so its internal generic rank is
$j_{14}(q)=q^7+2q^6-2q^2$. The unproved higher-dimensional step is
the joint rank increase
\begin{equation}
D_{14}(q)=2q^{10}-q^6-2q^4-3q^2+4
\label{SM-eq:t2014remainingtest}
\end{equation}
from its two horizontal incidences. The certificate above establishes
$j_{14}(2)+D_{14}(2)=2192$, not the all-$q$ identity.
For \texttt{t2008}, injectivity of the twelve-site block is already
proved for all $q$, but its internal and external joint rank conditions
also remain to be proved at $q>2$.

For \texttt{t2015}, use the four-site path
\begin{equation}
B_{15}=\{(0,0,0),(0,0,1),(0,-1,2),(0,0,3)\}.
\label{SM-eq:t2015pathreduction}
\end{equation}
Its coarse square graph has seven-leg horizontal edges and single
vertical edges. Eleven independently checked open regions with
2632 paths establish the normal partition hypotheses for every
$q\ge2$ and $L_x,L_y\ge32$. The internal path maps are generically
injective, giving $j_{15}(q)=2q^7+2q^6-3q^2$. The horizontal endpoint
partitions have sizes $3+2+1+1$ and $1+2+2+2$, with singleton common
refinement. Therefore a sufficient local condition, certified below
at $q=2$ and still open for $q>2$, is
\begin{equation}
\rank[J_{B_{15}}\mid\Theta_-\mid\Theta_+]
=j_{15}(q)+D_{15}(q),
\label{SM-eq:t2015remainingtest}
\end{equation}
where $D_{15}(q)=2q^{14}-q^6-4q^4-3q^2+6$.
At $q=2$ the required joint rank is 33006; it is certified in
Appendix~\ref{SM-sec:open5qubits}. Normality alone does not settle this row.

For \texttt{t2019}, the translated blocks
\begin{equation}
\begin{split}
B_{19}={}&\{(0,0,t):0\le t\le4\}\\
       &\cup\{(0,-1,t):5\le t\le7\}
\end{split}
\label{SM-eq:t2019localreduction}
\end{equation}
have eight vertices, thirteen internal edges, and four four-leg
coarse incidences. At $q=2$, an independent exact local calculation
gives $P_B=640$, $g_B=46$, $\rank J_B=594$, and joint rank 1512
with 918 added nonlocal columns. An explicit minor of order 1512 has
determinants 23361 and 65501 modulo the same two primes. This settles
local completeness and simultaneous splitting for this particular
qubit block. It does not supply the injective partitions required for
its global coarse network. No global rank conclusion for
\texttt{t2019}, nor an obstruction to alternative partitions, follows
from this local certificate. Appendix~\ref{SM-sec:open5qubits} replaces
$B_{19}$ by a 24-site block that admits both partitions.

\subsection{Qubit certificates for triangular, centered-square, \texttt{t2015}, and \texttt{t2019}}
\label{SM-sec:open5qubits}
This subsection proves Eq.~\eqref{eq:genericrank} at $q=2$ for the
triangular lattice, centered-square, \texttt{t2015}, and \texttt{t2019}.
The geometric step holds for every $q\ge2$. The local algebraic step is
the criterion~\eqref{SM-eq:q2jointcriterion}, certified by exact minors at
$q=2$ and not extrapolated. King's all-$q$ row construction is separate
(Appendix~\ref{SM-sec:kingrows}).

\paragraph{Lattice blockings.}
Each blocking consists of one finite block $B_0$ of original sites and
a lattice $\Lambda$ of cell translations whose translates
$B_0+\lambda$ partition the sites. Every block therefore has the same
internal graph and incidences, and $\Lambda$ acts transitively on
blocks. All four coarse graphs are triangular, so three coarse-edge
orbits and one coarse-vertex orbit exist. In each block every site of
degree $d$ has at least $(d-1)/2$ neighbors in its own block. The
singleton case of Eq.~\eqref{SM-eq:routecut} therefore holds in every union
of blocks. This removes the single-site obstruction of
Appendix~\ref{SM-sec:triangulargeneric}.

In the triangular coordinates of Appendix~\ref{SM-sec:coordinates} we use
two blockings. The first is the hexagon
\begin{equation}
\begin{aligned}
H_{19}&=\{(x,y):\max(|x|,|y|,|x-y|)\le2\},\\
\Lambda_{19}&=\langle(5,2),(3,5)\rangle,
\end{aligned}
\label{SM-eq:hex19}
\end{equation}
where $\Lambda_{19}=\{(x,y):y\equiv8x\pmod{19}\}$. The second is the
twelve-site block
\begin{equation}
\begin{split}
H_{12}={}&\{(0,0),(\pm1,0),(0,1),(1,1),\\
&\ (0,-1),(1,-1),(0,-2),(-1,-1),\\
&\ (-1,-2),(-2,-1),(-2,-2)\}
\end{split}
\label{SM-eq:hex12}
\end{equation}
with $\Lambda_{12}=\langle(6,0),(4,2)\rangle$. For centered-square,
cell $(x,y)$ contains the square vertex $v_{xy}$ and the center
$c_{xy}$ of the face $[x,x+1]\times[y,y+1]$. The block is
\begin{equation}
\begin{split}
C_{12}={}&\{v_{xy}:x\in\{0,1\},\,y\in\{-1,0,1\}\}\\
&\cup\{c_{-1,-1},c_{-1,0},c_{0,-2},c_{0,-1},c_{0,0},c_{0,1}\}
\end{split}
\label{SM-eq:cs12}
\end{equation}
with $\Lambda=\langle(2,0),(1,3)\rangle$. For \texttt{t2019}, in the
dataset coordinates $(x,y,t)$, the block $T_{24}$ contains all eight
sites of cell $(0,0)$, sites $3$--$7$ of $(0,-1)$, sites $2,3,4,6,7$
of $(-1,0)$, sites $1,2$ of $(-1,1)$, sites $0,1$ of $(0,1)$, site $5$
of $(1,-1)$, and site $0$ of $(1,0)$, with $\Lambda=\langle(3,0),(1,1)\rangle$.
Table~\ref{SM-tab:open5blocks} lists the incidences. \texttt{t2015} keeps the
path blocking~\eqref{SM-eq:t2015pathreduction} and its certified partitions.

\begin{table}[!tbp]
\centering
\caption{Lattice blockings of Appendix~\ref{SM-sec:open5qubits}. ``Legs''
lists the original legs of the six coarse incidences of one block.
Every coarse graph is triangular. The template disc radius $\rho_O$ and
the block radius $r_B$ are plane distances with unit edges. The last
column is the systole bound of Eq.~\eqref{SM-eq:open5systole}.}
\label{SM-tab:open5blocks}
\footnotesize
\setlength{\tabcolsep}{3pt}
\centering
\begin{adjustbox}{max width=\columnwidth}
\begin{tabular}{lcccccc}\toprule
Block & Sites & Legs & Paths & $\rho_O$ & $r_B$ & $2(\rho_O+r_B+1)$\\
\midrule
$H_{19}$ & 19 & $5^6$ & 3990 & 21.79 & 2.00 & 49.59\\
$H_{12}$ & 12 & $4^6$ & 3528 & 20.78 & 1.53 & 46.62\\
$C_{12}$ & 12 & $3^4\,7^2$ & 3830 & 17.26 & 1.51 & 39.54\\
$T_{24}$ & 24 & $4^6$ & 3464 & 25.98 & 2.54 & 59.04\\
\bottomrule
\end{tabular}
\end{adjustbox}
\end{table}

\paragraph{Normal partitions for all stated periods.}
For each blocking the frozen templates consist of three parts. The
first is a cover disc $D$ of all blocks whose centers lie within a
radius $\rho_D$ of the origin block; in all four blockings these are
the origin block and its six coarse neighbors. The second is, for each coarse-edge orbit $uv$,
regions $A\ni u$ and $B\ni v$ whose only coarse edge is $uv$, together
with a disc $O$ of blocks around $u$. The third is, for the vertex
orbit $v$, regions $R\not\ni v$ and $R^+=R\cup\{v\}$ with a disc $O$.
Every block whose center lies within $\rho_D$ of $A\cup B$, or of $R^+$,
belongs to $O$. Edge-disjoint routes certify injectivity of $D$, $A$,
$B$, $O\setminus(A\cup B)$, $R$, $R^+$, $O\setminus R$, and
$O\setminus R^+$ for every $q$, with all boundary legs taken in the
infinite graph. Let the period lattice $\Gamma\subseteq\Lambda$ satisfy
\begin{equation}
s(\Gamma)>2(\rho_O+r_B+1),
\label{SM-eq:open5systole}
\end{equation}
where $s(\Gamma)$ is the shortest nonzero period vector in the plane
embedding. Then every template, with its neighbor shell, embeds without
identifications. A block $p$ outside $O$ lies at distance greater than
$\rho_D$ from every block of $A\cup B$, so the translate $D+p$ is
injective and misses $A\cup B$. The torus complement of $A\cup B$ is the
union of $O\setminus(A\cup B)$ and these translates, and it is injective
by the union lemma of Ref.~\cite{molnar2018normal}. The complements of
$R$ and $R^+$ are covered in the same way. Both partition hypotheses of
Theorem~4 of Ref.~\cite{molnar2018normal} therefore hold for every coarse
edge and vertex.

For rectangular triangular periods, the plane period vectors have
lengths $L_x,L_y$ at angle $2\pi/3$, and
$|m\mathbf L_x+n\mathbf L_y|^2=(nL_y-mL_x/2)^2+\tfrac34m^2L_x^2$. If
$6\mid L_x,L_y$ and $L_x,L_y\ge54$, then $s\ge\tfrac{\sqrt3}{2}\,54>46.62$.
If $19\mid L_x,L_y$ and $L_x,L_y\ge57$, the terms with $|m|\ne1$ are at
least $57^2$. For $|m|=1$, either $nL_y-mL_x/2$ is a nonzero multiple of
$19/2$, or $L_x=2|n|L_y\ge114$. In both cases $s^2\ge(19/2)^2+\tfrac34 57^2$,
so $s>50.2>49.59$. Both residue conditions give $\Gamma\subseteq\Lambda$.
For centered-square, $s=\min(L_x,L_y)$, and $\Gamma\subseteq\Lambda$ holds
when $2\mid L_x$ and $6\mid L_y$. The reflection $(x,y)\mapsto(y,x)$ is an
automorphism, so the reflected blocking gives $6\mid L_x$, $2\mid L_y$.
For \texttt{t2019} the cell vectors have length three at angle $\pi/3$,
so $s\ge\tfrac{3\sqrt3}{2}\min(L_x,L_y)$, which exceeds $62.3>59.04$
for $L_x,L_y\ge24$, and $3\mid L_x,L_y$ gives $\Gamma\subseteq\Lambda$.

\paragraph{Local criterion at $q=2$.}
Table~\ref{SM-tab:open5local} lists the certified ranks of
Eq.~\eqref{SM-eq:q2jointcriterion}. Each audit uses fixed integer tensors
with entries in $\{-1,0,1\}$ and sampled output rows. It recomputes each
Jacobian block from sampled contractions and checks it against
directional derivatives. It also checks $J_B\Phi_B=0$ on the sampled
rows and compares
sampled nonlocal columns with named operators applied to the original
tensors. The Jacobian rank, the gauge rank, and the joint rank equal
their targets modulo 65521 and 65519. The pivot minor found modulo the
first prime is nonsingular modulo both. An independent replay rebuilds
every column of each frozen minor from the witness tensors alone:
Jacobian columns by tensor replacement, and nonlocal columns from fresh
evaluations of the block amplitude at the modified outputs. A random
sample of the nonlocal columns is also recomputed by applying the named
operators to the original tensors. The replay confirms nonsingularity
over both fields.

\begin{table}[!tbp]
\centering
\caption{Exact qubit block ranks for Eq.~\eqref{SM-eq:q2jointcriterion}.
$\sum_hc_h$ is the number of added nonlocal columns; the joint rank
equals $P_B-g_B+\sum_hc_h$ and is also the order of the frozen minor.}
\label{SM-tab:open5local}
\footnotesize
\setlength{\tabcolsep}{4pt}
\centering
\begin{adjustbox}{max width=\columnwidth}
\begin{tabular}{lcccccc}\toprule
Block & $n_B$ & $e_B$ & $P_B$ & $\rank J_B$ & $\sum_hc_h$ & Joint rank\\
\midrule
$H_{19}$ & 19 & 42 & 2432 & 2288 & 5940 & 8228\\
$H_{12}$ & 12 & 24 & 1536 & 1453 & 1377 & 2830\\
$C_{12}$ & 12 & 23 & 3264 & 3184 & 32814 & 35998\\
$T_{24}$ & 24 & 51 & 1920 & 1744 & 1377 & 3121\\
$B_{15}$ & 4 & 3 & 384 & 372 & 32634 & 33006\\
\bottomrule
\end{tabular}
\end{adjustbox}
\end{table}

\paragraph{Global qubit conclusion.}
The lifting argument of Appendix~\ref{SM-sec:open4qubits} applies without
change. The certified partitions give the normal theorem on the coarse
network. Equation~\eqref{SM-eq:q2jointcriterion} and the endpoint refinement
split every coarse generator into original-edge generators; the
refinement consists of single edges because the graphs are simple. The
residual block variations are internal gauges. At $q=2$ this proves
$\ker J_A=\im\Phi_A$ for triangular tori with $6\mid L_x,L_y\ge54$ or
$19\mid L_x,L_y\ge57$; for centered-square with $2\mid L_x\ge40$ and
$6\mid L_y\ge42$, or with $x$ and $y$ exchanged; for \texttt{t2015}
with $L_x,L_y\ge32$; and for \texttt{t2019} with $3\mid L_x,L_y\ge24$.
With $M=L_xL_y$ cells,
\begin{equation}
\begin{aligned}
\rank J_{\rm tri}&=1+118M,\\
\rank J_{\rm cs}&=1+524M,\\
\rank J_{\texttt{t2015}}&=1+347M,\\
\rank J_{\texttt{t2019}}&=1+569M
\end{aligned}
\label{SM-eq:open5ranks}
\end{equation}
at $q=2$. On locally faithful quotients the WGS-to-generic rank
increases are $22M$, $44M$, $64M$, and $120M$, respectively.

\paragraph{Scope and the earlier king searches.}
The residue conditions come from the chosen lattices $\Lambda$, and the
period bounds from the templates. Neither is a physical threshold, and
below or outside them Proposition~\ref{SM-prop:genericbounds} still bounds
$\nu_{\rm generic}$. The seven-site triangular hexagon and eight-site
centered-square blocks also satisfy the singleton condition. For each,
however, an integer program for the three routing flows finds no edge
partition that keeps the far regions fixed and frees every block within
about three coarse spacings of the edge. We record this as a failed
search, not an obstruction theorem. For king, a scan of Euclidean
Voronoi blocks of up to 24 sites finds that every block whose sites all
have at least four in-block neighbors has an incidence of at least seven
legs. The two twelve-site candidates fail
the same search. Larger blocks require joint qubit tests beyond the
sizes above. These failed searches are bypassed by the analytic row
construction of Appendix~\ref{SM-sec:kingrows}, which proves king completeness
without a dense contact-rank test. The all-$q$ statement remains open for all four families
treated here.

\subsection{King: row blocking and analytic splitting for every dimension}
\label{SM-sec:kingrows}
For the king graph, with the horizontal, vertical, and both diagonal
nearest-neighbor edges of the square grid, the following sufficient
period bounds establish generic completeness for every $q\ge2$.
Together with the preceding constructions, this completes the
thirty-nine-family catalogue, with six families restricted to $q=2$.
\begin{proposition}
\label{SM-prop:kinggeneric}
Independent tensors with $D=d=q\ge2$ on a rectangular king torus have
$\nu_{\rm generic}=0$ if $L_x\ge4$ and $L_y\ge18$, or with the axes
exchanged. With $N=L_xL_y$, their affine generic rank is
\begin{equation}
 \rank J_{\rm generic}=1+N(q^9-4q^2+3).
 \label{SM-eq:kinggenericrank}
\end{equation}
\end{proposition}

\paragraph{Injective bands and both coarse partitions.}
Block each entire horizontal row. The coarse graph is the cycle
$C_{L_y}$, with physical dimension $q^{L_x}$ and bond dimension
$q^{3L_x}$. A consecutive band of $w$ rows, with $6\le w\le L_y-6$,
has $6L_x$ boundary legs. Here is an explicit routing of its lower
boundary legs, with the band's rows numbered $0,\ldots,w-1$.
For the three legs incident at $(x,0)$, use respectively the paths
\begin{equation}
\begin{gathered}
 (x,0),\qquad (x,0),(x,1),\\
 (x,0),(x+1,1),(x+1,2).
\end{gathered}
 \label{SM-eq:kingbandroutes}
\end{equation}
Assign these to the three exterior neighbors in any fixed order.
All horizontal coordinates are taken modulo $L_x$. Reflection in the
horizontal midline gives the upper-boundary routes, ending in rows
$w-1,w-2,w-3$. All terminal sites are distinct and all used internal
edges are distinct, including across the horizontal seam. Thus every
such band is generically injective for every $q$ by the routing witness.

At each coarse edge, partition the cycle into three consecutive bands
of lengths $6,6,L_y-12$, with the distinguished edge between the first
two. It is their only coarse edge. At a coarse vertex $v$, let $R$
be the six consecutive rows immediately preceding $v$; $R$, $R+v$,
and their complements have lengths $6,7,L_y-6,L_y-7$, respectively,
and the same routes apply. This verifies both normal-PEPS partition
hypotheses of Ref.~\cite{molnar2018normal}. No injectivity of a single
row as a coarse tensor is assumed. In particular, the singleton cut
condition applies to these bands, not to the atomic row blocks:
a boundary site has five neighbors inside its band, although only
two neighbors lie in its own row.

\paragraph{An analytic row witness.}
We prove the splitting statement with the more general external
group dimension $Q\ge2q-1$; king uses $Q=q^3$.
Let the row length be $L\ge4$, with all site indices cyclic. At site
$i$, the two internal horizontal indices are $a_i,b_i\in\{0,\ldots,q-1\}$,
contracted as $b_i=a_{i+1}$. Write $u_i,d_i\in\mathbb Z_Q$ for the
upper and lower external groups, and $p_i\in\{0,\ldots,q-1\}$ for
the physical index. Set $t_d=d+1$, and set $\sigma(1)=1$ and
$\sigma(p)=0$ for all other $p$. Use the integer tensors
\begin{equation}
 A[p,u,d;a,b]=\delta_{a,p}\,
 \delta_{u,d+\sigma(p)}\,t_d^b.
 \label{SM-eq:kingrowtensor}
\end{equation}
The $q^2$ columns, grouped by physical index $p=a$, are independent
Vandermonde columns. Thus these grouped site maps are injective.
The injective-network theorem on the simple internal cycle gives
internal tangent completeness, with gauge rank $Lq^2-1$ and
Jacobian rank $Lq^3Q^2-Lq^2+1$.

The row contraction is explicitly
\begin{equation}
 F(\mathbf p,\mathbf u,\mathbf d)=
 \prod_i\delta_{u_i,d_i+\sigma(p_i)}
 \prod_i t_{d_i}^{p_{i+1}}.
 \label{SM-eq:kingrowamplitude}
\end{equation}
Let $K$ and $H$ be arbitrary operators on all upper and all lower
groups. We show that
\begin{equation}
\begin{gathered}
 (K+H)F\in\im J_{\rm row}\\
 \Longrightarrow\quad
 K,H\in\sum_i\operatorname{End}(\mathbb C^Q)_i.
\end{gathered}
 \label{SM-eq:kingrowsplit}
\end{equation}
The scalar identity is included in the sum. This is a simultaneous
statement, so cancellation between the two contacts is included.

Decompose $K$ into cyclic displacement sectors: in sector
$\boldsymbol\delta\in\mathbb Z_Q^L$, it sends an input
$\mathbf z$ to $\mathbf z+\boldsymbol\delta$, with coefficient
$k(\mathbf z)$. Pair this with the sector of $H$ sending input
$\mathbf z$ to $\mathbf z-\boldsymbol\delta$, with coefficient
$h(\mathbf z)$. At output
$\mathbf u=\mathbf d+\boldsymbol\sigma(\mathbf p)+\boldsymbol\delta$,
divide by the nonzero factor $\prod_i t_{d_i}^{p_{i+1}}$. The
combined action is
\begin{equation}
\begin{gathered}
 k(\mathbf d+\boldsymbol\sigma(\mathbf p))+
 h(\mathbf d+\boldsymbol\delta)
 \prod_i R_i(d_i)^{p_{i+1}},\\
 R_i(d)=\frac{t_{d+\delta_i}}{t_d}.
\end{gathered}
 \label{SM-eq:kingsector}
\end{equation}
All shifts of $d$ are modulo $Q$; the positive numbers $t_d$ themselves
are not taken modulo $Q$.

Replacing only tensor $i$ can break the constraint
$u_j=d_j+\sigma(p_j)$ only at $j=i$. In a sector supported on
$\{i\}$, its normalized output has precisely the form
\begin{equation}
 \sum_{a=0}^{q-1} f_{p_i,p_{i+1},a}(d_i)\,
                  t_{d_{i-1}}^{a-p_i},
 \label{SM-eq:kingrowtangent}
\end{equation}
where the functions $f$ are arbitrary. In the zero-displacement sector
one sums this expression over $i$. This follows directly by fixing
$b_i=p_{i+1}$ and allowing $a_i=a$ in the replaced tensor; the latter
changes the exponent at $i-1$ from $p_i$ to $a$.

\paragraph{Sectors with at least two displaced sites.}
There is no row tangent in these sectors. Taking $\mathbf p=0$ in
Eq.~\eqref{SM-eq:kingsector} gives
$h(\mathbf d+\boldsymbol\delta)=-k(\mathbf d)$.
Taking only $p_j=1$ then gives
\begin{equation}
 k(\mathbf d+e_j)=R_{j-1}(d_{j-1})k(\mathbf d).
 \label{SM-eq:kingholonomy}
\end{equation}
Choose $j$ with $\delta_{j-1}\ne0$. Iterating the shift of $d_j$
$Q$ times leaves $R_{j-1}$ unchanged. This positive ratio is not one,
so its $Q$th power is not one. Hence $k=h=0$.

\paragraph{Exactly one displaced site.}
Suppose only $\delta_i\ne0$. Varying $p_j$ from zero to one for
$j\notin\{i,i+1\}$ leaves Eq.~\eqref{SM-eq:kingrowtangent} and the
second term of Eq.~\eqref{SM-eq:kingsector} unchanged. Thus $k$ is
invariant under the unit shift of every such $d_j$ and depends only
on $d_i,d_{i+1}$. The equations with $\mathbf p=0$ and with only
$p_{i+1}=1$ now imply that
\begin{equation}
 k(d_i,d_{i+1}+1)-R_i(d_i)k(d_i,d_{i+1})
\end{equation}
is independent of $d_{i+1}$.
Since $R_i(d_i)^Q\ne1$, the unique cyclic solution of this recurrence
is constant in $d_{i+1}$. Therefore $k$ depends only on $d_i$, and
the zero-physical-index equation shows that
$h(\mathbf d+\boldsymbol\delta)$ depends only on $d_{i-1},d_i$.
For fixed $d_i$, the equations for all $p_i$ place its dependence on
$d_{i-1}$ in the intersection of the spaces
\begin{equation}
 \begin{gathered}
 V_p=\operatorname{span}\{t_d^{a-p}:0\le a<q\}
       \subset\operatorname{Fun}(\mathbb Z_Q),\\
 \bigcap_{p=0}^{q-1}V_p=\mathbb C1.
 \end{gathered}
 \label{SM-eq:kingpolynomialintersection}
\end{equation}
To verify the intersection, it suffices to use $p=0,q-1$: multiply
the two polynomial expressions by $t^{q-1}$. Their difference has
degree at most $2q-2<Q$ and vanishes at all $Q$ distinct $t_d$, so
it is zero as a polynomial and the original function is constant.
Consequently $h$ too depends only on site $i$.

\paragraph{The zero-displacement sector.}
Equation~\eqref{SM-eq:kingsector} now reads
\begin{equation}
 k(\mathbf d+\boldsymbol\sigma(\mathbf p))+h(\mathbf d)
 =\sum_i\sum_a f_{i;p_i,p_{i+1},a}(d_i)t_{d_{i-1}}^{a-p_i}.
 \label{SM-eq:kingdiagonal}
\end{equation}
For nonadjacent $i,j$ on $C_L$, take the mixed difference between
$p_i=0,1$ and $p_j=0,1$. Each summand on the right depends on at most
two adjacent physical indices, hence this gives
$\Delta_i\Delta_j k=0$, where $\Delta_i$ is the unit cyclic
difference in $d_i$. Decompose functions on $\mathbb Z_Q^L$ using,
on each coordinate, constants and zero-mean functions under uniform
averaging. Cyclic difference is invertible on the zero-mean factor.
Thus no component of $k$ can involve two nonadjacent coordinates.
Since $L\ge4$, every clique of $C_L$ is a vertex or an edge:
$k$ is a sum of constants, single-coordinate functions, and
zero-mean adjacent-pair functions. Setting $\mathbf p=0$ in
Eq.~\eqref{SM-eq:kingdiagonal} gives the same decomposition for $h$.

Project Eq.~\eqref{SM-eq:kingdiagonal} onto the zero-mean pair component
on $(i-1,i)$, averaging over the other coordinates. Only tangent
summand $i$ survives. It is independent of $p_{i-1}$; varying that
physical index forces the pair component of $k$ to be invariant
under the cyclic shift of its first coordinate. Its zero mean then
forces it to vanish. With all pair components of $k$ removed, the
pair component of $h$, as a function of $d_{i-1}$, lies in the
zero-mean projections of every $V_{p_i}$. Constants belong to all
$V_p$, so Eq.~\eqref{SM-eq:kingpolynomialintersection} says that these
projected spaces have zero intersection. The pair component of $h$
therefore vanishes as well. This proves Eq.~\eqref{SM-eq:kingrowsplit}.

\paragraph{Genericity and original-edge lifting.}
Choose fixed complements to the same-site operator sums in the two
contact algebras. Each complement has dimension
$c=Q^{2L}-LQ^2+L-1$. Equation~\eqref{SM-eq:kingrowsplit}, together with
internal tangent completeness at the witness, gives joint rank
\begin{equation}
 \rank[J_{\rm row}\mid\Theta_+\mid\Theta_-]
 =Lq^3Q^2-Lq^2+1+2c.
 \label{SM-eq:kingjoint}
\end{equation}
It attains the gauge upper bound plus the number of added columns,
so a maximal minor proves the same splitting on a nonempty
Zariski-open set. This step does not claim that an arbitrary
special-point splitting is automatically open.

Set $Q=q^3$. At each end of an inter-row contact, the generator now
lies in the sum of operator algebras on the triples of legs incident
to individual original sites. The intersection of the two endpoint
site partitions consists of individual bonds: no two original edges
have the same endpoints. The scalar/traceless support decomposition
therefore reduces every coarse generator to original-edge actions,
with the matched transpose and sign at its opposite endpoint.
Subtracting them leaves only internal row gauges. Intersect the
nonempty open splitting conditions with the band-injectivity,
maximal-gauge-rank, and smooth-generic-fiber conditions. The
normal-PEPS and generic-fiber argument of
Theorem~\ref{SM-thm:template} now proves Proposition~\ref{SM-prop:kinggeneric}.
Here $P=Nq^9$ and $E=4N$, giving Eq.~\eqref{SM-eq:kinggenericrank}.

\paragraph{Independent verification.}
The script \nolinkurl{research/king_rows/verify.py} checks the named
routes, both coarse partitions and periodic seams, evaluates the row
and single-replacement contractions independently, and checks the
polynomial intersection and displacement-sector joint ranks by exact
modular arithmetic. The smaller external dimension $Q=3$ at $q=2$
is used to check complete sector matrices of the same analytic lemma;
the actual king dimensions $Q=q^3$ are also checked in contraction
identities. These tests audit the formulas; the all-$q$ and
arbitrary-period conclusions follow from the argument above.

\section{Open-region kernels and minimality}
\label{SM-sec:sminimal}
\subsection{Why the external-leg definition is necessary}
Let $R$ be a connected induced region. The output of its contraction
contains physical variables $s_v$ for $v\in R$ and independent virtual output
indices on every cut half-edge. In particular, two cut legs ending at the
same exterior vertex are distinct output indices. Identifying them using
an exterior tensor would change the local question.

On each cut edge $e=vw$, $v\in R$, apply the invertible output change of
basis $F_{w,e}(t_e,a_e)$. This introduces an independent $q$-level variable
$t_e$ for that half-edge. The transformed open-region amplitude is
\begin{equation}
\begin{split}
\psi_R(s_R,t_\partial)
={}&q^{-|R|/2}\prod_{uv\in E(R)}W_{uv}(s_u,s_v)\\
&\times\prod_{e=vw\in\partial R}W_e(s_v,t_e).
\end{split}
\label{SM-eq:openamplitude}
\end{equation}
It is nowhere zero. An invertible output transformation preserves the
kernel of the open-region differential. No assertion about preservation
of a physical Hilbert metric is needed in this argument.

The single-site isomorphism now uses the extended star
\begin{equation}
\widetilde N_R[v]=\{v\}\cup N_R(v)\cup\{t_e:e\in\partial R, e\ni v\}.
\end{equation}
Every boundary variable is private to one center. A support containing
such a variable has only one center and therefore no relation. All
relations arise from supports contained entirely in $R$, with center
sets $C^R_S=\{v\in R:S\subseteq\{v\}\cup N_R(v)\}$.
Consequently the relation structure depends on the induced graph inside
$R$ and not on special coincidences in the exterior state.

\begin{lemma}[Open-region relations]
At a nondegenerate WGS, the kernel of the open-region differential is
generated by $\xi_{u,S}^{\boldsymbol k}-\xi_{v,S}^{\boldsymbol k}$ for each frequency assignment with $u,v\in R$ and
$S\subseteq N_R[u]\cap N_R[v]$. Its empty sector has dimension $|R|-1$.
The internal bond-gauge image has dimension $(q^2-1)|E(R)|+|R|-1$.
\end{lemma}
\begin{proof}
Apply the invertible output maps above and the one-site coordinate
isomorphisms. The sector map is Eq.~\eqref{eq:sector}; supports involving
private variables have multiplicity one. The scalar and singleton
argument from Sec.~\ref{sec:relations} applies to the connected graph
induced on $R$ and proves the internal gauge count. Undoing the
invertible changes of basis does not alter these kernel statements.
\end{proof}

A kernel from a proper subregion extends to a kernel on $R$ by changing
the same tensor blocks and setting all others to zero: contracting more
tensors cannot change the zero output. Define its new quotient as
\begin{equation}
\begin{split}
\mathcal D_R&=\im\Phi_R+
\sum_{\substack{R'\subsetneq R\\R'\text{ connected, induced}}}
\iota_{R'\to R}\ker J_{R'},\\
\mathcal Q_R&=\ker J_R/\mathcal D_R.
\end{split}
\label{SM-eq:newregion}
\end{equation}
All spaces in the denominator are embedded in the same parameter space.
The sum denotes a span, not a sum of dimensions. The unprojected
definition avoids incorrectly treating a local rescaling as an
environment-independent zero tensor.

\section{Normality and finite equivalences at WGS points}
\label{SM-sec:fiberdetails}
\subsection{Normality and the geometry of injective blocks}
\label{SM-sec:normalitydetails}
A nondegenerate WGS representation admits an injective blocking only if
the vertex set can be partitioned into connected blocks in each of
which every vertex has at most one external edge. On the infinite
square, triangular, and kagome lattices no finite region has this
property: on each of them the neighbors of a vertex come in antipodal
pairs $v\pm\boldsymbol d_i$ (two pairs on square and kagome, three on
triangular), so the vertex of a finite set that maximizes a generic
linear functional has one neighbor of every pair, hence at least two
neighbors, outside the set. The corresponding graph-state
representations are therefore not normal PEPS. On honeycomb, the
hexagons of one face color partition the vertices into injective
blocks whose coarse graph is the triangular lattice with one original
edge per coarse bond; the honeycomb WGS representation is normal. On
the star lattice $3.12^2$, likewise, the disjoint triangles are
injective blocks (each vertex has one external edge) with coarse
honeycomb graph and single-edge bonds, so the representation is normal
although $\nu=N(3r^2+2r^3)/3>0$.

Normality and the non-bond kernel are therefore independent
properties. On the infinite square, triangular, and kagome lattices,
the normal-PEPS theorem cannot be applied at the WGS point after any
partition into finite blocks. This corner argument does not exclude
winding blocks on finite tori, such as the square-lattice row bands
described after Corollary~\ref{cor:normal}. The mechanism is
the one noted in Ref.~\cite{molnar2018semiinjective} for purified
Gibbs states of commuting nearest-neighbor Hamiltonians on the square
lattice, which share the product-of-commuting-two-site-operators
structure of WGS: a corner of any finite block in the infinite lattice
has two external edges, and
Lemma~\ref{lem:injective} turns this observation into a criterion on
arbitrary graphs. The finite-equivalence problem in this setting
falls under semi-injective PEPS; Sec.~\ref{sec:integrability}
exhibits the finite non-bond equivalences explicitly. On the star
lattice the injective theorem applies to the blocked network and
confines every constant-state curve through $A_0$ to the coarse gauge
orbit; its tangent then lies in $\im\Phi_A$ plus the internal
open-region kernels of the triangle blocks, which is exactly where the
star-lattice non-bond kernel resides; the injective theorem thus
localizes every finite equivalence to those blocks. For generic tensors, injectivity is an open condition whenever an
injective witness exists. The total dimension inequality
$q^{|\partial R|}\le q^{|R|}$ is necessary but not sufficient:
internal bottlenecks must also satisfy the local cut conditions of
Eq.~\eqref{SM-eq:routecut}. The routing witnesses of Appendix~\ref{SM-sec:genericappendix} certify
nonzero injectivity minors on the specified regions.

\subsection{Character order, finite invariants, and nonlinear composition}
\label{SM-sec:finiteequivalencedetails}
Proposition~\ref{prop:finite} constructs exact
constant-state curves on independent sets of centers, and
Corollary~\ref{cor:lines} identifies their single-character
straight lines in the ray fiber.
The condition $\chi^2=1$ holds for every qubit character, and also for
characters with $k_w=q/2$ on every $w\in S$ when $q$ is even.
For a character of order greater than two, $1-t^2\chi^2$ is not constant,
so the straight line leaves the ray at second order, although the
exponential curve remains exact.

For a single qubit square relation there is also a direct finite
invariant. Ordinary invertible bond gauges preserve the matrix rank of
each physical slice across every bipartition of its virtual legs.
The original WGS slices have rank one. By Eq.~\eqref{eq:qexplicit},
the changed slice at either center of $A_0+t z_{S;u,v}$ is the sum of
the original product tensor and $\pm t$ times a second product tensor.
On each of the two support legs, their factors are linearly independent:
an invertible neighbor map takes them to the vectors
$[W_e(s,x)]_x$ and $[W_e(s,x)(-1)^x]_x$, which are independent since
both entries of $W_e(s,x)$ are nonzero. A bipartition separating these
two legs therefore has slice rank two for every $t\ne0$.
Consequently the ray-preserving line, and the corresponding rescaled
constant-state curve, cannot be obtained by ordinary bond gauges;
nonzero site rescalings do not change this rank either.

Exact composition does not imply that sums of such velocities generate
ray-preserving straight lines. For example, two square relations sharing
one center can have three jointly independent centers and distinct
qubit characters $a,b$. Along their summed straight-line direction the
three local factors are $1+t(a+b)$, $1-ta$, and $1-tb$, giving
\begin{equation}
\begin{split}
&[1+t(a+b)](1-ta)(1-tb)\\
&\qquad=1-t^2(2+ab)+t^3(a+b).
\end{split}
\end{equation}
This is generally not constant. The exponential factors
$e^{t(a+b)},e^{-ta},e^{-tb}$ instead have product one and supply the
required nonlinear corrections. On a bulk qubit square lattice the
$2N$ elementary square relations each give a ray-preserving line;
linear combinations with jointly independent centers integrate by
Proposition~\ref{prop:finite}, without implying an affine linear fiber
piece. On honeycomb there is no non-bond kernel, while the elementary
triangle and diamond relations have adjacent centers and are treated
in Sec.~\ref{sec:obstruction}.

\subsection{Second-order certificates and numerical continuation}
\label{SM-sec:obstructionchecks}
For the expansion in Eq.~\eqref{eq:secondorder},
the linearized equations of the exact-state fiber have kernel
$\ker J_A$, whereas those of the physical-ray fiber have kernel
$\ker J_{\perp,A}$. The latter allows $J_A\delta=c\Psi(A_0)$;
subtracting $c r_A$ fixes the amplitude to first order and gives an
affine-null representative, as in Eq.~\eqref{eq:affineprojective}.
For the directions $\delta\in\ker J_A$ considered here, a
constant-ray curve requires
$\Psi_2(\delta,\delta)\in\im J_A$, since the ray itself lies in
$\im J_A$. Equivalently, its projected Hessian obstruction is
\begin{equation}
\begin{aligned}
\mathfrak B_A(\delta,\delta)
&=(I-P_{\im J_A})D^2\Psi_A[\delta,\delta]\\
&=2(I-P_{\im J_A})\Psi_2(\delta,\delta)=0,
\end{aligned}
\end{equation}
where $P_{\im J_A}$ is the orthogonal physical projector onto the
affine image. This is a necessary second-order condition, not a
sufficient condition for integrability at all orders.
\paragraph{Certificate for Proposition~\ref{prop:obstruction}.}
All quantities are rational: $J_A$ is a $512\times128$ integer matrix,
$z_{1,2}$ are explicit rational vectors from Eq.~\eqref{eq:qexplicit},
and $\Psi_2$ is a finite contraction. Exact rational elimination gives
$\rank J_A=76$, $J_Az_i=0$, $\rank[J_A\,|\,\Psi_2(z_i,z_i)]=76$ for each
$i$ and also for the pair of orthogonal diagonals in adjacent squares,
and $\rank[J_A\,|\,\Psi_2(z_1+z_2,z_1+z_2)]=77$
(\nolinkurl{experiments/second_order_obstruction_exact.py}). The
obstruction is the mixed second-order contribution. If the ray fiber
were smooth, its tangent space would contain both $z_i$ and their sum,
and the sum would be the velocity of an analytic curve in that fiber,
contradicting the second-order condition.
Proposition~\ref{prop:finite} locates the obstruction. Two relations
whose center sets are jointly independent compose exactly, so their
sum integrates with the nonlinear corrections supplied by the
exponential family. In particular $\Psi_2(\delta,\delta)$ lies in
$\im J_A$ and is canceled by $J_A\delta'$; the second-order contribution
need not vanish on the uncorrected straight line.
The parallel pair of Proposition~\ref{prop:obstruction} has adjacent
centers, and there the cross contractions do not factorize into
characters. The same obstruction occurs for generic
non-bond directions on the $4\times4$ square torus, where the
component of $\Psi_2$ outside $\im J_A$ carries about half of its norm.

Triangle-based relations have adjacent centers, where the exponential
construction does not apply; for these, Gauss--Newton continuation
integrates every tested direction beyond second order: on the open triangle with one
external leg per vertex (the star-lattice block), the open kagome
bowtie, the open diamond with three external legs per vertex, and an
open strip of three triangles, a Gauss--Newton continuation finds, for
random non-bond kernel directions $\delta$ and displacements
$t\le0.2$, corrections $x(t)$ orthogonal to the complete kernel and to
the ray with $\|(I-\ket\psi\bra\psi)\Psi(A_0+t\delta+x)\|/\|\Psi\|
\le10^{-13}$ and $\|x\|\approx(0.04$--$0.1)\,t^2$; since $J_A$ is
injective on that complement, such a curve has velocity exactly
$\delta$. On the two edge-adjacent squares the same continuation
leaves a residual of order $t^2$ for the parallel pair and for generic
directions, and none for the orthogonal pair, in agreement with the
exact certificates
(\nolinkurl{experiments/nongauge_integrability.py}). A proof of
integrability for the triangle-based motifs, and a characterization
of the tangent cone where independent and adjacent centers mix, are
natural next steps.

\section{Euclidean QGT spectrum and conditioning}
\label{SM-sec:positive_spectrum}
The support theorem determines the rank and kernel independently of a
parameter metric. This section additionally determines the positive
spectrum in fixed balanced coordinates, and bounds it away from the
Fourier point. These results concern the exact tensor differential,
not a sampled metric or a solver's performance.

\subsection{Coordinates, exact zeros, and the positive condition number}
Fix the normalized WGS representation and balanced factors
\begin{equation}
\begin{gathered}
W_e=U_e\Sigma_eV_e^\dagger,\\
F_{u,e}=U_e\sqrt{\Sigma_e},\quad
F_{v,e}=\overline V_e\sqrt{\Sigma_e}.
\end{gathered}
\end{equation}
No site tensor is independently rescaled. All parameter inner products
are Euclidean in the raw complex tensor entries. Set
$\mathcal M_A=J_{\perp,A}^\dagger J_{\perp,A}$.
This is the projective Gram form of the QGT
\cite{provost1980riemannian,haegeman2014geometry}.
Its nullity is $g+\nu$. If $Q_g$ has orthonormal columns spanning
$\mathcal G_A^\perp$, then $Q_g^\dagger \mathcal M_A Q_g$ has exactly $\nu$ zeros
and the same positive eigenvalues as $\mathcal M_A$: every positive eigenvector
is orthogonal to the complete kernel and hence to $\mathcal G_A$.
An orthonormal complement of the full kernel removes those remaining
zeros without changing the positive eigenvalues. The quantity
$\kappa_+=\lambda_{\max}/\lambda_{\min}^{+}$ therefore measures the QGT
on physical parameter directions. It is the square of the corresponding
nonzero-singular-value condition number of $J_\perp$; these spectral
relations follow from the SVD \cite{golub2013matrix}.

Balanced SVD degeneracies allow unitary changes of virtual coordinates,
which preserve these eigenvalues. A general nonunitary bond gauge can
change them while preserving the state, physical tangent image, and
nullity. The spectrum is therefore assigned only in the stated tensor
coordinates~\cite{evenbly2018gauge}.

\subsection{Exact spectrum at the Fourier graph-state point}
Equip functions on $N[v]$ with the uniform probability inner product.
Let $L_v$ denote the map~\eqref{eq:localmap} from the Euclidean site parameter block
to this function space. At $W_eW_e^\dagger=qI$, each balanced half-edge
factor has all singular values $q^{1/4}$. The phase denominator in
Eq.~\eqref{eq:localmap} is a diagonal unitary on function values. For a fixed physical
value at the center, the numerator is a tensor product of those half-edge
maps, multiplied by $\sqrt q$. Including the uniform measure on
$q^{\deg(v)+1}$ configurations gives
\begin{equation}
L_vL_v^\dagger=q^{-\deg(v)/2}I_{\operatorname{Fun}(N[v])}.
\label{SM-eq:slocalisometry}
\end{equation}
Blocks with different center values have disjoint parameter and function
coordinates, so this identity applies to the whole site block.

Let $P_v$ project onto the physical vectors
$\chi_{S,\boldsymbol k}\psi$ with $S\subseteq N[v]$, including the
constant sector, and let $P_0=\ket\psi\bra\psi$. The different site parameter
blocks are orthogonal, and multiplication by the normalized WGS is
isometric from the global uniform function space into Hilbert space.
Consequently, on the physical tangent image,
\begin{equation}
\begin{gathered}
J_\perp J_\perp^\dagger=\sum_vq^{-\deg(v)/2}(P_v-P_0),\\
\lambda_{S,\boldsymbol k}=\sum_{v\in C_S}q^{-\deg(v)/2},\qquad S\ne\varnothing.
\end{gathered}
\label{SM-eq:sframespectrum}
\end{equation}
The nonzero spectra of $J_\perp J_\perp^\dagger$ and $\mathcal M_A$ agree
\cite{golub2013matrix}.
In parameter space, normalized positive eigenvectors can be chosen as
$J_\perp^\dagger(\chi_{S,\boldsymbol k}\psi)/\sqrt{\lambda_{S,\boldsymbol k}}$.
Each support has $(q-1)^{|S|}$ frequency copies.
The formula is exact for every finite graph at this balanced Fourier point,
including irregular graphs and short periodic identifications.

For a $z$-regular graph with distinct closed neighborhoods, a full
closed-neighborhood support has only its own center, so $m_{\min}=1$.
A singleton has $z+1$ centers, and no nonempty support has more centers
than any of its singletons. Hence $m_{\max}=z+1$ and
$\kappa_+=z+1$, independent of $N$ and $q$. Graphs with coincident
closed neighborhoods use the exact multiplicity formula instead.
For qubits this point is $\phi=\pi$; for arbitrary $q$ the edge phase
$2\pi/q$ gives $W_eW_e^\dagger=qI$.

\subsection{Local-frame assembly away from that point and uniform bounds}
At general nondegenerate WGS phases, Fourier transform each map~\eqref{eq:localmap}
with the uniform-measure convention. Let $C_v$ be its coefficient matrix
with the constant row removed, and $E_v$ embed its nonconstant local
character rows into the global support union. The finite matrix
\begin{equation}
B=\sum_vE_vC_vC_v^\dagger E_v^\dagger
\label{SM-eq:sframeassembly}
\end{equation}
is the physical frame expressed in the orthonormal character basis.
It is positive definite on the $K_q(G)-1$ allowed characters and has
exactly the nonzero eigenvalues of $\mathcal M_A$. This construction retains
all local tensor-entry directions and does not replace the Euclidean
metric by the metric of the inverse-map coordinates.

The squared singular values of $L_v$ are bounded by
\begin{equation}
\ell_v=\prod_{e\ni v}\frac{\sigma_{\min}(W_e)}q,
\qquad u_v=\prod_{e\ni v}\frac{\sigma_{\max}(W_e)}q.
\end{equation}
This follows from the tensor-product singular values in the numerator;
the phase denominator and the normalized Fourier transform are unitary.
Removing the constant row preserves these quadratic-form bounds on the
remaining rows. On the union space, each character is present at least
once and at most $\Delta+1$ times, where $\Delta$ is the maximum degree.
Summing the local bounds in Eq.~\eqref{SM-eq:sframeassembly} therefore yields
\begin{equation}
\min_v\ell_v\le\lambda_{\min}(B)\le\lambda_{\max}(B)
\le(\Delta+1)\max_v u_v.
\label{SM-eq:sframebound}
\end{equation}
For a $z$-regular graph with the same weight matrix $W$ on every edge,
\begin{equation}
\kappa_+\le(z+1)
\left[\frac{\sigma_{\max}(W)}{\sigma_{\min}(W)}\right]^z.
\label{SM-eq:positiveconditionbound}
\end{equation}
It is uniform in $N$ at fixed $q$, degree, and nondegenerate edge phase.
It can deteriorate as edge weights approach degeneracy, or after a
nonunitary change of tensor gauge, and is not uniform in either.

\subsection{Perturbations away from the WGS rank-drop set}
\label{SM-sec:nearbygramproof}
The preceding bounds vary the weights within the WGS family. They
control its positive spectrum but do not describe how its additional
zeros open under arbitrary tensor perturbations. Here the graph, bond
dimensions, independent parameter entries, and Euclidean parameter
metric are fixed. Let $A(\varepsilon)$ be a $C^2$ path through a
nondegenerate WGS representation $A_0$, parametrized by real
$\varepsilon$, with $\Delta=A'(0)$ and $\Psi(A(\varepsilon))\ne0$
near zero. The gauge rank is locally constant: a maximal minor of
$\Phi_{A_0}$ remains nonzero, and the scalar-cycle upper bound holds
at every tensor. Adding the independent ray preserves constant
dimension. Thus a $C^2$ orthonormal frame $Q(\varepsilon)$ for
$\mathcal G_{A(\varepsilon)}^\perp$ exists locally. No assumption of
constant physical differential rank is made.

Write $\widetilde J(\varepsilon)=J_{\perp,A(\varepsilon)}Q(\varepsilon)$,
$K_0=\ker \widetilde J(0)$, and $R_0=\im \widetilde J(0)=\mathcal T_{A_0}$.
For an isometry $Z_0:\mathbb C^\nu\to K_0$, define
\begin{equation}
L_\Delta=(I-P_{R_0})\widetilde J'(0)Z_0,\qquad \nu=\dim K_0.
\label{SM-eq:openingmapSM}
\end{equation}
Changing either orthonormal frame only changes $L_\Delta$ by a
constant unitary on its domain: derivatives of frame changes produce
terms in $R_0$, which its projector removes. Its singular values
are consequently well defined for the chosen tensor metric and path.
General nonunitary reparametrizations need not preserve their
prefactors.

\paragraph{Proof of Proposition~\ref{prop:nearbygram}.}
Choose orthonormal domain and codomain decompositions
$K_0^\perp\oplus K_0$ and $R_0\oplus R_0^\perp$. The restriction
$D=\widetilde J(0)|_{K_0^\perp}:K_0^\perp\to R_0$ is invertible, and
\begin{equation}
\widetilde J(\varepsilon)=
\begin{pmatrix}
D+\varepsilon D_1 & \varepsilon C\\
\varepsilon F & \varepsilon L_\Delta
\end{pmatrix}+O(\varepsilon^2).
\label{SM-eq:openingblocks}
\end{equation}
First multiply on the right by $\exp(\varepsilon X)$, where
\begin{equation}
X=\begin{pmatrix}
0&-D^{-1}C\\
C^\dagger(D^{-1})^\dagger&0
\end{pmatrix}.
\end{equation}
Since $X^\dagger=-X$, this preserves singular values and removes the
upper-right block to first order. Multiplication on the left by
$\exp(\varepsilon Y)$, with
\begin{equation}
Y=\begin{pmatrix}
0&(D^{-1})^\dagger F^\dagger\\
-FD^{-1}&0
\end{pmatrix},
\end{equation}
likewise removes the lower-left block. The resulting matrix differs
from $\operatorname{diag}(D+\varepsilon D_1,\varepsilon L_\Delta)$
by $O(\varepsilon^2)$ in operator norm. The singular-value perturbation
bound~\cite{golub2013matrix} therefore gives, for the $\nu$ singular
values tending to zero, with zeros included when necessary,
\begin{equation}
s_i(\widetilde J(\varepsilon))
=|\varepsilon|\sigma_i(L_\Delta)+O(\varepsilon^2).
\end{equation}
Order both small spectra increasingly. Squaring gives
$\mu_i=c_i\varepsilon^2+O(|\varepsilon|^3)$, where
$c_i=\lambda_i(L_\Delta^\dagger L_\Delta)$. For $c_i=0$ the
stronger estimate $\mu_i=O(\varepsilon^4)$ follows directly.
If $K_0^\perp$ is empty the assertion follows from
$\widetilde J(\varepsilon)=\varepsilon \widetilde J'(0)+O(\varepsilon^2)$.
When $\nu>0$, $L_\Delta$ is injective, and $\widetilde J(0)\ne0$, all reduced modes
are positive at sufficiently small nonzero $\varepsilon$ and
\begin{equation}
\kappa\!\left(\widetilde J(\varepsilon)^\dagger \widetilde J(\varepsilon)\right)
\sim\frac{\|\widetilde J(0)\|_{\rm op}^2}
{\sigma_{\min}(L_\Delta)^2\varepsilon^2}.
\label{SM-eq:openingcondition}
\end{equation}
In the WGS setting of the main text, $\rank J_{\perp,A_0}\ge q-1>0$
(the single-site Fourier sectors already suffice), so $\widetilde J(0)\ne0$
is automatic. For a general matrix curve it is necessary in
Eq.~\eqref{SM-eq:openingcondition}: $\widetilde J(\varepsilon)=\varepsilon I$
has quadratic Gram eigenvalues but constant condition number.
Generic completeness alone does not imply injectivity of $L_\Delta$.
For example, the scalar matrix $\widetilde J(\varepsilon)=\varepsilon^2$
has full rank off zero but a quartic Gram eigenvalue.

\paragraph{Mixed Hessian and the displacement convention.}
No differentiated QR basis is required to evaluate the opening map.
Put $Z=Q(0)Z_0$. Frame-derivative terms are in $\im J_{\perp,A_0}$.
Moreover, normalization and horizontal-projector derivatives on a
projective-null column contribute only in the affine image
$\im J_{A_0}=\mathbb C\psi_0\oplus\mathcal T_{A_0}$. Hence
\begin{equation}
L_\Delta
=\frac{(I-P_{\im J_{A_0}})D^2\Psi_{A_0}[\Delta,Z]}
{\|\Psi(A_0)\|}.
\label{SM-eq:openinghessian}
\end{equation}
The columns of $Z$ need not themselves be affine-null. The extra
ray projection in Eq.~\eqref{SM-eq:openinghessian} accounts for this.
This is the mixed projected Hessian already underlying the
second-order fiber obstruction in Appendix~\ref{SM-sec:obstructionchecks}.
If $\Delta_v$ denotes the perturbation at site $v$, multilinearity
evaluates $D^2\Psi[\Delta,Z]$ by the sum of contractions replacing
two distinct site tensors; no finite-difference subtraction is needed.

In each straight-line scan below,
$A(\varepsilon)=A_0+\varepsilon\Delta$ with
$\|\Delta\|=\|A_0\|$, so $|\varepsilon|$ is the relative Frobenius
displacement of the concatenated tensors. It is not a gauge-invariant
distance between states. For a general path the phase-aligned
normalized-state distance has leading term
$|\varepsilon|\|J_{\perp,A_0}\Delta\|$ when this coefficient is
nonzero. If $\Delta\in\ker J_{\perp,A_0}$, physical displacement can
start at second or higher order.

\paragraph{Paths with persistent nullity.}
Let $\mathcal V$ apply fixed invertible physical matrices $V_v$ to
the site tensors, and let $V_{\rm phys}=\bigotimes_v V_v$.
Contraction and gauge action obey
\begin{equation}
J_{\mathcal V A}\mathcal V=V_{\rm phys}J_A,\qquad
\Phi_{\mathcal V A}=\mathcal V\Phi_A.
\end{equation}
Thus affine rank, projective rank, gauge rank, and $\nu$ are unchanged.
A smooth path of such transformations can leave the fixed-basis WGS
family while retaining its full non-bond nullity. If every $V_v$
is unitary, the Euclidean parameter transformation is also unitary
and the entire projective Gram spectrum is unchanged. Varying
nondegenerate WGS edge phases also preserves $\nu$, although positive
eigenvalues may change. Neither path conflicts with generic completeness.

For qubits, this gives an exact physical evolution. Set
\begin{equation}
H_{\rm d}=\sum_{uv\in E_G}J_{uv}Z_uZ_v+\sum_vh_vZ_v
\label{SM-eq:zzpersistent}
\end{equation}
with real couplings. For initial WGS edge phases $\varphi_{uv}$,
use $Z_uZ_v=(1-2s_u)(1-2s_v)$ to write the evolved amplitude,
up to a global phase, as
\begin{equation}
2^{-N/2}\exp\!\left[i\sum_{uv}\varphi_{uv}(t)s_us_v
+i\sum_v\theta_v(t)s_v\right],
\end{equation}
where $\varphi_{uv}(t)=\varphi_{uv}-4J_{uv}t$ and
$\theta_v(t)=2t(h_v+\sum_{u:uv\in E_G}J_{uv})$.
Its WGS lift retains the exact nullity on each interval where the
edge matrices remain invertible. A many-body diagonal term can
instead introduce phases outside this pairwise family. Nor does
staying in the WGS state family force an unrestricted solver to
choose a rank-preserving tensor lift.

\paragraph{Dependence on the velocity lift.}
Let $b=-i(H-\langle H\rangle)\psi_0$ be tangent and
$J_{\perp,A_0}\Delta=b$. Another lift is $\Delta+z$ for any
$z\in\ker J_{\perp,A_0}$. With the same initial orthonormal
kernel basis $Z$, Eq.~\eqref{SM-eq:openinghessian} gives
\begin{equation}
L_{\Delta+z}-L_\Delta
=\frac{(I-P_{\im J_{A_0}})D^2\Psi_{A_0}[z,Z]}{\|\Psi(A_0)\|}.
\label{SM-eq:openingliftdependence}
\end{equation}
This need not vanish for a non-bond $z$. The coefficients thus
characterize a lift in a specified tensor metric, not just a Hamiltonian.
Equation~\eqref{eq:lift} likewise requires a choice of its
coefficients $\alpha_{vS}$.

Let $A_{ZZ}(t)$ be a smooth WGS lift of
$e^{-itH_{ZZ}}\psi_0$, and let $\mathcal U_X(t)$ apply
$e^{-ith_v^XX}$ to each physical leg. The curve
$\mathcal U_X(t)A_{ZZ}(t)$ retains $\nu$ and has physical velocity
$-i(H_{ZZ}+\sum_vh_v^XX_v)\psi_0$ at zero. Its opening map vanishes
by the proposition. It is a first-order lift of the Ising velocity,
not its exact finite-time trajectory when the two parts do not commute.
The straight line with the same initial tensor velocity has the same
zero quadratic coefficients, but curvature is needed to conclude
that its zero modes persist exactly.

\paragraph{Ridge crossover and compatible forces.}
For a singular pair
$\widetilde J(\varepsilon)v_i=\sqrt{\mu_i}\,u_i$ and physical source $b$,
the compatible force is $\boldsymbol F=\widetilde J^\dagger b$. Its regularized parameter
and physical coefficients are respectively
\begin{equation}
\begin{aligned}
x_i&=\frac{\sqrt{\mu_i}\langle u_i,b\rangle}{\lambda+\mu_i},\\
(\widetilde Jx)_i&=\frac{\mu_i\langle u_i,b\rangle}{\lambda+\mu_i}.
\end{aligned}
\label{SM-eq:openingcompatible}
\end{equation}
At $\lambda=0$, a lifted mode can therefore require a parameter
coefficient of order $1/|\varepsilon|$, while its physical coefficient
remains bounded. This does not contradict the absence of a
$1/\lambda$ divergence for a compatible force at the exact WGS point.
An independent error $\eta_i$ in the right-hand side instead gives
\begin{equation}
\delta x_i=\frac{\eta_i}{\lambda+\mu_i},\qquad
(\widetilde J\delta x)_i=\frac{\sqrt{\mu_i}\eta_i}{\lambda+\mu_i}.
\label{SM-eq:openingerror}
\end{equation}
For $c_i>0$, the leading denominator is
$\lambda+c_i\varepsilon^2$, giving the crossover
$|\varepsilon|\sim\sqrt{\lambda/c_i}$ when it lies in the
perturbative regime. Equation~\eqref{SM-eq:openingerror} describes
sensitivity to an independent force error; it does not ascribe
incompatibility to sampling metric and force from the same derivative
matrix. Converting the displacement scale into a time interval
requires a specified physical evolution and its tensor path.
Removing only the current exact kernel preserves the physical image;
discarding modes that have become positive is an additional truncation.

\paragraph{Finite-graph spectral experiment.}
The producer \nolinkurl{experiments/wgs_nearby_spectrum.py} uses the
open $3\times3$ qubit square graph, with all $128$ tensor entries
independent. Its twelve edge phases are uniform draws from
$[0.6,2.4]$ using NumPy generator seed $20260920$, and the tensors
use the balanced SVD factors stated at the start of this section.
Three complex Gaussian directions, seeds $101,102,103$, are rescaled
to $\|\Delta\|=\|A_0\|$. Each path uses seventeen logarithmically
spaced values from $10^{-1}$ to $10^{-5}$. The current projective
Jacobian and its gauge/ray complement are recomputed at every point.
At zero the reduced dimension is $83$, its rank is $75$, and
$\nu=8$; all sampled nonzero displacements have rank $83$ at relative
singular-value threshold $10^{-10}$. The small eigenvalues are
squared singular values of this reduced Jacobian, avoiding the loss
of precision from explicitly forming and diagonalizing its Gram.

The opening map is computed by exact multilinear differentiation
of the tensor contraction, followed by analytic differentiation of
normalization and the horizontal projector. The independent
affine-image expression~\eqref{SM-eq:openinghessian} agrees to relative
error below $1.9\times10^{-15}$. Centered finite differences of the
state, raw Jacobian, and projective Jacobian converge quadratically
to their analytic derivatives when the step is halved.
The $24$ coefficients $c_i$ lie between $0.02900398$ and $0.21163185$.
Fits over the common window $10^{-5}\le\varepsilon\le10^{-3}$ give
exponents from $1.999721$ to $2.000415$, and the largest relative
coefficient error at $10^{-5}$ is $2.29\times10^{-5}$.

Direct edge-phase amplitudes independently check the base PEPS
contraction to vector error below $3.9\times10^{-15}$. Gauge
annihilation is checked throughout, and a finite nonunitary bond
gauge preserves the state and rank. As a rank-preserving control,
$e^{-i\theta Y_0}$ for $\theta=0.001,0.01,0.1$ makes the Born
probability nonuniform but leaves all eight zeros and the entire
Gram spectrum unchanged to relative error below $3.0\times10^{-15}$.
An additional two-by-two matrix test,
$B(\varepsilon)=\left(\begin{smallmatrix}1&\varepsilon\\0&\varepsilon^2
\end{smallmatrix}\right)$, checks that omitting the cokernel
projection would incorrectly predict quadratic rather than quartic
opening. Configuration, all $408$ small-mode data points, source
hashes, and checks are recorded in
\nolinkurl{manuscript/quantum/data/wgs_nearby_spectrum}. Figure~\ref{SM-fig:randomopening}
retains the three-direction comparison; Fig.~\ref{fig:nearbygram} uses
one representative direction alongside the physical lift below.

\begin{figure}[!tbp]
\centering
\includegraphics[width=\columnwidth]{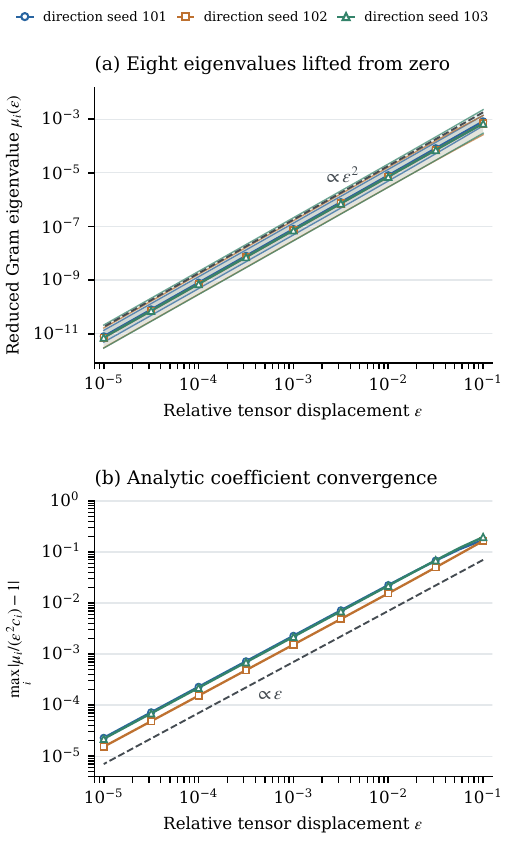}
\caption{Random-direction checks on the open $3\times3$ WGS square
graph. (a) Envelopes of the eight small Gram eigenvalues for each
seed, with geometric means of the middle pair. (b) Maximum relative
error from the eight analytic quadratic coefficients. The local-unitary
control preserves all eight zeros and the whole spectrum.}
\label{SM-fig:randomopening}
\end{figure}

\paragraph{Physical-path controls and minimum-norm lifts.}
The producer \nolinkurl{experiments/wgs_physical_paths.py} uses the
same graph, base phases, tensor metric, and seventeen displacements.
For $H=\sum_{uv\in E_G}Z_uZ_v+h\sum_vX_v$, with $h=0$ and $0.7$,
it computes the Euclidean minimum-norm projective lift
$\Delta_{\min}=J_{\perp,A_0}^{+}[-i(H-\langle H\rangle)\psi_0]$
using relative SVD cutoff $10^{-10}$. The scanned line is
$A_0+\varepsilon\Delta_{\min}/r$, where
$r=\|\Delta_{\min}\|/\|A_0\|$. The values of $r$ are $1.20825$
and $1.37729$, respectively; hence $\varepsilon/r$ is the associated
linear time parameter, not a time-integrated trajectory.
The physical velocities agree with independent computational-basis
Hamiltonian actions to vector error below $10^{-14}$.

For both lifts, six squared singular values of the opening map are
resolved as positive and two are below $10^{-28}$. The positive
coefficients range from $8.30\times10^{-4}$ to $7.26\times10^{-3}$
for $ZZ$ and from $9.10\times10^{-4}$ to $6.36\times10^{-3}$ for
Ising. Fits over $10^{-5}\le\varepsilon\le10^{-3}$ give quadratic
exponents $1.99790$--$1.99955$ and $1.99846$--$1.99931$;
the other two modes have exponents $3.99780$--$3.99816$ and
$3.99838$--$3.99844$, respectively. The positive Hessian coefficients
agree with $\mu_i/\varepsilon^2$ at $10^{-5}$ within relative errors
$1.16\times10^{-4}$ and $8.47\times10^{-5}$.
The quartic fits are numerical observations for these two prescribed
lines, not a general theorem for vanishing quadratic coefficients.
Squared singular values are used throughout, including modes too
small for the fixed relative rank threshold to resolve.

For an independent rank-preserving lift, keep each head factor
$F_{v,e}(0)$ fixed at its balanced initial value and set
\begin{equation}
\begin{aligned}
W_e(t)_{sr}&=W_e(0)_{sr}e^{-it(1-2s)(1-2r)},\\
F_{u,e}(t)&=W_e(t)F_{v,e}(0)^{-T}.
\end{aligned}
\label{SM-eq:zzsmoothfactors}
\end{equation}
This is a smooth, exact $ZZ$ path through the original tensors.
Following it by local $e^{-i0.7tX}$ rotations supplies the structural
Ising lift above. Their velocities are differentiated analytically
and checked by centered finite differences; both opening-map norms
are below $9\times10^{-15}$. At $t=0.001,0.003,0.01,0.03,0.1$,
both curves retain rank $75$, gauge/ray dimension $45$, and eight
residual Gram eigenvalues below $10^{-28}$. Direct computational-basis
$ZZ$ phases, followed where applicable by bit-flip rotations, check
the contracted states independently. The weights remain nondegenerate
at these times. Data and source hashes are in
\nolinkurl{manuscript/quantum/data/wgs_physical_paths}; tests also check
minimum-norm orthogonality to the complete projective kernel.
These scans test neither time integration nor large-system scaling;
the parameter space after bond expansion is outside their scope.

\section{Integer proof certificates}
\label{SM-sec:certificates}
An injectivity certificate records a finite region, every boundary input,
and a path to a distinct physical endpoint for each input. Paths use
no internal edge twice. Their delta-tensor construction proves an
injective boundary map for every integer $q\ge2$, as in
Eq.~\eqref{SM-eq:routeisometry}. The checker verifies inputs, graph edges,
capacities, endpoints, translation types, and the exterior cover.
Thus these records certify finite combinatorial statements in the proof;
they are not numerical singular-value estimates. Each verification
report cited below records the hashes of every file it checks, so the
certified version remains identifiable independently of later changes
to the repository. The accompanying frozen review bundle is identified by its SHA-256
manifest and archive checksum. Its common entry point is
\nolinkurl{reproducibility/README.md}, and the command
\texttt{python reproducibility/reproduce.py --suite all}
checks the manifest and replays the finite checks in a temporary copy.
The bundle includes the source, accepted certificates, and independent
verifiers; no partition search is needed for replay.
Appendix~\ref{SM-sec:reproducibility} explains the scopes of the replay
suites. Table~\ref{SM-tab:proofdependencies} maps the proof obligations
to their analytic arguments and certificate families.

\begin{table*}[!tbp]
\centering
\caption{Proof dependencies for the generic comparison. Each row specifies
normality/injectivity, internal completeness, coarse-edge splitting, and
extension to all stated periods. ``Direct'' means that no blocking or
composite-edge splitting is required. All period bounds and cell conventions
are those of Table~\ref{SM-tab:generic}; finite torus checks alone do not prove
an all-period statement. The last column concerns global generic completeness,
not merely the availability of a local witness.}
\label{SM-tab:proofdependencies}
\scriptsize
\setlength{\tabcolsep}{3pt}
\renewcommand{\arraystretch}{1.0}
\centering
\begin{adjustbox}{max width=\textwidth}
\begin{tabular}{>{\raggedright\arraybackslash}p{2.8cm}>{\raggedright\arraybackslash}p{2.8cm}>{\raggedright\arraybackslash}p{2.5cm}>{\raggedright\arraybackslash}p{3.4cm}>{\raggedright\arraybackslash}p{2.6cm}>{\raggedright\arraybackslash}p{0.9cm}}\toprule
Families / proof & Normal conditions and certificates & Internal block kernel & Composite-edge detector & All-period extension & Scope\\\midrule
Honeycomb; three subdivisions & Maximal-rank WGS witness & Direct & None & Exact finite-graph witness; bulk faithfulness for honeycomb & All $q$\\
Star, square-octagon, cross (App.~\ref{SM-sec:strivalent}) & Invertible triangle/square blocks & Injective grouped site maps; internal cycle theorem & Single original edges & Intact blocks and simple coarse quotient & All $q$\\
Square, kagome, snub-square (Apps.~\ref{SM-sec:archgenericproof}, \ref{SM-sec:m28generic}) & Original-edge and vertex partitions; m27/m28 routes & Direct & None & Translated injective covers & All $q$\\
Ruby; \texttt{t2005}, shuriken (Apps.~\ref{SM-sec:m28generic}, \ref{SM-sec:triangleextension}) & Triangle-block partitions; m28/blocked routes & Injective grouped triangle sites & Two-site double-bond detector; opposite-site annihilation & Translated covers on original legs & All $q$\\
Elongated triangular, maple-leaf (App.~\ref{SM-sec:m29generic}) & Dimer/triangle partitions; m29 routes & Injective grouped block sites & Joint dimer / asymmetric triangle triple-bond detectors & Covers with stated parity and vacancy cells & All $q$\\
Dice; \texttt{t2016}--\texttt{t2018} (App.~\ref{SM-sec:generic39proof}) & Original-edge and vertex partitions; open routes & Direct & None & Open-box exterior covers & All $q$\\
\texttt{t2001} (App.~\ref{SM-sec:triangleextension}) & Triangle/singleton partitions; blocked routes & Injective grouped triangle sites & Single original edges & Translated exterior covers & All $q$\\
\texttt{t2002}--\texttt{t2004}, \texttt{t2011}; checkerboard (App.~\ref{SM-sec:fiveextension}) & Triangle/dimer/$K_4$ partitions; open19 routes & Injective grouped block sites & Two-site double-bond, joint dimer triple-bond, or clique detector & Translated covers; checkerboard supercells & All $q$\\
\texttt{t2006,t2007} (App.~\ref{SM-sec:bowtieextension}) & Bowtie partitions; open14 routes & Internal normal partitions; center flattenings & Adjacent/nonadjacent leaf-pair double-bond detectors & Translated exterior covers & All $q$\\
\texttt{t2009,t2012} (App.~\ref{SM-sec:diamondextension}) & Diamond partitions; open12 routes & Injective grouped diamond sites & Joint double-bond / isolated triple-bond region detectors & Translated exterior covers & All $q$\\
\texttt{t2010,t2013} (App.~\ref{SM-sec:wheelextension}) & Injective wheel copy witness; geometry checks & Rim-dimer reduction; center and Bell-pair sectors & Two-channel joint dimer detector; external doubles for \texttt{t2013} & Offset bounds exclude aliasing for periods $\ge3$ & All $q$\\
\texttt{t2020} (App.~\ref{SM-sec:t2020extension}) & Diamond partitions; open10 routes & Injective grouped diamond sites & Joint unequal $1+2$ / $2+1$ region detector; endpoint refinement & Translated exterior covers & All $q$\\
\texttt{t2008,t2014} (App.~\ref{SM-sec:open4qubits}) & Injective twelve-site block / normal partitions; open4 routes & Explicit nonzero integer minors & Joint incidence minors; singleton endpoint refinement & Offset bound / translated exterior cover & $q=2$\\
\texttt{t2015} (Apps.~\ref{SM-sec:open4qubits}, \ref{SM-sec:open5qubits}) & Path-block normal partitions; open4 routes & Explicit nonzero integer minor & Joint incidence minor; singleton endpoint refinement & Translated exterior cover & $q=2$\\
Triangular, centered-square, \texttt{t2019} (App.~\ref{SM-sec:open5qubits}) & Lattice-block normal partitions; open5 routes & Explicit nonzero integer minors & Joint incidence minors; singleton endpoint refinement & Disc covers; residue classes and systole bound & $q=2$\\
King (App.~\ref{SM-sec:kingrows}) & Explicit six-row band routes; both cycle partitions & Injective grouped row sites; internal cycle theorem & Simultaneous displacement-sector row witness; endpoint refinement & $L_x\ge4$, $L_y\ge18$, or exchanged & All $q$\\
\bottomrule
\end{tabular}
\end{adjustbox}
\end{table*}

The snub-square open templates and verification records are retained in
\begin{verbatim}
artifacts/m27_global_generic_20260910/
\end{verbatim}
with source and certificate hashes. The square, kagome, and ruby
compressed templates are in
\begin{verbatim}
experiments/certificates/m28/
\end{verbatim}
and the elongated-triangular and maple-leaf templates are in
\begin{verbatim}
experiments/certificates/m29/
\end{verbatim}
The source modules for independent routing verification are
\begin{verbatim}
experiments.m27_verify_certificates
experiments.m28_verify
experiments.m29_verify
\end{verbatim}
From the repository root, set \texttt{PYTHONPATH=python:.} and invoke
\texttt{python -m} followed by the module name. The m28 and m29
verifiers read frozen templates without an integer-programming search.
The splitting constructions are implemented separately in
\texttt{experiments.m28\_ruby\_split} and
\texttt{experiments.m29\_split}; their analytic all-$q$ arguments are
included in the preceding appendix.
The original numerical experiments, additional finite-graph checks,
and their reports remain separate from these proof certificates.

The dice and three 2-uniform open-template certificates are stored in
\begin{verbatim}
research/generic_39_review/
\end{verbatim}
as \nolinkurl{dice_open.json}, \nolinkurl{t2016_open.json},
\nolinkurl{t2017_open.json}, and \nolinkurl{t2018_open.json}.
The read-only checker \nolinkurl{verify_certificates.py} in that
directory does not import the partition search or maximum-flow code.
It verifies 346 open-region certificates and 98\,719 paths, including
all 74 edge types and 30 vertex types, the inclusions in
Eq.~\eqref{SM-eq:extensioncover}, and the absence of boundary-leg
identifications. Dice adjacency is also checked against an independent
hub/rim construction. The saved \nolinkurl{verification.json} records
certificate hashes and checks on $30\times30$ and $31\times33$ tori;
Appendix~\ref{SM-sec:generic39proof} supplies the all-period covering proof.
From the repository root, run the checker with arguments
\texttt{dice t2016 t2017 t2018} and \texttt{PYTHONPATH=.:python}.

The same directory contains the triangle extensions as
\nolinkurl{t2001_blocking.json}, \nolinkurl{t2005_blocking.json},
and \nolinkurl{shuriken_blocking.json}, together with each graph's
\nolinkurl{*_blocked_open.json}. The independent read-only checker
\nolinkurl{verify_blocked.py} imports neither the blocking constructor
nor the partition search or maximum-flow code. It verifies the exact
original-site partition and adjacency, the ten double-bond isolation
conditions, all 56 coarse-edge and 28 coarse-vertex types, and 283
open regions with 75\,937 routing paths. It also checks the cover
in Eq.~\eqref{SM-eq:trianglecover} and boundary identifications on
$32\times32$ and $33\times35$ blocked tori, corresponding to
$64\times64$ and $66\times70$ original periods for shuriken.
Hashes and results are recorded in \nolinkurl{blocked_verification.json}.
Run this checker with arguments \texttt{t2001 t2005 shuriken} and
\texttt{PYTHONPATH=.:python}. Appendix~\ref{SM-sec:triangleextension}
proves the extension to all stated periods.

The separate \nolinkurl{audit_triangle_lift.py} checks the splitting
detector at $q=2,3$, obtaining ranks $9,64$ over two prime fields,
and compares two independent contraction formulas. At $q=2$, the
internal triangle has 96 parameters: its $512\times96$ Jacobian has
rank 85 and its gauge map has rank 11 over both fields, with
$J\Phi=0$ and finite gauge invariance checked exactly. These are
algebraic cross-checks; the all-$q$ statements follow from the analytic
splitting and fiber arguments, not extrapolation from small $q$.

The five further extensions are stored in
\begin{verbatim}
research/open19_extension/
\end{verbatim}
Each graph has a \nolinkurl{*_blocking.json} record and a
\nolinkurl{*_blocked_open.json} routing certificate. The independent
checker \nolinkurl{verify.py} reads these frozen files without importing
the blocking constructor, search, or maximum-flow code. It verifies
134 open regions and 29\,392 paths, covering all 27 coarse-edge and
12 coarse-vertex types, six double-bond isolation conditions, and
two triple-bond types with the required dimer incidence and endpoint
refinements. Identical coarse labels in neighboring cells are checked
as distinct block incidences. Checkerboard adjacency is additionally
checked directly by intersection of square-source edge endpoints.
The report \nolinkurl{blocked_verification.json} contains hashes and
checks on $32\times32$ and $33\times35$ blocked tori; the analytic
cover argument in Appendix~\ref{SM-sec:fiveextension} proves all stated
periods. Invoke the checker from the repository root with
\texttt{PYTHONPATH=.:python} and arguments
\texttt{t2002 t2003 t2004 t2011 checkerboard}.

The local \nolinkurl{audit.py} checks named contractions against
independent formulas over two prime fields. At $q=2$, the mixed-degree
triangle has 80 parameters, Jacobian rank 69, and gauge rank 11.
For the $K_4$ block it checks full-column site maps, gauge rank 21,
the differential gauge identity, and finite gauge invariance, without
computing the full block Jacobian rank. Its two-site detector has
direct rank 9 at $q=2$; at $q=3$, rank 64 follows from the checked
product factorization, with selected columns independently contracted.
The dimer detector is similarly checked at ranks 90 and 1280.
Finally, \nolinkurl{check_polynomials.py} compares the exact
support-census coefficients with Eq.~\eqref{SM-eq:fiverankdrops}.
These audits support the all-$q$ analytic proofs and do not replace them.

The bowtie blockings and open routes are stored in
\begin{verbatim}
research/open14_extension/
\end{verbatim}
as \nolinkurl{t2006_blocking.json}, \nolinkurl{t2007_blocking.json},
and their \nolinkurl{*_blocked_open.json} files. The independent
\nolinkurl{verify.py} checks the original adjacency, the five-site
bowties and their leaf pairs, all three double-bond detecting endpoints,
and all five coarse-edge and two coarse-vertex types. Its 25 open-region
certificates contain 3368 paths. It also checks the cover inclusions
and boundary identifications on $32\times32$ and $33\times35$ tori;
Appendix~\ref{SM-sec:bowtieextension} supplies the all-period proof.
Run this checker with arguments \texttt{t2006 t2007} and
\texttt{PYTHONPATH=.:python} from the repository root.

The exact local audit \nolinkurl{audit_bowtie.py} uses named tensor
indices and two prime fields. At $q=2$, the full bowtie Jacobian is
$8192\times160$ and has rank 138, while the gauge map has rank 22.
It checks all six center flattenings, $J\Phi=0$, and finite gauge
invariance. Adding the non-single-leg generators of every incident
external double bond raises the rank to 156 for \texttt{t2006} and
174 for \texttt{t2007}, accounting for nine independent directions
per incidence. These local audits agree with, but do not replace,
the all-$q$ splitting and internal normal-partition proofs.

The diamond blockings and open routes are stored in
\begin{verbatim}
research/open12_extension/
\end{verbatim}
as \nolinkurl{t2009_blocking.json}, \nolinkurl{t2012_blocking.json},
and their \nolinkurl{*_blocked_open.json} files. The independent
\nolinkurl{verify.py}, run with arguments \texttt{t2009 t2012}
and \texttt{PYTHONPATH=.:python}, checks original adjacency, all
four diamonds, the joint-double and isolated-triple incidence
conditions, 42 regions and 9934 paths. It verifies the eight
coarse-edge and four coarse-vertex types, the cover inclusions,
and boundary identifications on $32\times32$ and $33\times35$ tori.
Appendix~\ref{SM-sec:diamondextension} proves the all-period extension.

The named-index \nolinkurl{audit.py} checks the complete local
$4096\times192$ Jacobian at $q=2$: its rank is 174 and the gauge
rank is 18 over both primes 65521 and 65519. It verifies
$J\Phi=0$, finite gauge invariance, and all site-map ranks.
Adding every incident non-single-site coarse generator raises the
rank to 210 for each \texttt{t2009} block and 264 for the
\texttt{t2012} block. It also contracts the three-site witness
\eqref{SM-eq:diamondtriplewitness} and checks detector ranks 45 and
640 at $q=2,3$. The separate \nolinkurl{audit_joint.py} contracts
\eqref{SM-eq:diamonddoublewitness}, checks both tip quotients and the
disjoint spectator images, and obtains joint ranks 18 and 128.
Both audits use two prime fields and compare named contractions
with independent product formulas. The exact support-census
comparison in \nolinkurl{check_polynomials.py} verifies
Eq.~\eqref{SM-eq:diamondrankdrops}. These finite checks supplement
the analytic witnesses for every $q\ge2$.

The wheel and higher-degree diamond records are in
\begin{verbatim}
research/open10_extension/
\end{verbatim}
The independent \nolinkurl{verify_wheels.py} checks the original
adjacency, disjoint closed-neighborhood wheels, rim-dimer incidence,
and simple coarse quotients for \texttt{t2010} and \texttt{t2013}.
It checks $3\times3$, $3\times4$, and $5\times5$ tori, together
with the offset bounds proving no aliasing for every period at least
three. Wheel injectivity has the explicit copy-tensor witness in
Appendix~\ref{SM-sec:wheelextension}. The separate \nolinkurl{verify.py},
with argument \texttt{t2020}, checks its three diamond blocks, all
six triple-edge and three coarse-vertex types, both tip partitions,
and the single-edge endpoint refinement. Its 31 open regions contain
11\,976 paths, with cover inclusions and boundary checks on
$32\times32$ and $33\times35$ tori. All commands below use
\texttt{PYTHONPATH=.:python} from the repository root.

The named-index \nolinkurl{audit_wheel.py} computes the complete
$8192\times320$ local Jacobian at $q=2$. Its rank is 278 and the
gauge rank is 42 over both primes 65521 and 65519; adding the three
incident external double-bond actions raises the rank to 305.
It verifies all three dimer inverses, center flattenings, $J\Phi=0$,
and finite gauge invariance. The three-slice center rigidity witness
has ranks 31 and 161 at $q=2,3$. The separate
\nolinkurl{audit_dimer.py} contracts both channels of
Eq.~\eqref{SM-eq:wheeljointwitness}, checks that the opposite action
vanishes in each channel, and obtains joint ranks 18 and 128.

For each \texttt{t2020} diamond, \nolinkurl{audit_diamond.py}
checks an exact $2048\times384$ row submatrix of the full
$65536\times384$ Jacobian at $q=2$. Its rank is 366 and the gauge
rank is 18; adjoining all four incident non-single-site triple-bond
actions gives rank 546. These nonzero minors prove rank lower bounds,
with matching upper bounds supplied by gauge inclusion and the
number of added columns. Both fields, the differential gauge identity,
and finite gauge invariance are checked. The named contractions in
\nolinkurl{audit_region.py} verify Eq.~\eqref{SM-eq:t2020regionwitness}
and the left/right scalar-stabilizer code. The joint rank is checked
directly as 90 at $q=2$; at $q=3$, rank 1280 follows from the checked
factor ranks, disjoint images, and product identities. Finally,
\nolinkurl{check_polynomials.py} verifies
Eqs.~\eqref{SM-eq:wheelrankdrops} and \eqref{SM-eq:t2020rankdrop}
against the exact support census. These audits supplement the
all-$q$ proofs rather than extrapolating from small dimensions.

The qubit results and higher-dimensional reductions of
Appendix~\ref{SM-sec:open4qubits} are frozen in
\begin{verbatim}
research/open4_twouniform/
\end{verbatim}
The independent \nolinkurl{verify_injective.py} checks the twelve
original-leg paths of the \texttt{t2008} block and the offset bounds
for all periods at least three. The independent \nolinkurl{verify.py},
with arguments \texttt{t2014 t2015}, checks 22 open regions and
4704 paths, their exterior covers, and the boundary identifications
on $32\times32$ and $33\times35$ tori. Neither checker imports a
partition search or a flow solver. \nolinkurl{verify_geometry.py}
checks all four accepted blockings against the original integer graph,
the incidence partitions, and their endpoint refinements.

For each of \texttt{t2008}, \texttt{t2014}, and \texttt{t2019},
the \nolinkurl{*_q2_witness.json} file records every integer tensor
entry and sampled output row; \nolinkurl{*_q2_minor.json} records
an explicit square minor and both modular determinants.
The generator \nolinkurl{audit_block.py} also verifies the full
two-site differential gauge identities, finite gauge invariance,
and sampled Jacobian directional derivatives. The independent
\nolinkurl{verify_minor.py}, with arguments
\texttt{t2008 t2014 t2019}, reconstructs each minor by replacing
original tensors and applying named external-leg operators before
contraction. It reproduces the determinants without the generator's
flat-state bit-index implementation. Exact modular linear algebra
uses \texttt{python-flint}; the remaining dependencies and commands
are listed in the directory's \nolinkurl{README.md}.
The polynomial checker verifies the two qubit global rank counts,
the WGS comparisons, and the counts of columns in the unresolved
higher-dimensional tests. These minors prove the stated qubit claims;
they are not used as evidence for an unproved all-$q$ rank identity.

The certificates of Appendix~\ref{SM-sec:open5qubits} are frozen in
\begin{verbatim}
research/open5_remaining/
\end{verbatim}
For each lattice blocking, \nolinkurl{*_blocking.json} records the block,
its translation lattice, and all regrouped original edges, and
\nolinkurl{*_templates.json} records the cover disc, the three edge
templates, and the vertex template with every route. The independent
\nolinkurl{verify_templates.py} imports no search or flow code. It
rebuilds the four adjacencies, checks that the block tiles the sites,
checks every route and the separation of each edge template, verifies
that every disc $O$ contains the required cover neighborhoods, and
reports the bound~\eqref{SM-eq:open5systole}. For the four lattice blocks
and the \texttt{t2015} path, \nolinkurl{audit_big.py} generates the
qubit witnesses and minors, with ranks computed by an exact blocked
row echelon over $\mathrm{GF}(p)$ (\nolinkurl{rank_modp.c}). The
independent \nolinkurl{verify_minor_big.py} rebuilds every column of
each frozen minor from the witness tensors, as described in
Appendix~\ref{SM-sec:open5qubits}, and recomputes its rank over both fields. Commands are listed in the
directory's \nolinkurl{README.md}.

The king proof and its verification are in
\nolinkurl{research/king_rows/}. The NumPy-only script
\nolinkurl{verify.py} checks all band widths and translations on
$(4,18),(5,19),(7,23),(18,18)$ tori, including both coarse partition
hypotheses and periodic seams. It checks $2240$ explicit contraction
identities per prime, including actual king external-group dimensions
$Q=8,27$, and polynomial-space intersections. At $q=2,Q=3,L=4$, it
checks all $81$ displacement sectors of the analytic row lemma over
both $65521$ and $65519$, obtaining total row-Jacobian rank $273$
and joint rank $13329$. The all-$q$ statement uses the analytic proof
in Appendix~\ref{SM-sec:kingrows}, not extrapolation of these finite ranks.

\section{Projected ground-state response at the graph-state point}
\label{SM-sec:groundresponse}
At the qubit graph-state point, the lift of Eq.~\eqref{eq:lift} also computes a static response.
Write $\ket{e_S}=Z_S\ket{\psi_G}$ and take the stabilizer parent
$H_0=-\sum_vK_v$, where $K_v=X_v\prod_{w\in N(v)}Z_w$, including
all boundary stabilizers. Its unique ground state is $\psi_G$, with
$E_0=-N$, and $(H_0-E_0)e_S=2|S|e_S$ \cite{hein2006graph}.

\begin{corollary}[Projected ground-state response]
\label{SM-cor:groundresponse}
For $H(\lambda)=H_0+\lambda V$ with Hermitian $V$, put
$v_S=\langle e_S|V|\psi_G\rangle$. Let $\psi_0(\lambda)$ be its normalized
nondegenerate ground state for sufficiently small $\lambda$, with
$\psi_0(0)=\psi_G$, and let $A_0$ be the graph-state representation.
In the parallel phase convention,
\begin{equation}
\begin{gathered}
\ket{\psi'_0(0)}=-\sum_{S\ne\varnothing}\frac{v_S}{2|S|}\ket{e_S},\\
c_S^{(1)}=-\frac{v_S}{2|S|},\qquad S\in\K(G)\setminus\{\varnothing\}.
\label{SM-eq:groundresponse}
\end{gathered}
\end{equation}
The lift $\Xi c^{(1)}$ of Eq.~\eqref{eq:lift} realizes the
orthogonal projection of this response and minimizes the energy
through second order among smooth tensor paths through $A_0$. The
normalized contraction $\psi_{\rm lift}(\lambda)$ of
$A_0+\lambda \Xi c^{(1)}$ obeys
\begin{equation}
1-|\langle\psi_0(\lambda)|\psi_{\rm lift}(\lambda)\rangle|^2
=\lambda^2\sum_{S\notin\K(G)}\frac{|v_S|^2}{4|S|^2}+O(\lambda^3).
\label{SM-eq:groundresponse_error}
\end{equation}
If the entire first-order response is accessible, this error is
$O(\lambda^4)$.
\end{corollary}
\begin{proof}
Nondegenerate perturbation theory in the stabilizer eigenbasis gives
the first identity. For a normalized path with horizontal derivative
$\sum_S c_Se_S$, the quadratic energy coefficient is
$\sum_S[2|S||c_S|^2+2\operatorname{Re}(c_S^*v_S)]$, with sums over
allowed nonempty supports. Its minimizer is $c^{(1)}$.
The lift realizes this derivative; the squared norm of its omitted
orthogonal component gives Eq.~\eqref{SM-eq:groundresponse_error}.
When that component vanishes, the states differ only at second order
up to phase, giving fourth-order infidelity.
\end{proof}
The equivalence of energy minimization and orthogonal projection uses
that $H_0-E_0$ preserves $\T_{A_0}$, which an arbitrary parent
Hamiltonian need not do. The inverse excitation energies distinguish
this static response error from the unweighted dynamical residual of
Eq.~\eqref{eq:residual}: both are sums over the same
unsupported Fourier coefficients, with different weights.

\section{Nearest-neighbor \texorpdfstring{$XX$ and $YY$}{XX and YY} supports on the eleven tilings}
\label{SM-sec:yysupport}
Equation~\eqref{eq:xxsupport} writes
$X_uX_v\psi=-\chi_{B_{uv}}\psi$ at a graph state, with
$B_{uv}=N(u)\mathbin{\triangle}N(v)$ in open neighborhoods.
\begin{lemma}
\label{SM-lem:xxcenter}
On each of the eleven bulk Archimedean tilings, and on locally faithful
tori, no vertex is a center of $B_{uv}$ for any edge $uv$; hence
$B_{uv}\notin\K$ and every nearest-neighbor $XX$ term is orthogonal to
the physical tangent space at the graph state.
\end{lemma}
\begin{proof}
Both endpoints lie in $B_{uv}$, each endpoint has a neighbor exclusive
to it, and $|B_{uv}|\ge z+1$ for coordination number $z$. An endpoint
is therefore not a center, because it misses the other endpoint's
exclusive neighbors, and every other vertex of $B_{uv}$ misses one
endpoint. A center outside $B_{uv}$ would need $z+1$ neighbors and has
only $z$.
\end{proof}

Equation~\eqref{eq:yysupport} writes $Y_uY_v\psi=\chi_{B'_{uv}}\psi$
at the same point, with $B'_{uv}=B_{uv}\setminus\{u,v\}$ the set of
neighbors exclusive to one endpoint of the edge $uv$.
\begin{lemma}
\label{SM-lem:yycenter}
On each of the eleven bulk Archimedean tilings, and on locally faithful
tori, no vertex is a center of $B'_{uv}$ for any edge $uv$; hence
$B'_{uv}\notin\K$ and every nearest-neighbor $YY$ term is orthogonal
to the physical tangent space at the graph state.
\end{lemma}
\begin{proof}
A center of $B'_{uv}$ cannot be an endpoint, since each endpoint has an
exclusive neighbor. It cannot be an exclusive neighbor $w$ of $u$,
since $w$ would then share with $v$ the common neighbors $u$ and every
exclusive neighbor of $v$, so by Lemma~\ref{lem:shortcycles}(i) $v$ would have a single exclusive
neighbor $v_1$, so that $uv$ would border $z-2\ge1$ triangles and the
square face $uwv_1v$: an edge borders two faces, and $3.4.x$ is not a
vertex configuration of degree three. It cannot be a common neighbor,
which would force two triangles at a degree-three vertex or three at a
degree-four vertex. The remaining possibility, a vertex at distance two
from both endpoints adjacent to every exclusive neighbor, is excluded
on each of the eleven tilings by inspection of its cell, checked for
every edge orbit in
\nolinkurl{tests/unit/test_nongauge_integrability.py}.
\end{proof}

\section{Twelve-qubit trajectory: construction and numerical checks}
\label{SM-sec:trajectorydetails}
Figure~\ref{fig:dynamics} shows the blockade and its repair as actual
time traces. The graph is the periodic kagome lattice with $2\times2$
honeycomb source cells, twelve qubits and $24$ unit-coupling $XX$
bonds, whose $D=2$ graph-state representation has projective rank
$248$ with every $B_{uv}$ unsupported, so that
$\Var(H_{XX})=48$ and $\epsilon_{\rm TDVP}=1$. Integrating the exact
projected equation for all independent tensor entries
\cite{sorella2001generalized,tvmc1} with fourth-order Runge--Kutta
steps~\cite{hairer1993ode} and no regularization leaves the state
stationary to numerical precision, while the exact survival
probability at $t=0.3$ is $0.0315$ and the mean stabilizer is
$-0.653$. The family
of Eq.~\eqref{eq:gatefamily},
the all-edge version of the single-gate augmentation used in Appendix~\ref{SM-sec:mbqcprotocol}, has derivatives $-iX_uX_v\psi_G$ at
$\boldsymbol\alpha=0$ on every edge; each gate has operator Schmidt
rank two and can be absorbed into the bond tensors at bond dimension
at most four~\cite{verstraete2004valence,cirac2021matrix}. The
projected equation in these coordinates recovers the exact trajectory
to $2\times10^{-14}$ against an independent sparse matrix-exponential
integration~\cite{almohy2011exponential,virtanen2020scipy}. Because
the $XX$ terms commute, the exact trajectory lies in this family, so
the identical projected equation is exactly stationary or exactly
correct depending only on whether the tangent image contains
$-iX_uX_v\psi$.

The tensor displacement in the unaugmented evolution stays below
$4\times10^{-15}$. The plotted step is $0.025$; steps $0.05$ and
$0.1$ give the same physical results. In the augmented family, the
phase-aligned vector error stays below $2\times10^{-14}$.
All expectation values use exact Hilbert-space statistics.

\section{Sampled metric and force in gauge-only projection}
\label{SM-sec:sampledforce}
In the setting of Sec.~\ref{sec:projectiongap}, finite
sampling alone does not produce an incompatible force:
forming metric and force from the same centered derivative samples,
$\widehat{\mathcal M}=\mathcal O^\dagger\mathcal O/N_s$ and
$\widehat{\boldsymbol F}=\mathcal O^\dagger\mathcal h/N_s$
[Eq.~(8) of Ref.~\cite{wu2025peps}], gives
$\widehat{\boldsymbol F}\perp\ker\widehat{\mathcal M}$ in exact arithmetic, as also checked
with $64$ samples. Incompatibility requires separate approximations or
floating-point error. 
\section{A state-level \texorpdfstring{$XX$}{XX} obstruction}
\label{SM-sec:schmidtobstruction}
The tangent-image theorem concerns a specified tensor representation.
This section proves Corollary~\ref{cor:honeycombstate}
directly from the state. Throughout, $\ket G$ is the standard qubit
graph state with edge phase $\pi$, and the comparison PEPS use the
same fixed graph and assignment of physical sites. The exact Schmidt
spectrum below is not asserted for arbitrary WGS phases.

\begin{lemma}[Matching-cut obstruction]
\label{SM-lem:matchingcut}
Let $R|\bar R$ be a bipartition of a finite simple graph whose $c$
crossing edges form a matching. Suppose $uv$ is a crossing edge,
with $u\in R$, and each endpoint has a same-side neighbor incident to
no crossing edge. Put $r=2^c$. The Schmidt coefficients of
$\ket{\psi_\varepsilon}=e^{-i\varepsilon X_uX_v}\ket G$
are $|\cos\varepsilon|/\sqrt r$ and $|\sin\varepsilon|/\sqrt r$,
each with multiplicity $r$, omitting a group when its prefactor
vanishes. Its Schmidt rank is $2^{c+1}$ whenever
$\sin\varepsilon\cos\varepsilon\ne0$, exceeding the bound $2^c$
for any PEPS with all bond dimensions at most two on $G$.
For the normalized set $\mathcal P_2(G)$ of these PEPS states,
\begin{equation}
\max_{\phi\in\mathcal P_2(G)}
|\langle\phi|\psi_\varepsilon\rangle|^2
=\max\{\cos^2\varepsilon,\sin^2\varepsilon\}.
\label{SM-eq:matchingcutfidelity}
\end{equation}
\end{lemma}
\begin{proof}
Let $U_R$ and $U_{\bar R}$ be the products of controlled-$Z$ gates on
edges internal to their respective sides. Their product $U$ is a
bipartition-local unitary and satisfies $U^2=I$. Applying it to
$\ket G$ removes all internal graph-state gates, leaving $c$ disjoint
two-qubit graph states on the crossing edges, and $\ket+$ on every
vertex incident to no crossing edge. Each pair is maximally
entangled, so the $r$ Schmidt coefficients are all $1/\sqrt r$.
Denote the two Schmidt support spaces by $\mathcal S_R$ and
$\mathcal S_{\bar R}$.

Conjugation of the generator gives
\begin{equation}
\begin{aligned}
U X_uX_v U&=P_R\otimes P_{\bar R},\\
P_R&=X_u\prod_{w\in N(u)\cap R}Z_w,\\
P_{\bar R}&=X_v\prod_{w\in N(v)\cap\bar R}Z_w.
\end{aligned}
\label{SM-eq:matchingcutpauli}
\end{equation}
Choose free same-side neighbors $a\in R$, $b\in\bar R$ as in the
hypothesis. The factor $Z_a$ changes the fixed $\ket+_a$ to
$\ket-_a$, so $P_R\mathcal S_R\perp\mathcal S_R$.
Similarly $P_{\bar R}\mathcal S_{\bar R}\perp\mathcal S_{\bar R}$.
Both Pauli factors are unitary. The coefficient matrix of the
evolved state is therefore the direct sum of two maximally entangled
blocks weighted by $\cos\varepsilon$ and $-i\sin\varepsilon$.
This proves the spectrum without assuming a smooth tensor lift.

Cutting $c$ virtual indices expresses any bond-two PEPS as a sum of
at most $2^c$ product vectors across this bipartition, proving the
rank bound for every representation. The maximum squared overlap
of a normalized rank-at-most-$r$ state with a target is the sum of
the target's largest $r$ squared Schmidt coefficients, by Schmidt
truncation. Here that sum is the right-hand side of
Eq.~\eqref{SM-eq:matchingcutfidelity}. It is attained by $\ket G$ or
$X_uX_v\ket G$, both belonging to $\mathcal P_2(G)$.
\end{proof}

\subsection{A local cut for every honeycomb edge}
Use the unwrapped coordinates of Appendix~\ref{SM-sec:coordinates}.
For the edge $A_{0,0}B_{0,0}$ choose the ten-vertex region $R_0$
with
\begin{equation}
\begin{aligned}
R_0\cap A={}&\{A_{-2,1},A_{-1,0},A_{-1,1},\\
             &\hspace{8mm}A_{0,-1},A_{0,0}\},\\
R_0\cap B={}&\{B_{-2,0},B_{-2,1},B_{-1,-1},\\
             &\hspace{8mm}B_{-1,0},B_{0,-1}\}.
\end{aligned}
\label{SM-eq:honeycombcutregion}
\end{equation}
These vertices form two adjacent hexagons sharing
$A_{-1,0}B_{-1,0}$. Its endpoints have internal degree three; the
other eight vertices have internal degree two
(Fig.~\ref{SM-fig:honeycombcut}). The complete cut is
\begin{center}
\begin{tabular}{cc}
\toprule
Inside endpoint & Outside endpoint\\
\midrule
$B_{-2,0}$ & $A_{-2,0}$\\
$A_{-2,1}$ & $B_{-3,1}$\\
$B_{-2,1}$ & $A_{-2,2}$\\
$B_{-1,-1}$ & $A_{-1,-1}$\\
$A_{-1,1}$ & $B_{-1,1}$\\
$A_{0,-1}$ & $B_{0,-2}$\\
$B_{0,-1}$ & $A_{1,-1}$\\
$A_{0,0}$ & $B_{0,0}$\\
\bottomrule
\end{tabular}
\end{center}
All sixteen endpoints are distinct, so the cut is a matching.
The target's inside free neighbor is $B_{-1,0}$ and one outside
free neighbor is $A_{1,0}$.

\begin{figure}[!tbp]
\centering
\includegraphics[width=\columnwidth]{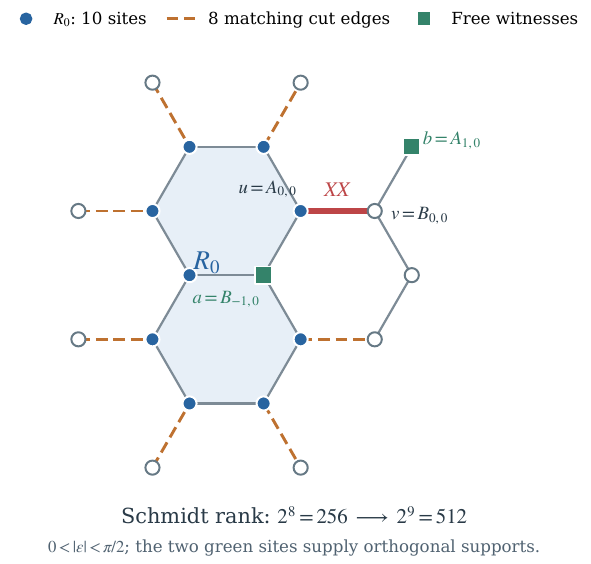}
\caption{A matching cut for the target honeycomb edge $uv$.
The blue region $R_0$ contains two adjacent hexagons. Its cut consists
of the seven orange dashed edges and the red target edge, with
sixteen distinct endpoints. The green sites $a=B_{-1,0}$ and
$b=A_{1,0}$ are same-side neighbors of $u$ and $v$, respectively,
and neither is incident on a cut edge. Their orthogonal flips create
the second Schmidt block. Only selected edges wholly outside $R_0$
are shown; every edge incident on $R_0$ is included.}
\label{SM-fig:honeycombcut}
\end{figure}

Writing $A_{xy}=(x,y,0)$ and $B_{xy}=(x,y,1)$, the map
\begin{equation}
\tau(x,y,s)=(-x-y-s,x,s)
\label{SM-eq:honeycombcutrotation}
\end{equation}
is an order-three automorphism of the infinite honeycomb graph.
It fixes $A_{0,0}$ and cycles its three neighboring $B$ vertices.
Thus $R_k=\tau^kR_0$, $k=0,1,2$, supplies all three edge directions.
The map acts on the finite coordinate templates before periodic
reduction; no rotational symmetry of a rectangular torus is assumed.

Set $v_k=\tau^kB_{0,0}$ and
$\mathcal Q_k=R_k\cup N(R_k)\cup N(v_k)$.
Here $N(R)=\bigcup_{w\in R}N(w)$ uses open neighborhoods.
These sets contain every vertex needed to determine the complete
cuts and their free neighbors. Their coordinate bounds are
\begin{center}
\begin{tabular}{cccc}
\toprule
$k$ & $x$ interval & $y$ interval & $|\mathcal Q_k|$\\
\midrule
0 & $[-3,1]$ & $[-2,2]$ & 20\\
1 & $[-1,2]$ & $[-3,1]$ & 20\\
2 & $[-2,2]$ & $[-1,2]$ & 20\\
\bottomrule
\end{tabular}
\end{center}
Each coordinate diameter is at most four. Reduction modulo any
$L_x,L_y\ge5$ is therefore injective on each $\mathcal Q_k$.
It cannot merge cut endpoints, create extra edges incident on
$R_k$, or identify an outside free neighbor with a cut endpoint.
Translating the three templates covers every nearest-neighbor edge.
Lemma~\ref{SM-lem:matchingcut} with $c=8$ proves the honeycomb
state-level exclusion and exact fidelity in the main text.

\subsection{A cylinder cut on the square torus}
\label{SM-sec:squarecut}
Lemma~\ref{SM-lem:matchingcut} does not require the region $R$ to be
finite in the infinite lattice or to be injective; a bipartition of
the torus into two cylinders suffices. On an $L_x\times L_y$ square
torus with $L_x\ge6$, take a horizontal edge $uv$ with
$u=(x_0,y_0)$, $v=(x_0+1,y_0)$, and let $R$ be the
$m=\lfloor L_x/2\rfloor$ columns $x_0-m+1,\ldots,x_0$. The crossing
edges are the $L_y$ horizontal edges between columns $x_0$ and $x_0+1$
and the $L_y$ horizontal edges between columns $x_0-m+1$ and $x_0-m$.
Because $m\ge3$ and $L_x-m\ge3$, these four columns are pairwise
distinct, so each vertex meets at most one crossing edge and the cut
is a matching containing $uv$. Vertical edges never cross. The
same-side neighbors $a=(x_0-1,y_0)\in R$ and $b=(x_0+2,y_0)\notin R$
lie in columns $x_0-1$ and $x_0+2$, which carry no crossing edge
because $m\ge3$ and $L_x-m\ge3$ respectively. Lemma~\ref{SM-lem:matchingcut}
with $c=2L_y$ gives Schmidt coefficients $|\cos\varepsilon|/2^{L_y}$
and $|\sin\varepsilon|/2^{L_y}$, each of multiplicity $2^{2L_y}$,
rank $2^{2L_y+1}$ for $\sin\varepsilon\cos\varepsilon\ne0$, and the
fidelity~\eqref{SM-eq:matchingcutfidelity} against every bond-two PEPS
on the torus. Vertical edges use rows in place of columns and need
$L_y\ge6$; hence $L_x,L_y\ge6$ covers every edge, proving Corollary~\ref{cor:squarestate}.

A dense check on the $6\times3$ torus ($18$ qubits, horizontal target
edge, $R$ the three columns ending at $u$) gives rank $64$ at
$\varepsilon=0$ and $\pi/2$ and rank $128$ at
$\varepsilon=0.1$, $\pi/4$, $1$, with the two predicted coefficient
values $\cos\varepsilon/8$ and $\sin\varepsilon/8$ at numerical
threshold $10^{-10}$; the script is
\nolinkurl{experiments/xx_schmidt_square_torus.py}.

\subsection{Scope and independent checks}
The Schmidt proof applies to any graph and edge meeting its cut
conditions, independently of normality. The absence of a non-bond
kernel at the honeycomb representation is not used, and square-torus
WGS representations are covered by Appendix~\ref{SM-sec:squarecut}. The cut conditions can fail: on kagome and
on the triangular lattice every edge lies in a triangle, and a
bipartition separating one vertex of a triangle from the other two
gives that vertex two crossing edges, so no nontrivial matching cut
exists on either graph and the lemma does not apply. There the
computed tangent image still excludes a differentiable tensor lift
through the WGS representation, but a state-level exclusion would
need a different argument.

The script \nolinkurl{experiments/xx_schmidt_obstruction.py} records
the three unwrapped coordinate certificates and checks all $2358$
edges of eighteen rectangular tori: $5\le L_x,L_y\le8$, together
with $(L_x,L_y)=(5,11),(11,5)$. The all-period conclusion follows
from the coordinate-diameter argument, not extrapolation of the
finite enumeration.
An independent dense check uses a sixteen-qubit $4\times2$
honeycomb torus and a width-two strip with four matching cut edges.
This torus is outside the bulk assumption, but its explicitly checked
cut satisfies Lemma~\ref{SM-lem:matchingcut}. Direct controlled-$Z$
amplitudes agree with a separately constructed, named-leg bond-two
PEPS contraction to vector error $1.8\times10^{-15}$.
The full Schmidt spectrum agrees with the two predicted blocks to
$2.0\times10^{-15}$ at
$\varepsilon=0,0.03,0.1,\pi/4,\pi/2$, with respective ranks
$16,32,32,32,16$ at numerical threshold $10^{-12}$.
The solvable endpoints, optimal fidelity, and rejection of cuts
missing either hypothesis are also tested. These checks and their
versions and source hashes are stored in
\nolinkurl{manuscript/quantum/data/xx_schmidt_obstruction}; they are separate
from the analytic proof and do not assert Lean formalization.

\section{Phase-dependent \texorpdfstring{$XX$}{XX} residual on the eleven tilings}
\label{SM-sec:xxscan}
\paragraph{Evaluation of the Fourier coefficients.}
For a qubit Pauli term, let $F$ contain its $X$ or $Y$ sites and let
$Z_P$ contain its $Z$ or $Y$ sites. Its amplitude ratio depends only on
$U_P=Z_P\cup\bigcup_{v\in F}N[v]$, so Eq.~\eqref{eq:energywalsh} can be
evaluated by a Walsh transform on the $2^{|U_P|}$ local configurations,
followed by aggregation of equal global supports before applying
Eq.~\eqref{eq:residual}. For an extensive sum of bounded-range terms at
fixed degree, this uses $O(N)$ bounded-size local tables. 
Equation~\eqref{eq:xxphi} gives the residual of a \emph{single}
$X_uX_v$ term as a function of the uniform edge phase $\phi$ on the
honeycomb lattice. With $c=\cos\phi$, the corresponding closed forms on
the square and triangular lattices are
$R^2_{\rm TDVP}=(1-c)^2p_{\rm sq}(c)/64$ and $(1-c)^2p_{\rm tri}(c)/64$ with
\begin{equation}
\begin{split}
p_{\rm sq}&=32c^{12}+128c^{11}+128c^{10}-128c^9\\
&\quad-304c^8-64c^7+208c^6+128c^5\\
&\quad-47c^4-56c^3+14c^2+56c+49,\\
p_{\rm tri}&=-(c^6+4c^5+6c^4-23c^2-52c-48),
\end{split}
\label{SM-eq:xxpoly}
\end{equation}
obtained by exact symbolic Walsh transforms and confirmed numerically
(Appendix~\ref{SM-sec:reproducibility}).
In all three cases
$R^2=O(\phi^4)$ and $\Var_\psi(X_uX_v)=O(\phi^2)$ as $\phi\to0$, so the
relative residual $\epsilon_{\rm TDVP}$ vanishes linearly at the
product state, increases monotonically on $[0,\pi]$ (verified
numerically), and reaches one only at $\phi=\pi$;
at $\phi=\pi/2$ it equals $2\sqrt7/7\approx0.76$, $\sqrt7/3\approx0.88$,
and $8\sqrt{3/255}\approx0.87$ on honeycomb, square, and triangular,
respectively. The graph-state point is thus the extreme case of a
phase-dependent partial blockade rather than an isolated accident.
Figure~\ref{SM-fig:tdvpscan} shows the corresponding scan for the full
Hamiltonian $H_{XX}=\sum_{uv}X_uX_v$ on all eleven Archimedean
tilings, evaluated by the local Walsh transform of
Corollary~\ref{cor:tdvp}: the coefficients of distinct edges on a
common character are summed before squaring, so this extensive
residual is not the sum of single-edge residuals, and interference
between edges makes the relative residual $\epsilon^2_{\rm TDVP}$
smaller than its single-edge value away from $\phi=\pi$ (on honeycomb
at $\phi=\pi/2$, $\epsilon^2\approx0.33$ for $H_{XX}$ against $4/7$ for
one edge). At $\phi=\pi$ every edge support $B_{uv}$ is distinct and
unsupported, so $\epsilon^2_{\rm TDVP}=1$ and $R^2_{\rm TDVP}/N=E/N$ on
every tiling; the Ising Hamiltonian has zero residual throughout. The
curves are smooth in $\phi$ and the blockade is complete only at the
graph-state point, in agreement with the closed forms.

\begin{figure*}[!tbp]
\centering
\includegraphics[width=\textwidth]{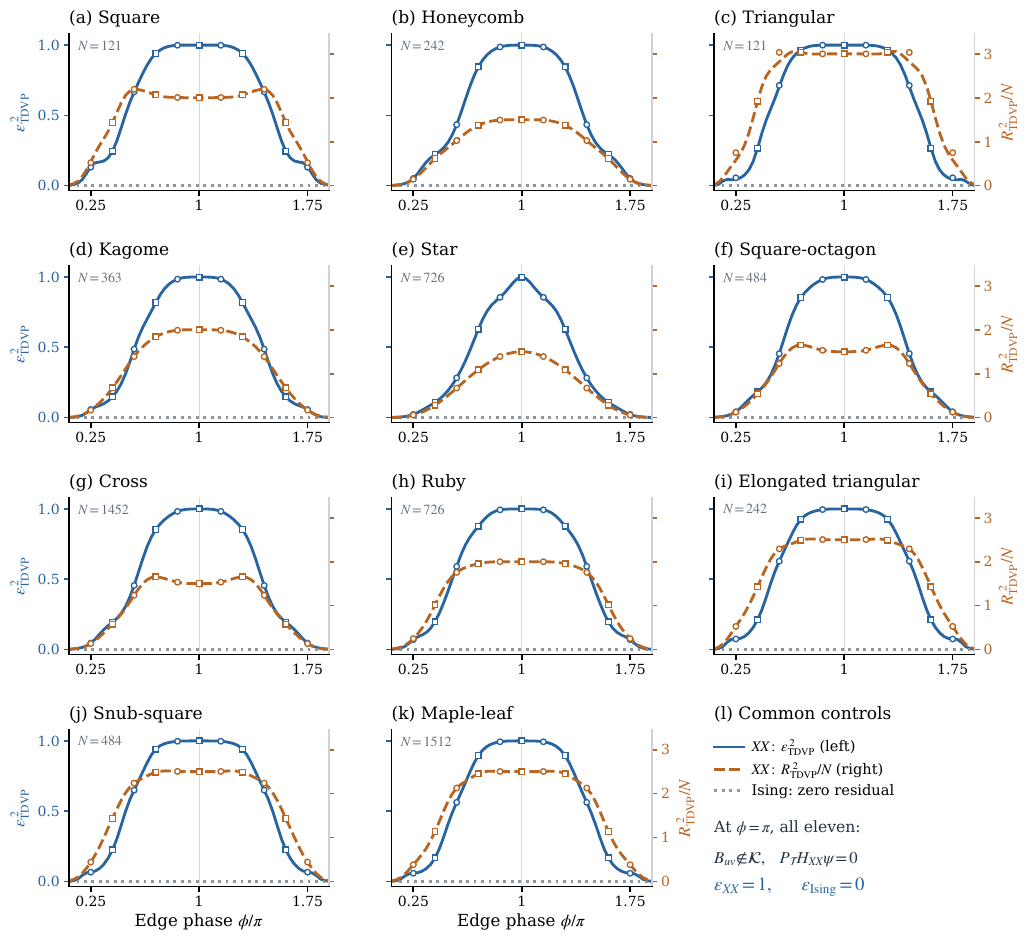}
\caption{Instantaneous TDVP residual on all eleven Archimedean
lattices at $q=2$ and uniform edge phase $\phi$. (a)--(k) For
unit-coupling $H_{XX}$, blue solid curves show the relative residual
squared $\epsilon^2_{\rm TDVP}$ (left axis) and orange dashed curves the
residual squared per site $R^2_{\rm TDVP}/N$ (right axis); axis ranges
are common to all panels. Curves use construction size $11$ with the
vertex count $N$ indicated; circles and squares use sizes $5,8$ for the
first four lattices and $7,9$ for the other seven, with the period
conventions of Appendix~\ref{SM-sec:coordinates}. For maple-leaf,
construction size $s$ means square periods $7(s-5)$ in each coordinate,
so $N=42(s-5)^2$ and the $s=11$ curve has $N=1512$.
Coefficients of distinct edges on a common support are added before
squaring. (l) The Ising Hamiltonian has zero residual throughout; at
$\phi=\pi$ every adjacent $XX$ support lies outside $\K$ on all eleven
lattices, giving zero tangent projection and unit relative residual.
The scans use integer Fourier polynomials checked against direct
local-energy tables (\nolinkurl{experiments/gauge_tdvp.py},
\nolinkurl{manuscript/quantum/data/gauge_tdvp}).}
\label{SM-fig:tdvpscan}
\end{figure*}

\section{A missing observable response and its measurement-based realization}
\label{SM-sec:mbqcprotocol}
\subsection{A missing direction witnessed by an observable}
\label{SM-sec:mbqccase}
A blocked direction is not only a residual: it is a first-order change
of a local observable that no tensor variation can produce. Consider
the six-qubit graph state on a diamond with two leaves,
\begin{equation}
\begin{gathered}
V=\{a,b,c,d,x,y\},\\
E_G=\{ab,ac,ad,bc,bd,ax,by\},\\
\ket{G_\diamond}=\prod_{uv\in E_G}{\rm CZ}_{uv}\ket+^{\otimes6},
\end{gathered}
\label{SM-eq:diamondresource}
\end{equation}
with its nondegenerate $q=2$ WGS representation $A_0$. The two largest
closed neighborhoods are $N[a]=V\setminus\{y\}$ and
$N[b]=V\setminus\{x\}$, so
\begin{equation}
\K(G)=\{S\subseteq V:\{x,y\}\not\subseteq S\}.
\label{SM-eq:diamondcoverage}
\end{equation}
Define the physical curve and a local Pauli-product observable by
\begin{equation}
\begin{gathered}
\ket{G_{\diamond,\epsilon}}=e^{-i\epsilon X_aX_b}\ket{G_\diamond},\\
Q=Z_aZ_bZ_xZ_y,\qquad M=-Y_bZ_cZ_dZ_x.
\end{gathered}
\label{SM-eq:mbqcerror}
\end{equation}

Corollary~\ref{cor:observableresponse} gives the vanishing tensor response
and the exact signal $\sin(2\epsilon)$ for this resource. Here its
missing velocity is $\ket t=iQ\ket{G_\diamond}$, with $Q$ as defined in
Eq.~\eqref{SM-eq:mbqcerror}.
Removing coordinate null directions leaves this obstruction unchanged.
The curve $\ket{G_{\diamond,\epsilon}}$ is produced by a coherent preparation error
on the shared edge of the measurement-based Hadamard primitive
described below~\cite{raussendorf2001oneway,gross2007novel}.
Figure~\ref{SM-fig:kernelimage} contrasts this response with the
triangle's full physical tangent.

\begin{figure}[!tbp]
\centering
\includegraphics[width=\columnwidth]{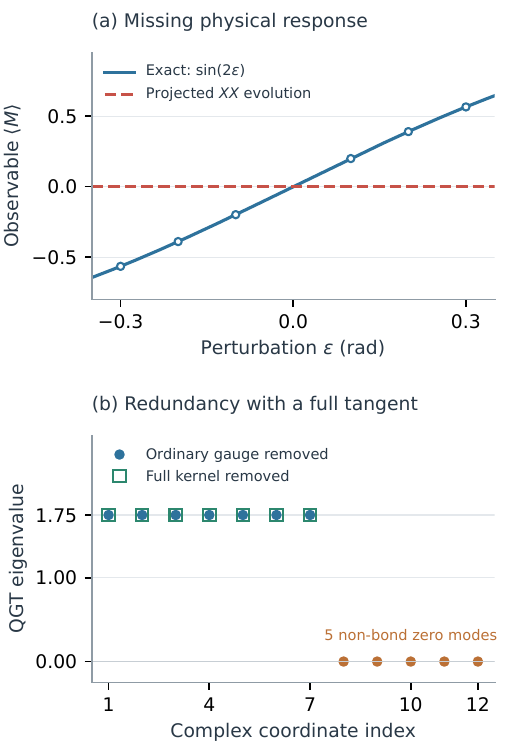}
\caption{Image versus kernel on two small graphs.
(a) Diamond with two leaves: exact response $\langle M\rangle=\sin(2\epsilon)$ of the
observable $M$ to the curve $e^{-i\epsilon X_aX_b}\ket{G_\diamond}$ of Eq.~\eqref{SM-eq:mbqcerror}, against the vanishing response of every
tensor variation at the WGS point (circles: independent dense-state
evaluations).
(b) Triangle: spectrum of the projective QGT after removing ordinary
gauge and the ray, showing seven positive modes and five non-bond zero
modes although the physical tangent is full; complete reduction retains
the seven positive modes (value $7/4$ in the normalization of
Eq.~\eqref{SM-eq:triangleqgt}).}
\label{SM-fig:kernelimage}
\end{figure}

\subsection{Measurement-based realization}
This appendix gives a measurement-based Hadamard primitive on the graph
of Eq.~\eqref{SM-eq:diamondresource}
\cite{raussendorf2001oneway,raussendorf2003cluster}. It realizes the
observable response in Corollary~\ref{cor:observableresponse}; its
calibration and logical-channel formulas require the additional circuit
calculation below. They are not needed for the image or kernel results.
All operations except the specified coherent preparation error are
ideal. This is a theoretical single-primitive example, not a hardware
demonstration or a universal resource by itself.

\begin{figure}[!tbp]
\centering
\includegraphics[width=\columnwidth]{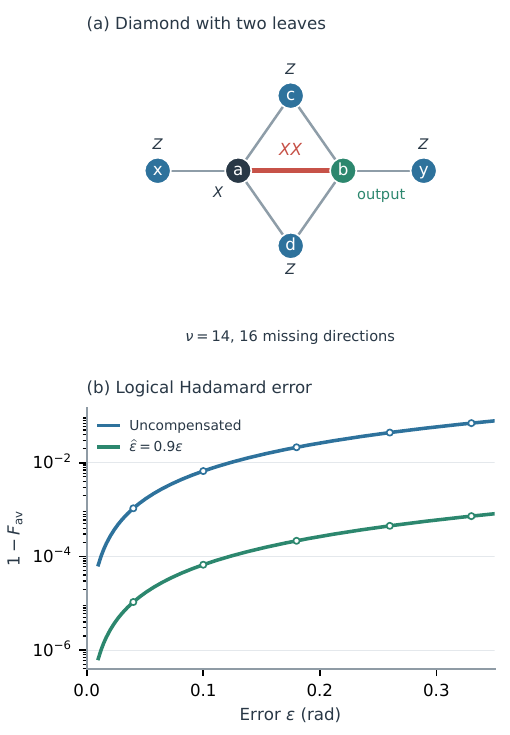}
\caption{Optional measurement-based realization of the observable
response. (a) Six-qubit resource, with the unwanted
$e^{-i\epsilon X_aX_b}$ on the red shared edge after ideal preparation.
Measure the four outer qubits in $Z$, then $a$ in $X$, and retain $b$
as output; the text specifies the Pauli corrections. (b) Average
Hadamard infidelity before compensation and with the illustrative
estimate $\widehat\epsilon=0.9\epsilon$. Lines give
Eq.~\eqref{SM-eq:mbqccompfidelity}; circles independently average the channel
over six input states. All 32 measurement branches are retained,
and all other operations are ideal.}
\label{SM-fig:mbqcprotocol}
\end{figure}

\subsection{Resource counts and finite preparation error}
Write $U_\epsilon=e^{-i\epsilon X_aX_b}$ for the perturbation in
Eq.~\eqref{SM-eq:mbqcerror}.
The shared edge $ab$ belongs to triangles $abc$ and $abd$, with leaves
$x$ on $a$ and $y$ on $b$ [Fig.~\ref{SM-fig:mbqcprotocol}(a)].
Equation~\eqref{SM-eq:diamondcoverage} has $|\K(G)|=48$ supports. The image theorem and the bond-gauge rank formula give
\begin{equation}
\begin{gathered}
P=88,\quad g_0=26,\quad \rank J_{A_0}=48,\\
\dim\T_{A_0}=47,\quad \nu=88-48-26=14.
\end{gathered}
\label{SM-eq:diamondranks}
\end{equation}
The projective parameter quotient by ordinary gauge and the state ray
still has $88-27=61$ coordinates, including 14 null directions.
Independently, the physical tangent misses 16 of the 63 nonempty Fourier
characters: precisely $Z_S\ket{G_\diamond}$ with $\{x,y\}\subseteq S$.
These are different defects. Removing the leaves gives an isolated
diamond with the same $\nu=14$ but a full physical tangent, because a
central vertex then has the entire graph as its closed neighborhood.
The leaves are essential to the present missing-direction example.

The identities $(X_aX_b)^2=Q^2=I$ give the exact evolution
\begin{equation}
\begin{gathered}
\ket{G_{\diamond,\epsilon}}=\cos\epsilon\ket{G_\diamond}+i\sin\epsilon\,Q\ket{G_\diamond},\\
|\langle G_\diamond|G_{\diamond,\epsilon}\rangle|^2=\cos^2\epsilon.
\end{gathered}
\label{SM-eq:diamondfinite}
\end{equation}
\subsection{Measurement protocol and exact output probabilities}
Assume $|\epsilon|<\pi/4$. For computation, replace the initial $\ket+_a$ by an arbitrary input
$\ket\varphi_a$, initialize the other five qubits in $\ket+$, and apply
the same ideal entanglers followed by $U_\epsilon$. First measure
$c,d,x,y$ in the $Z$ basis, obtaining bits
$z=(z_c,z_d,z_x,z_y)$. On the remaining qubits apply
$Z_a^{r_a}Z_b^{r_b}$, where
\begin{equation}
r_a=z_c+z_d+z_x,\qquad r_b=z_c+z_d+z_y\pmod2.
\end{equation}
Then measure $a$ in the $X$ basis with eigenvalue $(-1)^s$ and apply
$X_b^s$. The ideal logical gate is the Hadamard $U_H$. Retain the full
classical record and define $\sigma_{z,s}=(-1)^{z_x+z_y+s}$.
The corrected Kraus operator from input $a$ to output $b$ is
\begin{equation}
L_{z,s}(\epsilon)=\frac1{\sqrt{32}}
 e^{-i\sigma_{z,s}\epsilon X}U_H.
\label{SM-eq:mbqckraus}
\end{equation}
To see this, the four $Z$ measurements leave
$Z_a^{r_a}Z_b^{r_b}{\rm CZ}_{ab}\ket{\varphi,+}/4$ before the error.
The error commutes with those measurements, and conjugating it by the
Pauli corrections changes its sign to $(-1)^{z_x+z_y}$.
The subsequent $X$ measurement replaces $X_a$ by $(-1)^s$, giving
Eq.~\eqref{SM-eq:mbqckraus}. Every one of the 32 branches has probability
$1/32$, for every input and every $\epsilon$; no branch is postselected.

For calibration set $\ket\varphi=\ket+$, so the ideal output is $\ket0$.
Measure $Y_b$ on one ensemble and $Z_b$ on another, with outcomes
$m=\pm1$. The full joint probabilities are
\begin{align}
p_Y(z,s,m)&=\frac{1-m\sigma_{z,s}\sin(2\epsilon)}{64},\nonumber\\
p_Z(z,s,m)&=\frac{1+m\cos(2\epsilon)}{64}.
\label{SM-eq:mbqcstatistics}
\end{align}
Consequently the signed signal and companion quadrature are
\begin{equation}
\begin{gathered}
S=-\mathbb E[\sigma_{z,s}m_Y]=\sin(2\epsilon),\\
C=\mathbb E[m_Z]=\cos(2\epsilon).
\end{gathered}
\label{SM-eq:diamondsignals}
\end{equation}
Discarding the record before forming the signed signal gives
$\mathbb E[m_Y]=0$ and loses its error-sign information.

\subsection{The observable and information in the measurement record}
The signed signal in Eq.~\eqref{SM-eq:diamondsignals} equals the resource
expectation of $M=-Y_bZ_cZ_dZ_x$ in Eq.~\eqref{SM-eq:mbqcerror}.
Indeed the $s$ sign cancels the conjugation of $Y_b$ by $X_b^s$;
the remaining $Z_b^{r_b}$ correction gives the factor
$-(-1)^{z_c+z_d+z_x}$. Summing all records produces $M$.
The record-resolved $Y$ experiment has
\begin{equation}
\mathcal I_Y(\epsilon)=
\sum_{z,s,m}\frac{(\partial_\epsilon p_Y)^2}{p_Y}=4,
\qquad |\epsilon|<\pi/4.
\label{SM-eq:diamondfisher}
\end{equation}
This equals the pure-resource quantum Fisher information
$4\Var_{G_\diamond}(X_aX_b)=4$.

\subsection{One physical model parameter and calibrated control}
A minimal repair for this specified error is to append its physical gate:
\begin{equation}
\ket{\Psi_{\rm aug}(A,\eta)}
=e^{-i\eta X_aX_b}\ket{\Psi(A)}.
\label{SM-eq:diamondaugment}
\end{equation}
At $(A_0,0)$ the added derivative is $t$, orthogonal to the entire old
tangent. It supplies the one direction required for this error and
reproduces its finite evolution exactly at $A=A_0$; the same gate also
models arbitrary computational inputs. It does not supply all 16 missing
directions or eliminate the old parameter kernel. The identity
$\ket{G_{\diamond,\epsilon}}=e^{i\epsilon Q}\ket{G_\diamond}$ is useful at the calibration
state, but is not assumed for arbitrary inputs. The physical two-qubit
gate in Eq.~\eqref{SM-eq:diamondaugment} avoids that restriction.

Estimate the error from the measured quadratures using
\begin{equation}
\widehat\epsilon=\tfrac12\operatorname{atan2}(S,C).
\label{SM-eq:mbqcestimate}
\end{equation}
For subsequent computational runs, after the specified Pauli corrections,
apply $R_x(-2\sigma_{z,s}\widehat\epsilon)$ to the output, where
$R_x(\theta)=e^{-i\theta X/2}$. Equation~\eqref{SM-eq:mbqckraus} then holds
with $\epsilon$ replaced by $\Delta=\epsilon-\widehat\epsilon$.
Averaging over the retained branches gives
\begin{equation}
\begin{gathered}
\mathcal E_\Delta(\rho)=\cos^2\Delta\,\rho_H+
\sin^2\Delta\,X\rho_H X,\\
\rho_H=U_H\rho U_H^\dagger.
\end{gathered}
\end{equation}
Using the average gate fidelity~\cite{nielsen2002average},
\begin{equation}
1-F_{\rm av}(\mathcal E_\Delta,U_H)
=\frac23\sin^2\Delta=\frac23\Delta^2+O(\Delta^4).
\label{SM-eq:mbqccompfidelity}
\end{equation}
Without compensation $\Delta=\epsilon$. Exact calibration restores the
ideal gate within this error model. For $\epsilon=0.1$ the uncompensated
infidelity is $6.644\times10^{-3}$, whereas a residual
$|\Delta|=0.01$ gives $6.666\times10^{-5}$
[Fig.~\ref{SM-fig:mbqcprotocol}(b)]. A per-primitive infidelity budget $f\le2/3$ is
met on the small-residual branch if
$|\Delta|\le\arcsin\sqrt{3f/2}$.

The concrete prescription is therefore to include the missing physical
error parameter in the model, retain the measurement record, estimate
its signed quadratures, and implement the conditional output rotation.
Only the modeling repair changes the tangent image; ordinary gauge and
non-bond-kernel removal address parameter redundancy. Loss, stochastic
noise, readout bias, other coherent terms, and imperfect feedforward
require additional modeling. The stated fidelities concern this single
primitive and do not establish a fault-tolerance threshold.

\section{Triangle identifiability: coordinates and synthetic reconstruction}
\label{SM-sec:triangleexperiment}
The three-qubit triangle has a full physical tangent together with
five non-bond complex null directions. We give its exact counts,
finite constant-state family, and synthetic reconstruction below.
Figure~\ref{SM-fig:kernelimage} contrasts it with the missing physical
direction of Corollary~\ref{cor:observableresponse}.

\subsection{Redundant parameters with a full physical tangent}
\label{SM-sec:triangleidentifiability}
For $G_\triangle=K_3$, every closed neighborhood is the whole graph,
so Theorem~\ref{thm:star} gives $\rank J=8$ and
$\dim\T_A=7$. With $P=24$ and $g_0=11$,
\begin{equation}
\nu=24-8-11=5=P-g_0-q^N.
\label{SM-eq:trianglecounts}
\end{equation}
This saturates the dimension-counting lower bound of
Proposition~\ref{SM-prop:genericbounds}; the five null modes are not a
WGS rank drop. The implicit-function construction below integrates
them into a local family transverse to ordinary gauge and the ray.
Complete kernel removal selects a representative while leaving the
physical predictions of a pseudoinverse fit unchanged. Singular QGT
methods likewise regularize or invert the metric on its range without
enlarging the physical image
\cite{stokes2020qng,chen2024minsr,rende2024minsr}.
The finite-shot records below are synthetic binomial samples.

\subsection{Five residual complex zero modes and no missing tangent}
Let $G_\triangle=K_3$ and
$\ket{G_\triangle}={\rm CZ}_{ab}{\rm CZ}_{ac}{\rm CZ}_{bc}\ket{+++}$.
Use virtual bits $\alpha,\beta,\gamma$ on edges $ab,ac,bc$, respectively,
and the explicit tensors
\begin{align}
 A_a^s(\alpha,\beta)&=2^{-1/2}\delta_{s,\alpha}\delta_{s,\beta},\nonumber\\
 A_b^t(\alpha,\gamma)&=2^{-1/2}(-1)^{t\alpha}\delta_{t,\gamma},\nonumber\\
 A_c^u(\beta,\gamma)&=2^{-1/2}(-1)^{u(\beta+\gamma)}.
\label{SM-eq:triangletensors}
\end{align}
All 24 complex entries, including initially zero entries, are independently
variable. The ranks and non-bond nullity are those of
Eq.~\eqref{SM-eq:trianglecounts}; the physical tangent is full.

The Euclidean tensor-entry metric fixes a reproducible coordinate
convention. Choose orthonormal columns $(C,N)$ spanning the complement
of ordinary gauge and the state ray, with $C$ spanning its seven physical
complex directions and $N$ its five non-bond-kernel directions. Direct
contraction of Eq.~\eqref{SM-eq:triangletensors} gives
$4JJ^\dagger=7I_8$. Since gauge and ray directions are annihilated by
$J_\perp$, the QGT in this basis is
\begin{equation}
\bigl[J_\perp(C,N)\bigr]^\dagger J_\perp(C,N)
=\operatorname{diag}\bigl(\tfrac74 I_7,0_5\bigr).
\label{SM-eq:triangleqgt}
\end{equation}
Removing the complete kernel retains the seven positive modes
[Fig.~\ref{SM-fig:kernelimage}(b)]. Their common value depends
on this tensor normalization; the nullity and physical rank do not.
An integer certificate for $2J$ verifies the stated identity exactly.

\subsection{Finite parameter ambiguity at the same physical state}
Write a local affine slice in parameter space as
\begin{equation}
 A(\xi,\zeta)=A_0+C\xi+N\zeta,\qquad
 \xi\in\mathbb C^7,\quad\zeta\in\mathbb C^5.
\end{equation}
The equation
\begin{equation}
 (I-\ket{G_\triangle}\bra{G_\triangle})\Psi(A(\xi,\zeta))=0
\label{SM-eq:trianglefiber}
\end{equation}
has seven independent complex components. Its derivative with respect
to $\xi$ at zero is invertible onto the physical tangent, while its
$\zeta$ derivative is zero. The implicit-function theorem therefore gives
$\xi=\xi(\zeta)=O(\|\zeta\|^2)$ for all sufficiently small $\zeta$.
This is a five-complex-dimensional local family of parameters representing
exactly the same normalized state up to phase. Its initial tangents $N\zeta$
are transverse to ordinary gauge and the state ray. Thus the ambiguity
is not merely an absence of first-order measurement sensitivity.

We verify this family along a fixed, seeded complex unit vector
$w\in\mathbb C^5$, solving Eq.~\eqref{SM-eq:trianglefiber} at 17 values
$\zeta=uw$, $-0.08\le u\le0.08$. The largest aligned normalized-state
distance is below $7\times10^{-15}$, and the largest change in any of the
63 readout-adjusted Pauli means is below $9\times10^{-15}$.
At $u=0.08$ the parameter displacement is $0.0800011$, with physical
coordinate compensation $\|\xi\|=4.20\times10^{-4}$.
These values use the normalization in Eq.~\eqref{SM-eq:triangletensors}.

\begin{figure}[!tbp]
\centering
\includegraphics[width=\columnwidth]{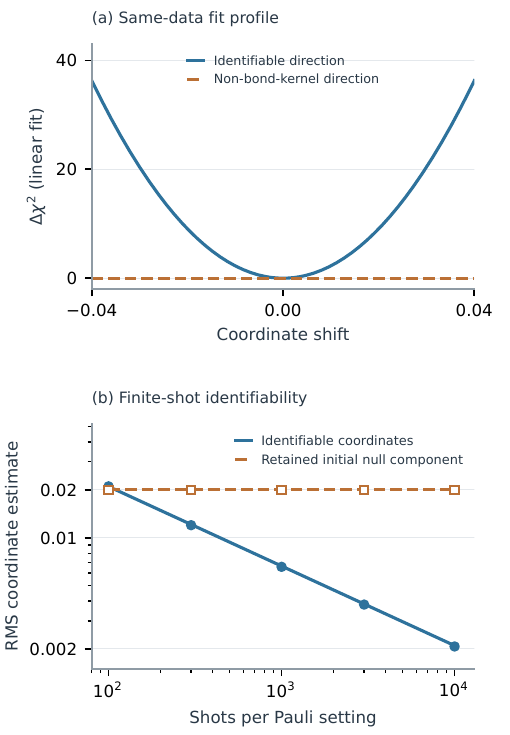}
\caption{Synthetic local reconstructions on the three-qubit triangle.
(a) The local least-squares profile from the same finite-shot record
($n=1000$ per Pauli word) is flat along a non-bond-kernel direction and
curved along an identifiable direction; $\Delta\chi^2$ includes the
factor $n$. (b) Across 512 synthetic reconstructions per shot count,
identifiable-coordinate RMS follows $1/\sqrt{28\kappa^2n}$, while the
retained initial null component has RMS 0.02. This plateau is an imposed
initialization scale, not a physical uncertainty or a statistical error floor.
Here $\kappa=0.9$, and total copies are $63n$. Both parameterizations
produce identical linear measurement predictions.}
\label{SM-fig:triangleidentifiability}
\end{figure}

\subsection{Complete Pauli measurements and finite-shot fitting}
Measure each of the 63 nonidentity three-qubit Pauli words $P_k$ on
$n$ independent copies. For a strictly positive likelihood at the
stabilizer state, independently flip each reported Pauli outcome with
known probability $0.05$, corresponding to contrast $\kappa=0.9$. The measured
mean and outcome probability are
\begin{equation}
\begin{gathered}
 \mu_k=\kappa\langle G_\triangle|P_k|G_\triangle\rangle,\\
 p_k(m)=\frac{1+m\mu_k}{2},\qquad m=\pm1.
\end{gathered}
\label{SM-eq:trianglemeasurement}
\end{equation}
There are $63n$ resource copies per reconstruction. The known contrast
preserves informational completeness and is included in the fit.

Use real coordinates $x=(\operatorname{Re}\xi,\operatorname{Im}\xi)$
and $z=(\operatorname{Re}\zeta,\operatorname{Im}\zeta)$, of dimensions
14 and 10. Let $\overline m_k$ be the sampled mean, and define the
whitened local linear fit by
\begin{equation}
\begin{gathered}
 y_k=\frac{\overline m_k-\mu_k}{\sqrt{1-\mu_k^2}},\qquad
 R_{ki}=\frac{\partial_i\mu_k}{\sqrt{1-\mu_k^2}},\\
 y\simeq R(x,z).
\end{gathered}
\label{SM-eq:trianglefit}
\end{equation}
The non-bond-kernel columns vanish. Complete Pauli orthogonality and
Eq.~\eqref{SM-eq:triangleqgt} give the stronger identity
\begin{equation}
\begin{gathered}
 R=(D,0),\qquad D^{\mathsf T}D=28\kappa^2I_{14},\\
 \mathcal I_n=nR^{\mathsf T}R.
\end{gathered}
\label{SM-eq:triangleinformation}
\end{equation}
To obtain the factor 28, the Pauli response obeys
$\sum_k(\delta\mu_k)^2=16\kappa^2\|\delta\psi\|^2$.
The seven nonzero Pauli expectations belong to stabilizers and have zero
tangent response, so whitening does not change this quadratic form.
The norm on the physical coordinates is $\|\delta\psi\|^2=7\|x\|^2/4$.
Consequently the measurement information has ten real zero modes after
ordinary gauge removal, and none after complete kernel removal.

For a common data record, least squares gives
\begin{equation}
 \widehat x=D^+y,\qquad \widehat z=z_0,
 \label{SM-eq:triangleestimator}
\end{equation}
where $z_0$ is the initial null component for ordinary Euclidean gradient
descent. The fully reduced fit gives the identical $\widehat x$ and the
identical predicted measurements. The loss is exactly flat in $z$ in
this local model [Fig.~\ref{SM-fig:triangleidentifiability}(a)]. Selecting a
minimum-norm pseudoinverse with zero initialization also sets $z_0=0$;
this is already a choice of representative. A quadratic prior can choose
a different null component, but the data supply no curvature in it.

At the reference state, independent binomial sampling gives the exact
covariance of this linear estimator,
\begin{equation}
 \operatorname{Cov}(\widehat x)=\frac{I_{14}}{28\kappa^2 n}.
 \label{SM-eq:trianglesampling}
\end{equation}
We generated 512 independent records at each of
$n=100,300,1000,3000,10000$, with fixed seed 61027. The initial null
coordinates were drawn independently with standard deviation 0.02 and
held fixed across the shot-count comparison. The RMS identifiable
coordinate estimate decreases from 0.02104 to 0.002074 as $n$ increases
by a factor of 100, consistent with the predicted values 0.02100 and
0.002100. The RMS null coordinate remains 0.02002
[Fig.~\ref{SM-fig:triangleidentifiability}(b)]. The two parameterizations'
linear measurement predictions agree to $2.1\times10^{-16}$.
The null spread measures dependence on initialization relative to the
chosen reference representation; it is not a physical uncertainty or
shot-noise amplification. These are local linear reconstructions, not
claims about a global nonlinear maximum-likelihood algorithm.

\section{Proof certificates and reproducibility}
\label{SM-sec:reproducibility}
\paragraph{What is proved analytically and what is machine-verified.}
The WGS image, kernel, and motif classification (Secs.~\ref{sec:star}--\ref{sec:periodic}), the injectivity lemma, the support criterion, and the
closed-form residuals follow from the analytic arguments in the main text
and Appendix~\ref{SM-sec:xxscan}.
The generic-completeness results of Theorem~\ref{thm:generic} combine
analytic all-$q$ witnesses with finite combinatorial certificates:
routing certificates that exhibit, for each required region, a path
system proving a nonzero injectivity minor for every integer $q\ge2$,
and, for the six qubit-only families, exact integer minors. Together
these establish all thirty-nine families under the stated period
conditions, with thirty-three valid for every $q\ge2$. These
certificates were produced with computer assistance and are verified by
independent checkers that read only the frozen certificate files and
re-derive every input, path, capacity, endpoint, translation type, and
exterior cover (Appendix~\ref{SM-sec:certificates}). Language-model
assistance was used in drafting proofs and text; every analytic step
was subsequently checked by the author, and no step of any proof relies
on an unverified model output. Floating-point rank samples and
dynamical benchmarks, including the trajectories of Fig.~\ref{fig:dynamics} (\nolinkurl{manuscript/quantum/data/gauge_dynamics},
\nolinkurl{experiments/gauge_dynamics.py}), are used only as
illustrations, never as proofs.

\paragraph{Independent checks.}
The common entry point of the frozen review bundle is
\nolinkurl{reproducibility/README.md}; its manifest and archive
checksum identify the reviewed source and data, and
\texttt{python reproducibility/reproduce.py --suite all} replays all
finite checks in a temporary workspace. The \texttt{core} suite checks
the WGS differential and physical consequences, \texttt{routes} the
accepted partitions, routing, covers, and geometry, and \texttt{minors}
the exact qubit minors. In addition: the cell counts of Theorem~\ref{thm:cell} are compared with independently enumerated
support unions on two finite tori for each of 28 periodic graphs at
$q=2,3,4,5,8$ (280 exact comparisons); the minimal-region polynomial
of Theorem~\ref{thm:minimalgeneral} is checked against sector quotients
on all 996 connected simple graphs with at most seven vertices; the
triangular-torus certificates of Sec.~\ref{sec:generic} use exact
modular rank computations at random integer points
(\nolinkurl{experiments/triangular_generic_rank.py}: $4\times3$ and
$4\times4$ tori at $q=2$, prime $8191$, ranks $1417$ and $1889$); the
$XX$ residuals of Eq.~\eqref{eq:xxphi} were obtained by exact
symbolic Walsh transforms (\nolinkurl{experiments/xx_residual_phase_symbolic.py})
and confirmed numerically; the second-order obstruction of Proposition~\ref{prop:obstruction} is certified by exact rational
elimination and the unobstructed cases are checked in floating point
(\nolinkurl{experiments/second_order_obstruction.py}); the exact
equivalences of Proposition~\ref{prop:finite} are verified to machine
precision at $q=2,3$, the straight lines of Eq.~\eqref{eq:squareline}
are checked against $(1-t^2)$, the Gauss--Newton continuations of
Appendix~\ref{SM-sec:obstructionchecks} are reproduced, and the $YY$ supports of
Eq.~\eqref{eq:yysupport} are checked to lie outside $\K$ for every
edge orbit of the eleven tilings
(\nolinkurl{experiments/nongauge_integrability.py},
\nolinkurl{tests/unit/test_nongauge_integrability.py}); the six-qubit and
triangle examples are checked by dense contraction, full Jacobians, and
direct matrix evolution; and the constructive response, local Pauli
evaluation, bond expansion, and local-frame spectrum have unit tests
comparing tensor derivatives with the analytic formulas. Plots are
generated from the same connectivity specifications.

\paragraph{Nearby spectra and the state-level obstruction.}
The perturbation result in Appendix~\ref{SM-sec:nearbygramproof} and the
matching-cut lemma in Appendix~\ref{SM-sec:schmidtobstruction} have analytic
proofs. Their numerical and finite combinatorial checks are stored in
\nolinkurl{manuscript/quantum/data/wgs_nearby_spectrum},
\nolinkurl{manuscript/quantum/data/wgs_physical_paths}, and
\nolinkurl{manuscript/quantum/data/xx_schmidt_obstruction}, with source hashes.
The replay report records runtime versions. The separate entry point
\texttt{python }\nolinkurl{reproducibility/reproduce_m1_m2.py} reruns their producers
and focused tests without altering the earlier frozen review bundle.
The finite scans support the stated examples; the general claims
follow from the proofs, and no Lean formalization is asserted here.

\paragraph{Force compatibility and imposed right-hand-side error.}
The dense experiment accompanying Fig.~\ref{fig:projectiongap} is reproduced with
\texttt{PYTHONPATH=.:python python -m}
\nolinkurl{research.wgs_gauge_projection_gap.condition_number_impact}
from the repository root. It uses $H=\sum_vZ_v$, checks its exact
physical velocity against dense least squares, and separately inserts
a seeded error of norm $10^{-6}\|\boldsymbol F\|$ in the residual non-bond kernel.
It also verifies the compatibility of a force and Gram matrix formed
from the same $64$ centered derivative samples. The checks in
\nolinkurl{tests/unit/test_projection_gap_compatibility.py} compare
ordinary and complete kernel reductions, the pseudoinverse limit,
the $1/\lambda$ amplification of the imposed error, and its vanishing
instantaneous physical image. These are fixed-representation tests,
not finite-time integration benchmarks or a sampling-noise model.

\end{document}